\documentclass[a4paper]{article}

\usepackage[left=3cm,right=3cm]{geometry}

\usepackage{booktabs}   
\usepackage{subcaption} 

\usepackage{amsmath}
\usepackage{amsthm}
\usepackage{thmtools} 
\usepackage{amsfonts}
\usepackage{amssymb}
\usepackage{mathtools}
\usepackage{bm}

\usepackage{tikz}
\usetikzlibrary{matrix}
\usetikzlibrary{cd}
\usepackage{quiver}

\usepackage{stmaryrd}
\usepackage{ebproof}
\usepackage{cmll}

\usepackage{varioref}
\usepackage{hyperref}
\usepackage[nameinlink, capitalize]{cleveref}

\definecolor{blue(pigment)}{rgb}{0.2, 0.2, 0.6}
\definecolor{darkgreen}{rgb}{0.0, 0.5, 0.0}
\definecolor{darkred}{HTML}{4A0005}

\hypersetup
{
     colorlinks=true,
     linkcolor=darkred, 
     filecolor=magenta,            
     urlcolor=blue(pigment),
     citecolor=darkgreen
}

\usepackage[colorinlistoftodos]{todonotes}

\newcommand\Tnoted[1]{\todo[color=pink,inline]{T:\the\day/\the\month/\the\year: #1}}

\title{Coherent differentiation with a monoid of degrees: \\
derivatives, Taylor expansion, finite differences\dots}

\author{Thomas Ehrhard\\
  Université Paris Cité, CNRS, Inria, IRIF, F-75013, Paris, France\\
  Aymeric Walch\\ Fédération ENAC ISAE-SUPAERO ONERA, Université de Toulouse, France.}

\newcommand\CMLLPAR{
\usepackage{cmll}
\newcommand\IPar{\mathord{\parr}}
}

\CMLLPAR
\input{mainnot}

\begin{document}

\maketitle

\begin{abstract}
Coherent differentiation has been introduced a few years ago as a
refinement of the (categorical semantics of the) differential
lambda-calculus and of differential linear logic, compatible with the
determinism of computation, and hence in categories where addition of
morphisms in the hom-sets is only partially defined.
In this setting differentiation is implemented by a kind of ``tangent
functor'', with the difference that it maps an object to its object of
summable tuples rather to a tangent bundle.
In our earlier work, these tuples were pairs in the case of simple
differentiation, and families indexed by integers when dealing with
Taylor expansion and the Faà~di~Bruno formula.

We push this idea further, considering now tuples indexed by an
arbitrary commutative partial monoid satisfying a nonnegativity
condition.
The monoid structure on these indices plays a crucial role in the
definition of the ``differential'' calculus associated with this
``tangent functor''.
In that way, and without changing the general categorical
infrastructure of coherent differentiation, this theory accommodates
new notions of ``differentiation'' of which we describe several
instances.
Among many others, such notions include various forms of finite
differences as well as Taylor expansion truncated at any fixed degree,
and finite Taylor expansion of a fixed finite degree with a remainder.


\end{abstract}


\section*{Introduction}

The differential lambda-calculus and differential linear logic
\cite{EhrhardRegnier02,EhrhardRegnier06d}, as well as the description
of their associated categorical structures
\cite{BluteCockettSeely06,Fiore07,BluteCockettSeely09,Manzonetto12,Ehrhard15c,BluteCockettLemaySeely20},
arose from the observation that in several categorical models of
linear logic (\(\LL\)), non-linear morphisms admit derivatives because
they can be understood as some kind of analytic functions: this is the
case in the relational model, in Köthe spaces \cite{Ehrhard02} and
finiteness spaces \cite{Ehrhard00b}.
As in Calculus, the description of these derivatives makes a strong use
of addition: in the differential lambda-calculus there are reduction
rules similar in spirit to the Leibniz Law \((uv)'=u'v+uv'\).
As far as categorical models are concerned such summations correspond
to additions of morphisms in hom-sets and therefore categorical models
of the differential calculus (and of differential \(\LL\)) are
required to be additive categories, meaning that they are enriched in
commutative monoids.
From a computational point of view, this is a strong assumption which
means for instance that, in the ``object of booleans'' of our models,
the superposition \(\True+\False\) of the two boolean values must
exist, which amounts to saying that our models must feature a strong
form of non-determinism.

From this point of view the differential lambda-calculus (and
differential \(\LL\)) are unsatisfactory approaches to the
differentiation of programs/proofs as we do not see \emph{a priori}
why such an extension of the standard functional paradigm should
entail the non-determinism of computations.
This consideration is enforced by the observation that there are
perfectly well-behaved models of \(\LL\), such as probabilistic
coherence spaces, where non-linear morphisms are clearly analytic
functions for which derivatives can be defined.
Concretely, the issue is that, if we define these derivatives naively,
there is no way to make them morphisms in the considered category.
We give a concrete example of this situation.
In probabilistic coherence spaces, the object \(\Sone\) (unit of the
tensor product which is also the standard interpretation of the
\(\mathtt{unit}\) data-type) is the interval
\(\Intcc 01\subseteq\Real\), whose elements should be understood as
the probability of a closed program of type \(\mathtt{unit}\) to
terminate.
A non-linear morphism \(\Sone\to\Sone\) is a function
\(f:\Intcc01\to\Intcc01\) which can be written as a power series with
non-negative coefficients: there is a family
\((a_n\in\Realp)_{n\in\Nat}\) such that
\(\forall x\in\Intcc01\ f(x)=\sum_{n=0}^\infty a_nx^n\).
Clearly, a family \((a_n\in\Realp)_{n\in\Nat}\) defines such a
function \(f:\Intcc01\to\Intcc01\) iff \(\sum_{n\in\Nat}a_n\leq1\).
Then the derivative \(f'(x)=\sum_{n=0}^\infty (n+1)a_{n+1}x^n\) is a
monotonic function \(\Interco01\to\Realp\) which is not necessarily
bounded and \emph{a fortiori} does not extend in general to a morphism
\(\Sone\to\Sone\) in our non-linear category.
Take for instance \(f(x)=1-\sqrt{1-x}\) for which the non-negative
\(a_n\)'s are easy to compute using the usual series expansion of
\(\sqrt{1+y}\) at \(0\), and for which we have
\(f'(x)=\frac 1{2\sqrt{1-x}}\).
This example is particularly relevant because it is the
interpretation of the following recursively defined probabilistic
(call-by-name) program \(F:\mathtt{unit}\to\mathtt{unit}\):
\begin{align*}
  \mathtt{let\ rec}\  F\ x
  =\mathtt{if}\ \mathtt{bernouilli}(0.5)\ 
  \mathtt{then}\ x\ \mathtt{else}\ (F\ x);(F\ x)
\end{align*}
so it cannot be discarded as pathological.

\paragraph*{Coherent differentiation.}
To escape this aporia, we cannot change the definition of derivatives,
but we can put some restrictions on the use we make of them.
For this we are guided by the \emph{general} definition of the
derivative of a (sufficiently regular) function \(f:E\to F\) where
\(E\) and \(F\) are finite dimensional \(\Real\)-vector spaces: it is
the unique map \(E\to\cL(E,F)\) (where \(\cL(E,F)\) is the
\(\Real\)-vector space of linear maps from \(E\) to \(F\)) such that
\begin{align*}
  f(x+u)=f(x)+f'(x)\cdot u+o(\Norm u)
\end{align*}
where \(\Norm{\_}\) is any of the (all equivalent) norms which define
the usual topology of \(E\) and \(\cdot\) stand for the application of
a linear map to its argument.
In the case \(E=\Real\) the usual convention is to use implicitly the
canonical isomorphism \(\cL(\Real,F)\Isom F\) mapping \(\phi\) to
\(\phi\cdot 1\) and consider that \(f'(x)\in F\), so that the formula
above reads \(f(x+u)=f(x)+uf'(x)+o(|u|)\).

Now, coming back to our analytic maps \(f:\Intercc01\to\Intercc01\)
(with non-negative coefficients) the crucial observation is that our
source and target spaces \(\Intercc01\) are not vector spaces, but we
can consider \(\Intercc01\) as a \emph{partial} monoid: given
\(x,u\in\Intercc 01\), \(x+u\) is defined as usual if the
value of this sum is \(\leq 1\) and is undefined otherwise.
Moreover, when \(x+u\leq 1\), we have \(f(x)+uf'(x)\leq f(x+u)\leq 1\)
because all the coefficients \(a_n\) defining \(f\) are 
nonnegative.
Setting
\(\Sfun\Intercc01=\{(x,u)\in\Intercc01^2\St x+u\in\Intercc01\}\)
(which can easily be seen as a probabilistic coherence space) we can
therefore define a map \(\Dfun f:\Sfun\Intercc01\to\Sfun\Intercc01\)
by \(\Dfun f(x,u)=(f(x),uf'(x))\) and this operation on functions,
formally similar to the tangent bundle functor of differential
geometry, is well defined and functorial.

This simple observation generalizes easily to all probabilistic
coherence spaces and analytic functions between them%
\footnote{See \cref{ex:pcoh-cart-smc-elem-sum} for a survey on
  probabilistic coherence spaces.}%
: given an analytic morphism \(t:X\to Y\) and \(x,u\in\Pcoh X\) we can
always define \(t'(x)\cdot u\) but we can only say that it is an
element of \(\Realpcto{\Web Y}\) taking possibly \(\infty\) values in
some coordinates.
However if we assume moreover that \(x+u\in\Pcoh X\) then it is
granted that \(t(x)+t'(x)\cdot u\in\Pcoh Y\) and we can again define
an analytic \(\Dfun t:\Sfun X\to\Sfun Y\) where \(\Sfun X\) is a
probabilistic coherence space such that, up to a simple isomorphism,
\(\Pcoh{\Sfun X}=\Eset{(x,u)\in\Pcoh X^2\St x+u\in\Pcoh X}\) and then,
up to this iso, we have
\(\Dfun t(x,u)=(t(x),t'(x)\cdot u)\in\Pcoh{\Sfun Y}\).

Analyzing further the situation, one of the authors 
observed in \cite{Ehrhard23a} that this operation
\(X\mapsto\Sfun X\) on probabilistic coherence spaces is actually a
functor \(\Sfun:\PCOH\to\PCOH\) (where \(\PCOH\) is the category of
probabilistic coherence spaces and \emph{linear} morphisms which are
actually matrices generalizing sub-stochastic matrices).
He observed also that, in almost all known cases, this functor is
\emph{representable}: \(\Sfun\_=\Limplp\Dbimongen\_\) where
\(\Dbimongen=\With\Sone\Sone\) (we use \(\Sone\) for the unit of the
tensor product, \(\with\) for the categorical product and
\(\Limpl\_\_\) for the internal hom functor associated with the fact
that \(\PCOH\) is a symmetric monoidal closed category).
It is useful to be a bit more pedantic here:
\(\Dbimongen=\Bwith_{i\in\Mdeg}\Sone\) where the elements of
\(\Mdeg=\{0,1\}\) are the two indices we use to specify the two
components of the summable pairs \((x,u)=(x_0,x_1)\in\Pcoh{\Sfun X}\).
It is important to notice that there are jointly epic canonical
injections \((\Win i\in\PCOH(\Sone,\Mdeg))_{i\in\{0,1\}}\) mapping
\(1\in\Pcoh\Sone=\Intercc01\) to \((1,0)\) and \((0,1)\) respectively
(both pairs are elements of \(\Pcoh\Dbimongen\Isom\Intercc01^2\)).
Then one essential feature of the functor \(\Sfun\), and hence of the
object \(\Dbimongen\), is that they completely describe the notion of
summability in the homsets \(\PCOH(X,Y)\): two morphisms
\(t_0,t_1\in\PCOH(X,Y)\) are summable iff there is a morphism
\(t\in\PCOH(\Tens X\Dbimongen,Y)\) such that
\(t_i=t\Compl\Tensp X{\Win i}\) for \(i=0,1\), keeping again implicit
the canonical iso between \(X\) and \(\Tens X\Sone\).
And actually \(t_0+t_1=t\Compl\Tensp X\Delta\) where
\(\Delta\in\PCOH(\Sone,\Dbimongen)\) is the diagonal morphism.
More explicitly, this means that \(t_0\) and \(t_1\) are summable iff
the function \(s:\Intercc01^2\to\PCOH(X,Y)\) mapping
\((\lambda_0,\lambda_1)\) to \(\lambda_0t_0+\lambda_1t_1\) is
well-defined, and then of course \(t_0+t_1=s(1,1)\).

To understand the differentiation operator \(\Dfun\), we must recall
that, as a model of \(\LL\), the category \(\PCOH\) is equipped with a
resource comonad \(\oc\) and that the category of non-linear morphisms
(here, the analytic functions between probabilistic coherence space) is
the Kleisli category \(\Kl\PCOH\) of this comonad: the analytic maps
from \(X\) to \(Y\) can be seen as elements of
\(\Kl\PCOH(X,Y)=\PCOH(\Oc X,Y)\).
Our differentiation functor \(\Dfun\) can therefore be understood as a
functor \(\Kl\PCOH\to\Kl\PCOH\), which extends the functor
\(\Sfun:\PCOH\to\PCOH\) because the derivative of a linear map is the
map itself.
It is well known that such extensions are in bijective correspondence
with the distributive laws \(\Oc\Comp\Sfun \Natto \Sfun\Comp\Oc\) (see
\cref{sec:distr-law-comonad}) and, because of the representability of
the functor \(\Sfun\), as explained in \cref{sec:coalg-to-dl}, each
\(\oc\)-coalgebra structures on \(\Dbimongen\) induces such a
distributive law.
So, at the end of the day, our coherent differentiation infrastructure
boils down to a morphism \(\Ddiff\in\PCOH(\Dbimongen,\Oc\Dbimongen)\)
satisfying two commutations expressing that \((\Dbimongen,\Ddiff)\) is a
\(\oc\)-coalgebra, that is, an object of the Eilenberg-Moore category
\(\Em\PCOH\).
The fact that the functor \(\Dfun\) induced in that way by this
coalgebra structure \(\Ddiff\) actually computes the derivatives
manifests itself by the specific commutative \(\ITens\)-comonoid
structure it induces on \(\Dbimongen\), as any \(\oc\)-coalgebra
structure does on its carrier object in a model of \(\LL\).
This ``differential'' comonoid structure consists of a counit which is
the first projection \(\Proj0\in\PCOH(\Dbimongen,\Sone)\) and of a
comultiplication
\(\Dbimonm\in\PCOH(\Dbimongen,\Tens\Dbimongen\Dbimongen)\) that is
determined by the two following equations
\begin{align*}
  \Dbimonm\Compl\Win0&=\Tens{\Win0}{\Win0}\\
  \Dbimonm\Compl\Win1&=\Tens{\Win1}{\Win0}+\Tens{\Win0}{\Win1}
\end{align*}
keeping implicit the unique canonical iso between \(\Sone\) and
\(\Tens\Sone\Sone\).
It is essential to observe here that the addition used in the second
equation is well-defined in
\(\PCOH(\Sone,\Tens\Dbimongen\Dbimongen)
\Isom\Pcohp{\Tens\Dbimongen\Dbimongen}\).
These constructions can be generalized to all models of \(\LL\) where
hom-sets have a partially defined addition which can be characterized
by means of the representable \(\Sfun\) functor, and it turns out that
there are a lot of such models (most of the ones we know actually).
A noticeable example is Girard's category of coherence spaces with the
\emph{free} exponential whose definition uses finite multicliques.
A typical non-example is
the category of coherence spaces again, but with the original Girard's
exponential based on finite cliques (whose associated Kleisli category
is isomorphic to the category of coherence spaces and stable
functions).

In \cite{EhrhardWalch25} we observed that, by taking now
\(\Dbimongen=\Bwith_{i\in\Nat}\Sone\) (for which we have jointly epic
injections \((\Win i\in\PCOH(\Sone,\Dbimongen))_{i\in\Nat}\), with
\(\Win i\) mapping \(1\) to
\[(\overbrace{0,\dots,0}^{{i}\times},1,0,\dots)
\in\Pcoh\Dbimongen\Isom\Intercc01^\Nat\,\]
and choosing a \(\oc\)-coalgebra structure on \(\Dbimongen\) whose
associated comonoid has the first projection \(\Proj0\) as counit and
a comultiplication characterized by
\begin{align*}
  \forall i\in\Nat\quad
  \Dbimonm\Compl\Win i=\sum_{l+r=i}\Tens{\Win l}{\Win r}
\end{align*}
we get a completely similar theory where the associated
``differentiation'' functor \(\Dfun\) computes the Faà~di~Bruno
formula which has to be understood as a functorial generalization of
the usual Taylor expansion.

As in the ``binary'' case above of differentiation, the functor
\(\Sfun\) (and hence the object \(\Dbimongen\)) completely describes the
summability structure of the hom-sets \(\PCOH(X,Y)\): a countable
family \((t_i\in\PCOH(X,Y))_{i\in\Nat}\) is summable iff there is a
\(t\in\PCOH(\Tens X\Dbimongen, Y)\) such that
\(\forall i\in\Nat\ t_i=t\Compl\Tensp X{\Win i}\), and then
\(\sum_{i\in \Nat}t_i=t\Compl\Tensp X\Delta\).
As in the previous theory of differentiation, this construction can be
done in most models of \(\LL\) we know.

\paragraph*{Contribution of the paper.}
In the present paper we generalize these two strikingly similar
constructions.
The main parameter of this generalization is a commutative partial
monoid%
\footnote{See \cref{def:pcm,sec:summability-structure} for precise
  definitions.}%
\(\Mdeg\) having at least \(2\) elements.
It must also satisfy a \Mpositivity{} requirement
(\(i+_\Mdeg j=0\Implies i=0\)) whose purpose will become clear in our
developments.
The elements of \(\Mdeg\) should be considered as \emph{degrees} in the
sense of differentiation.

Then a model of coherent \(\Mdeg\)-differentiation is a model \(\cL\)
of \(\LL\) whose hom-sets have a commutative partial monoid structure,
and where this structure can be characterized by means of an endofunctor
$\Sfun$ equipped with projections (see \cref{sec:summability-structure}).
In most cases, this functor is \emph{representable} in the same way as 
described in our previous examples:
$\Sfun = \Limpl{\Dbimongen}{\_}$ where \(\Dbimongen=\Bwith_{i\in\Mdeg}\Sone\) 
is an object of \(\cL\)
\footnote{Actually, it is not essential that \(\Dbimongen\) be the
  product \(\Bwith_{i\in\Mdeg}\Sone\), it suffices that the jointly
  epic \(\Win i\in\cL(\Sone,\Dbimongen)\) exist as well as a diagonal
  morphism \(\Delta\in\cL(\Sone,\Dbimongen)\), related as explained above
  with the partial sum operation in the hom-sets of \(\cL\),
  see~\cite{TassonWalch26}.}.
We also require \(\Dbimongen\) to have a
\(\oc\)-coalgebra structure whose associated commutative comonoid has
\(\Proj0\) (the projection associated with the neutral element of
\(\Mdeg\)) as counit, and a comultplication \(\Dbimonm\) such that
\begin{align*}
  \Dbimonm\Compl\Win i=\sum_{l+_\Mdeg r=i}\Tens{\Win l}{\Win r}\,,
\end{align*}
where the \((\Win i\in\PCOH(\Sone,\Mdeg))_{i\in\Mdeg}\) are the usual
jointly epic injections.
This \(\oc\)-coalgebra structure induces a distributive law
\(\oc\Comp\Sfun\Natto\Sfun\Comp\oc\) between the functor \(\Sfun\) and
the comonad \(\oc\)
which allows us to extend the functor \(\Sfun\) to a functor
\(\Dfun:\Kl\cL\to\Kl\cL\) which inherits from \(\Dbimongen\) a
commutative strong monad structure.
This comonoid structure on \(\Dbimongen\) induces a monad structure on
\(\Sfun\), which extends to a monad structure on \(\Dfun\) and even to
a bimonad structure when \(\Mdeg\) is a total monoid in the sense that
addition of elements of \(\Mdeg\) is always defined, generalizing an
observation made in~\cite{EhrhardWalch25}.

Our first two constructions (coherent differentiation and coherent
Taylor expansion) are instances of this general theory upon taking
\(\Mdeg=\Eset{0,1}\) with \(1+_\Mdeg1\) undefined and
\(\Mdeg=(\Nat,\mathord+)\) respectively.
In the paper we exhibit other examples and here is a short description
of three of them.
In each case we pick an
\(f\in\Kl\cL(X,Y)\) that we consider intuitively as an analytic
function \(X\to Y\) (as in the case \(\cL=\PCOH\)) and we describe
\(\Dfun f\in\Kl\cL(\Sfun X,\Sfun Y)\) as an analytic function mapping
a summable family \((x_i)_{i\in\Mdeg}\) to a summable family
\((y_i)_{i\in\Mdeg}\).
\begin{enumerate}
\item\label{item:binary-differences} \(\Mdeg=\Eset{0,1}\) with \(1+_\Mdeg1=1\).
  Then
  \(\Dfun f(x_0,x_1)=(f(x_0),f(x_0+x_1)-f(x_0))\).
  Interestingly this functor exists in the category of coherence
  spaces equipped with Girard's original resource comonad \(\oc\) based
  on finite cliques, even though there is no satisfactory notion of
  differentiation for this comonad, see \cref{sec:diff-in-LL}.
\item \(\Mdeg=\Nat\) with \(i+_\Mdeg j=\max(i,j)\). Then
  \(\Dfun
  f(x_0,x_1,\dots)
  =(f(x_0),f(x_0+x_1)-f(x_0),f(x_0+x_1+x_2)-f(x_0+x_1),\dots)\).
\item \(\Mdeg=\Eset{0,\dots,N}\) with \(i+_\Mdeg j=\min(i+j,N)\) then
  \[
    \Dfun f(0,x_1,\underbrace{0,\dots,0}_{(N-1)\times})
    =\Big(f(0),f'(0)\cdot x_1,\frac12 f''(0)\cdot(x_1,x_1),\dots,
    f(x_1)-\sum_{n=0}^{N-1}\frac1{\Factor
      n}f^{(n)}(0)\cdot(x_1,\dots,x_1)\Big)
  \]
  which shows that \(\Dfun f\) contains a Taylor expansion of \(f\) of
  order \(N-1\) with a remainder.
  Actually the expression of \(\Dfun f(x_0,\dots,x_N)\) is more
  complicated and is a Faà~di~Bruno formula of order \(N-1\), with a
  remainder.
\end{enumerate}

\paragraph{Perspectives.}
These examples suggest that \DDIFF{} widely subsumes the standard
notions of higher derivatives, including finite differences, and
various kinds of combinations between the two notions.
And there are many choices of \(\Mdeg\) for which understanding \DDIFF{}
could be challenging and interesting: think for instance of the
nonnegative rational numbers, equipped with the usual addition, or of
the set of nonempty finite subsets of \(\Nat\) equipped with set
addition.
Moreover since our models of coherent \DDIFF{} contain all the basic
ingredients required for interpreting Turing-complete functional
programming languages in call-by-name and call-by-value, including
general recursive types, it will most probably be possible to extend
such programming languages with syntactic \(\Mdeg\)-differential
operations for arbitrary choices of \Mpositive{} \Pmonoid{} \(\Mdeg\).
This means that such extensions will feature ``differential''
constructs and reductions much more general than those considered in
the differential extensions of the \(\lambda\)-calculus considered so
far. The differential application of the 
differential $\lambda$-calculus constraints the argument of the function 
to be used linearly, that is, exactly once. Our generalization 
to \DDIFF{} introduces other possible notions of resource sensitivity. 

\paragraph{Contents.}
\cref{sec:diff-in-LL} is an appetizer where we consider the category
of coherence spaces and stable functions and we exhibit a functor from
this category to itself which maps each coherence space \(E\) to a
coherence space \(\Sfun E\) whose cliques (configurations) can be seen
as disjoint pairs \((x_0,x_1)\) of cliques of \(E\) whose union is a
clique.
This functor maps each stable function \(f:E\to F\) to the stable
function \(\Dfun f:\Sfun E\to\Sfun F\) which maps \((x_0,x_1)\) to
\((f(x_0),f(x_0\cup x_1)\setminus f(x_0))\).
We also explain how this difference operation provides another
characterization of stability (in the sense of Berry and Girard).
In the sequel of the paper we shall see that this construction is an
instance of our general theory, and more precisely of the situation
described in \cref{item:binary-differences} above.

In \cref{sec:categorical-prelim} we introduce the categorical notions
that we need.
In particular we develop a notion of symmetric (co)monoid and bimonoid
in a monoidal category which is not necessarily symmetric, with
applications to categories of functors.
This general setting allows us to explain the representable theory in
a principled way.
We also recall the necessary definitions and basic results on (op)lax
functors, categorical models of \(\LL\) and mates.

In \cref{sec:summability} we introduce the notion of strong partial
commutative monoid (\Pmonoid).
%
We also define the notion of \Pmonoid-enriched category, which is just
a specialization to \Pmonoid{} of the general notion of enriched
category.

In \cref{sec:summability-structure} we develop the theory of \SUMBLTY{}
structures in its general functorial form.
Given a category \(\cL\) where each hom-set has a structure of
\Pmonoid{}, a \SUMBLTY{} structure is a triple
\((\Sfun,(\Sproj i)_{i\in\Mdeg},\Ssum)\) where
\(\Sfun:\cL\to\cL\) is a functor, \(\Sproj {i,X}\in\cL(\Sfun X,X)\)
are jointly monic natural transformations and
\(\Ssum_X\in\cL(\Sfun X,X)\) is a natural transformation.
The main assumption is that a family of morphisms
\((f_i\in\cL(X,Y))_{i\in\Mdeg}\) is summable (in the sense of the
\Pmonoid-structure of \(\cL(X,Y)\)) iff there is
\(h\in\cL(X,\Sfun Y)\) such that \(f_i=\Sproj i\Compl h\) for all
\(i\in\Mdeg\), and then \(\sum_{i\in\Mdeg}f_i=\Ssum\Compl h\).
We prove that such a category is automatically \Pmonoid-enriched.
In \cref{sec:representable-sum} we focus on the representable case
that we described above more thoroughly, where \(\Sfun\) is the
functor \(\Limpl\Dbimongen\_\) for \(\Dbimongen=\Bwith_{i\in\Mdeg}\One\).

In \cref{sec:diff-structure} we develop a \DDIFF{} theory in terms of
a distributive law between the functor \(\Sfun\) and the resource
comonad \(\oc\) of our model \(\cL\) of \(\LL\).
This distributive law allows us to extend the functor \(\Sfun\) into a
functor \(\Dfun:\Kl\cL\to\Kl\cL\) that implements
\DDIFF{} on non-linear morphisms.
In the representable case, this distributive law boils down to a
\(\oc\)-coalgebra structure on \(\Dbimongen\).
We notably give a set of axioms relating this distributive law (or
\(\oc\)-coalgebra structure) with the algebraic \Pmonoid{} structure
of \(\Mdeg\).
In that way, we determine the ``differential'' part of our
\(\Mdeg\)-differentiation theory.

\cref{sec:diff-lambda-cat} explains how these summability and \DDIFF{}
structures interact with the product, the function space construction
of \(\cL\), and with the fixpoint operators when available in
\(\Kl\cL\), as it is the case in most models of \(\LL\) which have
such a coherent generalized difference structure.
A programming language featuring general recursion, value types, and 
coherent differentiation (in the sense of~\cite{Ehrhard23a})
has already been introduced in~\cite{Ehrhard23b}. This language 
has good operational properties in spite of its heavy 
syntax. Our \cref{sec:closed-differential} shows that our general theory 
of \DDIFF{} can be carried to a programming language similarly.

\paragraph{Comparison with our previous articles on coherent
  differentiation.}
This article is the third long account of coherent differentiation
following~\cite{Ehrhard23a} and~\cite{EhrhardWalch25}, after its
discovery in~2021.
The present article essentially subsumes the two previous ones with
one major exception: in~\cite{EhrhardWalch25} we also study the
``cartesian'' theory of coherent Taylor expansion in the spirit
of~\cite{EhrhardWalch23}, an aspect of the theory we chose not to
address here for keeping the present article of a reasonable length.
Such a cartesian theory is of course perfectly possible for \DDIFF{}
and can be developed along the lines of~\cite{EhrhardWalch25} without
any surprise. 

The main contribution of this paper wrt.~the two previous ones is of
course the introduction of \DDIFF{} parameterized by an arbitrary
\Mpositive{} \Pmonoid{} \(\Mdeg\) thus widely enlarging the scope of
applications of our theory, including now finite differences for
instance.
This is why we try to present as many examples as possible, including
some not based on the multiset exponential.

Another contribution of our paper lies in the presentation of 
the representable theory. The explicit computations done in~\cite{Ehrhard23a}
were identified as an instance of a mate construction in~\cite{EhrhardWalch25}.
We continue along this line, and present the representable theory as
the transportation of a symmetric monoid structures (see
\cref{sec:symmetric-monoids}) along strong functors. This approach is
more concise and conceptually cleaner than the approach developed
in~\cite{EhrhardWalch25}, in which all structural maps had to be
considered separately.

In parallel to this work, one of the authors published
in~\cite{TassonWalch26} a generalization of the coherent Taylor
expansion of \cite{EhrhardWalch25} to a more liberal notion of partial
sums called PCMs (from which our terminology SPCM is directly
inspired).
The goal of this generalization is to encompass models in which sums are not 
necessarily positive: $x+y = 0$ does not imply $x=y=0$. These
include the historical 
models of the differential $\lambda$-calculus, which are the categories 
of finiteness spaces~\cite{Ehrhard00b} and Köthe spaces~\cite{Ehrhard02}.
This generalization to PCM also applies to
\DDIFF{}, but we keep here the SPCM-based theory for two reasons.
First, it is a bit simpler.
Second, and more importantly, we focus here on the semantics of 
programming languages featuring both differentiation and general recursion
(see~\cref{sec:fixpoints}). It turns out that the semantics of general 
recursion hardly seems compatible with non-positive PCMs (see 
\cref{ex:fixpoint-models}).

Concerning the axiomatization of summability structures, we depart
from \cite{Ehrhard23a} and follow the guidelines
of~\cite{EhrhardWalch25}: we start from the assumption that the
hom-sets of the categorical model \(\cL\) of \(\LL\) are equipped with
a \Pmonoid{} structure, and the summability structure characterizes
them.
This axiomatization splits in two cases depending on whether \(\Mdeg\)
is finite or not whereas in~\cite{EhrhardWalch25} only the infinite
case was considered; surprisingly the finite case is more complicated
to deal with.
In the initial presentation of summability structures in
\cite{Ehrhard23a}, the \Pmonoid{} structures of hom-sets were
byproducts of the summability structure thus requiring specific axioms
guaranteeing that the induced partial addition is associative and
commutative.
The new approach initiated in~\cite{EhrhardWalch25} simplifies the
theory by turning these ``axioms'' into consequences of the \Pmonoid{}
structure of hom-sets.
It is clear however that these two ways of axiomatizing summability
structures are equivalent%
\footnote{The initial presentation keeps some heuristic virtue: it
  often allows to find the suitable notion of summability in hom-sets.
  In coherence space for instance, a family of morphisms (presented as
  cliques) is summable if their union is a clique, and if they are
  pairwise disjoint:  the first condition comes immediately to one's
  mind but the second one is more hidden.}.

\paragraph*{Related work.}
Another approach to resource sensitivity is that of
graded Linear Logic~\cite{BreuvartPagani15,BrunelGaboardiMazzaZdancewic14}.
The key idea of grading is to replace the $\Oc$ exponential with a
parameterized version $\Oc_r$, were $r$ lies in an ordered semiring.
For example, a morphism $f : \Limpl{\Oc_n X}{Y}$ (where $n$ lies in
the semiring $(\Nat, +, \times)$ equipped with its standard order) 
describes a program that
can access its input at most $n$ times.
It would be quite natural to investigate the link between grading and
differentiation, even more so that our \DDIFF{} theory can
now accommodate for very general monoids. 
One direction could be to follow the spirit
of~\cite{BreuvartKerjeanMirwasser26,LemayVienney23} and express a
theory of \DDIFF{} as a distributive law 
with regard to a graded exponential.
Alternatively, we conjecture that a category equipped with a
\DDIFF{} admits a graded exponential, assuming the
existence of sufficient colimits.

\tableofcontents

\section{Differences in LL} %
\label{sec:diff-in-LL}

This section develops a motivating concrete instance of the general
theory presented in the paper, based on the category of coherence
spaces and stable functions.
The stable semantics of the lambda-calculus and of PCF has been
introduced by Gérard~Berry in \cite{Berry78} as a refinement of Scott
semantics featuring stronger properties of determinism.
On coherence spaces, which are special cases of Berry's dI-domains,
Jean-Yves Girard observed in \cite{Girard86,Girard87} that the
coherence space \(\Simpl EF\) of stable function from the coherence
space \(E\) to the coherence space \(F\) can be written
\(\Limpl{\Ocg E}{F}\) where \(\Limpl{E'}{F}\) is the coherence space
of linear functions (stable functions commuting with all existing
lubs) from \(E'\) to \(F\) and \(\Ocg\) is a new construction on
coherence spaces, a resource modality called exponential: this is one
of the origins of \(\LL\).

We provide a new characterization%
\footnote{As far as we know this result does not appear in the
  available literature on stable functions, in spite of the fact that
  this notion is more than 40 years old.} %
of stable functions on coherence spaces as a strong form of
monotonicity involving the difference between \(\subseteq\)-comparable
cliques.
Then we exhibit a functorial difference operator on stable functions
which is directly related to this characterization of stability, and
show that this operation preserves stability.
Formally this operation is similar to the ``tangent bundle functor''
on smooth manifolds which involves computing a differential.
In our functor this differential is replaced by a subtraction, and
actually the stable functions on coherence spaces have no well-behaved
derivatives as explained in~\cite{Ehrhard26}.
In the context of coherence spaces, derivatives can be obtained by
replacing the Girard exponential with another one which uses multisets
instead of sets~\cite{EhrhardWalch25}: this differential operation can
also be presented using the framework of the present article.

A coherence space (CS) is a pair \(E=(\Web E,\mathord{\Coh E})\) where
\(\Web E\) is a set, the \emph{web}, and \(\Coh E\) is a binary,
reflexive and symmetric relation on \(\Web E\), the \emph{coherence
  relation}.
A clique of \(E\) is a subset \(x\) of \(\Web E\) such that
\(\forall a,a'\in x\ a\Coh Ea'\).
Given CSs \(E\) and \(F\), one defines a CS \(\Limpl EF\) by
\(\Web{\Limpl EF}=\Web E\times\Web F\) and
\((a,b)\Coh{\Limpl EF}(a',b')\) if
\(a\Coh Ea'\Implies(b\Coh Fb'\text{ and }b=b'\Implies a=a')\), and the
category \(\COH\) has CSs as objects, and
\(\COH(E,F)=\Cl{\Limpl EF}\), identity morphisms being the diagonal
relations and composition being the standard composition of relations.

If \(E\) and \(F\) are coherence spaces, a function
\(f:\Cl E\to\Cl F\) is Scott continuous if it is increasing and
commutes with the unions of directed sets of cliques, that is, given
\(S\subseteq\Cl E\) directed, one has
\(f(\Union S)\subseteq\Union f(S)\).
A Scott continuous function \(f:\Cl E\to\Cl F\) is stable if, for all
\(x\in\Cl E\) and \(b\in f(x)\) the set
\(\Eset{x'\subseteq x\St b\in f(x')}\) has a unique minimal element
(intuitively: the part of \(x\) that \(f\) has used for computing
\(b\)).

The category \(\COH\) is a model of \LL{} (in particular, it is a
resource category) and the exponential considered first by Girard
\(\Ocg E\) is defined as follows:
\(\Web{\Ocg E}=\{\Set{\List a1n} \St n\in\Nat\text{ and }\{\List
a1n\}\in\Cl E\}\).

The main feature of this exponential is that its Kleisli category
\(\Klcat\COH{\Ocg}\) is isomorphic to the category of coherence spaces
and stable functions by the following correspondence:
with any \(t\in\Klcat\COH\Ocg(E,F)=\Cl{\Limpl{\Ocg E}F}\) one
associates the stable function %
\(\Fun t:\Cl E\to\Cl F\) defined by %
\(\Fun t(x)=\Set{b\in\Web F\St\exists x_0\subseteq x\ (x_0,b)\in
  t}\) %
and the mapping \(t\mapsto\Fun t\) is a bijection between %
\(\Klcat\COH\Ocg(E,F)\) and the set of stable functions
\(\Cl E\to\Cl F\).

Such a \(t\in\Klcat\COH\Ocg(E,F)\) is linear if all its elements
\((x_0,b)\) are such that \(x_0\) is a singleton, so that these linear
stable maps are just the same thing as elements of \(\COH(E,F)\).
We use the linear algebraic notation \(\Matappa tx\) to denote the
application of such a linear \(t\) to \(x\in\Cl E\), that is
\(\Matappa tx=\Set{b\in\Web F\St \exists a\in x\ (a,b)\in t}\).

\subsection{Stability and differences}
Given \(u\in\Cl E\), one can define a coherence space \(E_u\) as
follows:
\(\Web{E_u}=\{a\in\Web E\St \forall a'\in u\ a\Scoh E{a'}\}
\subseteq\Web E\setminus u\) and \(a\Coh{E_u}a'\) iff \(a\Coh{E}a'\)
for all \(a,a'\in\Web{E_u}\).
Notice that
\(\Cl{E_u}=\{x\in\Cl E\St x\cap u=\emptyset\text{ and }x\cup u\in\Cl
E\}\).
If \(f:\Cl E\to\Cl F\) we can define \(\Delta_uf:\Cl{E_u}\to\Cl F\) by
\(\Delta_u f(x)=f(x\cup u)\setminus f(x)\).

\begin{proposition}
  \label{prop:stab-charact}
  Let \(f:\Cl E\to\Cl F\) be Scott continuous.
  The following properties are equivalent.
  \begin{enumerate}
  \item\label{prop:stab-char1} \(f\) is stable;
  \item\label{prop:stab-char2} \(f\) is conditionally multiplicative,
    that is: for all \(x,x'\in\Cl E\), if \(x\cup x'\in\Cl E\) then
    \(f(x)\cap f(x')\subseteq f(x\cap x')\);
  \item\label{prop:stab-char3} for all \(u\in\Cl E\) the function
    \(\Delta_u f\) is increasing.
  \end{enumerate}
\end{proposition}
\begin{proof}
  We only prove \cref{prop:stab-char2} \(\Equiv\)
  \cref{prop:stab-char3}, the equivalence between the two first items
  being well known.
  Assume first that \(f\) is conditionally multiplicative.
  Let \(x\subseteq x'\in\Cl{E_u}\) and let
  \(b\in\Delta_u f(x)=f(x\cup u)\setminus f(x)\).
  By monotonicity we have \(b\in f(x'\cup u)\) and assume that \(b\in f(x')\).
  Since \(x\cup u\cup x'=x'\cup u\in\Cl E\) and
  \(b\in f(x\cup u)\cap f(x')\), we have
  \(b\in f((x\cup u)\cap x')=f(x)\) which is a contradiction.
  It follows that \(b\notin f(x')\) and hence \(b\in\Delta_u f(x')\)
  as required.

  Assume now that \cref{prop:stab-char3} holds, let \(x,x'\in\Cl E\)
  be such that \(x\cup x'\in\Cl E\) and let \(b\in f(x)\cap f(x')\),
  we must prove that \(b\in f(x\cap x')\), so assume that this is not
  the case.
  Let \(u=x\setminus(x\cap x')\), we have \(b\in\Delta_u f(x\cap x')\)
  since \((x\cap x')\cup u=x\).
  Notice that \(x'\cap u=\emptyset\) and hence \(x'\in\Cl{E_u}\) and
  since \(x\cap x'\subseteq x'\) we must have \(b\in\Delta_u f(x')\)
  which implies \(b\notin f(x')\), contradiction.
  Hence we must have \(b\in f(x\cap x')\).
\end{proof}

This result provides a new understanding of stability, related to a
simple operation on functions similar to differentiation, but based on
differences rather than derivatives.
More specifically, given a coherence space \(E\) we can define a
coherence space \(\Sfun E\) by \(\Web{\Sfun E}=\{0,1\}\times\Web E\)
and \((i,a)\Coh{\Sfun E}(i',a')\) if \(a\Coh Ea'\) and
\(i\not=i'\Implies a\not=a'\).
This means that, up to a trivial order isomorphism,
\(\Cl{\Sfun E}=\{(x,u)\in\Cl E^2\St x\cap u=\emptyset\text{ and }x\cup
u\in\Cl E\}\) so that, in particular,
\(\Cl{E_u}=\Eset{x\St (x,u)\in\Cl{\Sfun E}}\).
If \(f:\Cl E\to\Cl F\) is any function, we can define a function
\begin{align*}
  \Dfun f:\Sfun E &\to\Sfun F\\
  (x,u)&\mapsto \Eset0\times f(x)\cup\Eset 1\times(f(x\cup u)\setminus f(x))
\end{align*}
that is \(\Dfun f(x,u)=(f(x),f(x\cup u)\setminus f(x))\) up to the
trivial iso mentioned above.

\begin{proposition}
  If \(f:\Cl E\to\Cl F\) is Scott continuous and stable then so is
  \(\Dfun f\) and \(\Dfun\) is a functor.
\end{proposition}
\begin{proof}
  The fact that \(\Dfun f\) is increasing comes from \cref{prop:stab-charact}.
  Scott continuity of \(\Dfun f\) trivially results from that of
  \(f\), so let us check that \(\Dfun f\) is stable.
  Let \((x,u),(x',u')\in\Cl{\Sfun E}\) be such that
  \((x,u)\cup(x',u')\in\Cl{\Sfun E}\); this means that
  \(x\cup u\cup x'\cup u'\in\Cl E\) and that
  \((x\cup x')\cap(u\cup u')=\emptyset\), so that
  \((x,u)\cup(x',u')=(x\cup x',u\cup u')\) and
  \((x,u)\cap(x',u')=(x\cap x',u\cap u')\).
  We have
  \begin{align*}
    \Dfun f(x,u)\cap\Dfun f(x',u')
    &=(f(x)\cap f(x'),(f(x\cup
      u)\setminus f(x))\cap(f(x'\cup u')\setminus f(x')))\\
    &=(f(x\cap x'),(f(x\cup
      u)\setminus f(x))\cap(f(x'\cup u')\setminus f(x')))
  \end{align*}
  so it suffices to prove that
  \begin{align*}
    (f(x\cup
    u)\setminus f(x))\cap(f(x'\cup u')\setminus f(x'))\subseteq
    f((x\cap x')\cup (u\cap u'))\setminus f(x\cap x')
  \end{align*}
  which results obviously from the fact that
  \((x\cup u)\cap(x'\cup u')=(x\cap x')\cup(u\cap u')\)
  and from the monotonicity and stability of \(f\). 

  Last if \(f:\Cl E\to\Cl F\) and \(g:\Cl F\to\Cl G\) are stable
  (actually we only need \(f\) to be increasing) we have
  \begin{align*}
    (\Dfun g\Comp\Dfun f)(x,u)
    &=\Dfun g(f(x),f(x\cup u)\setminus f(x))\\
    &=(g(f(x)),g(f(x)\cup(f(x\cup u)\setminus f(x)))\setminus g(f(x)))\\
    &=(g(f(x)), g(f(x\cup u))\setminus g(f(x)))
    \text{\quad because }f(x)\subseteq f(x\cup u)\\
    &=\Dfun(g\Comp f)(x,u)
  \end{align*}
  so that \(\Dfun\) is a functor.
\end{proof}

To conclude, we can observe that the functor \(\Dfun\) is a monad and
also a comonad.
The monad structure is given by the stable functions
\(\Dmonu:\Cl E\to\Cl{\Sfun E}\) (the unit of the monad) and
\(\Dmonm:\Cl{\Sfun^2E}\to\Cl{\Sfun E}\) (the multiplication) given by
\(\Dmonu(x)=(x,\emptyset)\) and
\(\Dmonm((x_{00},x_{01}),(x_{10},x_{11}))= (x_{00},x_{01}\cup
x_{10}\cup x_{11})\), using again the description of \(\Cl{\Sfun E}\)
as a set of pairs of cliques.
Similarly the comonad structure is given by the stable functions
\(\Dsum:\Cl{\Sfun E}\to\Cl E\) (counit) such that
\(\Dsum(x_0,x_1)=x_0\cup x_1\) and
\(\Dlift:\Cl{\Sfun E}\to\Cl{\Sfun^2E}\) (comultiplication) given by
\(\Dlift(x_0,x_1)=((x_0,\emptyset),(\emptyset,x_1))\).
These monad and comonad structures of \(\Dfun\) satisfy some
additional compatibility conditions which turns \(\Dfun\) into a
bimonad.
This additional structure on the functor \(\Dfun\) will be a byproduct
of our general theory.

Proving for instance that \(\Dmonm\) is a natural transformation
involves painful computations that we do not need to do because this
property is granted by the general theory of \DDIFF{} of which this
example is an instance.

\section{Categorical preliminaries}
\label{sec:categorical-prelim}

The purpose of this preliminary section is to gather, as compactly as
possible, all the categorical material needed to understand the paper
(most standard notions such as adjunctions, monads, monoidal
categories, are assumed to be known by the reader).
This section is admittedly a bit dry and we think that the reader
can skip it and use it when required: in the paper, at the places where
these notions are used, we have inserted precise references to the
corresponding definitions and results listed in the present section.

In all the monoidal categories considered in this section, we use
\(\ITens\) for the monoidal product, \(\Sone\) for the monoidal unit,
\(\Leftu\) and \(\Rightu\) for the left and right unitors and
\(\Assoc\) for the associator of type
\(\Tens{\Tensp{X_1}{X_2}}{X_3}\to\Tens{X_1}{\Tensp{X_2}{X_3}}\).
When the monoidal category is symmetric, we use \(\Sym\) for the
braiding, which consists of a natural transformation  
\(\Sym_{X_1,X_2}:\Tens{X_1}{X_2}\to\Tens{X_2}{X_1}\)
such that 
\(\Sym_{X_2,X_1}\Compl\Sym_{X_1,X_2}=\Id_{\Tens{X_1}{X_2}}\).

\begin{lemma}
  Let
  \(\cL\) be a monoidal category and let \(f\in\cL(X,\Sone)\) and
  \(g\in\cL(\Sone,Y)\).
  The following diagram commutes.
  \begin{center}
    \label{lemma:sone-tens-comp}
    \begin{tikzcd}
      \Tens X\Sone
      \ar[r,"\Tens fg"]\ar[d,swap,"\Rightu_X"]
      &
      \Tens\Sone Y
      \ar[d,"\Leftu_Y"]
      \\
      X
      \ar[r,"g\Compl f"']
      &
      Y
    \end{tikzcd}
  \end{center}
\end{lemma}
\begin{proof}
  Use the naturality of \(\Leftu\) and \(\Rightu\), and the fact that
  \(\Leftu_\Sone=\Rightu_\Sone\).
\end{proof}

\subsection{Notations}
\label{sec:notations}
A finite multiset of elements of a set \(A\) is a function %
\(m:A\to\Nat\) such that \(\Supp m=\{a\in A\St m(a)\not=0\}\) is
finite.
We use \(\Msetempty\) for the empty multiset such that
\(\Supp\Msetempty=\emptyset\) and standard algebraic notations
\(m_1+m_2\) and \(\sum_{i\in I}m_i\) (for \(I\) finite) for the
pointwise addition of multisets.

We write \(m=\Mset{\List a1k}\in\Mfin A\) the multiset defined by
\(m(a)=\Card{\Eset{j\St a_j=a}}\), and given \(i\in I\), we define %
\(\Indactms im=\Mset{(i,a_1),\dots,(i,a_k)}\in\Mfin{I\times A}\).

If \(M=\Mset{\List m1k}\in\Mfin{\Mfin A}\), we set
\(\sum M=\sum_{i=1}^km_k\in\Mfin A\).

\subsection{Commutative and symmetric (co)monoids}
Since the notions of monoid and comonoid in a monoidal category will
be very useful in this paper, let us start with some reminders about
them.

\subsubsection{Monoids and commutative monoids}

\begin{definition}
  Let \(\cF\) be a monoidal category.
  A monoid in \(\cF\) is a structure
  \(M=(\Monoidca M,\Monoidu M,\Monoidm M)\) where \(\Monoidca M\) is
  an object of \(\cF\), \(\Monoidu M\in\cF(\Sone,\Monoidca M)\),
  \(\Monoidm M\in\cF(\Tens{\Monoidca M}{\Monoidca M},\Monoidca M)\)
  subject to the following commutations where we keep the associator
  implicit:
  \begin{center}
    \begin{tikzcd}
      \Tens{\Sone}{\Monoidca M}\ar[r,"\Tens{\Monoidu M}{\Monoidca M}"]
      \ar[dr,swap,"\Leftu_{\Monoidca M}"]
      &[1em]
      \Tens{\Monoidca M}{\Monoidca M}
      \ar[d,"\Monoidm M"]
      \\
      &
      \Monoidca M
      \end{tikzcd}
    \quad
    \begin{tikzcd}
      \Tens{\Monoidca M}{\Sone}\ar[r,"\Tens{\Monoidca M}{\Monoidu M}"]
      \ar[dr,swap,"\Rightu_{\Monoidca M}"]
      &[1em]
      \Tens{\Monoidca M}{\Monoidca M}
      \ar[d,"\Monoidm M"]
      \\
      &
      \Monoidca M
    \end{tikzcd}
      \quad
    \begin{tikzcd}
      \Monoidca M\ITens\Monoidca M\ITens\Monoidca M
      \ar[r,"\Tens{\Monoidca M}{\Monoidm M}"]
      \ar[d,swap,"\Tens{\Monoidm M}{\Monoidca M}"]
      &[1em]
      \Tens{\Monoidca M}{\Monoidca M}\ar[d,"\Monoidm M"]
      \\
      \Tens{\Monoidca M}{\Monoidca M}\ar[r,"\Monoidm M"']
      &
      \Monoidca M
    \end{tikzcd}
  \end{center}
  Assume moreover that \(\cF\) is a symmetric monoidal category.
  Then \(M\) is a commutative monoid if we also have
  \(\Monoidm M\Compl\Sym_{\Monoidca M,\Monoidca M}=\Monoidm M\).
  If \(M\) and \(N\) are 
  monoids, a morphism from \(M\) to \(N\) is an
  \(f\in\cF(\Monoidca M,\Monoidca N)\) such that the following diagrams
  commute
    \begin{center}
    \begin{tikzcd}
      &[-2em]\Sone\ar[dl,swap,"\Monoidu M"]\ar[dr,"\Monoidu N"]&[-2em]
      \\
      \Monoidca M\ar[rr,"f"]
      &&
      \Monoidca N
    \end{tikzcd}
    \quad
    \begin{tikzcd}
      \Tens{\Monoidca M}{\Monoidca M}\ar[r,"\Tens ff"]\ar[d,swap,"\Monoidm M"]
      &
      \Tens{\Monoidca N}{\Monoidca N}\ar[d,swap,"\Monoidm N"']
      \\
      \Monoidca M\ar[r,"f"']
      &
      \Monoidca M
    \end{tikzcd}
  \end{center}
  We use \(\CMON\cF\) for the category of commutative monoids in
  \(\cF\) and monoid morphisms.
\end{definition}

  We define dually (commutative) comonoids as a tuple
  \(C=(\Cmonoidca C,\Cmonoidu C,\Cmonoidm C)\),
  with \(\Cmonoidu C\in\cF(\Sone,\Cmonoidca C)\) and
  \(\Cmonoidm C\in\cF(\Cmonoidca C,\Tens{\Cmonoidca C}{\Cmonoidca
    C})\) subject to similar commutations.
  We use \(\CCMON\cF\) 
  for the category of commutative comonoids in
  \(\cF\) and comonoid morphisms.

\begin{lemma} \label{lemma:commutative-monoid-product}
  Assume that \(\cF\) is an SMC.
  Given objects \(M_1\) and \(M_2\) of \(\CMON\cF\), one can define a
  commutative monoid \(\Tens{M_1}{M_2}\) in \(\cF\) by taking
  \(\Monoidca{\Tens{M_1}{M_2}}=\Tens{\Monoidca{M_1}}{\Monoidca{M_2}}\)
  with a structure of commutative monoid whose unit and multiplication
  are defined as the following compositions of morphisms in \(\cF\),
  keeping the associator implicit:
  \begin{center}
    \begin{tikzcd}
      \Sone\ar[r,"\Leftu_\Sone=\Rightu_\Sone"]
      &[0.8em]
      \Tens\Sone\Sone\ar[r,"\Tens{\Monoidu{M_1}}{\Monoidu{M_2}}"]
      &[1.4em]
      \Tens{\Monoidca{M_1}}{\Monoidca{M_2}}
    \end{tikzcd}
  \end{center}
  \begin{center}
    \begin{tikzcd}
      \Monoidca{M_1}\ITens\Monoidca{M_2}\ITens\Monoidca{M_1}\ITens\Monoidca{M_2}
      \ar[r,"\Monoidca{M_1}\ITens\Sym\ITens\Monoidca{M_2}"]
      &[3em]
      \Monoidca{M_1}\ITens\Monoidca{M_1}\ITens\Monoidca{M_2}\ITens\Monoidca{M_2}
      \ar[r,"\Tens{\Monoidm{M_1}}{\Monoidm{M_2}}"]
      &[2em] \Tens{\Monoidca{M_1}}{\Monoidca{M_2}}
    \end{tikzcd}
  \end{center}
  Dually, given objects \(C_1\) and \(C_2\) of \(\CCMON\cF\), one can
  define a commutative comonoid \(\Tens{C_1}{C_2}\) in \(\cF\) by
  tacking
  \(\Cmonoidca{\Tens{C_1}{C_2}}=\Tens{\Cmonoidca{C_1}}{\Cmonoidca{C_2}}\),
  with unit \(\Cmonoidu{\Tens{C_1}{C_2}}\) and multiplication
  \(\Cmonoidm{\Tens{C_1}{C_2}}\) defined dually to the corresponding
  operations of the monoid \(\Tens{M_1}{M_2}\).
\end{lemma}

\begin{theorem} \label{thm:monoid-cocartesian} 
  Assume that \(\cF\) is an SMC.
  The category \(\CMON\cF\) is cocartesian:
  \begin{itemize}
  \item the initial object is \(\Sone\) with the obvious monoid
    structure induced by \(\Leftu\) and \(\Rightu\), and for any
    object \(M\) of \(\CMON\cF\), the unique element of
    \(\CMON\cF(\Sone,M)\) is \(\Monoidu M\);
  \item the coproduct of \(M_1\) and \(M_2\) is the monoid
    \(\Tens{M_1}{M_2}\) equipped with the injections
    \((\Tinj_i\in\CMON\cF(M_i,\Tens{M_1}{M_2}))_{i=1,2}\) where
    \(\Tinj_1=\Tensp{\Monoidca{M_1}}{\Monoidu{M_2}}
    \Compl\Inv{\Rightu_{\Monoidca{M_1}}}\) and similarly for
    \(\Tinj_2\);
  \item given \((f_i\in\CMON\cF(M_i,M))_{i=1,2}\), the unique
    \(f\in\CMON\cF(\Tens{M_1}{M_2},M)\) such that
    \((f\Compl\Tinj_i=f_i)_{i=1,2}\) is
    \(f=\Monoidm M\Compl\Tensp{f_1}{f_2}\) and we set
    \(\Cotuplet{f_1,f_2}=f\).
  \end{itemize}
  The category \(\CCMON\cF\) is cartesian: \label{thm:comonoid-cartesian}
  \begin{itemize}
  \item the terminal object is \(\Sone\) with the obvious comonoid
    structure induced by \(\Leftu\) and \(\Rightu\), and for any
    object \(C\) of \(\CCMON\cF\), the unique element of
    \(\CCMON\cF(C,\Sone)\) is \(\Cmonoidu C\);
  \item the product of \(C_1\) and \(C_2\) is the comonoid
    \(\Tens{C_1}{C_2}\) equipped with the projections
    \((\Tproj_i\in\CMON\cF(\Tens{C_1}{C_2},C_i))_{i=1,2}\) where
    \(\Tproj_1={\Rightu_{\Cmonoidca{C_1}}}
    \Compl\Tensp{\Cmonoidca{C_1}}{\Cmonoidu{C_2}}
    \) and similarly for \(\Tproj_2\);
  \item given \((f_i\in\CCMON\cF(C,C_i))_{i=1,2}\), the unique
    \(f\in\CCMON\cF(C,\Tens{C_1}{C_2})\) such that
    \((\Tproj_i\Compl f=f_i)_{i=1,2}\) is
    \(f=\Tensp{f_1}{f_2}\Compl\Cmonoidm C\) and we set
    \(\Tuplet{f_1,f_2}=f\).
  \end{itemize}
\end{theorem}

\subsubsection{\SYmmonoid{}s}
\label{sec:symmetric-monoids}

The goal of this section is to define a notion of \emph{symmetric (co)monoid}
that we will apply in two different contexts.
\begin{enumerate}
\item In symmetric monoidal categories (typically, derived from a
  model of \(\LL\)); we describe in that way the standard notions of
  commutative (co)monoid
\item In categories of functors with various kinds of additional
  structures, where the monoidal product is functor composition; we
  describe in that way (co)monads equipped with a symmetry.
\end{enumerate}
For our purpose, it is convenient to develop a common language for
dealing with both situations which will be connected by means of a
functor as we shall see.
This notion of symmetric monoid was first introduced 
in \cite{Grandis01} in a different setting, but with the 
same underlying motivation.

\begin{definition}
  \label{def:symmonoid}
  A \emph{\Symmonoid{}} is a structure
  \(M=(\Monoidca M,\Monoidu M,\Monoidm M,\Monoidc M)\) where
  \((\Monoidca M,\Monoidu M,\Monoidm M)\) is a monoid
  and
  \(\Monoidc M\in\cF(\Tens{\Monoidca M}{\Monoidca M},\Tens{\Monoidca
    M}{\Monoidca M})\), subject to the following
  commutations where we keep the associator implicit:
  \begin{center}
    \begin{tikzcd}
      \Tens{\Monoidca M}{\Monoidca M}
      \ar[rr,"\Monoidc M"]\ar[dr,swap,"\Monoidm M"]
      &[-3em]&[-3em]
      \Tens{\Monoidca M}{\Monoidca M}\ar[dl,"\Monoidm M"]
      \\
      & {\Monoidca M} &
    \end{tikzcd}
    \quad
    \begin{tikzcd}
      \Tens\Sone{\Monoidca M}\ar[r,"\Leftu"]
      \ar[d,swap,"\Tens{\Monoidu M}{\Monoidca M}"]
      &[-1em]
      \Monoidca M\ar[r,"\Inv{\Rightu}"]
      &[-1em]
      \Tens{\Monoidca M}\Sone
      \ar[d,"\Tens{\Monoidca M}{\Monoidu M}"]
      \\
      \Tens{\Monoidca M}{\Monoidca M}
      \ar[rr,"\Monoidc M"']
      &&
      \Tens{\Monoidca M}{\Monoidca M}
    \end{tikzcd}
  \end{center}
  \begin{center}
    \begin{tikzcd}
      \Monoidca M\ITens\Monoidca M\ITens\Monoidca M
      \ar[r,"\Tens{\Monoidca M}{\Monoidc M}"]
      \ar[d,swap,"\Tens{\Monoidm M}{\Monoidca M}"]
      &[1em]
      \Monoidca M\ITens\Monoidca M\ITens\Monoidca M
      \ar[r,"\Tens{\Monoidc M}{\Monoidca M}"]
      &[1em]
      \Monoidca M\ITens\Monoidca M\ITens\Monoidca M
      \ar[d,"\Tens{\Monoidca M}{\Monoidm M}"]
      \\
      \Tens{\Monoidca M}{\Monoidca M}
      \ar[rr,"\Monoidc M"']
      &&
      \Tens{\Monoidca M}{\Monoidca M}
    \end{tikzcd}
  \end{center}
  \begin{center}
    \begin{tikzcd}
      \Monoidca M\ITens\Monoidca M\ITens\Monoidca M
      \ar[r,"\Tens{\Monoidca M}{\Monoidc M}"]
      \ar[d,swap,"\Tens{\Monoidc M}{\Monoidca M}"]
      &[1em]
      \Monoidca M\ITens\Monoidca M\ITens\Monoidca M
      \ar[r,"\Tens{\Monoidc M}{\Monoidca M}"]
      &[1em]
      \Monoidca M\ITens\Monoidca M\ITens\Monoidca M
      \ar[d,"\Tens{\Monoidc M}{\Monoidca M}"]
      \\
      \Monoidca M\ITens\Monoidca M\ITens\Monoidca M
      \ar[r,"\Tens{\Monoidca M}{\Monoidc M}"']
      &
      \Monoidca M\ITens\Monoidca M\ITens\Monoidca M
      \ar[r,"\Tens{\Monoidc M}{\Monoidca M}"']
      &
      \Monoidca M\ITens\Monoidca M\ITens\Monoidca M
    \end{tikzcd}
  \end{center}
  \begin{center}
    \begin{tikzcd}
      \Tens{\Monoidca M}{\Monoidca M}
      \ar[rr,"\Id"]\ar[dr,swap,"\Monoidc M"]
      &[-4.5em]&[-4.5em]
      \Tens{\Monoidca M}{\Monoidca M}
      \\
      & \Tens{\Monoidca M}{\Monoidca M}\ar[ur,swap,"\Monoidc M"] &
    \end{tikzcd}
  \end{center}
  The first diagram is the notion of commutativity we want to axiomatize.
  In the language of monads, the second and third diagrams express
  that \(\Monoidc M\) is a distributive law between \(\Monoidca M\)
  and itself; they express how \(\Monoidc M\) swaps the unit and
  multiplication of the monoid with arbitrary elements of the monoid.
  The fourth diagram is a Yang-Baxter commutation~\cite{JoyalStreet93}
  and the last one means that \(\Monoidc M\) is involutive%
  \footnote{This means that we are generalizing the symmetric monoidal
    situation and not the braided monoidal situation. This assumption
    as well as the ``commutativity'' commutations could be lightened
    if we were looking for the most general theory, but this is not
    the point of this section.}.

  If \(M\) and \(N\) are \Symmonoid{}s, a morphism from \(M\) to
  \(N\) is a morphism of monoid from $M$ to $N$ subject to the following 
  additional commutation.
  \begin{center}
    \begin{tikzcd}
      \Tens{\Monoidca M}{\Monoidca M}\ar[r,"\Tens ff"]\ar[d,swap,"\Monoidc M"]
      &
      \Tens{\Monoidca N}{\Monoidca N}\ar[d,"\Monoidc N"]
      \\
      \Tens{\Monoidca M}{\Monoidca M}\ar[r,"\Tens ff"']
      &
      \Tens{\Monoidca N}{\Monoidca N}
    \end{tikzcd}
  \end{center}
  We use \(\SMON\cF\) for the category defined in that way and one
  defines similarly \(\SCMON\cF\), the category of symmetric comonoids
  defined dually.
\end{definition}

Observe that any commutative monoid $M$ in a symmetric monoidal 
category is a symmetric monoid, taking $\Monoidc M$ to be the 
braiding $\Sym_{\Monoidca M, \Monoidca M}$ of the symmetric monoidal structure.
One defines in that way a full and faithful functor
from \(\CMON\cF\) to \(\SMON\cF\).
The main feature of \Symmonoid{}s is that for any \(n\in\Nat\), the
object \(\Tpower{\Monoidca M}n\) can be equipped with a canonical
structure of \Symmonoid{}.
This construction generalizes to symmetric monoids 
the product of commutative monoids given 
in~\cref{lemma:commutative-monoid-product}.
For instance, for \(n=2\), the unit, multiplication and symmetry of
this \Symmonoid{} are defined as the following compositions of morphisms
in \(\cF\), keeping the associator implicit:
\begin{center}
  \begin{tikzcd}
    \Sone\ar[r,"\Leftu"]
    &
    \Tens\Sone\Sone\ar[r,"\Tens{\Monoidu M}{\Monoidu M}"]
    &[1em]
    \Tens{\Monoidca M}{\Monoidca M}
  \end{tikzcd}
\end{center}
\begin{center}
  \begin{tikzcd}
    \Tpower{\Monoidca M}4
    \ar[r,"\Monoidca M\ITens\Monoidc M\ITens\Monoidca M"]
    &[3em]
    \Tpower{\Monoidca M}4
    \ar[r,"\Tens{\Monoidm M}{\Monoidm M}"]
    &[2em]
    \Tens{\Monoidca M}{\Monoidca M}
  \end{tikzcd}
\end{center}
\begin{center}
  \begin{tikzcd}
    \Tpower{\Monoidca M}4
    \ar[r,"\Monoidca M\ITens\Monoidc M\ITens\Monoidca M"]
    &[3em]
    \Tpower{\Monoidca M}4
    \ar[r,"\Tens{\Monoidc M}{\Monoidc M}"]
    &[2em]
    \Tpower{\Monoidca M}4
    \ar[r,"\Monoidca M\ITens\Monoidc M\ITens\Monoidca M"]
    &[3em]
    \Tpower{\Monoidca M}4
  \end{tikzcd}
\end{center}
Intuitively these operations are just those of the original monoid
\(M\), performed componentwise.
When \(M\) is a monad, this construction corresponds to the iterated
monad structure described in~\cite{Cheng11}.

One defines of course dually the notion of \Symcomonoid{}
\(C=(\Cmonoidca C,\Cmonoidu C,\Cmonoidm C,\Cmonoidc C)\), with similar
properties, so that \(\Tpower{\Cmonoidca C}n\) has canonically a
structure of \Symcomonoid{}.
We use \(\SCMON\cF\) for the corresponding categories (whose morphisms
are defined similarly to those of \(\SMON\cF\)).
%
\begin{definition}
  \label{def:symmetric-bimonoid}
  A \Symbimonoid{} is a structure
  \(M=(\Monoidca M,\Monoidu M,\Monoidm M,\Cmonoidu M,\Cmonoidm
  M,\Monoidc M)\) such that
  \begin{itemize}
  \item \(\Bmonoid M=(\Monoidca M,\Monoidu M,\Monoidm M,\Monoidc M)\)
    is a symmetric monoid;
  \item  \(\Bcmonoid M=(\Monoidca M,\Cmonoidu M,\Cmonoidm M,\Monoidc M)\)
    is a symmetric comonoid;
  \item \(\Monoidu M\in\SCMON\cF(\Sone,\Bcmonoid M)\) and
    \(\Monoidm M\in\SCMON\cF(\Tens{\Bcmonoid M}{\Bcmonoid M},\Bcmonoid
    M)\);
  \item \(\Cmonoidu M\in\SMON\cF(\Bmonoid M,\Sone)\) and
    \(\Cmonoidm M\in\SMON\cF(\Bmonoid M,\Tens{\Bmonoid M}{\Bmonoid
      M})\).
  \end{itemize}
  Spelling out the definitions, one observes that the two last
  conditions boil down to the following commutations.
\begin{center}
  \begin{tikzcd}
    \Sone\ar[rr,"\Id"]\ar[dr,swap,"\Monoidu M"]
    &[-2em]&[-2em]
    \Sone
    \\
    &\Monoidca M\ar[ru,swap,"\Cmonoidu M"]&
  \end{tikzcd}
  \quad
  \begin{tikzcd}
    \Sone\ar[d,swap,"\Monoidu M"]\ar[r,"\Inv\Leftu"]
    &
    \Tens\Sone\Sone\ar[d,"\Tens{\Monoidu M}{\Monoidu M}"]
    \\
    \Monoidca M\ar[r,"\Cmonoidm M"']
    &
    \Tens{\Monoidca M}{\Monoidca M}
  \end{tikzcd}
  \quad
  \begin{tikzcd}
    \Tens{\Monoidca M}{\Monoidca M}\ar[r,"\Monoidm M"]
    \ar[d,swap,"\Tens{\Cmonoidu M}{\Cmonoidu M}"]
    &
    \Monoidca M\ar[d,"\Cmonoidu M"]
    \\
    \Tens\Sone\Sone\ar[r,"\Leftu"']
    &
    \Sone
  \end{tikzcd}
\end{center}
\begin{center}
  \begin{tikzcd}
    \Tens{\Monoidca M}{\Monoidca M}\ar[r,"\Monoidm M"]
    \ar[d,swap,"\Tens{\Cmonoidm M}{\Cmonoidm M}"]
    &
    \Monoidca M\ar[r,"\Cmonoidm M"]
    &
    \Tens{\Monoidca M}{\Monoidca M}
    \\
    \Tpower{\Monoidca M}4
    \ar[rr,"\Monoidca M\ITens\Monoidc M\ITens\Monoidca M"']
    &&
    \Tpower{\Monoidca M}4
    \ar[u,swap,"\Tens{\Monoidm M}{\Monoidm M}"]
  \end{tikzcd}
\end{center}
\end{definition}

\begin{definition}
  Assume that \(\cF\) is symmetric (with braiding \(\Sym\)).
  A bicommutative bimonoid is a tuple
  \(M=(\Monoidca M,\Monoidu M,\Monoidm M,\Cmonoidu M,\Cmonoidm M)\)
  such that 
  \((\Monoidca M,\Monoidu M,\Monoidm M,\Cmonoidu M,\Cmonoidm
  M,\Sym_{\Monoidca M,\Monoidca M})\) is a \Symbimonoid.
\end{definition}

\subsubsection{Lax and oplax monoidal functor}

We now introduce monoidal functors, which are functors between 
monoidal categories that preserve the monoidal structures. 
In particular, it is well known that lax monoidal functors preserve 
monoids~\cite[Sec~6.2.]{Mellies09} and dually oplax monoidal 
functors preserve comonoids. We prove in this section 
(\cref{thm:symmonoid-strong-mondoidal-functor-image}) that similarly, 
strong monoidal functors preserve both \Symmonoid{}s, 
\Symcomonoid{}s and \Symbimonoid{}s. This result will be a key ingredient 
to build bimonoids in our categories of interest.

\begin{definition}
  \label{def:smf}
  A lax monoidal functor from a monoidal category $\smtuple$ to
  another monoidal category $\smtuplebis$ is a tuple
  $F = (\Monlaxca F, \Monlaxz, \Monlaxt)$ where
  $\Monlaxca F : \smcat \arrow \smcatbis$ is a functor,
  $\Monlaxz \in \smcatbis(\smonebis, \Monlaxca{F} \smone)$ is a
  morphism, and
  $\Monlaxt_{X,Y} \in \smcatbis(\Monlaxca{F} X \smbis \Monlaxca{F} Y,
  \Monlaxca{F} (X \sm Y))$ is a natural transformation such that
  $\Monlaxz$ and $\Monlaxt$ are compatible with the monoidal
  structures:
  \[
    \begin{tikzcd}
      {\Monlaxca{F} X \smbis \smonebis} &[1em] {\Monlaxca{F} X \smbis \Monlaxca{F} \smone } &[-1em] {\Monlaxca{F}(X \sm \smone)} \\
      {\Monlaxca{F} X} & & {\Monlaxca{F} X} \arrow["{\Monlaxca{F} X
        \smbis \Monlaxz}", from=1-1, to=1-2] \arrow["\Monlaxt",
      from=1-2, to=1-3] \arrow["{\smunitRbis}"', from=1-1, to=2-1]
      \arrow["{\Monlaxca{F} \smunitR}", from=1-3, to=2-3]
      \arrow["{\Id}"', from=2-1, to=2-3]
    \end{tikzcd} \
    \begin{tikzcd}
      {\smonebis \smbis \Monlaxca{F} X} &[1em] {\Monlaxca{F} \smone \smbis \Monlaxca{F} X} &[-1em] {\Monlaxca{F}(\smone \sm X)} \\
      {\Monlaxca{F} X} & & {\Monlaxca{F} X} \arrow["{\Monlaxz \smbis
        \Monlaxca{F} X}", from=1-1, to=1-2] \arrow["\Monlaxt",
      from=1-2, to=1-3] \arrow["{\smunitLbis}"', from=1-1, to=2-1]
      \arrow["{\Monlaxca{F} \smunitL}", from=1-3, to=2-3]
      \arrow["{\Id}"', from=2-1, to=2-3]
    \end{tikzcd}
  \] 
  \[ 
    \begin{tikzcd}
      {(\Monlaxca{F} X \smbis \Monlaxca{F} Y) \smbis \Monlaxca{F} Z}
      &[1em] {\Monlaxca{F} (X \sm Y) \smbis \Monlaxca{F} Z} &
      {\Monlaxca{F} ((X \sm Y) \sm Z)} \\
      {\Monlaxca{F} X \smbis (\Monlaxca{F} Y \smbis \Monlaxca{F} Z)} &
      {\Monlaxca{F} X \smbis \Monlaxca{F}(Y \sm Z)} & {\Monlaxca{F} (X
        \sm (Y \sm Z))} \arrow["{\smassocbis}"', from=1-1, to=2-1]
      \arrow["{\Monlaxca{F} \smassoc}", from=1-3, to=2-3]
      \arrow["{\Monlaxt \smbis \Monlaxca{F} Z}", from=1-1, to=1-2]
      \arrow["\Monlaxt", from=1-2, to=1-3] \arrow["{\Monlaxca{F} X
        \smbis \Monlaxt}"', from=2-1, to=2-2] \arrow["\Monlaxt"',
      from=2-2, to=2-3]
    \end{tikzcd}
  \]  
  A lax monoidal functor between two symmetric monoidal categories is
  \emph{symmetric} if its lax monoidal structure is compatible with
  the symmetries:
  \[
    \begin{tikzcd}
      {\Monlaxca{F} X \smbis \Monlaxca{F} Y} & {\Monlaxca{F}(X \sm Y)} \\
      {\Monlaxca{F} Y \smbis \Monlaxca{F} X} & {\Monlaxca{F}(Y \sm X)}
      \arrow["{\Monlaxt_{X, Y}}", from=1-1, to=1-2]
      \arrow["{\smsymbis_{\Monlaxca{F} X, \Monlaxca{F} Y}}"',
      from=1-1, to=2-1] \arrow["{\Monlaxt_{Y, X}}"', from=2-1, to=2-2]
      \arrow["{\Monlaxca{F} \smsym_{X, Y}}", from=1-2, to=2-2]
    \end{tikzcd}
  \]	
\end{definition}

Observe that a (commutative) monoid $M$ in a (symmetric) monoidal
category $\cL$ is exactly the same thing as a lax (symmetric) monoidal
functor $M : \One \arrow \cL$, where $\One$ is the category with one
object and one morphism.

\begin{definition}
  \label{def:monoidal-natural-trans}
  A monoidal natural transformation between two monoidal functors
  $(\Monlaxca{F}, \Monlaxz_{F}, \Monlaxt_{F})$ and
  $(\Monlaxca{G}, \Monlaxz_{G}, \Monlaxt_{G})$ is a natural
  transformation $\alpha : \Monlaxca F \Natto \Monlaxca G$ such that
  the following diagrams commutes
  \[
    \begin{tikzcd}
      {\smonebis} & {\smonebis} \\
      {\Monlaxca{F} \smone} & {G \smone}
      \arrow["{\Monlaxz_{\Monlaxca{F}}}"', from=1-1, to=2-1]
      \arrow["{\Monlaxz_{G}}", from=1-2, to=2-2]
      \arrow["{\alpha_{\smone}}"', from=2-1, to=2-2]
      \arrow["{\Id_{\smonebis}}", from=1-1, to=1-2]
    \end{tikzcd} \qquad
    \begin{tikzcd}
      {\Monlaxca{F} A \smbis \Monlaxca{F} B} & {\Monlaxca{F}(A \sm B)} \\
      {G A \smbis G B} & {G(A \sm B)} \arrow["{\alpha \smbis
        \alpha}"', from=1-1, to=2-1] \arrow["\Monlaxt_{\Monlaxca{F}}",
      from=1-1, to=1-2] \arrow["\Monlaxt_{G}"', from=2-1, to=2-2]
      \arrow["\alpha", from=1-2, to=2-2]
    \end{tikzcd}
  \]
\end{definition}

For every symmetric monoidal categories 
$\smcat$ and $\smcatbis$,
let $\FUNMONLAX{\smcat}{\smcatbis}$ be the category 
whose objects are the symmetric monoidal functors
$\smcat \arrow \smcatbis$, and 
whose morphisms are the monoidal natural transformations.
In particular, the category $\CMON \cL$ of commutative monoids of a
SMC $\cL$ is trivially isomorphic to $\FUNMONLAX{1}{\cL}$.

\begin{definition}
  Dually, an oplax monoidal functor from a monoidal category
  $\smtuple$ to another monoidal category $\smtuplebis$ is a tuple
  $F = (\Monoplaxca F, \Monoplaxz, \Monoplaxt)$ where
  $\Monoplaxca F : \smcat \arrow \smcatbis$ is a functor,
  $\Monoplaxz \in \smcatbis(\Monlaxca{F} \smone, \smonebis)$ is a
  morphism, and
  $\Monoplaxt_{X,Y} \in \smcatbis(\Monoplaxca{F} (X \sm Y),
  \Monoplaxca{F} X \smbis \Monlaxca{F} Y)$ is a natural transformation
  such that $\Monoplaxz$ and $\Monoplaxt$ are compatible with the
  monoidal structure (the commutations are dual to the commutations of
  lax monoidal functors given in \cref{def:smf}).
\end{definition}

Every definition on lax monoidal functors can be adapted to oplax
monoidal functors.
We use $\FUNMONOPLAX{\smcat}{\smcatbis}$ for the category whose
objects are the oplax symmetric monoidal functors from $\smcat$ to
$\smcatbis$, and whose morphisms are the monoidal natural
transformations.
Observe that the category $\CCMON{\cL}$ of commutative comonoids is
trivially isomorphic to $\FUNMONOPLAX{1}{\cL}$.

\begin{definition} \label{def:strong-monoidal-functor}
  A strong monoidal functor is a lax and an oplax 
  monoidal functor whose lax and oplax structures are inverse of each other.
  A strong monoidal functor $F$ is \emph{strict}
  if $\smonebis = \Monlaxca{F} \smone$, and for all objects $X$ and $Y$
  $\Monlaxca{F} X \smbis \Monlaxca{F} Y =
    \Monlaxca{F} (X \sm Y)$, and if the lax and oplax structure 
    are the identity morphisms.
\end{definition}

A crucial property of strong monoidal functors is that they preserve 
both symmetric monoids, symmetric comonoids
and symmetric bimonoids. 

\begin{theorem}
  \label{thm:symmonoid-strong-mondoidal-functor-image}
  Let \(\cL\) and \(\cF\) be monoidal categories and let
  \(F:\cL\to\cF\) be a functor equipped with a strong monoidal
  structure given by \(\Funmonoidz\in\cF(\Sone,F(\Sone))\) and
  \(\Funmonoidt_{U,V}\in\cF(\Tens{F(U)}{F(V)},F(\Tens UV))\).
  If \(M\) is a \Symmonoid{} of \(\cL\) then \(F(\Monoidca M)\) can be
  canonically equipped with a \Symmonoid{} structure; we use simply
  \(F(M)\) for the corresponding \Symmonoid{}.
  Moreover, this correspondence defines a functor
  \(\SMON\cL\to\SMON\cF\). A similar result holds 
  for \Symcomonoid{}s and \Symbimonoid{}s.
\end{theorem}

\begin{proof}
  One defines \(\Monoidu{F(M)}\), \(\Monoidm{F(M)}\) and
  \(\Monoidc{F(M)}\) as the following compositions of morphisms in
  \(\cF\)
  \begin{center}
    \begin{tikzcd}
      \Sone\ar[r,"\Funmonoidz"]
      &
      F(\Sone)\ar[r,"F(\Monoidu M)"]
      &
      F(\Monoidca M)
    \end{tikzcd}
    \quad
    \begin{tikzcd}
      \Tens{F(\Monoidca M)}{F(\Monoidca M)}\ar[r,"\Funmonoidt"]
      &[-1em]
      F(\Tens{\Monoidca M}{\Monoidca M})\ar[r,"F(\Monoidm M)"]
      &[1em]
      F(\Monoidca M)
    \end{tikzcd}
  \end{center}
  \begin{center}
    \begin{tikzcd}
      \Tens{F(\Monoidca M)}{F(\Monoidca M)}
      \ar[r,"\Funmonoidt"]
      &
      F(\Tens{\Monoidca M}{\Monoidca M})
      \ar[r,"F(\Monoidc M)"]
      &
      F(\Tens{\Monoidca M}{\Monoidca M})
      \ar[r,"\Invp\Funmonoidt"]
      &
      \Tens{F(\Monoidca M)}{F(\Monoidca M)}
    \end{tikzcd}
  \end{center}
  The verification of the required commutations are routine
  diagrammatic computations.
\end{proof}

\subsubsection{Monads and comonads}

\label{sec:monads}

As promised in \cref{sec:symmetric-monoids}, we now 
describe monads, comonads and bimonads as monoids, comonoids and 
bimonoids.
If \(\cC\) and \(\cD\) are two categories we use \(\FUN\cC\cD\) for
the category of functors \(\cC\to\cD\), with natural transformations
as morphisms.
For any category \(\cC\), the category \(\FUN\cC\cC\) is a strict
monoidal category with functor composition as monoidal product and the
identity functor as monoidal unit.
The monoidal product of two natural transformations
\((\phi_i\in\FUN\cC\cC(F_i,G_i))\) is their ``horizontal composition''
\(\phi_1\Compv\phi_2\in\FUN\cC\cC(F_1\Compfun F_2,G_1\Compfun
G_2)\) whose definition is standard,
\(\phi_1\Compv\phi_2=(\phi_1G_2)\Compl (F_1\phi_2)=(G_1\phi_2)\Compl(\phi_1F_2)\)
(by naturality).

\begin{definition} \label{def:monad}
A monad on \(\cL\) is a monoid \(M=(\Monadca M,\Monadu M,\Monadm M)\)
in the category \(\FUN\cL\cL\) and similarly, a comonad is a comonoid
\(C=(\Cmonadca C,\Cmonadu C,\Cmonadm C)\) in \(\FUN\cL\cL\).
\end{definition}

We use the standard notations $\Klcat \cL M$ for the Kleisli category
of $M$ and $\Emcat \cL M$ for its Eilenberg-Moore category (also known
as the category of $M$-algebras).
We use the same notations for comonads.

\begin{definition}
  \label{def:sym-monad}
  A \Symmonad{} on \(\cL\) is a \Symmonoid{}
  in \(\FUN\cL\cL\),
  see \cref{def:symmonoid}, and a \Symcomonad{} is a \Symcomonoid{} in
  \(\FUN\cL\cL\).
  A \Symbimonad{} is a \Symbimonoid{} in \(\FUN\cL\cL\), see
  \cref{def:symmetric-bimonoid}.
\end{definition}

As explained in \cref{sec:symmetric-monoids} and \cite{Cheng11},
each symmetric monad $M$ can be iterated to give a symmetric monad $M^n$.
And similarly for symmetric comonads and symmetric bimonads.

\begin{remark}
  Our symmetric bimonads
  \(M=(\Monadca M,\Monadu M,\Monadm M,\Cmonadu M,\Cmonadm M,\Monadc
  M)\) are exactly the $\Monadc M$-bimonads of
  \cite{MesablishviliWisbauer2011} in which $\Monadc M$ is involutive
  and satisfies the Yang Baxter equation.
\end{remark}

Monads and comonads on a monoidal category $\cL$ can be equipped with
an additional structure that expresses compatibility between the monad
and the monoidal structure.
Observe that the composition of two (op)lax monoidal functors has a
canonical (op)lax monoidal structure, and thus $\FUNMONLAX{\cL}{\cL}$
and $\FUNMONOPLAX{\cL}{\cL}$ are monoidal categories with the same
monoidal structure as $\FUN \cL \cL$.

\begin{definition}
  A lax symmetric monoidal monad on a SMC $\cL$ is a monoid in
  $\FUNMONLAX{\cL}{\cL}$.
  A lax symmetric monoidal comonad on $\cL$ is a comonoid in
  $\FUNMONLAX{\cL}{\cL}$.

  An oplax symmetric monoidal monad on $\cL$ is a monoid in
  $\FUNMONOPLAX{\cL}{\cL}$.
  An oplax symmetric monoidal comonad on $\cL$ is a comonoid in
  $\FUNMONOPLAX{\cL}{\cL}$.
\end{definition}

These compatibility conditions allow extending the monoidal 
product of $\cL$ to the Kleisli category of the (co)monad, 
or its Eilenberg-Moore category.

\begin{theorem} \label{thm:sm-extension} The following assertions hold.
  \begin{enumerate}
    \item If $M$ is a lax symmetric monoidal monad on $\cL$, then
    the Kleisli category $\Klcat{\cL}{M}$ of $M$ is symmetric 
    monoidal (\cite[p.30]{Day70}).
    \item If $C$ is a lax symmetric monoidal comonad on $\cL$, then
    the Eilenberg-Moore category $\Emcat{\cL}{C}$ is symmetric 
    monoidal \cite[Sec.~3.4]{Wisbauer08}.
    \item If $M$ is an oplax symmetric monoidal monad on $\cL$, then
    the Eilenberg-Moore $\Emcat{\cL}{M}$ of $M$ is symmetric 
    monoidal \cite[Prop.~1.4]{Moerdijk09}.
    \item If $C$ is an oplax symmetric monoidal comonad on $\cL$, then
    the Kleisli category $\Klcat{\cL}{C}$ of $C$ is symmetric 
    monoidal.
  \end{enumerate}
\end{theorem}

\begin{remark}
\cref{thm:sm-extension} can be improved into a bijection between 
lax (resp.~oplax) monoidal structures and 
\emph{extensions} (resp.~\emph{liftings}) to the
Kleisli category (resp.~the Eilenberg-Moore category) of the
monad, see \cite[Thm.~7.1]{Moerdijk09}. But we will not use 
this bijection here.
\end{remark}

For example, the monoidal product of two objects in the Kleisli
category $\Klcat{\cL}{M}$ of a lax symmetric monoidal monad $M$ is the
same as the product in $\cL$. The product of two morphisms
$f_i \in \Klcat{\cL}{M}(X_i, Y_i)$ for \(i=1,2\) is given by the
following composition in $\cL$.
\[
  \begin{tikzcd}[column sep = large] {X_1 \ITens X_2} & {\Monadca{M}
      Y_1 \ITens \Monadca{M} Y_2} & {\Monadca{M}(Y_1 \ITens Y_2)}
    \arrow["{f_1 \ITens f_2}", from=1-1, to=1-2]
    \arrow["{\Monlaxt_{Y_1, Y_2}}", from=1-2, to=1-3]
  \end{tikzcd}
\]
The monoidal product in the category of coalgebras of a comonad is defined
similarly, we will see an explicit definition in 
\cref{sec:coalg-comonoids} for the exponential modality of \(\LL\). 

\subsection{Resource categories and coalgebras of the resource modality}

\label{sec:ll}

We recall in this section material on the categorical semantics 
of \(\LL\). We consider in this paper models of \(\LL\) 
with cartesian products.

\begin{definition} %
  \cite{Seely89,Bierman95}
  \label{def:resource-category}
  A \emph{resource category} is a category \(\cL\) equipped with the
  following structures and satisfying the following properties.
  \begin{itemize}
  \item \(\cL\) is an SMC. 
  \item \(\cL\) is cartesian that is, any finite family
    \((X_l)_{l\in L}\) of objects of \(\cL\) has a cartesian product %
    \((\Bwith_{l\in L}X_l,(\Proj l)_{l\in L})\) where the
    \((\Proj q\in\cL(\Bwith_{l\in L}X_l,X_q))_{q\in L}\) are the
    projections; we use \(\Stop\) for the terminal object and if
    \((f_l\in\cL(Y,X_l))_{l\in L}\) we use \(\Tuple{f_l}_{l\in L}\) for
    the unique element of \(\cL(X,\Bwith_{l\in L}X_l)\) such that
    \(\Proj q\Tuple{f_l}_{l\in L}=f_q\) for all \(q\in L\).
  \item \(\cL\) is equipped with a \emph{resource modality}, which is
    a tuple %
    \((\Oc,\Der{},\Digg{},\Seelyz,\Seelyt)\) where \(\Oc:\cL\to\cL\)
    is a functor, \(\Der{}\) (dereliction) and \(\Digg{}\) (digging)
    are respectively the counit and the comultiplication of a comonad
    structure on this functor and \((\Seelyz\in\cL(\Sone,\Oc\Stop))\) and
    \(\Seelyt_{X_1,X_2}\in\cL(\Tens{\Oc X_1}{\Oc
      X_2},\Oc{\Withp{X_1}{X_2}})\) turn \(\Oc{}\)
    into a strong symmetric monoidal comonad from the SMC
    \((\cL,\Stop,\mathord{\IWith})\) to the SMC
    \((\cL,\Sone,\mathord{\ITens})\).
    These isos are called the \emph{Seely isomorphisms} of \(\cL\).
    The iso $\Seelyt$ should also satisfy an additional commutation,
    given in \cite[Eq.~(72)]{Mellies09}.
  \end{itemize}
\end{definition}

If \(\cL\) is a resource category, the resource modality induces
a Kleisli category %
\(\Kloc\cL\) whose objects are those of \(\cL\) and where %
\(\Kloc\cL(X,Y)=\cL(\Oc X,Y)\).
In this category the identity morphisms are \(\Klocid_X=\Der X\) %
and composition of \(f\in\Kloc\cL(X,Y)\) and \(g\in\Kloc\cL(Y,Z)\) %
is defined by \(g\Comp f=g\Compl(\Oc f)\Compl\Digg X\).
The basic intuition in this situation is that the morphisms of \(\cL\)
are linear whereas \(\Kloc\cL\) is a category of nonlinear morphisms.
Here the word ``linear'' can be used in its algebraic and its computer
science meaning.
This intuition is supported by the fact that there is a functor %
\(\Derfun:\cL\to\Kloc\cL\) which acts as the identity on objects and
maps \(f\in\cL(X,Y)\) to \(\Derfun(f)=f\Compl\Der X\in\Kloc\cL(X,Y)\).
This functor is not necessarily faithful (it is, in most known
categorical models of \LL), but it should nevertheless be considered
as a kind of ``inclusion'' of \(\cL\) into the larger \(\Kloc\cL\).

\begin{convention}
  In this section, we assume that \(\cL\) is a resource category.
\end{convention}

\subsubsection{Coalgebras and comonoids}
\label{sec:coalg-comonoids}

As in \cref{def:monad}, we use \(\Em\cL\) for the Eilenberg-Moore
category of the comonad \(\Oc\).
We write a $\Oc$-coalgebra $P$ as a pair $(\Coalgca{P}, \Coalgst{P})$,
here $\Coalgca{P}$ is an object of $\cL$ and
$\Coalgst{P} \in \cL(\Coalgca{P}, \Oc \Coalgca{P})$.
Let \(\Emforg:\Em\cL\to\cL\) be the forgetful functor, and 
let $\Emfree : \cL \arrow \Em \cL$ be the free coalgebra functor 
that maps an object $X$ to the $\Oc$-coalgebra $(\Oc X, \Digg{X})$ and 
a morphism $f \in \cL(X, Y)$ to $\Oc f$.

Remember (from~\cite{Mellies09}, for instance) that the functor
\(\Oc\_\) inherits, from its strict symmetric monoidal structure
\((\Seelyz,\Seelyt)\) from the SMC \((\cL,\mathord{\IWith})\) to the
SMC \((\cL,\mathord{\ITens})\), a lax symmetric monoidal structure
\((\ExpMon0,\ExpMon2)\) from the SMC \((\cL,\mathord{\ITens})\) to
itself.
As explained in \cref{thm:sm-extension}, this lax monoidal structure 
endows the category $\Em \cL$
with a symmetric monoidal structure. The unit is given by the $\Oc$-coalgebra 
$(\Sone, \ExpMon0)$ that we simply denote as \(\Sone\).
The monoidal product of two \(\Oc\)-coalgebras \(P_1\) and \(P_2\)
is the pair %
\((\Tens{\Coalgca{P_1}}{\Coalgca{P_2}},h)\) where \(h\) is defined as
the following composition of morphisms
\begin{equation*}
  \begin{tikzcd}
    \Tens{\Coalgca{P_1}}{\Coalgca{P_2}}
    \ar[r,"\Tens{\Coalgst{P_1}}{\Coalgst{P_2}}"]
    &[1.4em]
    \Tens{\Oc{\Coalgca{P_1}}}{\Oc{\Coalgca{P_2}}}
    \ar[r,"\ExpMon2_{\Coalgca{P_1},\Coalgca{P_2}}"]
    &[1.4em]
    \Oc{\Tensp{\Coalgca{P_1}}{\Coalgca{P_2}}}
  \end{tikzcd}.
\end{equation*}

The point of a resource modality in Linear Logic is to endow 
objects with an erasure and a duplication operations, as follows.

\begin{proposition} \label{prop:resource-is-linear}
For all resource category $\cL$ and object $X$ of $\cL$, 
the objects $\Oc X$ has a canonical
commutative comonoid structure given by 
\begin{equation} \label{eq:canonical-comonoid}  
    \Weak X : \begin{tikzcd}[ampersand replacement = \&]
    \Oc{X}\ar[r,"\Oc 0"]
    \&[-1em]
    \Oc\Top\ar[r,"\Invp{\Seelyz}"]
    \&[1em]
    \Sone
  \end{tikzcd} \qquad 
  \Contr X : \begin{tikzcd}[ampersand replacement = \t] 
      \Oc{X}\ar[r,"{\Oc{\Tuple{\Id,\Id}}}"]
      \t[1em]
      \Oc{(X \with X)}\ar[r,"\Invp{\Seelyt}"]
      \t[1em]
      \Tens{\Oc{X}}{\Oc{X}}
    \end{tikzcd}
\end{equation}
Furthermore, $\Digg{X}$ is a comonoid morphism from 
$(\Oc X, \Weak X, \Contr X)$ to $(\Oc \Oc X, \Weak{\Oc X}, \Contr{\Oc X})$, 
$\Weak X$ is a coalgebra morphism from 
$(\Oc X, \Digg X)$ to $1$, and $\Contr X$ is a coalgebra morphism from 
$(\Oc X, \Digg X)$ to  $(\Oc X, \Digg X) \ITens (\Oc X, \Digg X)$. 
\end{proposition}

\cref{prop:resource-is-linear} above means that every resource category is a \emph{linear category}, 
in the sense of~\cite{Bierman94}.
The proof of this implication can be found in~\cite{Bierman95}.
Now, this comonoid structure that is defined on free coalgebras can be extended to all coalgebras, as 
follows.

\begin{proposition} \label{prop:coalgebra-to-comonoid}
  Every coalgebra $P$ has a canonical comonoid structure (in 
  the category $\Em \cL$) given by
  the following compositions in $\cL$
  \begin{equation} \label{eq:coalgebra-to-comonoid}
    \begin{split}
    \Coalgw P &: 
    \begin{tikzcd}[ampersand replacement = \t]
    \Coalgca P\ar[r,"\Coalgst P"]
    \t
    \Oc{\Coalgca P}
    \ar[r, "\Weak{\Coalgca P}"]
    \t[1em]
    \Sone
  \end{tikzcd} \\
    \Coalgc P &: 
    \begin{tikzcd}[ampersand replacement = \t]
      \Coalgca P\ar[r,"\Coalgst P"]
      \t
      \Oc{\Coalgca P}
      \ar[r,"\Contr{\Coalgca P}"]
      \t[1em]
      \Tens{\Oc{\Coalgca P}}{\Oc{\Coalgca P}}
      \ar[r,"\Tens{\Der{\Coalgca P}}{\Der{\Coalgca P}}"]
      \t[3em]
      \Tens{\Coalgca P}{\Coalgca P}
    \end{tikzcd}
  \end{split}
  \end{equation}
  and every coalgebra morphism $f \in \Em \cL(P, Q)$ 
  is a comonoid morphism from $(\Coalgca{P}, \Coalgw P, \Coalgc P)$
  to  $(\Coalgca{Q}, \Coalgw Q, \Coalgc Q)$.
  We get in that way a functor 
  \(\Coalcomo:\Em\cL\to\CCMON\cL\) that maps an 
  object \(P\) of \(\Em\cL\) to \((\Coalgca P,\Coalgw P,\Coalgc P)\)
  and a morphism \(f\in\Em\cL(P,Q)\) to \(f\).
\end{proposition}

\begin{proof}
The proof is not trivial and can be found in~\cite[Sec.~7.4]{Mellies09}.
\end{proof}

\begin{remark} \label{rem:same-comonoids}
Observe that for all free coalgebra $\Emfree(X) = (X, \Digg X)$, the
commutative comonoid $(X, \Coalgw {\Emfree X}, \Coalgc {\Emfree X})$ is defined as follows
  \begin{equation*} 
    \begin{split}
    \Coalgw {\Emfree X}&: \begin{tikzcd}[ampersand replacement = \t] 
    \Oc X \ar[r,"\Digg X"]
    \t
    \Oc{\Oc X}
    \ar[r, "\Weak{\Oc X}"]
    \t[1em]
    \Sone
  \end{tikzcd} \\
  \Coalgc {\Emfree X}&: \begin{tikzcd}[ampersand replacement = \t]  
      \Oc X \ar[r,"\Digg X"]
      \t
      \Oc{\Oc X}
      \ar[r,"\Contr{\Oc X}"]
      \t[1em]
      \Tens{\Oc{\Oc X}}{\Oc{\Oc X}}
      \ar[r,"\Tens{\Der{\Oc X}}{\Der{\Oc X}}"]
      \t[2em]
      \Tens{\Oc X}{\Oc X}
    \end{tikzcd}
  \end{split}
  \end{equation*}
  and it immediately follows from the fact that $\Digg X$ is a comonoid 
  morphism that $\Coalgw{\Emfree{X}} = \Weak X$ and $\Coalgc {\Emfree X} = \Contr X$.
  So the comonoid structure on the coalgebras given by
  \cref{prop:coalgebra-to-comonoid} is a conservative generalization
  of the comonoid structures on the free coalgebras given in
  \cref{eq:canonical-comonoid}.
\end{remark}

\begin{proposition} \label{prop:em-excl-cartesian}
The category $\Em \cL$ is cartesian. The terminal 
object is the coalgebra \((\Sone,\ExpMon0)\) (that we will 
often write $1$ for conciseness), and the unique 
coalgebra morphism from $P$ to $1$ is  
$\Coalgw P$.
The projection 
\(\Coalgproj i\in\cL(\Coalgca{\Tens{P_1}{P_2}},\Coalgca{P_i})\) 
is given by the following composition (for \(i=1\), the case \(i=2\) is of course
similar):
\begin{center}
  \begin{tikzcd}
    \Tens{\Coalgca{P_1}}{\Coalgca{P_2}}
    \ar[r,"\Tens{\Coalgca{P_1}}{\Coalgw{P_2}}"]
    &[2em]
    \Tens{\Coalgca{P_1}}{\Sone}\ar[r,"\Rightu"]
    &[-1em]
    \Coalgca{P_1}
  \end{tikzcd}
\end{center}
Given \((f_i\in\Em\cL(Q,P_i))_{i=1,2}\), the unique morphism
  \(g\in\Em\cL(Q,\Tens{P_1}{P_2})\) such that
  \(\Coalgproj i\Compl g=f_i\) for \(i=1,2\) is given by 
  \(\Tensp{f_1}{f_2}\Compl\Coalgc Q\).
\end{proposition}

\begin{proof}
  As explained in~\cite[Cor.~17]{Mellies09}, a cartesian category is
  the same thing as a monoidal category where every object is equipped
  with a commutative comonoid structure subject to additional
  conditions.
  As shown in \cite[Sec.~7.4]{Mellies09}, these conditions are
  satisfied by the comonoids $(P, \Coalgw P, \Coalgc P)$, this is why
  $\Em \cL$ is cartesian.
\end{proof}

\subsubsection{The Lafont situation}

\label{sec:Lafont}

As pointed out in~\cite{Mellies09}, there is a particularly interesting
situation which occurs quite often in concrete models of \(\LL\):
this is the case where the resource modality is free.

\begin{definition}
  \label{def:Lafont}
  A SMC \(\cL\) is \emph{Lafont} if the forgetful functor
  \(\Forgcmon:\CCMON\cL\to\cL\) has a right adjoint.
  We use \(\Freecmon\) for (a choice for) this right adjoint,
  \(\Freecmonu:\Id\Natto\Freecmon\Compfun\Forgcmon\) for the unit of
  the adjunction and
  \(\Freecmoncu:\Forgcmon\Compfun\Freecmon\Natto\Id\) for its counit.
  This adjunction induces a comonad on \(\cL\) whose underlying
  functor is \(\Oc=\Forgcmon\Compfun\Freecmon\), whose counit is
  \(\Der{}=\Freecmoncu\) and comultiplication is
  \(\Digg{}=\Forgcmon\Freecmonu\Freecmon\).
\end{definition}

Spelling out one of the characterizations of adjunctions in the
present setting, we get the following useful presentation of this
adjunction in the language of universal properties.
\begin{lemma}
  \label{lemma:Lafont-universal}
  A SMC \(\cL\) is Lafont if and only if for each object \(X\) of
  \(\cL\) there is a pair \((\Freecmon X,\Der X)\) where \(\Freecmon X\)
  is an object of \(\CCMON\cL\) and
  \(\Der X\in\cL(\Monoidca{\Freecmon X},X)\), with the following
  universal property: for each object \(C\) of \(\CCMON\cL\) and each
  \(f\in\cL(\Monoidca C,X)\), there is a unique morphism
  \(\Promcm f\in\CCMON\cL(C,\Freecmon X)\) such that
  \(\Der X\Compl\Promcm f=f\).
  This universally defined operation on morphisms allows one to
  retrieve the functorial action of \(\Freecmon\): given
  \(f\in\cL(X_1,X_2)\), we have
  \(\Freecmon f=\Promcm{(f\Compl\Der{X_1})}\).
\end{lemma}
This lemma expresses exactly that \(\Freecmon X\) is the cofree
commutative comonoid cogenerated by \(X\).

\begin{theorem}[\cite{Bierman94}] \label{thm:lafont-is-ressource}
  If \(\cL\) a Lafont SMC which is cartesian, then, equipped with the
  induced comonad \(\Oc\), \(\cL\) is a resource category.
  Furthermore, the comonoid $(\Oc X, \Weak X, \Contr X)$ given in
  \cref{eq:canonical-comonoid} coincides with the free commutative
  comonoid $\Freecmon X$.
\end{theorem}

Such a category is called a Lafont resource category.
Notice that, for an SMC, being a Lafont resource category is a
property, and not an additional structure.
It turns out that many interesting and non-additive models of \LL{}
are Lafont resource categories, here are a few examples but there are
many others:
\begin{itemize}
\item the category of coherence spaces with the multiset based
  exponential;
\item the category of hypercoherence spaces with the multiset based
  exponential\footnote{This exponential has not been formally
    introduced as far as we know but is easy to describe.};
\item the category of nonuniform coherence spaces equipped with the
  Boudes exponential~\cite{BucciarelliEhrhard99,Boudes11};
\item the category of probabilistic coherence spaces with its unique
  known exponential~\cite{DanosEhrhard08,CrubilleEhrhardPaganiTasson16}.
\end{itemize}

It is widely accepted among the \(\LL\) community that for all Lafont
SMC $\cL$, the associated category of commutative comonoids
$\CCMON \cL$ is isomorphic its Eilenberg-Moore category $\Em \cL$.
However, and to our surprise, this results does not seem to be
properly proved anywhere.
Thus, we provide here a proof of this isomorphism. We rely on the 
following standard result on adjunctions.

\begin{proposition} \label{prop:em-terminal}
  Let $\cL$ be a category equipped with a comonad 
  $(\Oc, \Der, \Digg{})$. Given an adjunction
  \[ 
    \begin{tikzcd}
      \cC && \cL
      \arrow[""{name=0, anchor=center, inner sep=0}, "L", curve={height=-12pt}, from=1-1, to=1-3]
      \arrow[""{name=1, anchor=center, inner sep=0}, "M", curve={height=-12pt}, from=1-3, to=1-1]
      \arrow["\dashv"{anchor=center, rotate=-90}, draw=none, from=0, to=1]
    \end{tikzcd} 
  \]
  such that $\Oc = LM$, there exists a unique functor 
  $F : \cC \arrow \Em \cL$ such that 
  the following diagrams commute
  \[ 
  \begin{tikzcd}
    {\cC} & \cL \\
    & {\Em\cL}
    \arrow["F"', from=1-1, to=2-2]
    \arrow["L", from=1-1, to=1-2]
    \arrow["\Emforg"', from=2-2, to=1-2]
  \end{tikzcd} \qquad  
    \begin{tikzcd}
    {\cC} & \cL \\
    & {\Em\cL}
    \arrow["F"', from=1-1, to=2-2]
    \arrow["M"', from=1-2, to=1-1]
    \arrow["\Emfree", from=1-2, to=2-2]
  \end{tikzcd}\]
  where we recall that $\Emforg : \Em \cL \arrow \cL$ 
  is the forgetful functor, and 
  $\Emfree : \cL \arrow \Em \cL$ is the free coalgebra functor 
  that maps an object $X$ to the $\Oc$-coalgebra $(\Oc X, \Digg{X})$ and 
  a morphism $f \in \cL(X, Y)$ to $\Oc f$.
\end{proposition}

\begin{proof}
 This is the dual of a standard result on 
  adjunctions and monads, see for example 
  \cite[Chap.~6,Sec.~3,Thm.~1]{Maclane71} for the result 
  on monads.
\end{proof}

Assume that $\cL$ is Lafont.
By \cref{prop:em-terminal} above there exists 
a functor $\CmonEmFun : \CCMON\cL \arrow \Em\cL$ 
such that the following diagrams commute
\[
  \begin{tikzcd}
    {\CCMON\cL} & \cL \\
    & {\Em\cL} \arrow["\CmonEmFun"', from=1-1, to=2-2]
    \arrow["\Forgcmon", from=1-1, to=1-2] \arrow["\Emforg"', from=2-2,
    to=1-2]
  \end{tikzcd} \qquad  
  \begin{tikzcd}
    {\CCMON\cL} & \cL \\
    & {\Em\cL} \arrow["\CmonEmFun"', from=1-1, to=2-2]
    \arrow["\Freecmon"', from=1-2, to=1-1] \arrow["\Emfree", from=1-2,
    to=2-2]
  \end{tikzcd}
\]
Explicitly, the functor $\CmonEmFun$ maps a comonoid $C$ to the
coalgebra $(\Cmonca C, \Promcm{\Id_C})$ where
$\Promcm{\Id_{\Cmonca C}} \in \cL(\Cmonca C, \Oc{\Cmonca C})$ is given
by \cref{lemma:Lafont-universal}; this is a coalgebra by naturality of
the unit \(\Freecmonu\) of the adjunction
$\Forgcmon \dashv \Freecmon$.
As required by the leftmost diagram, $\CmonEmFun$ does nothing on
morphisms, in particular every comonoid morphism
$f \in \CCMON\cL(C, D)$ is also a coalgebra morphism
$f \in \Em \cL (\CmonEmFun C, \CmonEmFun D)$.
Conversely, as stated in \cref{prop:coalgebra-to-comonoid} there is a
functor \(\Coalcomo:\Em\cL\to\CCMON\cL\) that maps a coalgebra \(P\)
to the comonoid \((\Coalgca P,\Coalgw P,\Coalgc P)\), and does nothing
on the morphism.

\begin{theorem}
  \label{thm:Lafont-comonoid-coalgebra-iso}
  Assume that $\cL$ is a Lafont SMC. Then the functors
  $\CmonEmFun : \CCMON\cL \arrow \Em\cL$ and
  \(\Coalcomo:\Em\cL\to\CCMON\cL\) are inverse of each other, and thus
  $\CCMON\cL$ and $\Em \cL$ are isomorphic.
\end{theorem}
\begin{proof}
  \Proofcase
  First, let us prove that
  $\Coalcomo{}\, \CmonEmFun = \Idfun_{\CCMON \cL}$.  Since this
  functor acts trivially on morphisms, it suffices to prove that
  $\Coalcomo\,\CmonEmFun C = C$ for all comonoid
  $C = (\Cmonca C, \Ccomonu C, \Ccomonm C)$.
  By definition of $\Coalcomo$ and $\CmonEmFun$, the structural maps
  of the comonoid $\Coalcomo\, \CmonEmFun C$ are given by the
  following compositions in $\cL$.
  \begin{equation*}
    \begin{split}
      &\begin{tikzcd}[ampersand replacement = \t] 
        \Cmonca C \ar[r,"\Promcm{\Id_{\Cmonca C}}"]
        \t
        \Oc{\Cmonca C}
        \ar[r, "\Weak{\Cmonca C}"]
        \t[1em]
        \Sone
      \end{tikzcd} \\
      &\begin{tikzcd}[ampersand replacement =
         \t] 
         \Cmonca C \ar[r,"\Promcm{\Id_{\Cmonca C}}"] \t \Oc{\Cmonca C}
         \ar[r,"\Contr{\Cmonca C}"] \t[1em] \Tens{\Oc{\Cmonca
             C}}{\Oc{\Cmonca C}} \ar[r,"\Tens{\Der{\Cmonca
             C}}{\Der{\Cmonca C}}"] \t[2em] \Tens{\Cmonca C}{\Cmonca
           C}
       \end{tikzcd}
    \end{split}
  \end{equation*}
  Now, recall from \cref{lemma:Lafont-universal} that
  $\Der{\Cmonca C} \Compl \Promcm{\Id_{\Cmonca C}} = \Id_{\Cmonca C}$,
  and that $\Promcm{\Id_{\Cmonca C}}$ is a comonoid morphism from $C$
  to $\Freecmon C$.
  Furthermore,
  $\Freecmon C = (\Cmonca{C}, \Weak{\Cmonca C}, \Contr{\Cmonca C})$ as
  stated in \cref{thm:lafont-is-ressource}.
  Thus,
  $\Weak{\Cmonca C} \Compl \Promcm{\Id_{\Cmonca C}} = \Ccomonu C$ and
  $(\Der{\Cmonca C} \ITens \Der{\Cmonca C}) \Compl \Contr{\Cmonca C}
  \Compl \Promcm{\Id_{\Cmonca C}} = (\Der{\Cmonca C} \ITens
  \Der{\Cmonca C}) \Compl (\Promcm{\Id_{\Cmonca C}} \ITens
  \Promcm{\Id_{\Cmonca C}}) \Compl \Ccomonm C = \Ccomonm C$.
  In other word, $\Coalcomo \CmonEmFun C = C$.

  \Proofcase
  We now prove that $\CmonEmFun \Coalcomo = \Idfun_{\Em \cL}$.
  To that end, we first prove that the two following diagrams commute.
  \[
    \begin{tikzcd}
      {\Em \cL} & \cL \\
      {\CCMON \cL} & \cL \\
      {\Em \cL} & \cL \arrow["\Emforg", from=1-1, to=1-2]
      \arrow["\Coalcomo"', from=1-1, to=2-1] \arrow["\Idfun",
      from=1-2, to=2-2] \arrow["\Forgcmon", from=2-1, to=2-2]
      \arrow["\CmonEmFun"', from=2-1, to=3-1] \arrow["\Idfun",
      from=2-2, to=3-2] \arrow["\Emforg"', from=3-1, to=3-2]
    \end{tikzcd} \qquad
    \begin{tikzcd}
      {\Em \cL} & \cL \\
      {\CCMON \cL} & \cL \\
      {\Em \cL} & \cL
      \arrow["\Coalcomo"', from=1-1, to=2-1]
      \arrow["\Emfree"', from=1-2, to=1-1]
      \arrow["{(*)}"{description}, draw=none, from=1-2, to=2-1]
      \arrow["\Idfun", from=1-2, to=2-2]
      \arrow["\CmonEmFun"', from=2-1, to=3-1]
      \arrow["\Freecmon"', from=2-2, to=2-1]
      \arrow["\Idfun", from=2-2, to=3-2]
      \arrow["\Emfree", from=3-2, to=3-1]
    \end{tikzcd}
  \]
  The leftmost diagram is trivial, since
  $\Coalgca{\CmonEmFun \Coalcomo P} = \Coalgca{P}$ and both
  $\CmonEmFun$ and $\Coalcomo$ acts trivially on morphisms.
  The bottom part of the rightmost diagram commutes by definition of
  $\CmonEmFun$ which leaves only the top part $(*)$.
  To prove that it commutes, it suffices to check that for each object
  $X$, $\Coalcomo\, \Emfree X = \Freecmon X$.
  But $\Freecmon\, X = (\Oc X, \Weak X, \Contr X)$ by
  \cref{thm:lafont-is-ressource}, and
  $\Coalcomo \Emfree X = (\Oc X, \Coalgw{\Oc X}, \Coalgc{\Oc})$ by
  definition.
  Thus we need to prove that
  $(\Oc X, \Coalgw{\Oc X}, \Coalgc{\Oc}) = (\Oc X, \Weak X, \Contr
  X)$, which we have already done in \cref{rem:same-comonoids}.
  Now, these two diagrams mean that the functor $\CmonEmFun \Coalcomo$
  is the unique functor $\Em \cL \arrow \Em \cL$ that satisfies the
  conditions given in \cref{prop:em-terminal}. But $\Idfun_{\Em \cL}$
  also fulfills these conditions, and thus we must have
  $\CmonEmFun \Coalcomo = \Idfun$.
\end{proof}

Recall that for each SMC $\cL$, $\CCMON \cL$ is cartesian, with
terminal object $(\Sone, \Id_{\Sone}, \Leftu_{\Sone})$ and product
$\Tens{C_1}{C_2}$ described in \cref{thm:comonoid-cartesian}.
Furthermore, if $\cL$ is Lafont then it is a resource category by
\cref{thm:lafont-is-ressource} and $\Em \cL$ is cartesian with
terminal coalgebra $(1, \ExpMon0)$ and product coalgebra
$\Tens{P_1}{P_2}$ described in \cref{prop:em-excl-cartesian}.
It turns out that $\CmonEmFun : \CCMON\cL \arrow \Em\cL$ preserves
these products.

\begin{proposition}
  \label{prop:comonoid-to-coalgebra-cartesian}
  The functor $\CmonEmFun : \CCMON\cL \arrow \Em\cL$ is a strict
  monoidal functor (\cref{def:strong-monoidal-functor}):
  $\CmonEmFun (\Sone, \Id_{\Sone}, \Leftu_{\Sone}) = (1, \ExpMon0)$
  and
  $\CmonEmFun (\Tens{C_1}{C_2}) = \Tens{\CmonEmFun C_1}{\CmonEmFun
    C_2}$.
\end{proposition}

\begin{proof}
  Isomorphisms between cartesian categories preserve products.
\end{proof}

In particular, $\CmonEmFun$ is strong monoidal. This implies by
\cref{thm:symmonoid-strong-mondoidal-functor-image} that this functor
sends symmetric monoids/comonoids/bimonoids in $\CCMON\cL$ to
symmetric monoids/comonoids/bimonoids in $\Em\cL$.
This means that the Lafont situation provides a very convenient
setting to build bimonoids in $\Emcat{\cL}{\Oc}$.
As we will see in \cref{sec:coalg-to-dl} our \DDIFF{} theory is often
expressed by such bimonoids.

\subsection{Distributive laws between functors and comonads}

\label{sec:distr-law-comonad}

In our introductory example \cref{sec:diff-in-LL}, we described a  
calculus of differences by means of an endofunctor $\Dfun$ (which has, 
in fact, a symmetric bimonad structure)
on the Kleisli category of the set based exponential of coherence spaces. 
Our theory of coherent differences generalizes this construction
to the Kleisli category of other models of \(\LL\), by means
of a \emph{distributive law}~\cite{PowerWatabane02}.

Let \(\DL\cL\Oc\) be the specialization to the $\Oc$ comonad 
of the \(2\)-category of comonads and 
distributive laws \cite[Sec.~5]{PowerWatabane02}.
Concretely, an object of this category is a pair \((F,d)\) where
\(F:\cL\to\cL\) is a functor and \(d:\Oc F\Natto F\Oc\) is a natural
transformation (called a \emph{distributive law}) 
subject to the two following commutations in \(\FUN\cL\cL\).
\begin{equation} \label{eq:oc-dl}
  \begin{tikzcd}
    \Oc F
    \ar[rr,"d"]\ar[rd,swap,"\Der{}F"]
    &[-2em]&[-2em]
    F\Oc
    \ar[ld,"F\Der{}"]
    \\
    &
    F
    &
  \end{tikzcd}
  \quad
  \begin{tikzcd}
    \Oc F
    \ar[rr,swap,"d"]\ar[d,"\Digg{}F"]
    &&
    F\Oc
    \ar[d,"F\Digg{}"]
    \\
    \Oc\Oc F
    \ar[r,"\Oc d"']
    &
    \Oc F\Oc
    \ar[r,"d\Oc"']
    &
    F\Oc\Oc
  \end{tikzcd}
\end{equation}
A morphism in \(\DL\cL\Oc((F,d),(G,e))\) is a natural transformation \(\phi:F\Natto G\) such
that the following diagram commutes in \(\FUN\cL\cL\)
\begin{equation} \label{eq:oc-dl-morphism}
  \begin{tikzcd}
    \Oc F\ar[r,"\Oc\phi"]\ar[d,swap,"d"]
    &
    \Oc G
    \ar[d,"e"]
    \\
    F\Oc
    \ar[r,"\phi\Oc"']
    &
    G\Oc
  \end{tikzcd}
\end{equation}
  We define in the same way a category 
  $\coDL \cL \Oc$ whose objects are pairs \((F,d^*)\) where
  \(F:\cL\to\cL\) is a functor and \(d^*:F \Oc\Natto \Oc F\)
  is a natural transformation subject to commutations similar 
  to those of~\cref{eq:oc-dl}.
  A morphism in $\coDL \cL \Oc ((F, d^*), (G, e^*))$ is a natural
  transformation $\phi : F \Natto G$ subject to a commutation similar
  to that of \cref{eq:oc-dl-morphism}.

The main purpose of distributive laws is to extend 
functors to Kleisli categories and to lift functors 
to Eilenberg-Moore categories. Our functor of 
\DDIFF{},
which is a functor in the Kleisli category of the exponential comonad, will be 
defined by such an extension.

\begin{definition} \label{def:extension} 
Let $F : \cL \arrow \cL$ be an endofunctor on $\cL$.
An \emph{extension} of $F$ to the Kleisli category of $\Oc$
is a functor $\Extfun{F} : \Kl\cL \arrow \Kl\cL$
such that the following diagram commutes. 
\[ \begin{tikzcd}
      \cL\ar[r,"F"]\ar[d,swap,"\Derfun"]
      &
      \cL\ar[d,"\Derfun"]
      \\
      \Kl\cL\ar[r,"{\Extfun{F}}"']
      &
      \Kl\cL
    \end{tikzcd}
\] 
In particular, $\Extfun{F} X = F X$.
We define in that way a category whose objects are pairs $(F, \Extfun{F})$
of an endofunctor and its extension, and whose morphisms from 
$(F, \Extfun{F})$ to $(G, \Extfun{G})$ are the natural transformations
$\alpha : F \Natto G$ such that the family of morphisms
$\beta_X = \Derfun(\alpha_X)$ is a natural transformation 
$\Extfun{F} \Natto \Extfun{G}$.
\end{definition}

\begin{definition} \label{def:lifting}
Let $F : \cL \arrow \cL$ be an endofunctor on $\cL$. 
A \emph{lifting} of $F$ to the Eilenberg-Moore category of $\Oc$
is a functor $\Liftfun{F} : \Em\cL \arrow \Em\cL$
such that the following diagram commutes.
\[ 
    \begin{tikzcd}
      \Em \cL\ar[r,"{\Liftfun{F}}"]\ar[d,swap,"\Emforg"]
      &
      \Em \cL\ar[d,"\Emforg"]
      \\
      \cL\ar[r,"F"']
      &
      \cL
    \end{tikzcd}
\]
In particular, for all $\Oc$-coalgebra $P$, 
$\Coalgca{\Liftfun F P} = F \Coalgca P$.
We define in that way a category whose objects are pairs $(F, \Liftfun{F})$
of an endofunctor and its lifting, and whose morphisms 
from $(F, \Extfun{F})$ to $(G, \Extfun{G})$ are the natural transformations
$\alpha : F \Natto G$ such that there exists 
a (necessarily unique) natural transformation 
$\beta : \Liftfun F \Natto \Liftfun G$
such that for all $\Oc$-coalgebra $P$, 
$\alpha_{\Coalgca P} = \Emforg(\beta_P)$.
\end{definition}

\begin{theorem} \cite{PowerWatabane02} \label{thm:extension-dl-iso}
  The categories $\EXTCAT{\Oc}$ and $\DL\cL\Oc$ are isomorphic. 
  The categories $\LIFTCAT{\Oc}$ and $\coDL \cL \Oc$ are isomorphic.
\end{theorem}

More precisely, for all endofunctor $F : \cL \arrow \cL$
there is a bijection between distributive laws 
$d : \Oc F \Natto F \Oc$ and extensions 
$\Extfun{F} : \Kl \cL \arrow \Kl \cL$ of $F$.
In particular, the extension $\Extfun F$ associated to 
a distributive law $(F, d)$ is given by 
\[ \Extfun F(X)=F(X) \qquad  
\Extfun F (f)=(Ff)\Compl d\in\Kl\cL(FX,FY) \text{ for all } 
f\in\Kl\cL(X,Y)=\cL(\Oc X,Y). \]
Furthermore, for all distributive laws $(F, d)$ and 
$(G, e)$
and their associated extensions $(F, \Extfun{F})$ and $(G, \Extfun{G})$,
$\alpha$ is a morphism from $(F, d)$ to $(G, e)$ in $\DL\cL\Oc$
if and only if it is a morphism from $(F, \Extfun{F})$ to 
$(G, \Extfun{G})$ in $\EXTCAT{\Oc}$. The isomorphism between 
$\LIFTCAT{\Oc}$ and $\coDL \cL \Oc$ can be described similarly.
We can thus build two chains of functors
\begin{equation} \label{eq:dl-chain-of-functors}
\begin{tikzcd}
	{\DL \cL \Oc} & {\EXTCAT \Oc} & {\FUN{\Kl \cL}{\Kl \cL}}
	\arrow["\sim", from=1-1, to=1-2]
	\arrow["{\Extforg{\Oc}}", from=1-2, to=1-3]
\end{tikzcd} 
\qquad 
\begin{tikzcd}
	{\coDL \cL \Oc} & {\LIFTCAT \Oc} & {\FUN{\Em \cL}{\Em \cL}}
	\arrow["\sim", from=1-1, to=1-2]
	\arrow["{\Liftforg{\Oc}}", from=1-2, to=1-3]
\end{tikzcd} 
\end{equation}
where in the above,
\begin{enumerate}
\item $\Extforg{\Oc} : \EXTCAT{\Oc} \arrow \FUN {\Kl \cL} {\Kl \cL}$
  sends $(F, \Extfun F)$ to $\Extfun F$, and
  $\alpha \in \EXTCAT \Oc((F, \Extfun F), (G, \Extfun G))$ to
  $\Derfun \alpha : \Extfun F \Natto \Extfun G$.
\item $\Liftforg{\Oc} : \LIFTCAT{\Oc} \arrow \FUN {\Em \cL} {\Em \cL}$
  sends $(F, \Liftfun F)$ to $\Liftfun F$, and
  $\alpha \in \LIFTCAT{\Oc}((F, \Liftfun F), (G, \Liftfun G))$ to the
  unique $\beta: \Liftfun F \Natto \Liftfun G$ such that
  $\Emforg \beta = \alpha$.
\end{enumerate}
In that way, distributive laws provide a convenient way to build
functors and their associated structure in the Kleisli and the
Eilenberg-Moore category of $\Oc$.
We now want to use this construction to build symmetric (co)monads in
$\Kl \cL$. To that end, it suffices to show that these chains of
functors are strict monoidal functors, for a monoidal structure that
we now define.

Observe that $\EXTCAT \Oc$ is a strict monoidal 
category, defining the product of $(F_1, \Extfun{F}_1)$ with 
$(F_2, \Extfun {F}_2)$ as $(F_1 F_2, \Extfun{F}_1 \Extfun{F}_2)$, and defining 
the product of natural transformations by their horizontal 
composition as in \cref{sec:monads}. Similarly, 
$\LIFTCAT \Oc$ is also a strict monoidal category.
Observe that $\Extforg{\Oc} : \EXTCAT{\Oc} \arrow \FUN {\Kl \cL} {\Kl \cL}$ 
and $\Liftforg{\Oc} : \LIFTCAT{\Oc} \arrow \FUN {\Em \cL} {\Em \cL}$
are trivially strict monoidal functors.

\begin{lemma}
  \label{lemma:dist-law-monoidal-struct}
  Given two objects \((F,d)\) and \((G,e)\) of \(\DL\cL\Oc\), we can
  define a new object 
  \[ \Tens{(F,d)}{(G,e)}=(FG, \phi) \text{ with } 
  \phi = \begin{tikzcd}
	{\Oc F G} & {F \Oc G} & {F G \Oc}
	\arrow["{d G}", from=1-1, to=1-2]
	\arrow["{F e}", from=1-2, to=1-3]
\end{tikzcd} \] 
  of \(\DL\cL\Oc\) and in that way we turn the category \(\DL\cL\Oc\)
  into a strict monoidal category.
  The unit of this monoidal category is \(\Sone=(\Idfun,\Id)\) where
  \(\Idfun:\cL\to\cL\) is the identity functor and \(\Id\) is the
  identity natural transformation on the functor \(\mathord\Oc\).
  The category $\coDL \cL \Oc$ is also strict monoidal, with a similar construction.
\end{lemma}

\begin{proposition}
The isomorphisms $\DL \cL \Oc \simeq \EXTCAT \Oc$ and 
$\coDL \cL \Oc \simeq \LIFTCAT \Oc$ are strict monoidal functors.
\end{proposition}

\begin{proof}
It suffices to check that the extension of $F G$ induced 
by the distributive law 
$\Tens{(F,d)}{(G,e)}$ given in \cref{lemma:dist-law-monoidal-struct}
is equal to $\Extfun F \Extfun G$, and similarly for liftings.
This is a standard result on compositionality of distributive 
laws.
\end{proof}

In that way, the chain of functors described in
\cref{eq:dl-chain-of-functors} yields a chain of strict monoidal
functors. Thus, by \cref{thm:symmonoid-strong-mondoidal-functor-image}
every symmetric monoid in $\DL \cL \Oc$ (which is the same thing as a
symmetric monad on $\cL$ equipped with a distributive law with regard
to $\Oc$) induces a symmetric monoid in $\FUN {\Kl \cL} {\Kl \cL}$
(which is the same thing as a symmetric monad on $\Kl \cL$).  Similar
results hold for symmetric comonoid and symmetric bimonoids.
The symmetric bimonad of coherent differences will be given by this construction.

\subsection{From monoids to monads}

As explained in the previous section, we want to describe our 
\DDIFF{} by a symmetric bimonad on $\Kl \cL$, which is constructed from 
a symmetric bimonoid on the category 
$\DL \cL \Oc$ of distributive laws.
Our goal in this section is to provide a recipe to build such bimonoid
as some kind of reader/writer effect compatible with the duplication
and erasure of \(\LL\).

More precisely, we first define in \cref{sec:monoid-to-monad} two 
strong monoidal functors 
\[
  \TensL : \cL \arrow \FUN{\cL}{\cL}
  \qquad \Sfung : \Op \cL \arrow \FUN \cL \cL
\]
The functor $\TensL$ maps an object $U$ of $\cL$ to the functor
$\Tens{\_}U$, and $\Sfung$ maps an object $U$ of $\cL$ to the functor
$\Limpl{U}{\_}$ (where $\Limpl{\_}{\_}$ is the internal hom functor).
Since these functors are strong, we know
by~\cref{thm:symmonoid-strong-mondoidal-functor-image} that they send
symmetric (co)monoids in $\cL$ and $\Op \cL$ to symmetric (co)monoids
in $\FUN \cL \cL$.
In other word, we can build (co)monads in $\cL$ out of
(co)monoids. This construction revisits the definition of the
well-known \emph{reader} and the \emph{writer} (co)monads.
Then, we prove in \cref{sec:coalg-to-dl} that these two functors
induce two strong monoidal functors
\[
  \CAcoDL \Oc : \Em \cL \arrow \coDL \cL \Oc 
  \qquad \CADL \Oc : \Op{(\Em \cL)} \arrow \DL{\cL}{\Oc}
\]
meaning that the reader/writer (co)monads associated to commutative
(co)monoids that also feature a $\Oc$-coalgebra structure are in fact
compatible with the exponential of \(\LL\).
In that way, we get two chains of strong monoidal functors
\begin{equation} \label{eq:full-chain-of-functors}
  \begin{tikzcd}[row sep = tiny]
    {\CCMON \cL} & {\Em \cL} & {\coDL \cL \Oc}
    & {\LIFTCAT \Oc} & {\FUN {\Em \cL}{\Em \cL}} \\
    {\Op {\CCMON{\cL}}} & {\Op{(\Em \cL)}} & {\DL \cL \Oc}
    & {\EXTCAT \Oc} & {\FUN {\Kl \cL}{\Kl \cL}}
	\arrow["\CmonEmFun", dashed, from=1-1, to=1-2]
	\arrow["{\CAcoDL \Oc}", from=1-2, to=1-3]
	\arrow["\sim", from=1-3, to=1-4]
	\arrow["{\Liftforg \Oc}", from=1-4, to=1-5]
  \arrow["{\Op \CmonEmFun}", dashed, from=2-1, to=2-2]
	\arrow["{\CADL \Oc}", from=2-2, to=2-3]
	\arrow["\sim", from=2-3, to=2-4]
	\arrow["{\Extforg \Oc}", from=2-4, to=2-5]
\end{tikzcd} 
\end{equation}
by precomposing $\CAcoDL \Oc$ and $\CADL \Oc$ with the strict monoidal
functor $\CmonEmFun$ defined in \cref{sec:Lafont} (only when $\cL$ is
Lafont, this is why the corresponding arrows are dashed), and
postcomposing it with the chain of strong monoidal functors described
in \cref{eq:dl-chain-of-functors}.
Thus, we can carve out from any commutative bimonoid in $\Em \cL$ (or
$\CCMON \cL$ when $\cL$ is Lafont) a reader/writer effect in $\Kl \cL$
and $\Em \cL$.
As we will see in \cref{sec:representable-diff}, many concrete models
of \(\LL\) feature a \DDIFF{} built in such a way.

The proof that $\TensL : \cL \arrow \FUN{\cL}{\cL} $ induces a functor
$\CAcoDL \Oc : \Em \cL \arrow \coDL \cL \Oc$ is quite straightforward,
but the proof that $\Sfung : \Op \cL \arrow \FUN \cL \cL $ induces a
functor $\CADL \Oc : \Op{(\Em \cL)} \arrow \DL \cL \Oc$ is not.
In order to prove this result, we use a construction called the
\emph{mate isomorphism} that we introduce in \cref{sec:mates}.

\subsubsection{Reader and writer monads}

\label{sec:monoid-to-monad}

\begin{figure}
  \begin{center}
  \begin{tabular}{c | c c c c}
    Object $U$ & & $\TensL(U) = U \ITens \_$
    & & $\Sfung(U) = \Limpl{U}{\_}$ \\ \\
    Monoid & & Writer monad & & Comonad \\ 
    Comonoid \qquad & & Comonad & & Reader monad \\ 
    Bimonoid & & Bimonad & & Bimonad
  \end{tabular}
\end{center}
\caption{(Co)monads induced by (co)monoids}
\label{fig:reader-writer}
\end{figure}

\paragraph{(Co)monad induced by a tensor product.}
Assume that $\cL$ is a symmetric monoidal 
category. 
The bifunctor $\Tens{\_}{\_}: \cL \times \cL \arrow \cL$ induces a
functor $\TensL : \cL \to \FUN \cL \cL$ which maps an object
$X \in \Obj{\cL}$ to the functor $\Tens{\_}{X}$, and a morphism
$f \in \cL(X, Y)$ to the natural transformation
$\Tens{\Id}{f} : \Tens{\_}{X} \Natto \Tens{\_}{Y}$.
Observe that the right unitor and the associator
\[
  \Rightu_X^{-1} \in \cL(X, \Tens X \Sone) \qquad \Assoc_{X, U, V}
  \in \cL\big((X \ITens U) \ITens V, X \ITens (U \ITens V) \big)
\]
are also natural transformations
$\Rightu^{-1} : \Idfun \Natto \TensL(\Sone)$ and
$\Assoc : \TensL(V) \TensL(U) \Natto \TensL (\Tens U V)$.
 The following proposition is an immediate consequence of
 coherence conditions in monoidal categories.
\begin{proposition}
The triple \((\TensL,\Rightu^{-1},\Assoc)\) is a strong monoidal
  functor from \(\cL\) to \(\FUN\cL\cL\).\footnote{
    Note that the monoidal structure on \(\FUN\cL\cL\) considered here is given by 
  composition of functors in the reverse order compared 
  to the structure described in~\cref{sec:monads}.
  Fortunately, taking such reverse monoidal structure does not change 
  the definition of monads, comonads, nor their symmetric structure.
  }
\end{proposition}

Since the functor \(\TensL:\cL\to\FUN\cL\cL\) is strong monoidal, we
know by \cref{thm:symmonoid-strong-mondoidal-functor-image} that, for
each commutative monoid
\(M = (\Monoidca{M}, \Monoidu{M}, \Monoidm{M})\), the functor
\(\TensL{\Monoidca M}\) has a structure of a symmetric monoid in the
monoidal category \(\FUN\cL\cL\), that is, it is a symmetric monad.
For instance, the monad multiplication
$\TensL{\Monoidca M} \TensL{\Monoidca M} \Natto \TensL{\Monoidca M}$
is given by the following natural transformation.
\[
  \begin{tikzcd}[column sep = large]
	{\Monoidca M \ITens (\Monoidca M \ITens X)} 
        & {X \ITens (\Monoidca M \ITens \Monoidca M)} & {X \ITens
          \Monoidca M} \arrow["{\Assoc_{X, \Monoidca M, \Monoidca
            M}}", from=1-1, to=1-2] \arrow["{X \ITens \Monoidm M}",
        from=1-2, to=1-3]
  \end{tikzcd}
\]
This monad is known as the \emph{writer monad}, also called the
\emph{monoid action monad}.
Similarly, every commutative comonoid $C$ 
induces a symmetric comonad 
$\TensL \Cmonoidca{C}$, and every symmetric bimonoid induces 
a symmetric bimonad.

\paragraph{(Co)monads induced by the internal hom.}
When \(\cL\) is a symmetric monoidal category (SMC) which is closed
(SMCC), we use \((\Limpl XY,\Evlin\in\cL(\Tens{\Limplp XY}X,Y))\) for
(a choice for) the internal hom object from \(X\) to \(Y\) and, given
\(f\in\cL(\Tens ZX,Y)\), we use \(\Curlin f\) for the corresponding
morphism in \(\cL(Z,\Limpl XY)\).
This operation is bijective, and we use \(\Curlininv\) for its inverse:
if \(f\in\cL(Z,\Limpl XY)\) then
\(\Curlininv f=\Evlin\Compl \Tensp fX\).

Let \(\cL\) be an SMCC.
The internal hom functor \(\Limpl{}{}:\Op\cL\times\cL\to\cL\) can be
curried as a functor \(\Sfung:\Op \cL \arrow \FUN\cL\cL\).
In other words \(\Sfung(U)=\Limpl U\_\) and similarly for morphisms.
There are natural isomorphisms \(\Sfungmonz:\Id\Natto\Sfung(\Sone)\)
and
\((\Sfungmont)^{U,V}:\Sfung(U)\Compfun\Sfung(V)\Natto\Sfung(\Tens
UV)\), which are defined in the obvious way, using the monoidal
closeness of \(\cL\).
For instance, given an object \(X\) of \(\cL\), the morphism
\((\Sfungmont)^{U,V}_X \) 
is
\[
  (\Sfungmont)^{U,V}_X = \Curlin(f)
  \in\cL\bigg(\Sfung(U)\Big(\Sfung(V)(X)\Big),\Sfung(\Tens
  UV)(X)\bigg)
\] 
where $f$ is given by the following composition
\[
  \begin{tikzcd}
    {(\Limpl{U}{(\Limpl{V}{X})}) \ITens (U \ITens V) } &
    {(\Limpl{U}{(\Limpl{V}{X})} \ITens U) \ITens V} & {(\Limpl V X)
      \ITens V} & X \arrow["{\Assoc^{-1}}", from=1-1, to=1-2]
    \arrow["{\Evlin \ITens V}", from=1-2, to=1-3] \arrow["\Evlin",
    from=1-3, to=1-4]
  \end{tikzcd}.
\]
The following result holds by standard SMCC computations.
\begin{proposition}
  The triple \((\Sfung,\Sfungmonz,\Sfungmont)\) is a strong monoidal
  functor from \(\Op\cL\) to \(\FUN\cL\cL\).
\end{proposition}

Since the functor \(\Sfung:\Op\cL\to\FUN\cL\cL\) is strong monoidal,
we know by \cref{thm:symmonoid-strong-mondoidal-functor-image} that,
for any object \(C = (\Cmonoidca{C}, \Cmonoidu{C}, \Cmonoidm{C})\) of
\(\CCMON\cL = \CMON{\Op \cL}\), the functor \(\Sfung{\Cmonoidca C}\)
has a structure of symmetric monad in the monoidal category
\(\FUN\cL\cL\).
For instance, the monad multiplication is given 
by the following natural transformation.
\[
  \begin{tikzcd}[column sep = large]
    {\Limpl{\Cmonoidca
        C}{\Limpl{\Cmonoidca C}{X}}} &[1em] {\Limpl{\big(\Cmonoidca C
        \ITens \Cmonoidca C \big)}{X}} & {\Limpl{\Cmonoidca C}{X}}
    \arrow["{(\Sfungmont)^{\Cmonoidca C, \Cmonoidca C}_{X}}",
    from=1-1, to=1-2] \arrow["{\Limpl{\Cmonoidm C}{X} }", from=1-2,
    to=1-3]
  \end{tikzcd}
\]

\begin{example}
  It is well known (see \cite[Prop.~16]{Mellies09} for instance) that
  every object $R$ in a cartesian category $\cC$ has a canonical
  commutative comonoid structure, with comultiplication given by the
  pairing $\Pair{\Id_R}{\Id_R} \in \cC(R, R \times R)$.
  Assuming that $\cC$ is cartesian closed, the symmetric monad
  structure $\Sfung R$ associated to this comonoid structure
  corresponds to the well known \emph{reader monad}.
\end{example} 

Similarly, every symmetric commutative monoid $M$ induces a symmetric
comonad $\Sfung \Monoidca{M}$, and every symmetric bimonoid induces a
symmetric bimonad.

\subsubsection{The mate isomorphism.}

\label{sec:mates}
Observe that for all object $U \in \Obj{\cL}$, the functors
$\TensL(U)$ and $\Sfung(U)$ are adjoints, through the adjunction
$U \ITens \_ \dashv \Limpl{U}{\_}$ of SMCCs.
This adjunction induces a bijection between natural transformations
involving $\TensL(U)$ and natural transformations involving
$\Sfung(U)$ called the \emph{mate isomorphism}~\cite{Kelly74} that we
now briefly explain.
Consider two adjunctions
\[
  \begin{tikzcd}
    \cC & \cD \arrow[""{name=0, anchor=center, inner sep=0}, "L",
    curve={height=-12pt}, from=1-1, to=1-2] \arrow[""{name=1,
      anchor=center, inner sep=0}, "R", curve={height=-12pt},
    from=1-2, to=1-1] \arrow["\dashv"{anchor=center, rotate=-90},
    draw=none, from=0, to=1]
  \end{tikzcd}
  \qquad
  \begin{tikzcd}
    {\cC'} & {\cD'} \arrow[""{name=0, anchor=center,
      inner sep=0}, "{L'}", curve={height=-12pt}, from=1-1, to=1-2]
    \arrow[""{name=1, anchor=center, inner sep=0}, "{R'}",
    curve={height=-12pt}, from=1-2, to=1-1]
    \arrow["\dashv"{anchor=center, rotate=-90}, draw=none, from=0,
    to=1]
  \end{tikzcd}.
\]
Then for all functors $H : \cC \arrow \cC'$ and $K : \cD \arrow \cD'$,
there is a \emph{mate isomorphism}~\cite[Prop.~2.1]{Kelly74} which is
a bijection between natural transformations $\Lmate : L' H \Natto K L$
and natural transformations $\Rmate: K R \Natto R' K$.
Furthermore, this bijection is compositional,
see~\cite[Prop.~2.2]{Kelly74} for a definition of compositionality
that we will not fully spell out here.
It turns out that the structures on $\TensL$ and 
$\Sfung$ we described in the previous section are mates.
For example, for all monoid  
$M = (\Monoidca{M}, \Monoidu M, \Monoidm M)$, the natural 
transformations 
\[
  \Id \ITens \Monoidm M : \TensL(\Monoidca M \ITens \Monoidca M)
  \Natto \TensL \Monoidca M
  \quad \text{ and }
  \quad \Limpl{\Monoidm
    M}{\Id} : \Sfung \Monoidca M \Natto \Sfung(\Monoidca M \ITens
  \Monoidca M)
\]
are mates, it suffices to unfold the definition of mates
\cite[Prop.~2.1]{Kelly74} for the adjunctions
$\TensL(\Monoidca M \ITens \Monoidca M) \dashv \Sfung(\Monoidca M
\ITens \Monoidca M)$ and
$\TensL \Monoidca M \dashv \Sfung \Monoidca M$, taking
$H = K = \smash{\Idfun_{\cL}}$.
Similarly, the strong monoidal structures 
\[
  \Assoc^{-1} : \TensL(U \ITens V) \Natto \TensL(V) \TensL(U)
  \quad \text{ and }
  \quad
  \Sfungmont : \Sfung(U) \Sfung(V) \Natto \Sfung(U \ITens V)
\]
are mates, also by definition. We use here the adjunction
$\TensL(V) \TensL(U) \dashv \Sfung(U) \Sfung(V)$ given by composition
of the adjunctions $\TensL(V) \dashv \Sfung(V)$ and
$\TensL(U) \dashv \Sfung(U)$.
Finally, the induced monad multiplication on $\TensL\Monoidca{M}$ 
and the induced comonad comultiplication on $\Sfung \Monoidca{M}$
\[
  \begin{tikzcd}[column sep = large, row sep = tiny]
	{\Monoidca M \ITens (\Monoidca M \ITens X)} 
    & {X \ITens (\Monoidca M \ITens \Monoidca M)}
    &[3em] {X \ITens \Monoidca M} \\
  {\Limpl{\Monoidca M}{X}}
    & {\Limpl{\big(\Monoidca M \ITens \Monoidca M \big)}{X} \ \ \ }
    & {\Limpl{\Monoidca M}{\Limpl{\Monoidca M}{X}}} 
	\arrow["{\Assoc_{X, \Monoidca M, \Monoidca M}}", from=1-1, to=1-2]
	\arrow["{X \ITens \Monoidm M}", from=1-2, to=1-3]
	\arrow["{\big((\Sfungmont)^{\Monoidca M, \Monoidca M}_{X}\big)^{-1}}", from=2-2, to=2-3]
	\arrow["{\Limpl{\Monoidm M}{X} }", from=2-1, to=2-2]
  \end{tikzcd}
\]
are mates, by compositionality of the mate
construction~\cite[Prop.~2.2]{Kelly74}.
Observe indeed that these two natural transformations are compositions
of natural transformations that are mates of each other.
A similar story holds for all the structure on $\TensL$ and $\Sfung$
we described in this section, as summarized in \cref{fig:mates}.

\begin{figure}
  \begin{center}
  \begin{tabular}{c | c c c c}
    & & Structure on $\TensL$
    & & Structure on $\Sfung$ \\ \\ \hline \\
    Strong 
      & & $\Rightu : \TensL(\Sone) \Natto \Idfun$ 
      & & $\Sfungmonz : \Idfun \Natto \Sfung(\Sone)$ \\ 
    monoidality
      & & $\Assoc^{-1} : \TensL(U \ITens V) \Natto \TensL(V) \TensL(U)$
      & & $\Sfungmont : \Sfung(U) \Sfung(V) \Natto \Sfung(U \ITens V)$ \\ \\
    Monoid 
      & & $\Id \ITens \Monoidu M : \TensL(\Sone) \Natto \TensL \Monoidca M$
      & & $\Limpl{\Monoidu M}{\Id} : \Sfung \Monoidca M \Natto \Sfung(\Sone)$ \\
    $(\Monoidca{M}, \Monoidu{M}, \Monoidm{M})$
      & & $\Id \ITens \Monoidm M : \TensL(\Monoidca M \ITens \Monoidca M)
        \Natto \TensL \Monoidca M$
      & & $\Limpl{\Monoidm M}{\Id} : \Sfung \Monoidca M
        \Natto \Sfung(\Monoidca M \ITens \Monoidca M)$ \\  
      & & \textbf{Induced monad} & & \textbf{Induced comonad} \\ \\
    Comonoid
      & & $\Id \ITens \Cmonoidu C : \TensL(\Cmonoidca C) \Natto \TensL(\Sone)$
      & & $\Limpl{\Cmonoidu C}{\Id} : \Sfung(\Sone) \Natto \Sfung \Cmonoidca C$\\
    $(\Cmonoidca{C}, \Cmonoidu{C}, \Cmonoidm{C})$
      & &$\Id \ITens \Cmonoidm C : \TensL \Cmonoidca C 
        \Natto \TensL(\Cmonoidca C \ITens \Cmonoidca C)$
      & & $\Limpl{\Cmonoidm C}{\Id} : \Sfung(\Cmonoidca C \ITens \Cmonoidca C)
        \Natto \Sfung \Cmonoidca C$ \\
      & & \textbf{Induced comonad} & & \textbf{Induced monad} \\ \\ 
    Symmetry $\Sym$ 
      & & $\Id \ITens \Sym : \TensL(X \ITens Y) \Natto \TensL(Y \ITens X)$
      & & $\Limpl{\Sym}{\Id} : \Sfung(Y \ITens X) \Natto \Sfung(X \ITens Y)$ \\
      & & \textbf{Induced symmetry} & & \textbf{Induced symmetry} \\
  \end{tabular}
\end{center}
\caption{Mate isomorphism. All natural transformations are mates by
  definition of mates~\cite[Prop.~2.1]{Kelly74}, except structures in
  \textbf{bold} that are mates by compositionality of
  mates~\cite[Prop.~2.2]{Kelly74}.}
\label{fig:mates}
\end{figure}

\subsubsection{From coalgebras to distributive laws}
\label{sec:coalg-to-dl}

We now prove that $\TensL : \cL \arrow \FUN{\cL}{\cL} $ extends to a
functor $\CAcoDL \Oc : \Em \cL \arrow \coDL \cL \Oc$ and that
$\Sfung : \Op \cL \arrow \FUN \cL \cL $ extends to a functor
$\CADL \Oc : \Op{(\Em \cL)} \arrow \DL \cL \Oc$.
As we will see, the proof for $\CAcoDL \Oc$ is rather straightforward,
and the proof for $\CADL \Oc$ relies on the mate isomorphism.
For all object \(P = (\Coalgca{P}, \Coalgst{P})\) of \(\Em\cL\),
we can define a natural transformation
$\Coalgdlmate P : \TensL(\Coalgca P) \Oc \Natto \Oc \TensL (\Coalgca P)$ as 
\begin{equation} \label{eq:Sdiff-elem-left-mate}
  \begin{tikzcd}
    {\Oc X \ITens \Coalgca{P}}
    &[1em] {\Oc X \ITens \Oc  \Coalgca{P}}
    & {\Oc(X \ITens  \Coalgca{P})}
    \arrow["{X \ITens \Coalgst P}", from=1-1, to=1-2]
    \arrow["\ExpMon2", from=1-2, to=1-3]
\end{tikzcd}.
\end{equation}

\begin{proposition} 
  \label{prop:coalg-to-distr-law-mate}
  The following assertions hold.
  \begin{enumerate}
  \item For all objects \(P\) of \(\Em\cL\), the natural
    transformation \(\Coalgdlmate P\) is a distributive law; in other
    words the pair $(\TensL \Coalgca P, \Coalgdlmate{P})$ is an object
    of $\coDL \cL \Oc$.
  \item For each morphism of $\Oc$-coalgebras \(f\in\Em\cL(P,Q)\),
    \(\TensL f : \TensL{\Coalgca P} \Natto \TensL{\Coalgca Q}\) is a
    morphism of distributive law (that is, a morphism in
    $\coDL{\cL}{\Oc}$) from $(\TensL{\Coalgca P}, \Coalgdlmate{P})$ to
    $(\TensL{\Coalgca Q}, \Coalgdlmate{Q})$
  \item The strong monoidal structure of $\TensL$,
    $\Rightu^{-1} : \Idfun \Natto \TensL(\Sone)$ and
    $\Assoc : \TensL \Coalgca{Q} \TensL \Coalgca{P} \Natto \TensL
    \Coalgca{\Tens P Q}$ are also morphisms of distributive laws,
    respectively from $(\Idfun, \Id)$ to $(\TensL 1, \Coalgdlmate{1})$
    and from
    $(\TensL \Coalgca{Q}, \Coalgdlmate{Q}) \ITens (\TensL \Coalgca{Q},
    \Coalgdlmate{Q})$ to
    $(\TensL \Coalgca{P \ITens Q}, \Coalgdlmate{P \ITens Q})$.
\end{enumerate}
In that way, we have defined a strong monoidal functor
\(\CAcoDL\Oc:\Em\cL \arrow \coDL\cL\Oc\).
\end{proposition}

\begin{proof}
  The commutations can be checked by straightforward computations.
  Alternatively, we can use the isomorphism between distributive laws
  and liftings of \cref{thm:extension-dl-iso} to prove the result
  without explicit computations.
  \begin{enumerate}
  \item The monoidal structure of $\Em \cL$ induced by the lax
    monoidal structure $(\ExpMon{0}, \ExpMon{2})$ of $\Oc$ by
    \cref{thm:sm-extension} is a lifting (see \cref{def:lifting}) of
    the monoidal structure of $\cL$\footnote{In fact, lax monoidal
      structures on a comonad can be understood as distributive laws}.
    Then, the natural transformation $\Coalgdlmate{P}$ is the
    distributive law associated to the lifting
    $\_ \ITens P : \Em \cL \arrow \Em \cL$ of the functor
    $\_ \ITens \Coalgca P : \cL \arrow \cL$.
  \item Similarly, for every morphism of $\Oc$-coalgebras
    $f \in \Em \cL(P, Q)$, the natural transformation
    $\Id \ITens f : \_ \ITens P \Natto \_ \ITens Q$ in $\Em \cL$ is a
    lifting of the natural transformation
    $\Id \ITens f : \_ \ITens \Coalgca{P} \Natto \_ \ITens
    \Coalgca{Q}$ in $\cL$.
    Thus, $\Id \ITens f$ is a morphism of distributive laws from
    $(\TensL{\Coalgca P}, \Coalgdlmate{P})$ to
    $(\TensL{\Coalgca Q}, \Coalgdlmate{Q})$.
  \item As explained in our first item, the monoidal structure of
    $\Em \cL$ is a lifting of the monoidal structure of $\cL$. As
    such, $\Rightu^{-1} : \Idfun \Natto \TensL(\Sone)$ and
    $\Assoc : \TensL \Coalgca{Q} \TensL \Coalgca{P} \Natto \TensL
    \Coalgca{\Tens P Q}$ have liftings in $\Em \cL$ and thus are
    morphisms of distributive laws by \cref{thm:extension-dl-iso}.
\end{enumerate}
This concludes the proof.
\end{proof}

For each object \(P = (\Coalgca{P}, \Coalgst{P})\) of \(\Em\cL\), we
can define a natural transformation
$\Coalgdl P : \Oc \Sfung(\Coalgca P) \Natto \Sfung(\Coalgca P) \Oc$
that is mate of the natural transformation
$\Coalgdlmate P : \TensL(\Coalgca P) \Oc \Natto \Oc \TensL (\Coalgca
P)$.
This natural transformation is defined as the currying of
\begin{equation} \label{eq:Sdiff-elem-right-mate}
\begin{tikzcd}
  {\Oc (\Limpl{\Coalgca{P}}{X}) \ITens \Coalgca{P}}
  &[1em] {\Oc((\Limpl{\Coalgca{P}} X) \ITens \Coalgca{P}) }
  & { \Oc X}
	\arrow["{\Coalgdlmate{P}_{\Limpl{\Coalgca{P}} X}}", from=1-1, to=1-2]
	\arrow["{\Oc \Evlin}", from=1-2, to=1-3]
\end{tikzcd}.
\end{equation}
The following result is an immediate consequence of
\cref{prop:coalg-to-distr-law-mate} and of the compositionality of
mate.  We refer the reader to \cite[Sec.~7.3]{EhrhardWalch25} for a
detailed explanation on how the mate isomorphism relates distributive
laws and their morphisms.

\begin{proposition}\label{prop:coalg-to-distr-law} 
  The following assertions hold.
  \begin{enumerate}
  \item For all objects \(P\) of \(\Em\cL\), the natural
    transformation \(\Coalgdl P\) is a distributive law; in other
    words the pair $(\Sfung \Coalgca P, \Coalgdl{P})$ is an object of
    $\DL \cL \Oc$.
  \item For each morphism of $\Oc$-coalgebras \(f\in\Em\cL(P,Q)\),
    \(\Sfung f : \Sfung{\Coalgca Q} \Natto \Sfung{\Coalgca P}\) is
    also a morphism of distributive law (that is, a morphism in
    $\DL{\cL}{\Oc}$) from $(\Sfung{\Coalgca Q}, \Coalgdl{Q})$ to
    $(\Sfung{\Coalgca Q}, \Coalgdl{Q})$.
  \item For each morphism of the strong monoidal structure of $\Sfung$
    are also morphisms of distributive laws.
  \end{enumerate}
  In that way, we have defined a strong monoidal contravariant functor
  \(\CADL\Oc:\Em\cL \Tocontra \DL\cL\Oc\).
\end{proposition}

To sum up, every coalgebra $P$ induces a distributive law
$\Coalgdlmate P : \TensL(\Coalgca P) \Oc \Natto \Oc \TensL (\Coalgca
P)$, which by the mate isomorphism induces a distributive law
$\Coalgdl P : \Oc \Sfung(\Coalgca P) \Natto \Sfung(\Coalgca P) \Oc$.
In that way we get the chain of strong monoidal functors of
\cref{eq:full-chain-of-functors}.

\begin{remark}
Conversely, every distributive law 
$\Coalgdl P : \Oc \Sfung(\Coalgca P) \Natto \Sfung(\Coalgca P) \Oc$ induces by the same 
mate isomorphism a distributive law  
$\Coalgdlmate P : \TensL(\Coalgca P) \Oc \Natto \Oc \TensL (\Coalgca P)$, which in turn 
induces a morphism
\[ \Coalgst P = \begin{tikzcd}
    {\Coalgca{P}}
    & {1 \ITens \Coalgca{P}}
    &[1em] {\Oc 1 \ITens \Coalgca{P}}
    & {\Oc(1 \ITens \Coalgca{P})}
    & {\Oc \Coalgca{P}}
	\arrow["{\Leftu^{-1}_{\Coalgca P}}", from=1-1, to=1-2]
	\arrow["{\ExpMon0 \ITens \Id_P}", from=1-2, to=1-3]
	\arrow["{\Coalgdlmate P_1}", from=1-3, to=1-4]
	\arrow["\Oc \Leftu_{\Coalgca P}", from=1-4, to=1-5]
\end{tikzcd}, \]
and it is rather straightforward to check that $(\Coalgca{P}, \Coalgst{P})$ is a $\Oc$-coalgebra.
In other word, for all object $\Coalgca{P}$ there is in fact a
bijection between $\Oc$-coalgebra structures on $\Coalgst{P}$,
distributive laws
$\Coalgdlmate P : \TensL(\Coalgca P) \Oc \Natto \Oc \TensL (\Coalgca
P)$, and distributive laws
$\Coalgdl P : \Oc \Sfung(\Coalgca P) \Natto \Sfung(\Coalgca P) \Oc$.
\end{remark}

\section{Categories enriched over \Pmonoid{}s}
\label{sec:summability}

Differential categories~\cite{BluteCockettSeely06} provide an axiomatization
of differentiation in categories enriched in commutative monoids.
As explained in our introduction, we want to consider instead a theory
of differentiation on categories enriched in \emph{partial}
commutative monoids, that we define in this section.

\subsection{\PMonoid{}s and \Pimonoid{}s}

Let $\Pmuniv$ be an infinite countable set of labels, chosen once and
for all. We can take $\Pmuniv = \Nat$, but no particular properties or
structures of $\Nat$ will be used.
Let us fix a set of index sets
$\Pmindex \in \Set{\Part{\Pmuniv}, \Partfin{\Pmuniv}}$.
Choosing $\Pmindex = \Part{\Pmuniv}$ will allow us to describe partial 
monoids with at most countable sums, and taking 
$\Pmindex = \Partfin{\Pmuniv}$ will allow us to describe partial monoids 
with finite sums.

For all set $M$ and index set $I \in \Pmindex$,
an $I$-\textbf{indexed family} of 
elements of $M$ is a function $\Vect x : I \arrow M$, written 
$\Vect x = (x_i)_{i \in I}$.
We consider a set $M$  with a partial function 
$\Sigma$ called the \emph{sum} that maps 
indexed families on $M$ to elements of $M$.
An indexed family 
$(x_i)_{i \in I}$ is \emph{summable} if it is in the domain of $\Sigma$,
and we write its image as $\sum_{i \in I} x_i$.

\begin{remark}
Observe that in this definition we do not consider families indexed by arbitrary 
sets, but only by sets $I \in \Pmindex$.
In that way, we avoid a quantification on the proper class of all
at most countable sets, as opposed to the traditional definition of partial
monoids in literature~\cite{ArbibManes86,Haghverdi01,Hines13}.
This choice is by no mean restrictive: we will 
recover in \cref{def:summability-arbitrary-index} a definition of summability on 
families indexed by arbitrary sets.
\end{remark}

\begin{definition} \label{notation:sum-def}
  For any  expression
  $e$ and $e'$ involving partial sums, we write 
  $e \sumsub e'$ (or $e' \sumsubinv e$) if whenever $e$ 
  is defined, then $e'$ is also defined 
  and $e = e'$
  We write  $e \sumiff e'$ if $e \sumsub e'$ and $e \sumsubinv e'$.
\end{definition}

\begin{definition}
    \label{def:pcm} 
    \labeltext{(Unary)}{ax:unary}
  \labeltext{(PA)}{ax:pa}
  The tuple $(M, \Pmindex, \Sigma)$ is a \emph{Strong Partial Commutative Monoid (\Pmonoid{})}
 if $M$ is non-empty and if the
  sum $\Sigma$ satisfies the following axioms. \begin{itemize}
  \item \ref{ax:unary} Unary sum axiom: every singleton family
    \((x_i)_{i\in I}\in M^I\) where \(I=\Eset j\) is summable, with
    sum \(x_j\).
  \item \ref{ax:pa} Partition associativity axiom: let
    $I \in \Pmindex$, let $(x_i)_{i \in I}\in M^I$ be an indexed
    family, let $K\in\Pmindex$, and let $\{I_k\}_{k \in K}$ a
    partition of $I$ (we use the word ``partition'' in a slightly non
    standard way: the $I_k$'s can be empty).
    Then
    \[
      \sum_{i \in I} x_i \sumiff \sum_{k \in K} \left(\sum_{i \in
          I_k} x_i \right)\,.
    \]
  \end{itemize}
An \Pmonoid{} $(M, \Pfmindex, \Sigma)$ is called \emph{finitary}, and an 
\Pmonoid{} $(M, \Pimindex, \Sigma)$ is called \emph{infinitary}. 
\end{definition}

We use the terminology \Pmonoid{} instead of the terminology
$\Sigma$-monoid from~\cite{Haghverdi01} as a reference to the PCMs
of~\cite{Hines13}, which are a generalization of \Pmonoid{}s to a
weaker version of \ref{ax:pa}.

\begin{example}
  Every monoid is a \Pfmonoid{}. The set of \Mpositive{} real numbers
  $\Realp$ equipped with its usual summability is an \Pimonoid{}.
  The set $[0, 1]$ in which a family $(x_i)_{i \in I}$ is summable if
  its sum (in $\Realp$) is defined and belongs to $[0, 1]$ is both a
  \Pfmonoid{} (if we consider only finite indexed families) and an
  \Pimonoid{} (if we consider all at most countable indexed families).
\end{example}

Let us fix a \Pmonoid{} $(M, \Pmindex, \Sigma)$.
Observe that the neutral element $0$ is not given as part of the data,
but can be defined from the axioms by taking the sum over the empty
family, which is always defined by \cref{prop:subfamily-summable}
below.

\begin{lemma}
  \label{prop:subfamily-summable}
  Consider $I, I' \in \Pmindex$ such that $I' \subseteq I$. Then for
  all summable family $(x_i)_{i \in I}$, the family $(x_i)_{i \in I'}$
  is also summable.
  In particular, the empty indexed family is summable.
\end{lemma}

\begin{proof}
  The first statement follows from~\ref{ax:pa} on the family
  $(x_i)_{i \in I}$, used on the partition $\Set{I \setminus I', I'}$.
  In particular, let $x \in M$ ($M$ is non-empty by assumption), let
  \(I=\Eset i\in\Pmindex\) be a singleton set and \(\Vect x\in M^I\)
  be defined by \(x_i=x\).
  Then by \ref{ax:unary}, the singleton family $\Vect x$ is summable,
  and since $\emptyset \subseteq I$, the empty family is summable.
\end{proof}

This element $0$ behaves as a neutral element.
\begin{proposition}
  \label{prop:zero-neutral}
  Let $\Vect x = (x_i)_{i \in I}$ be an
  indexed family, and $I' \in \Pmindex$ be such that
  $\Supp{\Vect x} \subseteq I' \subseteq I$, where
  $\Supp{\Vect x} = \Set{i \in I \St x_i \neq 0}$.
  Then $\sum_{i \in I'} x_i \sumiff \sum_{i \in I} x_i$.
\end{proposition} 

\begin{proof}
  For all $i \in I$, define $J_i = \{i\}$ if $i \in I'$, and
  $J_i = \emptyset$ otherwise.
  Then, $\{J_i\}_{i \in I}$ is a partition of $I'$.  Thus, by
  \ref{ax:pa},
  $\sum_{i \in I'} x_i \sumiff \sum_{i \in I} \sum_{j \in J_i} x_j$.
  But for all $i \in I$, $\sum_{j \in J_i} x_j = x_i$, by
  \ref{ax:unary} if $x_i \neq 0$, and by definition of $0$ if
  $x_i = 0$.
  Thus, $ \sum_{i \in I'} x_a \sumiff \sum_{i \in I} x_i$. 
\end{proof}

\begin{remark}
  For $\Pmindex = \Pfmindex$, we can equivalently define an
  \Pmonoid{} as a tuple $(M, +, 0)$ such that for all $a, b, c \in M$,
  (1) $a + b \sumiff b + a$, 
  (2) \(0+a\) is defined for all \(a\in\Seta\) and \(0+a=a\), and
  (3) \((a+b)+c \sumiff a+(b+c)\).
  We prefer \cref{def:pcm} for dealing with both the finitary and infinitary 
  cases in a unified way.
\end{remark}

A crucial feature of \Pmonoid{}s is that the choice of the index set
is completely irrelevant.
Given two sets $I, J$ and a bijection \(\phi:I \arrow J\), we define a
$J$-indexed family $\Famact\phi \Vect x= (y_j)_{j \in J}$ by
$y_{j}= x_{\Inv\phi(j)}$.

\begin{proposition}
  \label{prop:reindexing-bijection}
  For all index sets $I,J \in \Pmindex$, for all bijection
  $\phi : I \arrow J$ and for all $I$-indexed family $\Vect x$,
  $\Famact{\phi} \Vect x$ is summable if and only if $\Vect x$ is
  summable, and the two sums are equal when defined.
\end{proposition}

\begin{proof}
  It suffices to apply \ref{ax:unary} and \ref{ax:pa} on the partition
  $\Set{\phi(i)}_{i \in I}$ of $J$.
\end{proof}

Because the choice of the index sets is irrelevant, we can now define
a notion of partial sums on arbitrary index sets, and not only on
elements of $\Pmindex$.

\begin{definition}
  \label{def:summability-arbitrary-index}
  Let $K$ be a set (finite if $\Pmindex = \Pfmindex$, at most
  countable if $\Pmindex = \Pimindex$), that does not necessarily
  belong to $\Pmindex$.  An indexed family $\Vect x = (x_k)_{k \in K}$
  is \emph{summable} if there exists $I \in \Pmindex$ and a bijection
  $\phi : K \arrow I$ such that
  $(y_i)_{i \in I} = \Famact \phi \Vect x$ is summable, and we set
  \[ \sum_{k \in K} x_k = \sum_{i \in I} (\Famact \phi \Vect x)_i \]
\end{definition}

This definition makes sense because the summability of $\Vect x$ and
its sum do not depend on the choice of the bijection.
Indeed, for all bijections $\phi_1 : K \arrow I_1$ and
$\phi_2 : K \arrow I_2$, $\phi_2 \circ \phi_1^{-1} : I_1 \arrow I_2$
is a bijection, and we can easily check that reindexing is functorial
so that
$\Famact {(\phi_2 \circ \phi_1^{-1})} {\Famact{\phi_1}} \Vect x =
\Famact{\phi_2}{\Vect x}$.
It follows by \cref{prop:reindexing-bijection} that
\[
  \sum_{i \in I} (\Famact{\phi_1} \Vect x)_i \sumiff \sum_{j \in J}
  (\Famact{\phi_2} \Vect x)_j\,.
\]
This extension of summability to arbitrary indexed families is also
conservative: for all $K \in \Pmindex$, a family $(x_k)_{k \in K}$ is
summable in the sense of \cref{def:summability-arbitrary-index} if and
only if it is summable in the \Pmonoid{}, and the two sums agree,
thanks to \cref{prop:reindexing-bijection}.

\begin{remark} \label{rem:pcm-generalization}
Observe that the properties \ref{ax:unary} and \ref{ax:pa}
also hold on this generalized definition of summability.
We recover in that way the traditional definitions of partial monoids
in the literature~\cite{ArbibManes86,Haghverdi01,Hines13}.
In particular, the proof of 
\cref{prop:zero-neutral,prop:reindexing} directly carry to indexed 
set that do not necessarily belong to $\Pmindex$.
\end{remark}

We now prove that summable families can also be reindexed by
injections.
Given two sets $I$ and $J$, an injection \(\phi:I \Injection J\) and a
family $\Vect x= (x_i)_{i \in I}$, we define a $J$-indexed family
$\Famact\phi \Vect x= (y_j)_{j \in J}$ by $y_{\phi(i)}= x_i$, and
$y_j = 0$ if $j\notin \phi(I)$.
This definition generalizes the reindexing by a bijection defined
in \cref{prop:reindexing-bijection}.

\begin{proposition}
  \label{prop:reindexing}
  For all index sets $I,J$ (not necessarily in $\Pmindex$), for all
  injection $\phi : I \Injection J$ and for all $I$-indexed family
  $\Vect x$, $\Famact{\phi} \Vect x$ is summable if and only if
  $\Vect x$ is summable, and the two sums are equal.
\end{proposition}

\begin{proof}
  Let  $\Vect y = \Famact{\phi} \Vect x = (y_j)_{j \in J}$.
  By definition of $\Famact \phi\Vect x$, we have
  $\Supp{\Vect y} \subseteq \phi(I) \subseteq B$, so by
  \cref{prop:zero-neutral,rem:pcm-generalization},
  $\sum_{j \in J} y_j \sumiff \sum_{j \in \phi(I)} y_j$.
  Now observe that the injection $\phi : I \arrow J$ induces a bijection 
  $\psi : I \arrow \phi(I)$, and that 
  $(y_j)_{j \in \phi(I)} = \Famact \psi \Vect x$. 
  Thus, by \cref{prop:reindexing-bijection,rem:pcm-generalization} we conclude that
  $\sum_{j \in \phi(I)} y_j \sumiff  \sum_{i \in I} x_i$.
\end{proof}

\subsection{Enrichment over \Pmonoid{}s}

\label{sec:spm-enrichement}

We define in this section a notion of categories enriched over \Pmonoid{}.
We can define a category $\SPM$ whose objects are the \Pfmonoid{}s and
whose morphisms in $\SPM(M, N)$ are given by the maps $f : M \arrow N$
such that for all finite family $(x_i)_{i \in I}$ of $M$,
\[
  f \left(\sum_{i \in I} x_i \right) \sumsub \sum_{i \in I} f(x_i)\,.
\]
Such a function will be called \emph{linear} or \emph{additive}.
Similarly, let $\SPIM$ be the category of \Pimonoid{}s and sum
preserving maps, for which we use the same terminology.
We will write $\SPMS$ to designate any of these categories
interchangeably.
We adopt the convention that the index sets $I, J, \ldots$ we consider
in this section are either finite (if we consider enrichment over
\Pfmonoid{}s) or at most countable (if we consider enrichment over
\Pimonoid{}s).

\begin{definition}
  Let $A_1, A_2, B$ be \Pmonoid{}s.
  A function $f : A_1 \times A_2 \arrow B$ is \emph{bilinear} if for
  all $a_1 \in A_1$ and $a_2 \in A_2$, the functions
  $f(a_1, \_) : A_2 \arrow B$ and $f(\_, a_2) : A_1 \arrow B$ are
  morphisms in
  $\SPMS$. 
\end{definition}

\begin{remark}
  Observe that a function is bilinear if and only for all indexed
  families $(a_1^i)_{i \in I}$ of $A_1$ and all indexed families
  $(a_2^j)_{j \in J}$ of $A_2$,
  \[
    f \Big(\sum_{i \in I} a_1^i, \sum_{j \in J} a_2^j \Big) \sumsub
    \sum_{i \in I, j \in J} f(a_1^i, a_2^j)\,.
  \]
  This is an immediate consequence of \ref{ax:pa}.
\end{remark}

We now define what is a category enriched over \Pmonoid{}s.

\begin{definition}
  \label{def:spm-enriched}
  A \emph{\Pmenriched{}} category is a category $\cL$ in which each
  homset of \(\cL\) is equipped with a \Pmonoid{} structure, and such
  that composition is bilinear: for all indexed families
  \((f_j\in\cL(X,Y))_{j\in J}\) and \((g_i\in\cL(X,Y))_{i\in I}\),
  \[
    \Big(\sum_{i \in I} g_i \Big) \Compl f \sumsub \sum_{i \in I} (g_i
    \Compl f) \qquad g \Compl \Big(\sum_{j \in J} f_j \Big) \sumsub
    \sum_{j \in J} (g \Compl f_j)\,.
  \]
\end{definition}
Observe that \Pmenriched{} categories have \emph{zero morphisms}: if
$0_{X, Y}$ is the zero of the \Pmonoid{} $\cL(X, Y)$ (that is, the sum
over the empty family), then it follows from the definition of
enrichment that for all $f \in \cL(X', X)$ and $g \in \cL(Y, Y')$,
$0_{X, Y} \Compl f = 0_{X', Y}$ and $g \Compl 0_{X, Y} = 0_{X, Y'}$.

As explained in~\cite{Hoshino12}, our definition of \Pmenriched{}
categories is the specialization to the category $\SPMS$ of the
standard notion of categories enriched over a monoidal
category~\cite{Kelly05}.
We now describe how the enrichment over \Pmonoid{}s should behave
with respect to the symmetric monoidal structure of $\cL$, the closure
and the cartesian product, whenever they exist.
These definitions are also the specialization to $\SPMS$ of standard
notions in enriched category theory.

\begin{definition}
  \label{def:pm-enriched-monoidal}
  A \emph{symmetric monoidal \Pmenriched{}} category is a
  \Pmenriched{} category equipped with a symmetric monoidal structure
  such that the tensor $\ITens$ is bilinear.
  That is, for all all indexed families
  \((f_j\in\cL(X_1,Y_1))_{j\in J}\) and
  \((g_i\in\cL(X_2,Y_2))_{i\in I}\),
  \[
    \Big(\sum_{i \in I} g_i \Big) \ITens f \sumsub \sum_{i \in I} (g_i
    \ITens f) \qquad g \ITens \Big(\sum_{j \in J} f_j \Big) \sumsub
    \sum_{j \in J} (g \ITens f_j)\,.
  \]
\end{definition}

The definition above is a bit redundant: the two equations above 
can be derived from each other by symmetry of the tensor product.

\begin{definition}
  \label{def:spm-enriched-smcc}
  \emph{A symmetric monoidal closed \Pmenriched{}} category is a
  symmetric monoidal \Pmenriched{} category which is closed with
  respect to its tensor product $\ITens$, and such that for all object
  $X$, $\Limpl X \_$ is linear.
  That is, for all indexed family $(f_i)_{i \in I}$ in $\cL(Y, Z)$,
  \[
    \Limpl X {\Big( \sum_{i \in I} f_i \Big)} \sumsub \sum_{i \in I}
    (\Limpl X {f_i})\,.
  \]
\end{definition} 

\begin{remark} \label{rem:cur-additive} %
  Observe that for all $f \in \cL(\Tens X Y, Z)$ and
  $g \in \cL(Y, Z)$,
  \[
    \Curlin(f) = (\Limpl Y f) \Compl \Curlin(\Id_{\Tens X Y}) \quad
    \Limpl X g = \Curlin(g \Compl \Evlin)
  \]
  so by bilinearity of composition, \cref{def:spm-enriched-smcc} is
  equivalent the additivity of $\Curlin$, which intuitively expresses
  that the sum of functions is defined pointwise:
  $\Curlin(\sum_{i \in I} f_i) \sumsub \sum_{i \in I} \Curlin(f_i)$.
  This alternative definition is the point of view adopted
  in~\cite{EhrhardWalch25}.
  We chose \cref{def:spm-enriched-smcc} instead because it is more in
  line with the standard theory of enriched categories and because it
  simplifies some proofs we do in \cref{sec:summability-tensor}.
\end{remark}

This definition can be further simplified. We can prove that the two morphisms given 
in~\cref{def:spm-enriched-smcc} are always equal, whenever they are defined. 
So we get the following characterization of symmetric monoidal closed \Pmenriched{}
categories.

\begin{proposition}
  Let $\cL$ be a symmetric monoidal \Pmenriched{} category that is closed.
  Then the following are equivalent:
  \begin{enumerate}
  \item $\cL$ is a symmetric monoidal closed \Pmenriched{} category;
  \item for all indexed families $(f_i)_{i \in I}$ in $\cL(Y, Z)$, if
    $(f_i)_{i \in I}$ is summable then $(\Limpl X f_i)_{i \in I}$ is
    summable in $\cL(\Limpl X Y, \Limpl X Y)$.
  \end{enumerate}
\end{proposition}

\begin{proof}
  The forward implication is obvious, since $(2)$ is just a sub-property of 
  $(1)$. For the reverse implication,
  observe that for all families
  $(g_i)_{i \in I}$ of $\cL(X, \Limpl Y Z)$,
  by \Pmonoid{} enrichment,
  \[
    \Curlin^{-1} \Big(\sum_{i \in I} g_i \Big) = \Evlin_{Y, Z} \Big(
    \big(\sum_{i \in I} g_i \big) \ITens Y \Big) \sumsub \sum_{i \in
      I} \big(\Evlin_{Y, Z} \Compl (g_i \ITens Y)\big) = \sum_{i \in
      I} \Curlin^{-1}(g_i)\,.
  \]
  Now, let $(f_i)_{i \in I}$ be a summable family of $\cL(X, Y)$. Then 
  $(\Limpl X {f_i})_{i \in I}$ is also summable by $(2)$ and
  \[
    \Curlin^{-1} \Big(\Limpl X {\big(\sum_{i \in I} f_i \big)}\Big) =
    \Curlin^{-1} \Big(\Curlin\Big(\big(\sum_{i \in I} f_i \big) \Compl
    \Evlin \Big) \Big) = \big(\sum_{i \in I} f_i \big) \Compl \Evlin
    \sumsub \sum_{i \in I} (f_i \Compl \Evlin)
  \]
  \[
    \Curlin^{-1} \Big(\sum_{i \in I} (\Limpl X {f_i})\Big) \sumsub
    \sum_{i \in I} \Curlin^{-1} (\Limpl X {f_i}) = \sum_{i \in I}
    \Curlin^{-1}(\Curlin(f_i \Compl \Evlin)) = \sum_{i \in I} (f_i
    \Compl \Evlin)\,.
  \]
  By bijectivity of $\Curlin^{-1}$, we deduce that
  $\Limpl X {\Big( \sum_{i \in I} f_i \Big)} 
  = \sum_{i \in I} (\Limpl X {f_i})$ and thus $(1)$ holds.
\end{proof}
In particular, we always have $\Limplp X 0 = 0$ and $\Curlin(0) = 0$, even 
if we do not assume \cref{def:spm-enriched-smcc}.

\begin{definition}
  \label{def:spm-enriched-cartesian}
  A \emph{cartesian \Pmenriched{}} category is a \Pmenriched{}
  category with finite products (written $\with$, see \cref{sec:ll}
  for the notations) such that for each finite set $L$, the functor
  $\Bwith_{l \in L} : \cL^L \arrow \cL$ is linear.
  That is, for each collection of indexed families
  $(f_i^l \in \cL(X_l, Y_l))_{i \in I}$ where $l$ ranges over $L$,
  \[ \Bwith_{l \in L} \big(\sum_{i \in I} f_i^l \big) \sumsub 
  \sum_{i \in I} \big(\Bwith_{l \in L} f_i^l) \]
 \end{definition}

 \begin{remark}
   \label{rem:pairing-additive}
   A counterpart of \cref{rem:cur-additive} also holds for
   products, since for all families of morphisms
   $(f_l \in \cL(X, Y_l))_{l \in L}$ and
   ${g_l \in \cL(X_l, Y_l)}_{l \in L}$,
   \[
     \Tuple{f_l}_{l \in L} = (\Bwith_{l \in L} f_l) \Compl
     \Tuple{\Id_X}_{l \in L} \quad \text{and} \quad \Bwith_{l \in L}
     g_l = \Tuple{g_l \Compl \Proj l}_{l \in L}
   \]
   so by bilinearity of composition, \cref{def:spm-enriched-cartesian}
   is equivalent to the additivity of tupling, expressing that the
   the sums in product objects are defined componentwise:
   $\Tuple{\sum_{i \in I} f_l^i}_{l \in L} \sumsub \sum_{i \in I}
   \Tuple{f_l^i}_{l \in L}$.
   This alternative definition is the point of view adopted
   in~\cite{EhrhardWalch25}.
   We chose \cref{def:spm-enriched-smcc} instead because it is more in
   line with the standard theory of enriched categories and because it
   simplifies some proofs in \cref{sec:summability-cartesian}.
\end{remark}

This definition can be further simplified as above for monoidal
closeness.
 \begin{proposition}
   Let $\cL$ be a \Pmenriched{} category that has finite products. The
   the following are equivalent.
   \begin{enumerate}
   \item $\cL$ is a cartesian \Pmenriched{} category.
   \item For all collection of indexed families $(f_i^l)_{i \in I}$ in
     $\cL(X_l, Y_l)$, if $(f_i^l)_{i \in I}$ is summable for all
     $l \in L$, then $(\Bwith_{l \in L} f_i^l)_{i \in I}$ is summable
     in $\cL(\Bwith_{l \in L} X_l, \Bwith_{l \in L} X_l)$.
   \end{enumerate}
 \end{proposition}

 \begin{proof}
  The forward implication is obvious, since $(2)$ is just a sub-property of 
  $(1)$. For the reverse implication,
  assume that $(f_i^l)_{i \in I}$ is a summable family of $\cL(X_l, Y_l)$ 
  for all $l \in L$.
  Then by our assumption, $(\Bwith_{l \in L} f_i^l)_{i \in I}$ is also
  summable and
  \[
    \Proj l \Compl \Big(\Bwith_{l \in L} \big(\sum_{i \in I} f_i^l
    \big) \Big) = \big(\sum_{i \in I} f_i^l \big) \Compl \Proj l
    \sumsub \sum_{i \in I} (f_i^l \Compl \Proj l)
  \]
  \[
    \Proj l \Compl \sum_{i \in I} \big(\Bwith_{l \in L} f_i^l)
    \sumsub \sum_{i \in I} \big(\Proj l \Compl \Bwith_{l \in L} f_i^l
    \big) = \sum_{i \in I} (f_i^l \Compl \Proj l)
  \]
  so by joint monicity of the $\Proj l$, the morphisms
  $\Bwith_{l \in L} \big(\sum_{i \in I} f_i^l \big) 
  = \sum_{i \in I} \big(\Bwith_{l \in L} f_i^l)$ and $(1)$ holds.
\end{proof}
 
In particular, we always have $\Bwith_{l \in L} 0 = 0$ and
$\Tuple{0}_{l \in L} = 0$, even if we do not
assume~\cref{def:spm-enriched-cartesian}.

\section{Summability structures} %
\label{sec:summability-structure}

Let \(\Mdeg\) be a \Pfmonoid{} (in the sense of \cref{def:pcm}) whose
elements are called \emph{degrees}, with neutral element \(0\) and
addition \(+\).
We assume that \(\Mdeg\) has at least two elements.
We assume moreover that \(\Mdeg\) is \Mpositive{} in the following
sense: if \(i,j\in\Mdeg\) satisfy \(i+j=0\) then \(i=0\) (and hence
\(j=0\) as well).
Notice that since \(\Mdeg\not=\Set0\), \(\Mdeg\) is not a group.
We consider two cases:
\begin{itemize}
\item The set \(\Mdeg\) is finite
\item The set \(\Mdeg\) is infinite and countable.
\end{itemize}

\begin{convention}
  We use the terminology \Dpmonoid{} for finitary SPCMs when
  \(\Mdeg\) is finite, and for infinitary SPCMs when \(\Mdeg\) is
  infinite.
\end{convention}

Here are a few examples of such monoids \(\Mdeg\) which induce
interesting ``differentiation'' theories, but there is a whole
universe of ``monoid-based differential theories'' to be explored.
Let \(n\in\Natc\).
\begin{itemize}
\item \(\Mdeg=\Eset{i\in\Nat\St i<n}\) with \(i+j\) defined as the
  ordinary sum of \(i\) and \(j\) if this sum is in \(\Mdeg \), and
  undefined otherwise;
  the associated differential theory will be that of a functorial
  generalization of Taylor expansion%
  \footnote{Corresponding to the Faà di Bruno formula} %
  of order \(n\).
\item \(\Mdeg=\Eset{i\in\Nat\St i\leq n}\) with \(i+j\) defined as the
  min of \(n\) and of the ordinary sum of \(i\) and \(j\);
  the associated differential theory will be that of a functorial
  generalization of Taylor expansion of order \(n - 1\), with a remainder
  if \(n<\infty\).
\item \(\Mdeg=\Eset{i\in\Nat\St i<n}\) with \(i+j=\max(i,j)\);
  the associated differential theory is a calculus of finite
  differences instead of derivatives.
\item \(\Mdeg=\Mfin{\Nat}\) the set of all finite multisets on $\Nat$,
  with $m + p$ defined as the pointwise sum of multisets; this case
  has to be worked out and corresponds probably to a theory of
  multi-variate Taylor expansion.
\end{itemize}

\begin{remark}
  As a set, \(\Mdeg\) can be infinite, but it 
  will always be a \Pfmonoid{} and never a \Pimonoid.
\end{remark}

\begin{convention}
  \label{conv:cardinality-indexing-sets}
  In the whole paper we stick to the following convention: the
  indexing sets denoted with letters \(I\), \(J\) and \(K\), with
  possibly subscripts or superscripts, are tacitly assumed to be
  finite when \(\Mdeg\) is finite, and at most countable when $\Mdeg$
  is infinite and countable.
\end{convention}

\begin{definition}
  \label{conv:cardinality-partial-mon-enrichement}
  A \emph{\Pmcat{}} is a category \(\cL\) in which 
  each homset is equipped
  with a \Dpmonoid{} structure.
  Furthermore, the zeros $0_{X, Y}$ of
  the \Dpmonoid{}s $\cL(X, Y)$ are assumed to be zero 
  morphisms: for all $f \in \cL(X', X)$ and $g \in \cL(Y, Y')$, 
  $0_{X, Y} \Compl f = 0_{X', Y}$ and  $g \Compl 0_{X, Y} = 0_{X, Y'}$.
\end{definition}

Observe that \Pmcats{} are \emph{not} assumed to be 
\Pmenriched{} (\cref{def:spm-enriched}).
Our goal here is to describe \Pmcats{} where summability is described
by a functor $\Sfun$ called a \emph{summability structure} of which
the enrichment will be a byproduct.

\begin{definition}
  \label{def:pre-summ-struct}
  A \emph{\PSUMBLTY{} structure} on a \Pmcat{} \(\cL\)
  consists in a tuple %
  \((\Sfun,(\Sproj i)_{i\in\Mdeg},\Ssum)\) where
  \begin{enumerate}
  \item \ref{ax:S-fun} \labeltext{(S-fun)}{ax:S-fun} %
    \(\Sfun:\cL\to\cL\) is a functor;
  \item \ref{ax:S-nat} \labeltext{(S-nat)}{ax:S-nat} %
    the \(\Sproj i\)'s as well as \(\Ssum\) are natural
    transformations \(\Sfun\Natto\Id\);
  \item \ref{ax:S-monic} \labeltext{(S-monic)}{ax:S-monic} %
    for all objects \(X\) of \(\cL\), the morphisms
    \((\Sproj{i,X}\in\cL(\Sfun X,X))_{i\in\Mdeg}\) are jointly
    monic. That is, for all $f,g \in \cL(X, \Sfun Y)$, if
    $\Sproj{i, Y} \Compl f = \Sproj{i, Y} \Compl g$ for all
    $i \in \Mdeg$ then $f = g$.
  \end{enumerate}
\end{definition}

\begin{definition}
  \label{def:summ-struct}
  A \emph{\SUMBLTY{} structure} on a \Pmcat{} \(\cL\) is a \PSUMBLTY{}
  structure \((\Sfun,(\Sproj i)_{i\in\Mdeg},\Ssum)\) satisfying
  \begin{enumerate}
  \item \ref{ax:S-witness} \labeltext{(S-witness)}{ax:S-witness}  a family
    \((f_i\in\cL(X,Y))_{i\in\Mdeg}\) is summable iff there is a
    \(h\in\cL(X,\Sfun Y)\) (necessarily unique) such that
    \(\forall i\in\Mdeg\,f_i=\Sproj i\Compl h\) and then
    \(\sum_{i\in\Mdeg}f_i=\Ssum\Compl h\);
  \item \ref{ax:S-sum} \labeltext{(S-sum)}{ax:S-sum}  and a family
    \((h_i\in\cL(X,\Sfun Y))_{i\in I}\) is summable if the
    family \((\Ssum\Compl h_i)_{i\in I}\) is summable.
  \end{enumerate}
  When a \(\Mdeg\)-indexed family \((f_i\in\cL(X,Y))_{i\in\Mdeg}\) is
  summable we use \(\Stuple{f_i}_{i\in\Mdeg}\) for the unique element
  of \(\cL(X,\Sfun Y)\) such that
  \(\forall i\in\Mdeg\ \Sproj i\Compl \Stuple{f_i}_{i\in\Mdeg}=f_i\)
  and call this morphism the \emph{witness} of the summability of the
  \(f_i\)'s.
  Such a tuple \((\cL,\Sfun,(\Sproj i)_{i\in\Mdeg},\Ssum)\) is called
  a \emph{\SUMCAT}.
\end{definition}

Defining summability by an endofunctor is standard in algebraic
programming semantics; this is the case in partially additive 
categories~\cite{ArbibManes86,ChoJacosetAl15}.
But in this approach the counterpart of our object $\Sfun X$ is given
by a coproduct.
As explained in \cite[Part.~2]{TassonWalch26},
our summability structures are more general, and enjoy an epi-mono
factorization with the coproduct and the product of $\cL$ (whenever
they exist).
\[  
  \begin{tikzcd}
    {\bigoplus_{i \in \Mdeg} X } & {\Sfun X} & {\Bwith_{i
        \in \Mdeg} X} \arrow["e", from=1-1, to=1-2] \arrow["{m}",
    from=1-2, to=1-3]
  \end{tikzcd}.
\]
When sums are total, this factorization collapses into an equality and
thus the category has \emph{biproducts}, either finite if $\Mdeg$ is
finite ---~this is a standard result, see Chapter 8 of
\cite{Maclane71}~--- or infinite if $\Mdeg$ is
infinite~\cite[Prop.~237]{EhrhardWalch25}.
In the next lemmas we assume that \(\cL\) is a \Pmcat{} \(\cL\)
equipped with a summability structure
\((\Sfun,(\Sproj i)_{i\in\Mdeg},\Ssum)\).

\begin{lemma}
  \label{lemma:sproj-summable}
  \label{lemma:Sfun-functor-stuple}
  For each \(f\in\cL(X,Y)\), the family
  \((f\Compl\Sproj i)_{i\in\Mdeg}\) is summable,
  \(\Sfun f=\Stuple{f\Compl\Sproj i}_{i\in\Mdeg}\) and
  \(\sum_{i\in\Mdeg}(f\Compl\Sproj i)=f\Compl\Ssum\).
  In particular \((\Sproj i)_{i\in\Mdeg}\) is summable,
  \(\Stuple{\Sproj i}_{\in\Mdeg}=\Id\) and
  \(\sum_{i\in\Mdeg}\Sproj i=\Ssum\).
\end{lemma}
\begin{proof}
  Direct consequence of the naturality of $\Sproj i$ and of \(\Ssum\),
  and of the definition of witnesses.
\end{proof}

\begin{lemma} %
  \label{lemma:summability-compl}
  If the family \((f_i\in\cL(X,Y))_{i\in\Mdeg}\) is summable and if %
  \(g\in\cL(U,X)\) and \(h\in\cL(Y,V)\) then the family %
  \((h\Compl f_i\Compl g)_{i\in\Mdeg}\) is summable with %
  \(\Stuple{h\Compl f_i\Compl g}_{i\in\Mdeg} =\Sfun
  h\Compl\Stuple{f_i}_{i\in\Mdeg}\Compl g\) and %
  \(\sum_{i\in\Mdeg}h\Compl f_i\Compl g= h\Compl
  (\sum_{i\in\Mdeg}f_i)\Compl g\).
\end{lemma}
\begin{proof}
  By naturality of $\Sproj i$ and $\Ssum$.
\end{proof}

\begin{theorem}
  Let \(\cL\) be a \Pmcat{} and let
  \((\Sfun,(\Sproj i)_{i\in\Mdeg},\Ssum)\) and
  \((\Sfun',(\Sproj i')_{i\in\Mdeg},\Ssum')\) be \SUMBLTY{} structures
  on \(\cL\).
  Then for each object \(X\) of \(\cL\) there is exactly one
  morphism %
  \(\phi_X\in\cL(\Sfun X,\Sfun' X)\) such that
  \((\Sproj i'\Compl\phi_X=\Sproj i)_{i\in\Mdeg}\).
  This morphism also satisfies \(\Ssum'\Compl\phi_X=\Ssum\).
  The family of morphisms \(\phi\) defined in that way is a natural
  isomorphism from \(\Sfun\) to \(\Sfun'\).
\end{theorem}
\begin{proof}
  Since the family \((\Sproj i\in\cL(\Sfun X,X))_{i\in\Mdeg}\) is
  summable, we have a unique %
  \(\phi_X=\Stuple{\Sproj i}'_{i\in\Mdeg}\in\cL(\Sfun X,\Sfun' X)\)
  characterized by
  \((\Sproj i'\Compl\phi_X=\Sproj i)_{i\in\Mdeg}\).
  Similarly we have
  \(\phi'_X=\Stuple{\Sproj i'}_{i\in\Mdeg}\in\cL(\Sfun' X,\Sfun X)\)
  uniquely characterized by
  \((\Sproj i\Compl\phi'_X=\Sproj i')_{i\in\Mdeg}\).
  Then for each \(i\in\Mdeg\) we have
  \(\Sproj i'\Compl\phi'_X\Compl\phi_X=\Sproj i\Compl\phi_X=\Sproj
  i'\) and hence \(\phi'_X\Compl\phi_X=\Id_{\Sfun X}\), and similarly
  for \(\phi_X\Compl\phi'_X\).
  Hence \(\phi_X\) is an iso whose inverse is \(\phi'_X\).
  The equation \(\Ssum'\Compl\phi_X=\Ssum\) results from
  \(\Ssum'=\sum_{i\in\Mdeg}\Sproj i'\) and
  \cref{lemma:summability-compl}, and similarly for \(\Ssum\).
  Last we prove naturality: let \(f\in\cL(X,Y)\), we have %
  \(\phi'_Y\Compl(\Sfun' f)=\Stuple{\Sproj
    i'}_{i\in\Mdeg}\Compl(\Sfun' f)=\Stuple{\Sproj i'\Compl(\Sfun'
    f)}_{i\in\Mdeg}=\Stuple{f\Compl\Sproj i'}_{i\in\Mdeg}=(\Sfun
  f)\Compl\Stuple{\Sproj i'}_{i\in\Mdeg}=(\Sfun f)\Compl\phi'_X\) by
  \cref{lemma:summability-compl}, which proves that \(\phi'\) is a
  natural transformation, the same holds of course for \(\phi\).
\end{proof}

\begin{remark}
  \label{rk:unique-summability-structure}
  This easy but important observation shows that, up to a unique
  natural iso, a \Pmcat{} \(\cL\) has at most one \SUMBLTY{}
  structure.
  So having a \SUMBLTY{} structure is actually a \emph{property} for
  a \Pmcat{}.
  This property implies in particular that \(\cL\) is enriched in
  \Pmonoid{}s, as we prove in \cref{sec:ss-induces-enrichement}.
\end{remark}

\subsection{\texorpdfstring{\SUMBLTY{}}{Summability} structures and
  \texorpdfstring{\Dpmonoid{}}{SPCM} enrichment}

\label{sec:ss-induces-enrichement}

We prove in this section the following theorem.

\begin{theorem}
  \label{thm:ss-induces-enrichement}
  Assume that $\cL$ is a \Pmcat{} equipped with a \SUMBLTY{}
  structure.
  Then $\cL$ is enriched over \Dpmonoid{}s, in the sense of
  \cref{def:spm-enriched}.
\end{theorem}

We have already seen in \cref{lemma:summability-compl} that the
existence of a \SUMBLTY{} structure implies some kind of bilinearity
of composition.
But this lemma deals only with $\Mdeg$-indexed families,
we generalize it to all formulas as follows.
\begin{enumerate}
\item First, we prove that for all $n \in \Nat$,
  $(\Sfun^n, (\Sproj{i_1} \Compl \cdots \Compl \Sproj{i_n})_{(i_1,
    \ldots, i_n) \in \Mdeg^n}, \Ssum^n)$ is a summability structure whose
  index set is given by the \Pmonoid{} $\Mdeg^n$ (with partial sums
  defined componentwise).
  In that way, \cref{lemma:summability-compl} above gives enrichment
  for all $\Mdeg^n$-indexed families.
\item For each indexing set \(I\) one can find \(n\in\Nat\) and an
  injection from \(I\) to \(\Mdeg^n\).
  In that way the general case reduces to $\Mdeg^n$-indexed families
  by~\cref{prop:reindexing}.
\end{enumerate}

We start with the first step. We introduce the following notations for convenience.

\begin{definition}
  If \(\Vect i\in\Mdeg^n\), we 
  define
  \(\Sproj{\Vect i}=\Sproj{i_1}\Compl\cdots\Sproj{i_n}
  \in\cL(\Sfun^n X,X)\). We define
  $\Ssum^n
  = \overbrace{\Ssum \Compl \cdots \Compl \Ssum}^{n\times}
  \in \cL(\Sfun^n X, X)$.
\end{definition}

\begin{proposition}
  \label{prop:iterated-summability-structure}
  The tuple $(\Sfun^n, (\Sproj{\Vect i})_{\Vect i \in \Mdeg^n}, \Ssum^n)$
  is a \SUMBLTY[\Mdeg^n] structure.
\end{proposition}

\begin{proof}
  \Proofcase{}
  We can easily prove by induction on $n$ that the morphisms
  \((\Sproj{\Vect i}\in\cL(\Sfun^n X,X))_{\Vect i\in\Mdeg^n}\) are
  jointly monic, and thus
  $(\Sfun^n, (\Sproj{\Vect i})_{\Vect i \in \Mdeg^n}, \Ssum^n)$ is a
  \PSUMBLTY[\Mdeg^n] structure.

  \Proofcase{}
  We prove that \((\Sfun^n, (\Sproj{\Vect i})_{\Vect i\in\Mdeg^n}, \Ssum^n)\) 
  satisfies \ref{ax:S-sum}, by induction on $n$.
  The case $n=1$ is satisfied by~\ref{ax:S-sum} 
  in the \SUMBLTY{} structure $(\Sfun, (\Sproj i)_{i \in \Mdeg}, \Ssum)$.
  For the inductive step, consider a family
  $(f_i\in\cL(X, \Sfun^{n+1} Y))_{i \in I}$ such that
  $(\Ssum \Compl \Ssum^{n} \Compl f_i\in\cL(X,Y))_{i \in I}$ is summable.
  Then, by~\ref{ax:S-sum} 
  in the \SUMBLTY{} structure $(\Sfun, (\Sproj i)_{i \in \Mdeg}, \Ssum)$,
  $(\Ssum^{n} \Compl f_i)_{i \in I}$ is summable in
  $\cL(X, \Sfun Y)$.
  Furthermore, we know by inductive hypothesis that
  \((\Sfun^n, (\Sproj{\Vect i})_{\Vect i\in\Mdeg^n}, \Ssum^n)\)
  satisfies \ref{ax:S-sum} and hence $(f_i)_{i \in I}$ is summable.

  \Proofcase{}
  Finally, we prove that
  \((\Sfun^n, (\Sproj{\Vect i})_{\Vect i\in\Mdeg^n}, \Ssum^n)\)
  satisfies \ref{ax:S-witness} by induction on $n$.
  The base case $n = 1$ holds by \ref{ax:S-witness} satisfied by the
  \SUMBLTY{} structure $(\Sfun, (\Sproj i)_{i \in \Mdeg}, \Ssum)$.
  For the inductive step, let
  $(f_{i_1, \ldots, i_n, i})_{(i_1, \ldots, i_n, i) \in \Mdeg^{n+1}}$
  be an indexed family of $\cL(X, Y)$.
  First, assume that
  $(f_{i_1, \ldots, i_n, i})_{(i_1, \ldots, i_n, i) \in \Mdeg^{n+1}}$
  is summable.
  Then, for each $i \in \Mdeg$,
  $(f_{i_1, \ldots, i_n, i})_{(i_1, \ldots, i_n) \in \Mdeg^n}$ is summable
  by \cref{prop:subfamily-summable}, so by inductive hypothesis there
  is a morphism $f_i \in \cL(X, \Sfun^n Y)$ such that
  \[
    \Sproj {i_1, \dots, i_n} \Compl f_i = f_{i_1, \ldots, i_n, i}
    \qquad \Ssum^n \Compl f_i = \sum_{(i_1, \ldots, i_n) \in \Mdeg^n}
    f_{i_1, \ldots, i_n, i}\,.
  \]
  By \ref{ax:pa}, $(\Ssum^n \Compl f_i)_{i \in \Mdeg}$ is summable and
  hence $(f_i)_{i \in \Mdeg}$ is summable since
  \((\Sfun^n, (\Sproj{\Vect i})_{\Vect i\in\Mdeg^n}, \Ssum^n)\)
  satisfies \ref{ax:S-sum}, as already proved.
  Since
  $\Sproj {i_1, \ldots, i_n} \Compl \Sproj i \Compl \Stuple{f_j}_{j
    \in \Mdeg} = \Sproj {i_1, \ldots, i_n} \Compl f_i = f_{i_1,
    \ldots, i_n, i}$, the morphism $\Stuple{f_i}_{i \in \Mdeg}$ is a
  witness for the family
  $(f_{i_1, \ldots, i_n, i})_{(i_1, \ldots, i_n, i) \in \Mdeg^{n+1}}$.

  Conversely, assume that such a witness $f \in \cL(X, \Sfun^{n+1} Y)$
  exists.
  Then by inductive hypothesis, for each $i \in \Mdeg$, the family
  $(f_{i_1, \ldots, i_n, i})_{(i_1, \ldots, i_n) \in \Mdeg^n}$ is
  summable with witness $\Sproj i \Compl f$ and sum
  \[
    \sum_{(i_1, \ldots, i_n) \in \Mdeg^n} f_{i_1, \ldots, i_n, i} =
    \Ssum^n \Compl \Sproj i \Compl f = \Sproj i \Compl \Sfun(\Ssum^n)
    \Compl f\,.
  \]
  Thus by the \ref{ax:S-witness} property in the \SUMBLTY{} structure
  $(\Sfun, (\Sproj i)_{i \in \Mdeg}, \Ssum)$, we have that the family
  $(\sum_{(i_1, \ldots, i_n) \in \Mdeg^n} f_{i_1, \ldots, i_n, i})_{i
    \in \Mdeg}$ is summable with witness $\Sfun(\Ssum^n) \Compl f$ and
  with sum
  $\Ssum \Compl \Sfun(\Ssum^n) \Compl f = \Ssum^n \Compl \Ssum \Compl
  f = \Ssum^{n+1} \Compl f$.
  Hence, by \ref{ax:pa},
  $(f_{i_1, \ldots, i_n, i})_{i_1, \ldots, i_n, i \in \Mdeg}$ is
  summable with sum $\Ssum^{n+1} \Compl f$.
\end{proof}

We now prove \cref{thm:ss-induces-enrichement} using our reindexing
technique.

\begin{proof}[Proof of \cref{thm:ss-induces-enrichement}]
  Let $(f_i)_{i \in I}$ be an indexed family of $\cL(X, Y)$, and let
  \(g\in\cL(U,X)\) and \(h\in\cL(Y,V)\).
  Whatever be the cardinality of $\Mdeg$, by our assumption on the
  index sets considered in this section, there is $n \in \Nat_{>0}$
  such that $\Card{I} \leq \Card{\Mdeg}^{n}$ (remember that
  \(\Card\Mdeg\geq 2\)).
  Let $\phi : I \Injection \Mdeg^{n}$ be an injection.
  Consider the family
  $(f'_{i_1, \ldots, i_{n}})_{(i_1, \ldots, i_{n}) \in \Mdeg^n} =
  \Famact \phi (f_i)_{i \in I}$.
  We know that $0$ is a zero morphism by definition of \Pmcats{},
  and thus:
  \[
    \forall (i_1, \ldots, i_n) \in \Mdeg^n \setminus \phi(I), \ g
    \Compl f'_{i_1, \ldots, i_n} \Compl h = g \Compl 0 \Compl h = 0\,.
  \]
  Therefore
  $(g \Compl f'_{i_1, \ldots, i_n} \Compl h)_{(i_1, \ldots, i_n) \in
    \Mdeg^n} = \Famact \phi (g \Compl f_i \Compl h)_{i\in I}$.
  We conclude by \cref{prop:reindexing} and by
  \cref{lemma:summability-compl} applied to the
  \SUMBLTY[\Mdeg^n] structure $\Sfun^n$ that
  \[
    g \Compl \Big(\sum_{i \in I} f_i \Big) \Compl h \sumiff g \Compl
    \Big(\Smashspace{\sum_{(i_1, \ldots, i_{n}) \in \Mdeg^n}} f'_{i_1,
      \ldots, i_n} \Big) \Compl h \sumsub \Smashspace{\sum_{(i_1,
        \ldots, i_n) \in \Mdeg^n}} (g \Compl f'_{i_1, \ldots, i_n} \Compl
    h) \sumiff \sum_{i \in I} (g \Compl f_i \Compl h).
  \]
  Taking $g = \Id_{Y}$ implies that composition is left additive, and
  taking $h = \Id_X$ implies that composition is right additive.
  So composition is bilinear and hence $\cL$ is \Pmenriched{}.
\end{proof}

\subsection{The monad (and comonad) structure of the summability functor}
\label{sec:ssum-monad-structure}

We assume that the \Pmcat{} \(\cL\) is equipped with a
\SUMBLTY{} structure and we use the notations introduced previously.
We prove in this section that $\Sfun$ has a canonical symmetric 
monad structure, a canonical symmetric 
comonad structure, and that when $\Mdeg$ is total these two structures 
form a symmetric bimonad (see \cref{def:sym-monad}).
Our coherent calculus of differences will be given in
\cref{sec:diff-structure} by a distributive law that extends this
bimonad to the Kleisli category $\Kl \cL$.
First, we prove that naturality in $\cL$ is preserved by sums and
witnesses.

\begin{proposition}
  \label{prop:summability-naturality}
  The sums and witnesses of natural transformations are natural
  transformations.
\end{proposition}
More precisely, let $\cL'$ be a category and $F, G : \cL' \arrow \cL$ 
be two functors. Let $I$ be an index set, and for 
each $i \in I$, let $\alpha^i : F \Natto G$ be a 
natural transformation.
We say that $(\alpha^i)_{i \in I}$ is 
summable if for each object $X$ of $\cL'$, 
$(\alpha^i_X)_{i \in I}$ is summable. 
Then, the family of morphisms 
$\sum_{i \in I} \alpha^i$ given by
$(\sum_{i \in I} \alpha^i)_X 
= \sum_{i \in I} \alpha^i_X \in \cL(F X, G X)$
is a natural transformation $F \Natto G$.
Furthermore, if $I = \Mdeg$, then the family
of morphisms $\Stuple{\alpha^i}_{i \in \Mdeg}$ given by
$(\Stuple{\alpha^i}_{i \in \Mdeg})_X = 
\Stuple{\alpha^i_X} \in \cL(F X, \Sfun G X)$
is a natural transformation $F \Natto \Sfun G$.

\begin{proof} Let $f \in \cL'(X, Y)$ be a morphism.
  Then, 
  \[ \Big( \sum_{i \in I} \alpha_Y^i \Big) \Compl F f  
  \sumsub \sum_{i \in I} (\alpha_Y^i \Compl F f) 
  = \sum_{i \in I} (G f \Compl \alpha_X^i) 
  \sumsubinv G f \Compl \Big(\sum_{i \in I} \alpha_X^i \Big) \]
  by enrichment over \Pmonoid{}s.
  So $\sum_{i \in I} \alpha^i$ is natural. Furthermore, if 
  $I = \Mdeg$.
  \[ \Stuple{\alpha_Y^i}_{i \in \Mdeg} \Compl F f 
  = \Stuple{\alpha_Y^i \Compl F f}_{i \in \Mdeg}
  = \Stuple{G f \Compl \alpha_X^i}_{i \in \Mdeg}
  = \Sfun G f \Compl \Stuple{\alpha_X^i}_{i \in \Mdeg} \] 
  so $\Stuple{\alpha^i}_{i \in \Mdeg}$ is a natural transformation 
  $F \Natto \Sfun G$.
\end{proof}

We first introduce the symmetry of the monad and the comonad structure.
This symmetry swaps the order of the projections.

\begin{lemma} %
  \label{lemma:sflip-invol-yb}
  For each object \(X\) of \(\cL\), there is a unique morphism
  \(\Sflip_X\in\cL(\Sfun^2X,\Sfun^2X)\) such that %
  \(\Sproj i\Compl\Sproj j\Compl\Sflip_X=\Sproj j\Compl\Sproj i\) %
  for all \(i,j\in\Mdeg\).
  This morphism is natural and such that \(\Sflip_X^2=\Id_{\Sfun^2X}\).
  Moreover, the natural transformation \(\Sflip\) satisfies the
  Yang-Baxter equation, that is
  \begin{center}
    \begin{tikzcd}
      \Sfun^3
      \ar[r,"\Sfun\Sflip"]\ar[d,swap,"\Sflip\Sfun"]
      &
      \Sfun^3
      \ar[r,"\Sflip\Sfun"]
      &
      \Sfun^3
      \ar[d,"\Sfun\Sflip"]
      \\
      \Sfun^3
      \ar[r,"\Sfun\Sflip"']
      &
      \Sfun^3
      \ar[r,"\Sflip\Sfun"']
      &
      \Sfun^3
    \end{tikzcd}
  \end{center}
\end{lemma}

\begin{proof}
  By \cref{thm:ss-induces-enrichement}, $(\Sproj j \Compl \Sproj i)_{(i,j) \in \Mdeg^2}$
  is summable, and thus $\Sflip$ exists by \cref{prop:iterated-summability-structure}, and 
  is natural by~\cref{prop:summability-naturality}.
  Observe that $\Sproj i \Compl \Sproj j \Compl \Sflip^2 
  = \Sproj j \Compl \Sproj i \Compl \Sflip 
  = \Sproj i \Compl \Sproj j$ so $\Sflip^2 = \Sflip$ 
  by joint monicity of the $\Sproj i \Compl \Sproj j$.
  We can check the
  Yang-Baxter property by a similar argument: let \(i,j,k\in\Mdeg\), we have
  \begin{align*}
    \Sproj i\Compl\Sproj j\Compl\Sproj k\Compl
    (\Sfun\Sflip)\Compl(\Sflip\Sfun)\Compl(\Sfun\Sflip)
    &=\Sproj i\Compl\Sproj j\Compl\Sflip
      \Compl\Sproj k\Compl(\Sflip\Sfun)\Compl(\Sfun\Sflip)\\
    &=\Sproj j\Compl\Sproj i
      \Compl\Sproj k\Compl(\Sflip\Sfun)\Compl(\Sfun\Sflip)\\
    &=\Sproj j\Compl\Sproj k
      \Compl\Sproj i\Compl(\Sfun\Sflip)
    =\Sproj j\Compl\Sproj k\Compl\Sflip
      \Compl\Sproj i\\
    &=\Sproj k\Compl\Sproj j\Compl\Sproj i
  \end{align*}
  and a similarly
  \(\Sproj i\Compl\Sproj j\Compl\Sproj k\Compl
  (\Sflip\Sfun)\Compl(\Sfun\Sflip)\Compl(\Sflip\Sfun)
  =\Sproj i\Compl\Sproj k\Compl\Sproj j
  \Compl(\Sfun\Sflip)\Compl(\Sflip\Sfun)
  =\Sproj i\Compl\Sproj k\Compl\Sflip\Compl\Sproj j
  \Compl(\Sflip\Sfun)
  =\Sproj k\Compl\Sproj i\Compl\Sproj j
  \Compl(\Sflip\Sfun)
  =\Sproj k\Compl\Sproj j\Compl\Sproj i\).
  We conclude that the Yang-Baxter equation holds by joint monicity
  of the $\Sproj i\Compl\Sproj j\Compl\Sproj k$.
\end{proof}

We now define the monad structure.
By \ref{ax:unary}, the family \((\Kronecker ij\Id_X)_{j\in\Mdeg}\) 
is summable in \(\cL(X,X)\) with sum $\Id_X$, so 
there exists a unique morphism
\(\Sin i\in\cL(X,\Sfun X)\) such that
\(\Sproj j\Compl\Sin i=\Kronecker ij\Id_X\), and 
 \(\Ssum\Compl\Sin i=\Id_X\).
This is a natural transformation by \cref{prop:summability-naturality}.

\begin{remark}
  Until now, the only information we have used about \(\Mdeg\) is its
  cardinality.
  Its \Pmonoid{} structure will only be used in the 
  definition of the monad multiplication.
\end{remark}

\begin{definition}
  \label{def:mdegsl}
  Given \(i\in\Mdeg\) we set
    \(\Mdegsl i=\Eset{(l,r)\in\Mdeg^2\St l+r=i}\).
  These sets are pairwise disjoint and
  \(\Mdegdom=\Union_{i\in\Mdeg}\Mdegsl i\) is a subset of \(\Mdeg^2\).
  This inclusion is an equality iff \(\Mdeg\) is total, that is, is a
  commutative monoid.
\end{definition}

Notice that our \Mpositivity{} assumption on \(\Mdeg\) means that
\(\Mdegsl 0=\Set{(0,0)}\).

\begin{lemma} %
  \label{lemma:ssum-monad-mult}
  There is a unique morphism
  \(\Sfunmonm_X\in\cL(\Sfun^2X,\Sfun X)\) such that %
  \(\Sproj
  i\Compl\Sfunmonm=\sum_{(l,r)\in\Mdegsl i}\Sproj
  l\Compl\Sproj r\).
  This morphism is natural and we have
  \(\Sfunmonm_X\Compl\Sflip_X=\Sfunmonm_X\).
\end{lemma}
\begin{proof}
  The family \((\Sproj l\Compl\Sproj r)_{(l,r)\in\Mdeg^2}\) is summable
  by \cref{thm:ss-induces-enrichement}.
  It follows by \cref{prop:subfamily-summable} that the family 
  \((\Sproj l\Compl\Sproj r)_{(l,r)\in\Mdegdom}\) is
  summable.
  Hence for each \(i\in\Mdeg\) the family
  \((\Sproj l\Compl\Sproj r)_{(l,r)\in \Mdegsl i}\) is summable as well as
  the family
  \((\sum_{(l,r)\in \Mdegsl i}\Sproj l\Compl\Sproj r)_{i\in\Mdeg}\) and we
  set
  \(\Sfunmonm_X=\Stuple{\sum_{(l,r)\in \Mdegsl i}\Sproj l\Compl\Sproj
    r}_{i\in\Mdeg}\). This is a natural transformation by 
    \cref{prop:summability-naturality}.
  Last we have %
  \(\Sproj i\Compl\Sfunmonm_X\Compl\Sflip_X =(\sum_{(l,r)\in \Mdegsl i}
  \Sproj l\Compl\Sproj r)\Compl\Sflip_X=\sum_{(l,r)\in \Mdegsl i}
  \Sproj l\Compl\Sproj r\Compl\Sflip_X=\sum_{(l,r)\in \Mdegsl i}
  \Sproj r\Compl\Sproj l=\Sproj i\Compl\Sfunmonm\).
\end{proof}

\begin{theorem} %
  \label{th:sfun-linear-monad}
  The tuple \((\Sfun,\Sin0,\Sfunmonm,\Sflip)\) is a \Symmonad{} (see
  \cref{def:sym-monad}) that we denote as \(\Sfunm\).
\end{theorem}
\begin{proof}[Proof sketch]
  The monad commutations are checked using as usual the joint monicity
  of the \(\Sproj i\)'s.
  For instance, one proves easily that
  \begin{align*}
    (\Sfun\Sfunmonm_X)\Compl\Sfunmonm_X
    =\Sfunmonm_{\Sfun X}\Compl\Sfunmonm_X
    =\Stuple{\sum_{\Biind{(l,j,r)\in\Mdeg^3}{l+j+r=i}}
      \Sproj l\Compl\Sproj j\Compl\Sproj r}_{i\in\Mdeg}\,.
  \end{align*}
  Next, we already now thanks to \cref{lemma:sflip-invol-yb}
  that $\Sflip$ is involutive and satisfies the Yang Baxter equation,
  and we already know thanks to \cref{lemma:ssum-monad-mult}
  that $\Sfunmonm \Compl \Sflip =  \Sfunmonm \Compl \Sflip$,
  so we only need to prove the two following commutations:
  \begin{center}
    \begin{tikzcd}
      &[-2em]
      \Sfun\ar[ld,swap,"\Sin0\Sfun"]\ar[dr,"\Sfun\Sin0"]
      &[-2em]
      \\
      \Sfun^2\ar[rr,"\Sflip"']
      &&
      \Sfun^2
    \end{tikzcd}
    \quad
    \begin{tikzcd}
      \Sfun^3\ar[r,"\Sfun\Sflip"]\ar[d,swap,"\Sfunmonm\Sfun"]
      &
      \Sfun^3\ar[r,"\Sflip\Sfun"]
      &
      \Sfun^3\ar[d,"\Sfun\Sfunmonm"]
      \\
      \Sfun^2\ar[rr,"\Sflip"']
      &&
      \Sfun^2
    \end{tikzcd}
  \end{center}
  We use the joint monicity if the $\Sproj i$.
  For all \(i,j\in\Mdeg\), we have
  \begin{align*}
    \Sproj i\Compl\Sproj j\Compl\Sflip\Compl(\Sin0\Sfun)
    &=\Sproj j\Compl\Sproj i\Compl(\Sin0\Sfun)
      =\Kronecker i0\Sproj j\\
    \Sproj i\Compl\Sproj j\Compl(\Sfun\Sin0)
    &=\Sproj i\Compl\Sin0\Sproj j=\Kronecker i0\Sproj j\\
    \Sproj i\Compl\Sproj j\Compl\Sflip\Compl(\Sfunmonm\Sfun)
    &=\Sproj j\Compl\Sproj i\Compl(\Sfunmonm\Sfun)
      =\Sproj j\sum_{(l,r)\in\Mdegsl i}\Sproj l\Compl\Sproj r
      =\sum_{(l,r)\in\Mdegsl i}\Sproj j\Compl\Sproj l\Compl\Sproj r\\
    \Sproj i\Compl\Sproj j\Compl(\Sfun\Sfunmonm)
    \Compl(\Sflip\Sfun)\Compl(\Sfun\Sflip)
    &=\Sproj i\Compl\Sfunmonm\Compl\Sproj j
      \Compl(\Sflip\Sfun)\Compl(\Sfun\Sflip)\\
    &=\sum_{(l,r)\in\Mdegsl i}\Sproj l\Compl\Sproj r
      \Compl\Sproj j\Compl(\Sflip\Sfun)\Compl(\Sfun\Sflip)\\
    &=\sum_{(l,r)\in\Mdegsl i}\Sproj l\Compl\Sproj j
      \Compl\Sproj r\Compl(\Sfun\Sflip)\\
    &=\sum_{(l,r)\in\Mdegsl i}\Sproj l\Compl\Sproj j
      \Compl\Sflip\Compl\Sproj r
      =\sum_{(l,r)\in\Mdegsl i}\Sproj j\Compl\Sproj l\Compl\Sproj r\,.
      \qedhere
  \end{align*}
\end{proof}

\begin{remark}
  \label{rem:kleisli-S-power-series}
  Let us provide some intuition on the Kleisli category
  $\Klcat{\cL}{\Sfunm}$ of this monad.
  A morphism $f \in \cL(X, \Sfun Y)$ is uniquely characterized by its
  projections $f_i = \Sproj i \Compl f$, and therefore we can
  represent $f$ by a formal power series
  \[
    f = \sum_{n \in \Mdeg} f_i \varepsilon^i
  \]
  whose coefficients lie in $\cL(X, Y)$, and where the exponents of
  the formal variable $\varepsilon$ lie in $\Mdeg$.
  The canonical functor $\cL \arrow \Klcat \cL \Sfunm$ maps a morphism
  $h \in \cL(X, Y)$ to $\Sin 0 \Compl h \in \cL(X, \Sfun Y)$, which is
  nothing but the constant power series $h$ since
  $\Sproj 0 \Compl \Sin 0 \Compl h = h$ and
  $\Sproj i \Compl \Sin 0 \Compl h = 0$ when $i \neq 0$.
  Then, the composition in the Kleisli category corresponds to a
  Cauchy product.
  More precisely, given $f \in \cL(X, \Sfun Y)$ and
  $g \in \cL(Y, \Sfun Z)$, their composition is given by
  \[
    g \Comp f = \Sfunmonm \Compl \Sfun g \Compl f = \sum_{n \in \Mdeg}
    \big(\sum_{i + j = n} g_j \Compl f_i \big) \varepsilon^n\,.
  \]
  When $\Mdeg = (\Nat, +)$, this composition corresponds to the
  standard multiplication of power series (defining the multiplication
  of two coefficients as their composition).
  When $\Mdeg = \Set{0, 1}$ with $1 +1$ undefined, we get that
  $f = f_0 + f_1 \varepsilon$, $g = g_0 + g_1 \varepsilon$, and
  $g \Comp f = g_0 \Compl f_0 + (g_1 \Compl f_0 + g_0 \Compl f_1)
  \varepsilon$.
  In other word, morphism can be considered as degree one polynomial
  over a formal variable $\varepsilon$ such that $\varepsilon^2 = 0$,
  and composition is the polynomial multiplications.
  This corresponds to Clifford's \emph{dual numbers}, that are nowadays
  widely used in automated differentiation~\cite{Baydin17}.
\end{remark}

We now introduce the comonad structure.
The comultiplication is given by a structure map of \(\Sfun\) that was
overlooked in~\cite{Ehrhard23a}.

\begin{lemma}
  There is a natural \emph{lift} morphism %
  \(\Sfunlift_X\in\cL(\Sfun X,\Sfun^2X)\) characterized by %
  \(\Sproj i\Compl\Sproj j\Compl\Sfunlift=\Kronecker ij\Sproj i\).
\end{lemma}
\begin{proof}
  For each \(i\in\Mdeg\) we have
  \(\Sin i\Compl\Sproj i\in\cL(\Sfun X,\Sfun X)\).
  Moreover \(\Ssum\Compl\Sin i\Compl\Sproj i=\Sproj i\) and hence the
  family \((\Sin i\Compl\Sproj i)_{i\in\Mdeg}\) is summable and we can
  set
  \(\Slift_X=\Stuple{\Sin i\Compl\Sproj i}_{i\in\Mdeg}\in\cL(\Sfun
  X,\Sfun^2X)\).
  This morphism is natural by \cref{prop:summability-naturality} and
  \(\Sproj i\Compl\Sproj j\Compl\Sfunlift=\Kronecker ij\Sproj i\)
  obviously holds.
\end{proof}

\begin{proposition}
  \label{prop:sfun-comonad}
  The triple \((\Sfun,\Ssum,\Sfunlift,\Sflip)\) is a \Symcomonad{}
  (see \cref{def:sym-monad}) that we denote as \(\Sfunc\).
\end{proposition}
\begin{proof}
    We prove first the three comonad commutations by postcomposing
  with projections.
  \begin{center}
    \begin{tikzcd}
      \Sfun X
      \ar[r,"\Slift"]
      \ar[dr,swap,"\Id"]
      &
      \Sfun^2 X
      \ar[d,"\Ssum\Sfun"]\\
      &
      \Sfun X
    \end{tikzcd}
    \quad
    \begin{tikzcd}
      \Sfun X
      \ar[r,"\Slift"]
      \ar[dr,swap,"\Id"]
      &
      \Sfun^2 X
      \ar[d,"\Sfun\Ssum"]\\
      &
      \Sfun X
    \end{tikzcd}
    \quad
    \begin{tikzcd}
      \Sfun X
      \ar[r,"\Slift"]
      \ar[d,swap,"\Slift"]
      &
      \Sfun^2 X
      \ar[d,"\Slift\Sfun"]
      \\
      \Sfun^2 X
      \ar[r,"\Sfun\Slift"']
      &
      \Sfun^3 X
    \end{tikzcd}
  \end{center}
  For each \(i\in\Mdeg\) we have
  \begin{align*}
    \Sproj i\Compl(\Ssum\Sfun)\Compl\Slift
    &=\Ssum\Compl(\Sfun\Sproj i)\Compl\Slift
      \text{\quad by nat.~of }\Ssum\\
    &=\Ssum\Compl\Stuple{\Sproj i\Compl\Sproj j}_{j\in\Mdeg}
      \Compl\Slift
    \text{\quad by \cref{lemma:Sfun-functor-stuple}}\\
    &=\Ssum\Compl\Stuple{\Sproj i\Compl\Sproj j\Compl\Slift}_{j\in\Mdeg}\\
    &=\Ssum\Compl\Stuple{\Kronecker ij\Sproj j}_{j\in\Mdeg}\\
    &=\Sproj i\\
    \Sproj i\Compl(\Sfun\Ssum)\Compl\Slift
    &=\Ssum\Compl\Sproj i\Compl\Slift
      \text{\quad by nat.~of }\Sproj i\\
    &=\Ssum\Compl\Sproj i\Compl
      \Stuple{\Sin j\Compl\Sproj j}_{j\in\Mdeg}\\
    &=\Ssum\Compl\Sin i\Compl\Sproj i\\
    &=\Sproj i
  \end{align*}
  and for \(i,j,k\in\Mdeg\) we have
  \begin{align*}
    \Sproj i\Compl\Sproj j\Compl\Sproj k\Compl(\Slift\Sfun)\Compl\Slift
    &=\Sproj i\Compl (\Kronecker jk \Sproj j) \Compl\Slift\\
    &=\Kronecker jk \Sproj i\Compl  \Sproj j \Compl\Slift\\
    &=\Kronecker jk \Kronecker ij \Sproj i\Compl \\
    \Sproj i\Compl\Sproj j\Compl\Sproj k\Compl(\Sfun\Slift)\Compl\Slift
    &=\Sproj i\Compl\Sproj j\Compl\Slift\Compl\Sproj k\Compl\Slift\\
    &=\Kronecker ij\Sproj i\Compl\Sproj k\Slift\\
    &=\Kronecker ij\Kronecker ik\Sproj i\,.
  \end{align*}
  Next,  we already now thanks to \cref{lemma:sflip-invol-yb}
  that $\Sflip$ is involutive and satisfies the Yang Baxter equation,
  so we only need to prove
  the three following commutations.
  \begin{center}
    \begin{tikzcd}
      &[-2em] \Sfun \ar[dl, "\Slift"'] \ar[dr, "\Slift"] &[-2em] \\ 
      \Sfun^2 \ar[rr, "\Sflip"'] & & \Sfun^2
    \end{tikzcd}
    \begin{tikzcd}
      \Sfun^2\ar[rr,"\Sflip"]\ar[dr,swap,"\Sfun\Ssum"]
      &[-2em]&[-2em]
      \Sfun^2\ar[ld,"\Ssum\Sfun"] 
      \\
      &\Sfun&
    \end{tikzcd}
    \quad
    \begin{tikzcd}
      \Sfun^2
      \ar[rr,"\Sflip"]
      \ar[d,swap,"\Sfun\Slift"]
      &&
      \Sfun^2
      \ar[d,"\Slift\Sfun"]
      \\
      \Sfun^3
      \ar[r,"\Sflip\Sfun"']
      &
      \Sfun^3
      \ar[r,"\Sfun\Sflip"']
      &
      \Sfunc^2\Sfun
    \end{tikzcd}
  \end{center}
  We prove them by joint monicity of the 
  $\Sproj i$.
  \begin{align*}
    \Sproj i \Compl \Sproj j \Compl \Sflip \Compl \Slift
    &= \Sproj j \Compl \Sproj i \Compl \Slift 
    = \Kronecker{i}{j} \Sproj i = 
    \Sproj i \Compl \Sproj j \Compl \Slift \\
    \Sproj i\Compl(\Ssum\Sfun)\Compl\Sflip
    &=\Ssum\Compl(\Sfun\Sproj i)\Compl\Sflip
      =\Ssum\Compl\Stuple{\Sproj i\Compl\Sproj j}_{j\in\Mdeg}\Compl\Sflip
      =\Ssum\Compl\Stuple{\Sproj j\Compl\Sproj i}_{j\in\Mdeg}
      =\Ssum\Compl\Sproj i
      =\Sproj i\Compl(\Sfun\Ssum)\\
    \Sproj i\Compl\Sproj j\Compl\Sproj k\Compl(\Slift\Sfun)\Compl\Sflip
    &=\Kronecker jk\Sproj i\Compl\Sproj j\Compl\Sflip
      =\Kronecker jk\Sproj j\Compl\Sproj i\\
    \Sproj i\Compl\Sproj j\Compl\Sproj k\Compl
    (\Sfun\Sflip)\Compl(\Sflip\Sfun)\Compl(\Sfun\Slift)
    &=\Sproj i\Compl\Sproj j\Compl\Sflip\Compl\Sproj k
      \Compl(\Sflip\Sfun)\Compl(\Sfun\Slift)
      =\Sproj j\Compl\Sproj i\Compl\Sproj k
      \Compl(\Sflip\Sfun)\Compl(\Sfun\Slift)
      =\Sproj j\Compl\Sproj k\Compl\Sproj i\Compl(\Sfun\Slift)\\
    &=\Sproj j\Compl\Sproj k\Compl\Slift\Compl\Sproj i
      =\Kronecker jk\Sproj j\Compl\Sproj i
      \qedhere
  \end{align*}
\end{proof}

\begin{theorem}
  \label{th:sfun-bimonad}
  If \(\Mdeg\) is total then the tuple
  \((\Sfun,\Ssum,\Sfunlift,\Sin0,\Sfunmonm,\Sflip)\) is a
  \Symbimonad{} \(\Sfunb\).
\end{theorem}
\begin{proof}
  We must prove the following commutations.
  The hypothesis that \(\Mdeg\) is total is used only for the last
  one.
  \begin{center}
    \begin{tikzcd}
      \Id\ar[r,"\Sin0"]\ar[rd,swap,"\Id"] &
      \Sfun\ar[d,"\Ssum"]\\
      & \Id
    \end{tikzcd}
    \quad
    \begin{tikzcd}
      \Sfun^2
      \ar[r,"\Sfunmonm"]
      \ar[d,swap,"\Slift\Compv\Slift"]
      &
      \Sfun\ar[r,"\Slift"]
      &
      \Sfun^2\\
      \Sfun^4\ar[rr,swap,"\Sfun\Sflip\Sfun"]
      &&
      \Sfun^4\ar[u,swap,"\Sfunmonm\Compv\Sfunmonm"]
    \end{tikzcd}
  \end{center}
  \begin{center}
    \begin{tikzcd}
      \Id
      \ar[r,"\Sin0"]\ar[d,swap,"\Sin0"]
      &
      \Sfun
      \ar[d,"\Slift"]
      \\
      \Sfun
      \ar[r,"\Sin0\Sfun"']
      &
      \Sfun^2
    \end{tikzcd}
    \quad
    \begin{tikzcd}
      \Sfun^2
      \ar[r,"\Sfun\Ssum"]\ar[d,swap,"\Sfunmonm"]
      &
      \Sfun
      \ar[d,"\Ssum"]
      \\
      \Sfun
      \ar[r,"\Ssum"']
      &
      \Id
    \end{tikzcd}
  \end{center}
  where \(\Compv\) is the horizontal composition of natural
  transformations (which is also used in the two last diagrams since
  \((\Sin0\Sfun)\Compl\Sin0=\Sin0\Compv\Sin0\) and
  \(\Ssum\Compl(\Sfun\Ssum)=\Ssum\Compv\Ssum\)).
  The first commutation results from the definition of \(\Sin0\).
  For the second one, we have
  \begin{align*}
    \Sproj i\Compl\Sproj j
    \Compl(\Sfunmonm\Compv\Sfunmonm)
    \Compl(\Sfun\Sflip\Sfun)\Compl(\Slift\Compv\Slift)
    &=\Sproj i\Compl\Sproj j
      \Compl(\Sfunmonm\Sfun)\Compl(\Sfun^2\Sfunmonm)\Compl(\Sfun\Sflip\Sfun)
      \Compl(\Slift\Compv\Slift)\\
    &=\Sproj i\Compl
      \big(\sum_{(l,r)\in\Mdegsl j}\Sproj l\Compl\Sproj r\big)
      \Compl(\Sfun^2\Sfunmonm)\Compl(\Sfun\Sflip\Sfun)
      \Compl(\Slift\Compv\Slift)\\
    &=\Sproj i\Compl\Sfunmonm\Compl
      \big(\sum_{(l,r)\in\Mdegsl j}\Sproj l\Compl\Sproj r\big)
      \Compl(\Sfun\Sflip\Sfun)
      \Compl(\Slift\Compv\Slift)\\
    &=\big(\sum_{(p,q)\in\Mdegsl i}\Sproj p\Compl\Sproj q\big)
      \Compl\big(\sum_{(l,r)\in\Mdegsl j}\Sproj l\Compl\Sproj r\big)
      \Compl(\Sfun\Sflip\Sfun)
      \Compl(\Slift\Compv\Slift)\\
    &=\big(\sum_{(p,q)\in\Mdegsl i,(l,r)\in\Mdegsl j}
      \Sproj p\Compl\Sproj q\Compl\Sproj l\Compl\Sproj r\big)
      \Compl(\Sfun\Sflip\Sfun)
      \Compl(\Slift\Compv\Slift)\\
    &=\big(\sum_{(p,q)\in\Mdegsl i,(l,r)\in\Mdegsl j}
      \Sproj p\Compl\Sproj l\Compl\Sproj q\Compl\Sproj r\big)
      \Compl(\Slift\Compv\Slift)\\
    &=\sum_{(p,q)\in\Mdegsl i,(l,r)\in\Mdegsl j}
      \Kronecker pl\Kronecker qr\Sproj p\Compl\Sproj q\\
    &=\Kronecker ij\sum_{(p,q)\in\Mdegsl i}\Sproj p\Compl\Sproj q
      =\Sproj i\Compl\Sproj j\Compl\Slift\Compl\Sfunmonm\,.
  \end{align*}
  For the third diagram, we have
  \begin{align*}
    \Sproj i\Compl\Sproj j\Compl\Slift\Compl\Sin0
    &=\Kronecker ij\Sproj i\Compl\Sin0
    =\Kronecker ij\Kronecker i0\Compl\Id \\
    \Sproj i\Compl\Sproj j\Compl(\Sin0\Sfun)\Compl\Sin0
    &=\Sproj i\Compl(\Kronecker j0\Id)\Sin0
    =\Kronecker i0\Kronecker j0\Id
  \end{align*}
  for all \(i,j\in\Mdeg\).
  For the last diagram, we have
  \begin{align*}
    \Ssum\Compl\Sfunmonm
    &=\Big(\sum_{i\in\Mdeg}\Sproj i\Big)\Compl\Sfunmonm
    =\sum_{i\in\Mdeg}\sum_{(l,r)\in\Mdegsl i}\Sproj l\Compl\Sproj r \\
    \Ssum\Compl(\Sfun\Ssum)
    &=\Ssum\Compl(\Ssum\Sfun)
    =\Big(\sum_{i\in\Mdeg}\Sproj i\Big)
    \Compl\Big(\sum_{i\in\Mdeg}\Sproj i\Big)
    =\sum_{(l,r)\in\Mdeg^2}\Sproj l\Compl\Sproj r
  \end{align*}
 so this diagram commutes by \ref{ax:pa} since
  \(\biguplus_{i\in\Mdeg}\Mdegsl i=\Mdeg^2\) by totality of \(\Mdeg\).
\end{proof}

\subsection{Compatibility of the \texorpdfstring{\SUMBLTY{}}{Summability} 
structure with the tensor product and the internal hom}

\label{sec:summability-tensor}

We assume in this section that the \Pmcat{} \(\cL\) is equipped with a
\SUMBLTY{} structure and is symmetric monoidal.
Thus, $\cL$ is a \Dpmenriched{} category, by
\cref{thm:ss-induces-enrichement}, but it has no reason to be a
symmetric monoidal \Dpmenriched{} category (as defined in
\cref{def:pm-enriched-monoidal}) in general.
We prove here that the necessary compatibility condition between the
\Dpmonoid{} structure and the monoidal closed structure boils down to
a commutative monad structure on the \SUMBLTY{} structure $\Sfun$.

\begin{proposition} \label{prop:ss-tensor}
  The following assertions are equivalent. \begin{itemize}
  \item The category $\cL$ is a symmetric monoidal \Dpmenriched{} category.
  \item For all objects $X_2$, 
  $0 \ITens X_2 = 0$, and $(\Sproj i \ITens X_2)_{i \in \Mdeg}$ 
  is summable, with sum $\Ssum \ITens X_2$.
  \end{itemize}
\end{proposition}

We say that the \SUMBLTY{} structure satisfies
\labeltext{(\(\Sfun\)-\(\ITens\))}{ax:stens} \ref{ax:stens} if it
satisfies one of the equivalent condition above.
If \ref{ax:stens} holds then we can define
  \begin{equation} \label{eq:sm-strength}
    \begin{split}
      \Sfunstr^1_{X_1, X_2} = \Stuple{\Sproj i \ITens X_2}_{i \in \Mdeg} 
    : \Sfun X_1 \ITens X_2 \arrow \Sfun (X_1 \ITens X_2) \text{ so that } 
    \Ssum \Compl \Sfunstr^1 = (\Ssum \ITens X_2) \\ 
    \Sfunstr^2_{X_1, X_2} = \Stuple{X_1 \ITens \Sproj i}_{i \in \Mdeg} 
    : X_1 \ITens \Sfun X_2 \arrow \Sfun (X_1 \ITens X_2) \text{ so that } 
    \Ssum \Compl \Sfunstr^2 = (X_1 \ITens \Ssum)
    \end{split}
  \end{equation}  
where $\Sfunstr^2$ above is given by
\(\Sfunstr^2_{X_1,X_2}
=(\Sfun\Sym_{X_2,X_1})\Compl\Sfunstr^1_{X_2,X_1}\Sym_{X_1,\Sfun X_2}\).

\begin{proof}
  The forward implication is trivial, by direct application of 
  \cref{def:pm-enriched-monoidal}
  on the $\Sproj i$'s and $\Id$.
  The proof 
  of the reverse direction is very similar 
  to the proof of \cref{thm:ss-induces-enrichement}.
  Let $g \in \cL(X_2, Y_2)$ be a morphism, and
  $(f_i)_{i \in I}$ be a 
  family of morphisms in $\cL(X_1, X_2)$.  
  With the same notations as in that proof, define
  $(f'_{i_1, \ldots, i_n})_{(i_1, \ldots, i_n)\in D^n} = \Famact \phi
  (f_i)_{i \in I}$.
  Then, by our assumption that 
  $0 \ITens \Id_{Y_2} = 0$, we get that 
  \[
    \forall (i_1, \ldots, i_n) \in \Mdeg^n \setminus \phi(I)\quad 
    f'_{i_1, \ldots, i_n} \ITens g = 
      (0 \ITens \Id_{Y_2}) \Compl (\Id_{X_1} \ITens g) 
      = 0 \Compl (\Id_{X_1} \ITens g) = 0
  \]
  so
  $(\Tens{f_{i_1, \ldots, i_n}} g)_{(i_1, \ldots, i_n) \in \Mdeg^n}
  = \Famact \phi (\Tens {f_i} g)$.
  It follows from \cref{prop:reindexing} that 
  \[
    \Big( \sum_{i \in I} f_i \Big) \ITens g \sumiff \Big(
    \Smashspace{\sum_{(i_1, \ldots, i_n) \in \Mdeg^n}} f'_{i_1, \ldots,
      i_n} \Big) \ITens g \quad \text{and} \quad
    \Smashspace{\sum_{(i_1, \ldots, i_n) \in \Mdeg^n}} (f'_{i_1, \ldots,
      i_n} \ITens g) \sumiff \sum_{i \in I} (f_i \ITens g)
  \]
  so it suffices to prove that
  $\Big(\sum_{(i_1, \ldots, i_n) \in \Mdeg^n} f'_{i_1, \ldots, i_n} \Big)
  \ITens g \sumsub \sum_{(i_1, \ldots, i_n) \in \Mdeg^n} (f'_{i_1, \ldots,
    i_n} \ITens g)$ to end the proof.
  To do so, let us assume that
  $(f'_{i_1, \ldots, i_n})_{(i_1, \ldots, i_n)\in\Mdeg^n}$ is summable.
  Then by \cref{prop:iterated-summability-structure}, there exists
  $f \in \cL(X_1,\Sfun^n Y_1)$ such that
  $\Sproj{i_1} \Compl \cdots \Sproj{i_n} \Compl f = f_{i_1, \ldots,
    i_n}$ and
  $\Ssum^n \Compl f = \sum_{(i_1, \ldots, i_n) \in \Mdeg^n} f'_{i_1,
    \ldots, i_n}$.
  Let us consider the morphism
  \[
    (\Sfunstr^1)^n_{X_1, X_2} \Compl (f \ITens g) : X_1 \ITens X_2
    \arrow \Sfun^n(Y_1 \ITens Y_2)
  \]
  where
  $(\Sfunstr^1)^n_{X_1, X_2} : \Sfun^n X_1 \ITens X_2 \arrow
  \Sfun^n(X_1 \ITens X_2)$ is given by successive iterations of
  $\Sfunstr^1$. We can now easily check that
  \begin{alignat*}{3}
    \Sproj {i_1, \ldots, i_n} \Compl(\Sfunstr^1)^n_{X_1, X_2} \Compl (f \ITens g)
    = (&& \Sproj {i_1, \ldots, i_n} & \ITens \Id_{Y_2}) 
                                      \Compl (f \ITens g)
                                      = f_{i_1, \ldots, i_n} \ITens g \\
    \Ssum^n \Compl \Compl(\Sfunstr^1)^n_{X_1, X_2} \Compl (f \ITens g)
    =( && \Ssum^n & \ITens \Id_{Y_2}) \Compl (f \ITens g)
                    = \Big( \Smashspace{\sum_{(i_1, \ldots, i_n) \in \Mdeg^n}}
                    f'_{i_1, \ldots, i_n} \Big) \ITens g 
  \end{alignat*}
  We conclude by \cref{prop:iterated-summability-structure} again that
  the family
  $(f_{i_1, \ldots, i_n} \ITens g)_{(i_1, \ldots, i_n) \in \Mdeg^n}$ is
  summable with sum
  $(\sum_{(i_1, \ldots, i_n) \in \Mdeg^n} f'_{i_1, \ldots, i_n} \Big)
  \ITens g $.
\end{proof}

We now assume that the \SUMBLTY{} structure 
satisfies \ref{ax:stens}. 
By \cref{prop:summability-naturality}, the two families of morphisms 
$\Sfunstr^1$ and $\Sfunstr^2$ are always natural transformations.

\begin{theorem}
  The following diagrams commute
  \begin{equation*}
    \begin{tikzcd}
      \Tens{X_1}{X_2}
      \ar[r,"\Tens{\Sin0}{X_2}"]
      \ar[dr,swap,"\Sin0"]
      &
      \Tens{\Sfun X_1}{X_2}
      \ar[d,"\Sfunstr^1"]
      \\
      &
      \Sfun\Tensp{X_1}{X_2}
    \end{tikzcd}
    \Treesep
    \begin{tikzcd}
      \Tens{\Sfun^2X_1}{X_2}
      \ar[r,"\Sfunstr^1"]
      \ar[d,swap,"\Tens{\Sfunmonm}{X_2}"]
      &
      \Sfun\Tensp{\Sfun X_1}{X_2}
      \ar[r,"\Sfun\Sfunstr^1"]
      &
      \Sfun^2\Tensp{X_1}{X_2}
      \ar[d,"\Sfunmonm"]
      \\
      \Tens{\Sfun X_1}{X_2}
      \ar[rr,"\Sfunstr^1"']
      &&
      \Sfun\Tensp{X_1}{X_2}
    \end{tikzcd}
  \end{equation*}
  \begin{equation*}
    \begin{tikzcd}
      &[-2em]
      \Tens{\Sfun X_1}{\Sfun X_2}
      \ar[ld,swap,"\Sfunstr^1_{X_1,\Sfun X_2}"]
      \ar[rd,"\Sfunstr^2_{\Sfun X_1,X_2}"]
      &[-2em]\\
      \Sfun\Tensp{X_1}{\Sfun X_2}
      \ar[d,swap,"\Sfun\Sfunstr^2_{X_1,X_2}"]
      &
      &
      \Sfun\Tensp{\Sfun X_1}{X_2}
      \ar[d,"\Sfun\Sfunstr^1_{X_1,X_2}"]\\
      \Sfun^2\Tensp{X_1}{X_2}
      \ar[rr,"\Sflip_{\Tens{X_1}{X_2}}"']
      &&
      \Sfun^2\Tensp{X_1}{X_2}
    \end{tikzcd}
  \end{equation*}
  and hence \((\Sfun,\Sfunstr^1)\) is a commutative strong monad.
\end{theorem}
\begin{proof}
  The diagrams are proven commutative using the joint monicity of the
  \(\Sproj i\)'s.
  Commutativity of the monad means that the following diagram commutes
  \begin{equation*}
    \begin{tikzcd}
      &[-2em]
      \Tens{\Sfun X_1}{\Sfun X_2}
      \ar[ld,swap,"\Sfunstr^1_{X_1,\Sfun X_2}"]
      \ar[rd,"\Sfunstr^2_{\Sfun X_1,X_2}"]
      &[-2em]\\
      \Sfun\Tensp{X_1}{\Sfun X_2}
      \ar[d,swap,"\Sfun\Sfunstr^2_{X_1,X_2}"]
      &
      &
      \Sfun\Tensp{\Sfun X_1}{X_2}
      \ar[d,"\Sfun\Sfunstr^1_{X_1,X_2}"]\\
      \Sfun^2\Tensp{X_1}{X_2}
      \ar[rd,swap,"\Sfunmonm_{\Tens{X_1}{X_2}}"]
      &&
      \Sfun^2\Tensp{X_1}{X_2}
      \ar[ld,"\Sfunmonm_{\Tens{X_1}{X_2}}"]\\
      &
      \Sfun\Tensp{X_1}{X_2}
    \end{tikzcd}
  \end{equation*}
  which results from the third commutation and \cref{lemma:ssum-monad-mult}.
\end{proof}

\begin{remark} \label{rem:Sfun-lax-tensor}
  It is standard \cite{Kock70} that the common value %
  \[
    \Sfunlaxt_{X_1,X_2}
    =\Sfunmonm_{\Tens{X_1}{X_2}}
    \Compl(\Sfun\Sfunstr^1_{X_1,X_2}) \Compl\Sfunstr^2_{\Sfun X_1,X_2}
    =\Sfunmonm_{\Tens{X_1}{X_2}} \Compl(\Sfun\Sfunstr^2_{X_1,X_2})
    \Compl\Sfunstr^1_{X_1,\Sfun X_2} \in\cL(\Tens{\Sfun X_1}{\Sfun
      X_2},\Sfun\Tensp{X_1}{X_2})
  \] %
  and the associated unit \(\Sin0\in\cL(\Sone,\Sfun\Sone)\) %
  turns \(\Sfun\) into a lax symmetric monoidal monad.
  Notice that this morphism is characterized by
  \begin{align*}
    \forall i\in\Mdeg\quad \Sproj i\Compl\Sfunlaxt
    =\sum_{(l,r)\in\Mdegsl i}\Tens{\Sproj l}{\Sproj r}\,.
  \end{align*}
  It follows from \cref{thm:sm-extension} that the Kleisli category
  $\Klcat{\cL}{\Sfun}$ of the monad $\Sfun$ is monoidal.
  Going back to our power series intuition developed in
  \cref{rem:kleisli-S-power-series}, the tensor of a power series
  $\sum_{i \in \Mdeg} f_i \varepsilon^i \in \Klcat{\cL}{\Sfun}(X_1,
  Y_1)$ with
  $\sum_{j \in \Mdeg} g_j \varepsilon^j \in \Klcat{\cL}{\Sfun}(X_2,
  Y_2)$ is given by the power series
  \[
    \Sfunlaxt_{Y_1,Y_2} \Compl (f \ITens g) = \sum_{n \in \Mdeg}
    \sum_{i+j = n} (f_i \ITens g_j) \varepsilon^n \in
    \Klcat{\cL}{\Sfun}(X_1 \ITens X_2, Y_1 \ITens Y_2)
  \]
  which is again a Cauchy product of power series, except that this
  time the multiplication on the coefficients is taken as the tensor
  in $\cL$ rather than the composition.
\end{remark}

We now assume that \(\cL\) is a symmetric monoidal \Dpmenriched{}
category, and is closed with respect to its monoidal product.
It has no reason to be a symmetric monoidal closed \Dpmenriched{}
category (\cref{def:spm-enriched-smcc}) in general, and again this
property boils down to a property on the \SUMBLTY{} structure $\Sfun$
expressing that \(\Sfun{\Limplp XY}\) and \(\Limpl X{\Sfun Y}\) are
isomorphic.
More precisely, notice that
we have a morphism \(\Sfunstr_{X,Y}\) defined as the following
composition of morphisms
\begin{equation} \label{eq:Sfunstri}
  \begin{tikzcd}
    \Tens{\Sfun{\Limplp XY}}{X}
    \ar[r,"\Sfunstr^1_{\Limpl XY,X}"]
    &[1.6em]
    \Sfun{\Tensp{\Limplp XY}{X}}
    \ar[r,"\Sfun\Evlin"]
    &
    \Sfun X
  \end{tikzcd}
\end{equation}
so that we can define %
\(\Sfunstri_{X,Y} =\Curlin{\Sfunstr_{X,Y}}\in\cL(\Sfun{\Limplp
  XY},\Limpl X{\Sfun Y})\).
This morphism is a standard construction~\cite{Kock71} oftentimes
called the \emph{pointwise structure} (induced here by $\Sfunstr^1$).
Observe that $\Sfunstri$ is the mate of the strength $\Sfunstr^1$, see
\cref{sec:mates}.

\begin{lemma} \label{prop:ss-pointwise-characterization}
  For all $i \in \Mdeg$, 
  $(\Limpl X {\Sproj i}) \Compl \Sfunstri = \Sproj i$
\end{lemma}
\begin{proof}
 We prove this equation by naturality of $\Curlin$.
 \begin{align*}
    (\Limpl X {\Sproj i}) \Compl \Sfunstri &= 
      (\Limpl X {\Sproj i}) \Compl \Curlin
      \Stuple{\Evlin \Compl (\Sproj i \ITens X)}_{i \in \Mdeg} \\
      &= \Curlin(\Sproj i \Compl 
        \Stuple{\Evlin \Compl (\Sproj i \ITens X)}) 
      & \quad \text{naturality of $\Curlin$}\\
        &= \Curlin(\Evlin \Compl (\Sproj i \ITens X)) \\
        &= \Curlin(\Curlin^{-1}(\Sproj i)) = \Sproj i \tag*{\qedhere}
 \end{align*}
\end{proof}

\begin{proposition}
  \label{prop:ss-closure}
  The following assertions are equivalent. \begin{enumerate}
  \item The category $\cL$ is a symmetric monoidal closed
    \Dpmenriched{} category.
  \item For all object $X$, the indexed family
    $(\Limpl X {\Sproj i})_{i \in \Mdeg}$ is summable in
    $\cL(\Limpl{X}{\Sfun Y}, \Limpl X Y)$.
  \item The morphism
    \(\Sfunstri_{X,Y} \in\cL(\Sfun{\Limplp XY},\Limpl X{\Sfun Y})\) is
    an iso.
  \end{enumerate}
  Furthermore, when one of the conditions above holds,
  $(\Sfunstri_{X, Y})^{-1} = \Stuple{\Limpl X {\Sproj i}}_{i \in
    \Mdeg}$.
\end{proposition}

We say that the \SUMBLTY{} structure satisfies
\labeltext{(\(\Sfun\)-\(\multimap\))}{ax:slimpl}\ref{ax:slimpl} when
one of the equivalent conditions above holds.
This is a \emph{property} of the summability structure, not an additional
structure.

\begin{proof}
  We prove each equivalence separately.

  \Proofcase{$1 \Leftrightarrow 2$.}
  The proof is very similar to that of
  \cref{thm:ss-induces-enrichement}, as well as that of
  \cref{prop:ss-tensor}.
  The forward direction is tautological.
  Conversely, assume that $(f_i)_{i \in I}$ is a summable family of 
  $\cL(Y, Z)$. We need to prove 
  that $(\Limpl X {f_i})_{i \in I}$ is summable in $\cL(\Limpl X Y, \Limpl X Z)$.
  For the sake of simplicity, we assume here that $I = \Mdeg$.
  The general case relies on the reindexing technique used in the
  proof of \cref{prop:ss-tensor}, using  $\Limplp X 0 = 0$.
  Define the morphism
  \[
    h = \Stuple{\Limpl Y {\Sproj i}}_{i \in \Mdeg} 
    \Compl (\Limpl X {\Stuple{f_i}_{i \in \Mdeg}}) \in 
    \cL(\Limpl X Y, \Sfun( \Limpl X Z))\,.
  \]
  Then we can check easily by functoriality of $\Limpl X {\_}$ that
  $\Sproj i \Compl h = \Limplp X {f_i}$, so
  $(\Limpl X {f_i})_{i \in \Mdeg}$ is summable.
  This concludes the proof.

  \Proofcase{$2 \Leftrightarrow 3$.}
  Monicity is preserved by right adjoint functors, and hence the
  morphisms $(\Limpl X {\Sproj i})_{i \in \Mdeg}$ are jointly
  monic.
  This fact can also be checked by hand in a rather straightforward
  way using the fact that $\Curlin$ is a bijection and is natural.
  Now, assume that $(\Limpl X {\Sproj i})_{i \in \Mdeg}$ is summable.
  Then, by \cref{prop:ss-pointwise-characterization},
  \[
    \Sproj i \Compl \Stuple{\Limpl X {\Sproj i}}_{i \in \Mdeg} \Compl
    \Sfunstri = (\Limpl X {\Sproj i}) \Compl \Sfunstri = \Sproj i
  \]
  \[ (\Limpl X {\Sproj i}) \Compl \Sfunstri \Compl \Stuple{\Limpl X
      {\Sproj i}}_{i \in \Mdeg} = \Sproj i \Compl \Stuple{\Limpl X
      {\Sproj i}}_{i \in \Mdeg} = \Limpl X {\Sproj i}
  \]
  so, by joint monicity of the $\Sproj i$ and of the
  $\Limpl X {\Sproj i}$, we conclude that $\Sfunstri$ is an iso with
  inverse $\Stuple{\Limpl X {\Sproj i}}_{i \in \Mdeg}$.
  Conversely, assume that $\Sfunstri$ is an iso.
  Then,
  \[
    \Limplp X {\Sproj i} = (\Limpl X {\Sproj i}) \Compl \Sfunstri
    \Compl (\Sfunstri)^{-1} = \Sproj i \Compl (\Sfunstri)^{-1}
  \]
  so $(\Limpl X {\Sproj i})_{i \in \Mdeg}$ is summable with witness
  $(\Sfunstri)^{-1}$.
\end{proof}

\subsection{\texorpdfstring{\SUMBLTY{}}{Summability} structure in a cartesian category}
\label{sec:summability-cartesian}
We assume now that \(\cL\) is cartesian, see \cref{sec:ll}
for the notations.
It has no reason to be a cartesian \Pmenriched{}
category in general, and again this property boils
down to a strong monoidal structure on $\Sfun$.
The functor $\Sfun$ has a canonical natural transformations
\[
  \Tuple{\Sfun \Proj l}_{l \in L} 
  \in \cL(\Sfun(\Bwith_{l \in L} X_l), \Bwith_{l \in L} \Sfun X_l)\,.
\] 

\begin{proposition}
  \label{prop:ss-cartesian}
  The following assertions are equivalent.
  \begin{enumerate}
  \item The category $\cL$ is a cartesian \Dpmenriched{} category.
  \item The indexed family $(\Bwith_{l \in L} \Sproj i)_{i \in \Mdeg}$
    is summable in
    $\cL(\Bwith_{l \in L} (\Sfun X), \Bwith_{l \in L} X)$ for all $L$.
  \item The morphism
    \(\Tuple{\Sfun \Proj l}_{l \in L} \in \cL(\Sfun(\Bwith_{l \in L}
    X_l), \Bwith_{l \in L} \Sfun X_l) \) is an iso for all $L$.
  \end{enumerate}
  Furthermore, when one of the conditions above holds,
  $\Tuple{\Sfun \Proj l}_{l \in L}^{-1} = \Stuple{\Bwith_{l \in L}
    \Sproj i}_{i \in \Mdeg}$.
\end{proposition}

\begin{proof}
  The proof is the same as the proof of \cref{prop:ss-closure}.
\end{proof}

We say that the \SUMBLTY{} structure satisfies
\labeltext{(\(\Sfun\)-\(\Bwith\))}{ax:swith}\ref{ax:swith} when one of
the equivalent conditions above holds.
In the sequel, when dealing with a (pre-)\SUMBLTY{} structure on a
cartesian category, we always assume that~\ref{ax:swith} holds.
We write
\begin{equation}
  \label{eq:Swithiso}
  \Swithiso_{X_1, X_2} = \Tuple{\Sfun \Proj 1, \Sfun \Proj 2}^{-1} 
  \in \cL(\Withp{\Sfun X_1}{\Sfun X_2}, \Sfun \Withp{X_1}{X_2}).
\end{equation}

\begin{proposition}
  \label{prop:S-strong-with}
  $(\Sfunm, (\Sin 0)_{\Top}, \Swithiso)$ is a strong monoidal monad.
\end{proposition}

\begin{proof}
  This can easily be checked by hand. Alternatively, by a standard
  result on monads~\cite[Prop.~22]{Mellies09}, the natural
  transformation
  $\Tuple{\Sfun \Proj{1}, \Sfun \Proj{2}} \in \cL(\Sfun
  \Withp{X_1}{X_2}, \Withp{\Sfun X_1}{\Sfun X_2})$ and the unique
  morphism in $\cL(\Sfun \Top, \Top)$ endow the monad $\Sfunm$ with
  an oplax monoidal structure.
  Then \ref{ax:swith} implies that these natural transformations are
  isos, with respective inverses $(\Sin 0)_{\Top}$ (as easily checked)
  and $\Swithiso$ (by definition).
  Hence, this oplax monoidal structure is strong monoidal.
\end{proof}

\section{\texorpdfstring{\SUMBLTY{}}{Summability}
  structures in the \Elementary{} case}
\label{sec:representable-sum}
%

It turns out that most non-trivial \SUMBLTY{} structures result from
very simple properties of the category \(\cL\): the \SUMBLTY{}
functor is given by $\Sfun = (\Limpl {\Dbimongen} {\_})$ for 
some specific object $\Dbimongen$.
Such a \SUMBLTY{} structures will be called \emph{\Elementary{}}.

\subsection{Joint epicity, joint monicity and currying}

\begin{definition}
  Let \(\cL\) be a category.
  A family of morphisms \(\Vect f=(f_i\in\cL(X_i,Y))_{i\in I}\) is
  jointly epic if it is jointly monic in \(\Op\cL\), that is: for all
  objects \(U\) of \(\cL\) and pairs of morphisms \(f,g\in\cL(Y,U)\),
  if \(\forall i\in I\ g\Compl f_i=h\Compl f_i\), then \(g=h\).

  Assume now that \(\cL\) is an SMC.
  A family of morphisms \(\Vect f=(f_i\in\cL(X_i,Y))_{i\in I}\) is
  completely jointly epic if, for all object \(U\) of \(\cL\), the
  family of morphisms \((\Tens{f_i}{U})_{i\in I}\) is jointly epic.
\end{definition}
By symmetry of the monoidal structure of \(\cL\),
\(\Vect f=(f_i\in\cL(X_i,Y))_{i\in I}\) is completely jointly epic iff
the family of morphisms \((\Tens{U}{f_i})_{i\in I}\) is jointly epic
for all objects \(U\) of \(\cL\).

\begin{lemma}
  \label{lemma:compl-joint-epic-smcc}
  In an SMCC \(\cL\), all jointly epic families are completely jointly
  epic.
\end{lemma}
\begin{proof}
  Let \(U\) be an object of \(\cL\) and assume that
  \((f_i\in\cL(X_i,Y))_{i\in I}\) is jointly epic.
  Let \(g,h\in\cL(\Tens XU,Y)\) be such that
  \(g\Compl{\Tensp{f_i}U}=h\Compl{\Tensp{f_i}U}\).
  For all \(i\in I\) we have
  \(\Curlinp g\Compl f_i=\Curlinp{g\Compl\Tensp{f_i}U}=\Curlinp h\Compl
  f_i\), and hence \(\Curlin g=\Curlin h\), which implies \(g=h\).
\end{proof}

\begin{lemma}
  \label{lemma:limpl-epic-to-monic}
  Let \(\cL\) be an SMCC and let
  \(\Vect f=(f_i\in\cL(X_i,Y))_{i\in I}\) be a family of morphisms.
  Then \(\Vect f\) is jointly epic iff for all objects \(U\) of
  \(\cL\) the family
  \(\Limplp{\Vect f}U=(\Limpl{f_i}U\in\cL(\Limpl
  YU,\Limpl{X_i}U))_{i\in I}\) is jointly monic.
\end{lemma}
\begin{proof}
  Let \(g,h\in\cL(Z,\Limpl YU)\) be such that
  \(\forall i\in I\ \Limplp{f_i}{U}\Compl g=\Limplp{f_i}{U}\Compl h\).
  We have
  \(\Limplp{f_i}{U}\Compl g =\Curlinp{\Evlin\Compl\Tensp{\Limplp
      YU}{f_i}}\Compl g =\Curlinp{\Evlin\Tensp g{f_i}}
  =\Curlinp{\Evlin\Compl\Tensp g{X_i}\Compl\Tensp{Z}{f_i}}\) and
  similarly for \(h\), and hence
  \(\Evlin\Compl\Tensp g{X_i}\Compl\Tensp{Z}{f_i}=\Evlin\Compl\Tensp
  h{X_i}\Compl\Tensp{Z}{f_i}\) for all \(i\in I\), that is
  \(\Curlininv(g)\Compl\Tensp{Z}{f_i}=\Curlininv(h)\Compl\Tensp{Z}{f_i}\),
  and since the \(\Tens{f_i}Z\) are jointly epic by
  \cref{lemma:compl-joint-epic-smcc}, we get \(g=h\) which shows that
  $\Limpl{\Vect f}U$ is jointly monic.

  Assume conversely that $\Limpl{\Vect f}U$ is jointly monic and let
  \(g,h\in\cL(Y,Z)\) be such that
  \(\forall i\in I\ g\Compl f_i=h\Compl f_i\).
  Since \(\Limplp{(g\Compl f_i)}{U}=\Limplp{f_i}U\Compl\Limplp{g}U\) and
  similarly for \(h\), we have \(\Limplp gU=\Limplp hU\).
  This holds for all \(U\), so we can take \(U=Z\) and by precomposing
  this morphism \(\Limplp ZZ\to\Limplp YZ\) with
  \(\Curlinp{\Leftu_Z}\in\cL(1,\Limpl ZZ)\) we get that \(g=h\).
\end{proof}

\subsection{\texorpdfstring{\SUMBLTY{}}{Summability} object}

\begin{convention}
  In this section, we assume that \(\cL\) is a closed \Pmsmc{}.
\end{convention}

\begin{definition}
  \label{def:pre-sumobj}
  \labeltext{(\(\Dbimongen\)-epi)}{ax:depi}\ref{ax:depi} A
  \emph{pre \Sumobj{}} is a tuple
  \((\Dbimongen,(\Win i)_{i\in\Mdeg},\Wdiag)\) where $\Dbimongen$ is an
  object of \(\cL\), \(\Wdiag\in\cL(\Sone,\Dbimongen)\), and the family
  of morphisms \((\Win i\in\cL(\Sone,\Dbimongen))_{i\in\Mdeg}\) is
  jointly epic.
\end{definition}

\begin{lemma} \label{lemma:pre-summability-object}
  Given an object \(\Dbimongen\) of \(\cL\) and morphisms
  \((\Win i\in\cL(\Sone,\Dbimongen))_{i\in\Mdeg}\) and
  \(\Wdiag\in\cL(\Sone,\Dbimongen)\), the tuple
  \((\Dbimongen,(\Win i)_{i\in\Mdeg},\Wdiag)\) is a pre-\Sumobj{} iff
  the functor \(\Sfun=\Sfung(\Dbimongen):\cL\to\cL\) (where $\Sfung$
  is the strong monoidal functor defined in
  \cref{sec:monoid-to-monad}), equipped with the natural
  transformations
  \((\Sproj i=\Invp{\Sfungmonz}\Compl\Sfung{(\Win
    i)}:\Sfun\Natto\Id)_{i\in\Mdeg}\) and
  \(\Ssum=\Invp{\Sfungmonz}\Compl\Sfung(\Wdiag):\Sfun\Natto\Id\),
  defines a pre-\SUMBLTY{} structure on the \Pmsmc{} $\cL$ (see
  \cref{def:pre-summ-struct}).
\end{lemma}
\begin{proof}
  Apply \cref{lemma:limpl-epic-to-monic}.
\end{proof}

\begin{lemma} %
  \label{lemma:tens-win-epic}
  Assume that \(\cL\) is equipped with a pre-\Sumobj{} for which we
  use the notations above.
  For any object \(X\) and any \(n\in\Nat\), the morphisms %
  \((X\ITens\Win{i_1}\ITens\cdots\ITens\Win{i_n} )_{\Vect
    i\in\Mdeg^n}\) %
  are jointly epic.
\end{lemma}
\begin{proof}
  Simple induction on \(n\).
\end{proof}

\begin{definition} %
  \label{def:sumobj}
  A pre \Sumobj{} in \(\cL\) is a \Sumobj{} if the induced
  pre-\SUMBLTY{} structure satisfies \cref{def:summ-struct}.
\end{definition}

The next two lemmas unfold the two conditions stated in
\cref{def:summ-struct} in the present setting.

\begin{lemma}
  \label{lemma:repr-summable-family}
  Let \((\Dbimongen,(\Win i)_{i\in\Mdeg},\Wdiag)\) be a \Sumobj{} in
  \(\cL\).
  Then a family \((f_i\in\cL(X,Y))_{i\in\Mdeg}\) is summable iff there
  is a \(h\in\cL(\Tens X\Dbimongen,Y)\) such that %
  \(f_i=h\Compl\Tensp\Dbimongen{\Win i}\Compl\Inv\Rightu\) for all
  \(i\in\Mdeg\), and then we have
  \(\sum_{i\in\Mdeg}f_i=f\Compl\Tensp\Dbimongen{\Wdiag}\Compl\Inv\Rightu\).
  We set \(\Stuplet{f_i}_{i\in I}=h\), so that %
  \(\Stuple{f_i}_{i\in I}=\A\Curlin{\Stuplet{f_i}_{i\in I}}\).
\end{lemma}
It makes sense to introduce the notation \(\Stuplet{f_i}_{i\in I}\)
because the \(\Win i\)'s are completely jointly epic.

\begin{lemma}
  \label{lemma:summability-param-family}
  Let \((\Dbimongen,(\Win i)_{i\in\Mdeg},\Wdiag)\) is a \Sumobj{} in
  \(\cL\) and \(I\) be an at most couuntable set which is finite if
  \(\Mdeg\) is finite.
  Then a family \((f_i\in\cL(\Tens X\Dbimongen,Y))_{i\in I}\) is
  summable if the family
  \((f_i\Compl\Tensp X\Wdiag\Compl\Inv\rho_X)_{i\in I}\) is
  summable.
  And a family \((f_i\in\cL(\Dbimongen,Y))_{i\in I}\) is summable if
  the family \((f_i\Compl\Wdiag\in\cL(\Sone,Y))_{i\in I}\) is
  summable.
\end{lemma}

\begin{lemma}
  \label{lemma:win-sumable}
  Let \((\Dbimongen,(\Win i)_{i\in\Mdeg},\Wdiag)\) be a \Sumobj{} in
  \(\cL\).
  The family of morphisms 
  \((\Win i\in\cL(\Sone,\Dbimongen))_{i\in\Mdeg}\) is summable
  and we have %
  \(\Stuplet{\Win
    i}_{i\in\Mdeg}=\Leftu\in\cL(\Tens\Sone\Dbimongen,\Dbimongen)\) and %
  \(\sum_{i\in\Mdeg}\Win i=\Wdiag\in\cL(\Sone,\Dbimongen)\).
\end{lemma}
\begin{proof}
  Obvious consequence of \cref{lemma:repr-summable-family}.
\end{proof}

\begin{proposition}
  \label{prop:sumobj-uniqueness}
  Let \((\Dbimongen,(\Win i)_{i\in\Mdeg},\Wdiag)\) and
  \((\Dbimongen',(\Win i')_{i\in\Mdeg},\Wdiag')\) be two \Sumobj{}s in a
  \Pmcat{} \(\cL\).
  Then there is a unique \(\phi\in\cL(\Dbimongen,\Dbimongen')\) such that
  \(\phi\Compl\Win i=\Win i'\) for all \(i\in\Mdeg\).
  Moreover \(\phi\) is an isomorphism and
  \(\phi\Compl\Wdiag=\Wdiag'\).
\end{proposition}
\begin{proof}
  We know that the \(\Win i'\) are summable in \(\cL(\Sone,\Dbimongen')\)
  and we take
  \(\phi=\Stuplet{\Win i'}_{i\in\Mdeg}\Compl\Inv{\Leftu_{\Dbimongen}}\).
  We have
  \(\phi\Compl\Win i=\Stuplet{\Win
    i'}_{i\in\Mdeg}\Compl\Tensp\Sone{\Win
    i}\Compl\Inv{\Leftu_\Sone}=\Win i'\) since
  \(\Leftu_\Sone=\Rightu_\Sone\).
  The uniqueness of \(\phi\) results from the joint epicity of the
  \(\Win i\)s.
  One defines similarly \(\phi'\in\cL(\Dbimongen',\Dbimongen)\) such that
  \(\phi'\Compl\Win i'=\Win i\) for each \(i\in\Mdeg\) and we deduce
  (using again the joint epicity of the \(\Win i\)s) that \(\phi'\) is
  the inverse of \(\phi\).
  Last, by \cref{lemma:summability-compl}, the morphisms
  \((\phi\Compl\Win i)_{i\in\Mdeg}\) are summable and have
  \(\phi\Compl\Wdiag\) as sum.
  Since \(\phi\Compl\Win i=\Win i'\) and since the \(\Win i'\) are
  summable with \(\Wdiag'\) as sum, we have
  \(\phi\Compl\Wdiag=\Wdiag'\).
\end{proof}

\begin{remark}
  \cref{prop:sumobj-uniqueness} means that, for an \Pmcat{} \(\cL\),
  the fact of possessing a \Sumobj{} is a \emph{property} and not a
  structure: we can speak of \emph{the} \Sumobj{} of \(\cL\), when it
  exists.
  This is of course completely similar to
  \cref{rk:unique-summability-structure} and motivates
  \cref{def:elementarily-summable-cat}.
\end{remark}

\begin{definition}
  \label{def:elementarily-summable-cat}
  A \emph{representably \SUMCAT} $\cL$ is a closed \Pmsmc{} which
  possesses a \Sumobj.
  We use \((\Dbimon,(\Win i)_{i\in\Mdeg},\Wdiag)\) for this
  summability object (which is uniquely defined up to unique iso as we
  have seen).
\end{definition}

\begin{remark}
  If \(\cL\) is additive (in the present context it suffices to assume
  that all the \Dpmonoid{} structures of the homsets are total), then
  \(\cL\) is enriched in commutative monoids and it is well-known that
  finite cartesian products are also coproducts.
  It is easy to prove, in that case, that
  \(\Dbimon=\Bplus_{i\in\Mdeg}\Sone=\Bwith_{i\in\Mdeg}\Sone\) and
  hence \(\Sfun X=\Limplp\Dbimon X\) is canonically isomorphic to
  \(\Bwith_{i\in\Mdeg}X\) and the \SUMBLTY{} structure is trivial:
  all families of morphisms \((f_i)_{i\in\Mdeg}\in\cL(X,Y)\) are
  summable with witness
  \(\Tuple{f_i}_{i\in\Mdeg}\in\cL(X,\Bwith_{i\in\Mdeg}Y)\) and sum
  \(\sum_{i\in\Mdeg}f_i\) (the addition provided by the enrichment).
\end{remark}

\begin{theorem}
  \label{th:rerpr-summable-tens-with-limpl}
  If \(\cL\) is \Elementarily{} \SUMBLE{} then~\ref{ax:swith}
  and~\ref{ax:stens} hold in \(\cL\).
  Moreover \ref{ax:slimpl} holds (remember that, as an SMC, \(\cL\) is
  assumed to be closed).
\end{theorem}
\begin{proof}[Proof sketch]
  The first property results from the fact that \(\Sfun\) is the right
  adjoint to the functor \(\Tens\_\Dbimon\) and as such preserves all
  existing limits.
  Next let \((f_i\in\cL(X,Y))_{i\in\Mdeg}\) be summable so that %
  \(\Stuplet{f_i}_{i\in\Mdeg}\in\cL(\Tens{X}{\Dbimon},Y)\), and let %
  \(g\in\cL(U,V)\).
  We have %
  \((\Tens{\Stuplet{f_i}_{i\in\Mdeg}}{g})\Compl\Tensp
  X{\Sym_{U,\Dbimon}} \in\cL(X\ITens U\ITens\Dbimon,\Tens YV)\).
  It is easily checked that %
  \((\Tens{\Stuplet{f_i}_{i\in\Mdeg}}{g})\Compl\Tensp X{\Sym_{U,\Dbimon}}
  =\Stuplet{\Tens{f_i}g}_{i\in\Mdeg}\).

  If \(X\) and \(X\) are objects of \(\cL\) then the morphism
  \(\Sfunstri_{X,Y}\in\cL(\Sfun{\Limplp XY},\Limpl X{\Sfun X})
  =\cL(\Limpl{\Dbimon}{\Limplp XY},\Limpl X{\Limplp\Dbimon Y})\) is
  \(\Curlin{(\Curlin\phi)}\) where \(\phi\) is the following
  composition of morphisms in \(\cL\)
  \begin{center}
    \begin{tikzcd}
      \Tens{(\Limpl{\Dbimon}{\Limplp XY})}{\Tens X\Dbimon}
      \ar[r,"\Tens\Id\Sym"]
      &[-0.5em]
      \Tens{(\Limpl{\Dbimon}{\Limplp XY})}{\Tens\Dbimon X}
      \ar[r,"\Tens\Evlin\Id"]
      &
      \Tens{\Limplp XY}{X}
      \ar[r,"\Evlin"]
      &[-1em]
      Y
    \end{tikzcd}
  \end{center}
  so that \(\Sfunstri_{X,Y}\) is an iso.
\end{proof}

Even if \(\Dbimon\) is not necessarily a \(\Mdeg\)-fold product of
\(\Sone\), it nevertheless has a projection \(\Dbimon\to\Sone\) for
each element of \(\Mdeg\), as we show now.
\begin{lemma}
  \label{lemma:sumobj-proj}
  Assume that \((\Dbimon,(\Win i)_{i\in D},\Wdiag)\) is a \Sumobj{} in
  \(\cL\).
  There is a family of morphisms
  \((\Proj i\in\cL(\Dbimon,\Sone))_{i\in\Mdeg}\) (called
  \emph{projections}) uniquely characterized by
  \(\Proj i\Compl\Win j=\Kronecker ij\Id_\Sone\).
\end{lemma}
\begin{proof}
  Let \(i\in\Mdeg\).
  The family of morphisms
  \((\Kronecker ij\Id_\Sone\in\cL(\Sone,\Sone))_{j\in\Mdeg}\) is
  summable because \(\cL(\Sone,\Sone)\) is an \Pmonoid.
  We set
  \(\Proj i=\Stuplet{\Kronecker
    ij\Id_\Sone}_{j\in\Mdeg}\Inv{\Leftu}\).
\end{proof}

\begin{remark}
  Very often, it turns out that the \Sumobj{} in \(\cL\) is the
  \(\Mdeg\)-fold cartesian product of \(\Sone\), that is
  \(\Dbimon=\Bwith_{i\in\Mdeg}\Sone\).
  This is the case of all the examples we consider in this paper.
  In that case \(\Wdiag=\Tuple{\Id_\Sone}_{j\in\Mdeg}\) is the
  ordinary diagonal morphism associated with the cartesian product and
  \(\Win i=\Tuple{\Kronecker ij\Id_\Sone}_{j\in\Mdeg}\).
  When the \Sumobj{} is such a product, it is easy to prove that the
  \(\Proj i\)s are the ordinary projection of the cartesian product.

  In our previous presentations of coherent differentiation, we
  restricted our attention to such product \Sumobj{}s.
  However, one of the authors found in~\cite{TassonWalch26} an example
  of a representably summable category where the \Sumobj{} cannot be
  described as the \(\Mdeg\)-fold cartesian product of \(\Sone\): the
  category of Köthe sequence space~\cite{Ehrhard02}.
  Such examples motivate our more general definition.
  Notice however that this generalization has no effect on the theory
  we develop around \(\Dbimon\): even in the case of product \Sumobj{}
  case, the only useful structure and properties of \(\Dbimon\) that
  we used in the earlier papers on coherent differentiation were those
  of \cref{def:pre-sumobj} and \cref{def:sumobj}.
\end{remark}

Notice that the family
\((\Proj i\in\cL(\Dbimon,\Sone))_{i\in\Mdeg}\) has no reason to be
summable.
However we have the following property.
\begin{lemma}
  \label{lemma:win-proj-summable}
  If \(\cL\) is \Elementarily{} \SUMBLE{} then the family
  \((\Win i\Compl\Proj i\in\cL(\Dbimon,\Dbimon))_{i\in\Mdeg}\) is
  summable, with \(\sum_{i\in\Mdeg}\Win i\Compl\Proj i=\Id_{\Dbimon}\).
\end{lemma}
\begin{proof}
  Observe that \(\Win i\Compl\Proj i\Compl\Wdiag=\Win i\) and apply
  \cref{lemma:summability-param-family,{lemma:win-sumable}}.
  To prove the equation, it suffices to precompose with \(\Win j\) for
  each \(j\in\Mdeg\) and use \cref{lemma:summability-compl} and the
  joint epicity of these morphisms.
\end{proof}

\begin{lemma}
  \label{lemma:sum-fam-proj}
  If the family \((f_i\in\cL(X,Y))_{i\in\Mdeg}\) is summable then the
  family
  \((\Rightu\Compl\Tensp{f_i}{\Proj i}\in\cL(\Tens
  X\Dbimon,Y))_{i\in\Mdeg}\) is also summable.
\end{lemma}

We can think of the ``elements'' of \(\Dbimon\) as good families
  \((\alpha_i)_{i\in\Mdeg}\) of coefficients%
  \footnote{The meaning of ``good'' depends on \(\cL\); for instance,
    if \(\cL=\PCOH\) (see \cref{ex:pcoh-cart-smc-elem-sum})
    this property means
    \(\forall i\ \alpha_i\in\Intercc01\subseteq\Real\).}, %
  and \cref{lemma:sum-fam-proj} says that if the family
  \((f_i\in\cL(X,Y))_{i\in\Mdeg}\) is summable, then the family
  \((\alpha_if_i\in\cL(X,Y))_{i\in\Mdeg}\) is also summable.
  
\begin{proof}
  By \cref{lemma:sone-tens-comp,lemma:win-proj-summable}, the family %
  \((\Tens{\Win i}{\Proj
    i}\in\cL(\Tens\Sone\Dbimon,\Tens\Dbimon\Sone))_{i\in\Mdeg}\) is
  summable.
  We have a witness
  \(h= \Stuplet{f_i}_{i\in\Mdeg}\in\cL(\Tens X\Dbimon,Y)\) of the
  summability of the \(f_i\)'s.
  Now
  \begin{align*}
    \Tens{f_i}{\Proj i}
    &=\Tens{(h\Compl\Win i\Compl\Inv\Rightu)}{\Proj i}\\
    &=\Tensp h\Sone
      \Compl(X\ITens{\Win i}\ITens{\Proj i})\Compl\Tensp{\Inv\Rightu}{\Dbimon}
  \end{align*}
  typed as follows
  \begin{center}
    \begin{tikzcd}
      \Tens X\Dbimon
      \ar[r,"\Tens{\Inv\Rightu}{\Dbimon}"]
      &[1em]
      X\ITens\Sone\ITens\Dbimon
      \ar[r,"X\ITens{\Win i}\ITens{\Proj i}"]
      &[2em]
      X\ITens\Dbimon\ITens\Sone
      \ar[r,"\Tens h\Sone"]
      &
      \Tens Y\Sone
    \end{tikzcd}
  \end{center}
  and we conclude the proof using
  \cref{th:rerpr-summable-tens-with-limpl,lemma:summability-compl}.
\end{proof}

\begin{lemma}
  \label{lemma:comp-proj-summable}
  If the family \((f_i\in\cL(\Sone,X))_{i\in\Mdeg}\) is summable then
  the family \((f_i\Compl\Proj i)_{i\in\Mdeg}\in\cL(\Dbimon,X)\) is
  summable.
\end{lemma}
\begin{proof}
  Apply \cref{lemma:sum-fam-proj,lemma:sone-tens-comp}.
\end{proof}

\subsection{Examples}

We describe the representable \SUMBLTY{} structures of coherence
spaces, that we have already encountered in \cref{sec:diff-in-LL}, and
of probabilistic coherence spaces that we will first define.

  \subsubsection{Coherence spaces}
  \label{ex:coh-cart-smc-elem-sum}
  Remember that the cartesian product in \(\COH\) of a family
  \((E_l)_{l\in L}\) is %
  \(E=\Bwith_{l\in L}E_l\) defined by
  \(\Web E=\Union_{l\in L}\Eset l\times\Web{E_l}\) with %
  \((l,a)\Coh E(l',a')\) if \(l=l'\Implies a\Coh{E_l}a'\) and
  projections %
  \(\Proj l=\Set{((l,a),a)\St l\in L\text{ and }a\in\Web{E_l}}
  \in\COH(E,E_l)\) as easily checked.
  Given \((t_l\in\COH(F,E_l))_{l\in L}\), the unique %
  \(\Tuple{t_l}_{l\in L}\COH(F,E)\) such that %
  \((\Proj q\Compl\Tuple{t_l}_{l\in L}=t_q)_{q\in L}\) is %
  \(\Tuple{t_l}_{l\in L}
  =\Set{(b,(l,a))\St l\in L\text{ and }(b,a)\in t_l}\).

  The object \(\Dbimon=\Bwith_{i\in\Mdeg}\Sone\) can be described by %
  \(\Web\Dbimon=\Mdeg\) with \(\forall i,j\in\Mdeg\ i\Coh\Dbimon j\)
  and we have %
  \(\Win i=\Eset{(\Sonelem,i)}\) and %
  \(\Wdiag=\Eset{(\Sonelem,i)\St i\in\Mdeg}\).
  It is easy to check that \(\COH\) is \Elementarily{} \SUMBLE{}, but
  we will not describe immediately its \Dpmonoid{} local structure; we
  prefer to show how it derives from the structure of \(\Dbimon\).
  This object induces the functor
  \(\Sfun=\Sfung\Dbimon:\COH\to\COH\) acting on objects by
  \(\Sfun E=\Limplp\Dbimon E\) and hence we have
  \(\Web{\Sfun E}=\Mdeg\times\Web E\) and %
  \((i,a)\Coh{\Sfun E}(i',a')\) if \(a\Coh Ea'\) and
  \(i\not=i'\Implies a\not=a'\).
  Therefore, up to a trivial order isomorphism, %
  \(\Cl{\Sfun E} =\Eset{(x_i)_{i\in\Mdeg}\in\Cl E^\Mdeg\St \forall
    i,j\in\Mdeg\ i\not=j\Implies x_i\cup x_j\in\Cl E \text{ and
    }x_i\cap x_j=\emptyset}\) (with the product order).
  Therefore we decide to equip the category \(\COH\) with the
  following \Dpmonoid{} local structure:
  a family \((s_i\in\COH(E,F))_{i\in I}\) is summable if
  \(\Union_{i\in I}s_i\in\COH(E,F)\) (which is not surprising) and
  \(s_i\cap s_j=\emptyset\) for \(i\not=j\) (which is more).
  And then \(\sum_{i\in I}s_i=\Union_{i\in I}s_i\).
  The corresponding witness is %
  \(\Stuple{s_i}_{i\in I}=\Eset{(a,(i,b))
    \St i\in\Mdeg\text{ and } (a,b)\in s_i}
  \in\COH(E,\Sfun F)\).

  \subsubsection{Probabilistic coherence spaces}
  \label{ex:pcoh-cart-smc-elem-sum}
  Another crucial example is that of \emph{probabilistic coherence
    spaces} (PCS) introduced in~\cite{Girard04a,DanosEhrhard08}.

  \begin{definition}
    \label{def:pcs-struct-def}
    A PCS is a pair \(X=(\Web X,\Pcoh X)\) where \(\Web X\) is a set
    and \(\Pcoh X\subseteq\Realpto{\Web X}\) that we consider as a
    poset (whose order relation is the product order) and that we
    assume to satisfy%
    \begin{itemize}
    \item \(\Pcoh X\) is non-empty;
    \item for all \(a\in\Web X\), the set
      \(\Set{x_a\St x\in\Pcoh X}\subseteq\Realp\) is bounded and is
      not reduced to \(\Set 0\);
    \item \(\Pcoh X\) is down-closed (that is if \(x\in\Pcoh X\) and
      \(y\in\Realpto{\Web X}\) satisfy \(y\leq x\) then
      \(y\in\Pcoh X\)) and closed under the lubs of monotone chains %
      (that is if \((x(n)\in\Pcoh X)_{n\in\Nat}\) is increasing
      componentwise, then %
      \((\sup_{n\in\Nat} x(n)_a)_{a\in\Web X}\in\Pcoh X\));
    \item \(\Pcoh X\) is closed under barycentric combinations, that
      is, %
      if \(x,y\in\Pcoh X\) and \(p\in\Intercc 01\), then %
      \((1-p)x+py\in\Pcoh X\) (the algebraic operations being defined
      pointwise).
    \end{itemize}
  \end{definition}

  \begin{remark}
    \label{rk:pcs-bipolar}
    In the literature, the definition of PCSs is most often
    based on a duality typical of \LL{}; the present definition is
    equivalent to the duality-based definition as shown
    in~\cite{Girard04a,Ehrhard22}, and is perhaps more intuitive.
    More precisely, given \(\cP\subseteq\Realpto L\) for some set
    \(L\), we can define
    \(\Orth\cP=\Set{x'\in\Realpto L\St\forall x\in\cP\ \Eval x{x'}\leq
      1}\subseteq\Realpto I\) where
    \(\Eval x{x'}=\sum_{l\in L}x_lx'_l\).
    Then, under the assumption that \(\Pcoh X\) fulfills the two first
    conditions of \cref{def:pcs-struct-def}, the two last conditions
    are equivalent to \(\Pcoh X=\Biorth{\Pcoh X}\).
  \end{remark}
  A morphism from \(X\) to \(Y\) is a matrix
  \(t\in\Realpto{\Web X\times\Web Y}\) such that for all
  \(x\in\Pcoh X\), one has \(\Matappa tx\in\Pcoh Y\) where
  \(\Matappa tx=(\sum_{a\in\Web X}t_{a,b}x_a)_{b\in\Web
    Y}\in\Realpto{\Web Y}\).
  These morphisms are in bijective correspondence with the functions
  \(f:\Pcoh X\to\Pcoh Y\) which are monotone, commute with lubs of
  monotone sequences and satisfy \(f((1-p)x+py)=(1-p)f(x)+pf(y)\): %
  given such an \(f\) we define its matrix
  \(t\in\Realpto{\Web X\times\Web Y}\) as follows.
  Let \((a,b)\in\Web X\times\Web Y\), then there is an
  \(\epsilon\in\Realpnz\) such that \(\epsilon\Base a\in\Pcoh X\)
  (where \(\Base a\in\Realpto{\Web X}\) is defined by
  \((\Base a)_{a'}=\Kronecker a{a'}\)) and we set
  \(t_{a,b}=\Inv\epsilon f(\epsilon\Base a)_b\) which does not depend
  on the choice of \(\epsilon\).
  In that way we have defined a category \(\PCOH\) where composition
  is defined by the usual product of matrices: if \(s\in\PCOH(X,Y)\)
  and \(t\in\PCOH(Y,Z)\) then \(\Matapp ts\in\PCOH(X,Z)\) is defined
  by \((\Matapp ts)_{a,c}=\sum_{b\in\Web Y}s_{a,b}t_{b,c}\).
  The identity morphism \(\Id\in\PCOH(X,X)\) is the diagonal matrix,
  \(\Id_{a,a'}=\Kronecker a{a'}\).

  The category \(\PCOH\) is an SMC.
  The tensor unit is \(\Sone=(\Set\Sonelem,\Intercc 01)\) (upon
  identifying \(\Realpto{\Set\Sonelem}\) with \(\Realp\)).
  Given \((x(i)\in\Pcoh{X_i})_{i=1,2}\), we set %
  \(\Tens{x(1)}{x(2)}
  =(x(1)_{a_1}x(2)_{a_2})_{(a_1,a_2)\in\Web{X_1}\times\Web{X_2}}
  \in\Realpto{\Web{X_1}\times\Web{X_2}}\).
  We can define \(\Tens{X_1}{X_2}\) as
  \((\Web{X_1}\times\Web{X_2},P)\) where \(P\) is the least subset of
  \(\Realpto{\Web{X_1}\times\Web{X_2}}\) which contains all the %
  \(\Tens{x(1)}{x(2)}\) for \((x(i)\in\Pcoh{X_i})_{i=1,2}\) and
  satisfies the two closure properties in the definition of PCS (the
  two last conditions).
  To describe more explicitly this operation, it is convenient to
  introduce the PCS \(\Limpl XY=(\Web X\times\Web Y,Q)\) where
  \(Q=\PCOH(X,Y)\).
  Then \(\Limpl X\Sone\) is (trivially) isomorphic to the PCS
  \(\Orth X\) where \(\Web{\Orth X}=\Web X\) and
  \(x'\in\Realpto{\Web X}\) belongs to \(\Pcoh{\Orth X}\) if, for all
  \(x\in\Pcoh X\), one has
  \(\Eval x{x'}=\sum_{a\in\Web X}x_ax'_a\leq 1\).
  It is then possible to prove that \(\Biorth X=X\) (a kind of
  ``bipolar theorem'', see \cref{rk:pcs-bipolar}).
  Then one has \(\Tens{X_1}{X_2}=\Orth{\Limplp{X_1}{\Orth{X_2}}}\)
  and, based on this property, when equipped with \(\Sone\) and
  \(\ITens\), the category \(\PCOH\) is an SMC.
  This SMC is closed with \((\Limpl XY,\Evlin)\) as internal hom from
  \(X\) to \(Y\), with
  \(\Evlin_{((a,b),a'),b'}=\Kronecker a{a'}\Kronecker b{b'}
  \in\cL(\Tens{\Limplp XY}{X},Y)\) as evaluation morphism, which
  satisfies of course \(\Matappa\Evlin{\Tensp tx}=\Matappa tx\) for
  \(t\in\PCOH(X,Y)\) and \(x\in\Pcoh X\).
  This SMCC is even \Staraut{} with \(\Sbot=\Sone\) as dualizing
  object (this essentially boils down to the fact that
  \(\Biorth X=X\)).
  
  The category \(\PCOH\) has all products: given a family
  \((X_l)_{l\in L}\) of PCSs, let \(X\) be defined by
  \(\Web X=\Union_{l\in L}\Set l\times\Web{X_l}\) and %
  \(x\in\Realpto{\Web X}\) belongs to \(\Pcoh X\) if %
  \(((x_{l,a})_{a\in\Web{X_i}}\in\Pcoh{X_l})_{l\in L}\).
  Equipped with %
  \((\Proj l\in\PCOH(X,X_l))_{l\in L}\) %
  defined by \((\Proj q)_{(l,a),a'}=\Kronecker ql\Kronecker a{a'}\), %
  \(X\) is easily seen to be the cartesian product of the \(X_l\)'s,
  and we set \(\Bwith_{l\in L}X_l=X\).
  Notice that \(\Pcohp{\Bwith_{l\in L}X_i}\) is isomorphic (for the
  order and for the barycentric structures) to
  \(\prod_{l\in L}\Pcoh{X_l}\).
  Given \((t_l\in\PCOH(Y,X_l))_{l\in L}\), the unique %
  \(\Tuple{t_l}_{l\in L}\in\PCOH(Y,\Bwith_{l\in L}X_l)\) such that %
  \((\Proj q\Compl\Tuple{t_l}_{l\in L}=t_q)_{q\in L}\) is given by
  \((\Tuple{t_l}_{l\in L})_{b,(q,a)}=(t_q)_{b,a}\) for all \(q\in L\),
  \(b\in\Web Y\) and \(a\in\Web{X_q}\).

  As in \cref{ex:coh-cart-smc-elem-sum}, we slightly postpone
  the definition of the \Pmcat{} structure of \(\PCOH\).
  The object
  \(\Dbimon=\Bwith_{i\in\Mdeg}\Sone\) is described by
  \(\Web\Dbimon=\Mdeg\) and \(\Pcoh\Dbimon=\Intercc01^\Mdeg\).
  Then \((\Win i\in\PCOH(\Sone,\Dbimon))_{i\in\Mdeg}\) and
  \(\Wdiag\in\PCOH(\Sone,\Dbimon)\) are given %
  by \(\Matappa{\Win i}u=(\Kronecker iqu)_{q\in\Mdeg}\) and %
  \(\Matappa\Wdiag u=(u)_{q\in\Mdeg}\), for all \(u\in\Intercc01\).
  Given a PCS \(X\), the PCS \(\Sfun X=\Limplp\Dbimon X\) is given by
  \(\Web{\Sfun X}=\Mdeg\times\Web X\) and an element of
  \(\Pcohp{\Sfun X}\) is a \(t\in\Realpto{\Mdeg\times\Web X}\)
  such that \(\Matappa t{\sum_{i\in\Mdeg}\Base i}\in\Pcoh X\) where
  \(\sum_{i\in\Mdeg}\Base i\in\Pcoh\Dbimon\) is the tuple
  \((1)_{i\in\Mdeg}\in\Intercc01^\Mdeg\).
  In other words
  \begin{align}
    \label{eq:ssum-pcoh-isom}
    \Pcohp{\Sfun X}\Isom\Set{(x(i))_{i\in\Mdeg}
    \in\Pcoh X^\Mdeg\St \sum_{i\in\Mdeg}x(i)\in\Pcoh X}\,,
  \end{align}
  this isomorphism holding for the poset and for the barycentric
  structures.
  With this identification, the morphisms %
  \(\Sproj i,\Ssum\in\PCOH(\Sfun X,X)\) (for all \(i\in\Mdeg\)) are
  characterized by %
  \(\Matapp{\Sproj i}{(x(q))_{q\in\Mdeg}}=x(i)\) and %
  \(\Matapp{\Ssum}{(x(i))_{i\in\Mdeg}}=\sum_{i\in\Mdeg}x(i)\).
  It follows easily from these observations that, if we want
  \cref{def:summ-struct} to apply to the current situation, we have to
  stipulate that a family of morphisms
  \((t(i)\in\PCOH(Y,X))_{i\in\Mdeg}\) is summable iff %
  \(\forall y\in\Pcoh Y\ \sum_{i\in\Mdeg}\Matappa{t(i)}y\in\Pcoh X\).
  This condition implies easily that the sum of matrices
  \(t=\sum_{i\in I}t(i)\) belongs to \(\PCOH(Y,X)\) iff we have
  \(\forall y\in\Pcoh Y\ \Matappa ty=\sum_{i\in I}\Matappa{t(i)}y\),
  and we turn in that way \(\PCOH\) into a \Pmcat.
  The witness of a summable family \((t(i)\in\PCOH(Y,X))_{i\in\Mdeg}\)
  is then \(s=\Stuple{t(i)}_{i\in\Mdeg}\in\PCOH(Y,\Sfun X)\) which is
  given as a matrix by \(s_{b,(i,a)}=t(i)_{b,a}\) and is characterized
  by %
  \(\Matappa sy=\Stuple{\Matappa{t(i)}y}_{i\in\Mdeg}\) %
  if we consider \cref{eq:ssum-pcoh-isom} as an equality.
  From this characterization of summability and witnesses, it follows
  easily that \(\PCOH\) is a \Elementarily{} summable category.

\subsection{The structures of the \texorpdfstring{\SUMBLTY{}}{summability} 
functor, in the representable case}

In this section we assume that \(\cL\) is a representably \SUMCAT{}.
Thanks to \cref{th:rerpr-summable-tens-with-limpl}, all the structures
presented in \cref{sec:ssum-monad-structure,sec:summability-tensor}
are available for the functor
\(\Sfun=\Sfung\Dbimon=\Limplp{\Dbimon}{\_}\).
The goal of this section is to show that these structures boil down to
very elementary (co)algebraic structures available on \(\Dbimon\)
that we will make explicit.
We describe first a comonoid structure $\Dbimon$, which will correspond 
trough the strong monoidal functor $\Sfung$ to the monad structure of 
$\Sfun$.

\begin{proposition}
  \label{prop:dbimon-comonoid}
  There is a unique \(\Dbimonm\in\cL(\Dbimon,\Tens\Dbimon\Dbimon)\)
  such that
  \[
    \forall i\in\Mdeg\quad \Dbimonm\Compl\Win i
    =\sum_{(l,r)\in\Mdegsl i}\Tens{\Win l}{\Win r}
  \]
  (keeping the iso %
  \(\Leftu_1=\Rightu_1\in\cL(\Tens{\Sone}{\Sone},\Sone)\) implicit).
  The triple \((\Dbimon,\Proj0,\Dbimonm)\) is a commutative comonoid
  that we denote as \(\Dbimonco\).
  Moreover \(\Win0\in\CMON\cL(\Sone,\Dbimonco)\).
\end{proposition}
\begin{proof}
  By \cref{lemma:win-sumable} the family
  \((\Win i\in\cL(\Sone,\Dbimon))_{i\in\Mdeg}\) is summable %
  (with \(\Stuplet{\Win i}_{i\in\Mdeg}=\Id_{\Dbimon}\) and %
  \(\sum_{i\in\Mdeg}\Win i=\Wdiag\)).
  Thus, the family of morphisms %
  \((\Tens{\Win i}{\Win j}
  \in\cL(\Tens\Sone\Sone,\Tens\Dbimon\Dbimon))_{(i,j)\in\Mdeg^2}\) %
  is summable by \cref{th:rerpr-summable-tens-with-limpl}.
  For each \(i\in\Mdeg\), by \ref{ax:pa},
  the family %
  \((\Tens{\Win l}{\Win r})_{(l,r)\in \Mdegsl i}\) is summable, and hence
  the morphisms %
  \((\sum_{(l,r)\in \Mdegsl i}\Tens{\Win l}{\Win r})_{i\in\Mdeg}\) are
  summable.
  This gives us a witness %
  \(\Stuplet{\sum_{l+r=i}\Tens{\Win l}{\Win r}}_{i\in\Mdeg}
  \in\cL(\Tens{\Tens\Sone\Sone}\Dbimon,\Tens\Dbimon\Dbimon)\)
  and we set %
  \[
    \Dbimonm =\Stuplet{\sum_{(l,r)\in\Mdegsl i}\Tens{\Win l}{\Win
        r}}_{i\in\Mdeg} \Compl\Inv\Leftu\Compl\Inv\Leftu
    \in\cL(\Dbimon,\Tens\Dbimon\Dbimon)
  \]
  which satisfies the announced equations by construction.
  Uniqueness results from the joint epicity of the \(\Win i\)'s.
  The fact that we define in that way a commutative comonoid is easily
  proved, again by joint epicity of \(\Win i\)'s.
  Let us check associativity, which is the most interesting property
  in this setting.
  We must prove that the following diagram commutes in \(\cL\):
  \begin{center}
    \begin{tikzcd}
      \Dbimon
      \ar[r,"\Dbimonm"]
      \ar[d,swap,"\Dbimonm"]
      &
      \Tens\Dbimon\Dbimon
      \ar[r,"\Tens{\Dbimonm}{\Dbimon}"]
      &
      \Tens{\Tensp\Dbimon\Dbimon}{\Dbimon}
      \ar[d,"\Assoc"]
      \\
      \Tens\Dbimon\Dbimon
      \ar[rr,"\Tens{\Dbimon}{\Dbimonm}"']
      &&
      \Tens{\Dbimon}{\Tensp\Dbimon\Dbimon}
    \end{tikzcd}
  \end{center}
  so let \(i\in\Mdeg\).
  We have
  \begin{align*}
    \Assoc\Compl\Tensp{\Dbimonm}{\Dbimon}\Compl\Dbimonm\Compl\Win i
    &=\Assoc\Compl\Tensp{\Dbimonm}{\Dbimon}\Compl
      \Big(\sum_{(l,r)\in \Mdegsl i}\Tens{\Win l}{\Win r}\Big)\\
    &=\Assoc\Compl
      \Big(\sum_{(l,r)\in \Mdegsl i}
      \Tens{(\Dbimonm\Compl\Win l)}{\Win r}\Big)\\
    &=\Assoc\Compl
      \Big(\sum_{(l,r)\in \Mdegsl i}
      \Tens{\big(\sum_{(p,q)\in \Mdeg_l}\Tens{\Win p}{\Win q}\big)}
      {\Win r}\Big)\\
    &=\Assoc\Compl
      \Big(\sum_{\Biind{(l,r)\in \Mdegsl i}{(p,q)\in \Mdeg_l}}
      \Tens{\Tensp{\Win p}{\Win q}}{\Win r}\Big)\\
    &=\sum_{\Biind{(l,r)\in \Mdegsl i}{(p,q)\in \Mdeg_l}}
      \Tens{\Win p}{\Tensp{\Win q}{\Win r}}
  \end{align*}
  and the commutation follows from the observation that 
  \begin{multline*}
    \Eset{(p,q,r)\in\Mdeg^3\St
      \exists l\in\Mdeg\ (p,q)\in \Mdegsl l\text{ and }(l,r)\in \Mdegsl i}\\
    =
    \Eset{(p,q,r)\in\Mdeg^3\St
      \exists l\in\Mdeg\ (q,r)\in \Mdegsl l\text{ and }(p,l)\in \Mdegsl i}
  \end{multline*}
  and, as usual, from the joint epicity of the \(\Win i\)'s.
  This equality of sets results from the assumption that \(\Mdeg\) is
  a \Pmonoid{} and \ref{ax:pa}.

  The commutations involving the counit \(\Proj0\) result from the
  fact that %
  \(\Proj i\Compl\Win j=\Kronecker ij\Id_\Sone\).
  Commutativity results from the symmetry of the sets \(\Mdegsl i\). 

  Last we have \(\Win0\in\CMON\cL(\Sone,\Dbimonco)\):
  \(\Proj0\Compl\Win0=\Id_\Sone\) readily results from the definition
  of \(\Win0\) as \(\Tuple{\Kronecker i0\Id_\Sone}_{i\in\Mdeg}\)
  and
  \(\Dbimonm\Compl\Win0=\sum_{(l,r)\in\Mdegsl 0}\Tens{\Win l}{\Win
    r}=\Tens{\Win0}{\Win0}\) since \(\Mdegsl0=\{(0,0)\}\) by our
  assumption that \(\Mdeg\) is \Mpositive.
  Notice that $\Win i$ has no reason to be a comonoid morphism for
  $i\not=0$.
\end{proof}

\begin{lemma}
  \label{lemma:dbimonm-sum-expr}
  We have
  \begin{align*}
    \Dbimonm=\sum_{i\in\Mdeg}\Big(
    \sum_{(l,r)\in\Mdegsl i}\Tens{\Win l}{\Win r}\Big)
    \Compl\Inv\Leftu\Compl\Proj i\,.
  \end{align*}
\end{lemma}
\begin{proof}
  First, we need to make sure that this sum is well defined.
  Observe that the family
  \(\big(\sum_{(l,r)\in\Mdegsl i}\Tens{\Win l}{\Win
    r}\in\cL(\Tens\Sone\Sone,\Tens\Dbimon\Dbimon)\big)_{i\in\Mdeg}\) is
  summable because the family
  \((\Tens{\Win l}{\Win r})_{(l,r)\in\Mdeg^2}\) is summable.
  Hence the family
  \(\Big(\big(\sum_{(l,r)\in\Mdegsl i}\Tens{\Win l}{\Win
    r}\big)\Compl\Inv\Leftu\Compl\Proj
  i\Big)\in\cL(\Dbimon,\Tens\Dbimon\Dbimon)_{i\in\Mdeg})\) is summable
  by \cref{lemma:comp-proj-summable}.
  Let \(f\in\cL(\Dbimon,\Tens\Dbimon\Dbimon)\) be their sum; given
  \(j\in\Mdeg\) we have
  \begin{align*}
    f\Compl\Win j
    &=\sum_{i\in\Mdeg}\Big(
      \sum_{(l,r)\in\Mdegsl i}\Tens{\Win l}{\Win r}\Big)
      \Compl\Inv\Leftu\Compl\Proj i\Compl\Win j\\
    &=\sum_{i\in\Mdeg}\Big(
      \sum_{(l,r)\in\Mdegsl i}\Tens{\Win l}{\Win r}\Big)
      \Compl\Inv\Leftu\Compl(\Kronecker ij\Id_\Sone)
      \text{\quad by \cref{lemma:sumobj-proj}}\\
    &=\Big(
      \sum_{(l,r)\in\Mdegsl j}\Tens{\Win l}{\Win r}\Big)
      \Compl\Inv\Leftu=\Dbimonm\Compl\Win j
  \end{align*}
  from which the announced result follows by joint epicity of the \(\Win j\)'s.
\end{proof}

We now describe a commutative monoid structure on $\Dbimon$.
%
\begin{theorem}
  \label{th:dbimon-monoid-bimonoid}
  There is a unique \(\Dbimonl\in\cL(\Tens\Dbimon\Dbimon,\Dbimon)\)
  such that %
  \(\Dbimonl\Compl\Tensp{\Win i}{\Win j} =\Kronecker ij\Win i\) %
  (keeping the iso %
  \(\Leftu_1=\Rightu_1\in\cL(\Tens{\Sone}{\Sone},\Sone)\) implicit).
  The triple \((\Dbimon,\Wdiag,\Dbimonl)\) is a commutative monoid
  that we denote as \(\Dbimonmo\).
  If \(\Mdeg\) is total then
  \((\Dbimon,\Wdiag,\Dbimonl,\Proj0,\Dbimonm)\) is a bicommutative
  bimonoid in the sense of \cref{def:symmetric-bimonoid}, that we denote as
  \(\Dbimonbi\).
  We recall that this means that the following additional diagrams commute.
\begin{center}
  \begin{tikzcd}
    \Sone\ar[r,"\Wdiag"]\ar[rd,swap,"\Id"]
    &
    \Dbimon\ar[r,"\Proj0"]
    &
    \Sone\ar[dl,"\Id"]
    \\
    &\Sone&
  \end{tikzcd}
  \quad
  \begin{tikzcd}
    \Tens{\Dbimon}{\Dbimon}
    \ar[r,"\Slift"]\ar[d,swap,"\Tens{\Dbimonm}{\Dbimonm}"]
    &[-1em]
    \Dbimon\ar[r,"\Dbimonm"]
    &[-1em]
    \Tens{\Dbimon}{\Dbimon}
    \\
    \Tens{\Tens{\Dbimon}{\Dbimon}}{\Tens{\Dbimon}{\Dbimon}}
    \ar[rr,"\Dbimon\ITens\Sym\ITens\Dbimon"']
    &&
    \Tens{\Tens{\Dbimon}{\Dbimon}}{\Tens{\Dbimon}{\Dbimon}}
    \ar[u,swap,"\Tens{\Slift}{\Slift}"]
  \end{tikzcd}
\end{center}
  \begin{center}
  \begin{tikzcd}
    \Tens{\Dbimon}{\Dbimon}
    \ar[r,"\Slift"]\ar[d,swap,"\Tens{\Proj0}{\Proj0}"]
    &
    \Dbimon\ar[d,"\Proj0"]
    \\
    \Tens\Sone\Sone\ar[r,"\Leftu_\Sone"']
    &
    \Sone
  \end{tikzcd}
  \quad
  \begin{tikzcd}
    \Sone\ar[r,"\Inv{\Leftu_\Sone}"]\ar[d,swap,"\Wdiag"]
    &
    \Tens\Sone\Sone\ar[d,"\Tens{\Wdiag}{\Wdiag}"]
    \\
    \Dbimon\ar[r,"\Dbimonm"']
    &
    \Tens{\Dbimon}{\Dbimon}    
  \end{tikzcd}
\end{center}
Only the last
one requires \(\Mdeg\) to be total.
\end{theorem}
\begin{proof}
  Remember that \(\Dbimon=\Bwith_{i\in\Mdeg}\Sone\), so we can set %
  \(\Dbimonl
  =\Tuple{\Leftu_\Sone\Compl\Tensp{\Proj i}{\Proj i}}_{i\in\Mdeg}
  \in\cL(\Tens\Dbimon\Dbimon,\Dbimon)\)
  which obviously satisfies the announced property, which characterizes
  \(\Dbimonl\) uniquely by \cref{lemma:tens-win-epic}.
  The fact that we define in that way a commutative monoid is easy to
  check.
  For instance (keeping implicit the associator \(\Assoc\) and the
  \(\Leftu\) isos) we have
  \(\Dbimonl^{(3)} =\Dbimonl\Compl\Tensp\Dbimonl\Dbimon
  =\Dbimonl\Compl\Tensp\Dbimon\Dbimonl
  =\Tuple{\Proj i\ITens\Proj i\ITens\Proj i}_{i\in\Mdeg}
  \in\cL(\Dbimon\ITens\Dbimon\ITens\Dbimon,\Dbimon)\).
  The commutations of the diagram involving the unit \(\Wdiag\) result
  from the fact that \(\Proj i\Compl\Wdiag=\Id_\Sone\) for \(i=0,1\).
  The first bimonoid diagram obviously commutes.
  For the second we use \cref{lemma:tens-win-epic}.
  We have
  \begin{align*}
    \Dbimonm\Compl\Dbimonl\Compl\Tensp{\Win i}{\Win j}
    &=\Dbimonm
      \Compl\Kronecker ij\Win i\\
    &=\Kronecker ij\sum_{(l,r)\in \Mdegsl i}
      \Tens{\Win l}{\Win r}
  \end{align*}
  and, setting \(\Sym_{2,3}=\Dbimon\ITens\Sym\ITens\Dbimon\),
  \begin{align*}
    \Tensp\Dbimonl\Dbimonl
    \Compl\Sym_{2,3}
    \Compl\Tensp\Dbimonm\Dbimonm
    \Compl\Tensp{\Win i}{\Win j}
    &=\Tensp\Dbimonl\Dbimonl
      \Compl\Sym_{2,3}
      \Compl\TenspB
      {\big(\sum_{(l,r)\in \Mdegsl i}
      \Tens{\Win l}{\Win r}\big)}
      {\big(\sum_{(l',r')\in \Mdeg_j}
      \Tens{\Win l}{\Win r}\big)}\\
    &=\Tensp\Dbimonl\Dbimonl
      \Compl\Sym_{2,3}
      \Compl
      \sum_{(l,r)\in \Mdegsl i,\,(l',r')\in \Mdeg_j}
      \Win l\ITens\Win r\ITens\Win{l'}\ITens\Win{r'}\\
    &=\Tensp\Dbimonl\Dbimonl
      \Compl
      \sum_{(l,r)\in \Mdegsl i,\,(l',r')\in \Mdeg_j}
      \Win l\ITens\Win{l'}\ITens\Win{r}\ITens\Win{r'}\\
    &=\sum_{(l,r)\in \Mdegsl i,\,(l',r')\in \Mdeg_j}
      \Kronecker l{l'}\Kronecker r{r'}
      \Win l\ITens\Win{r}\\
    &=\Kronecker ij\sum_{(l,r)\in \Mdegsl i}
      \Tens{\Win l}{\Win r}\,.
  \end{align*}
  The third bimonoid diagram is easily proven to commute, using the
  joint epicity of the \(\Tens{\Win l}{\Win r}\) for \(l,r\in\Mdeg\).
  For the last bimonoid diagram we have, using \cref{lemma:dbimonm-sum-expr},
  \begin{align*}
    \Tensp\Wdiag\Wdiag\Compl\Inv\Leftu
    &=\Big(\big(\sum_{l\in\Mdeg}\Win l\big)
      \ITens\big(\sum_{r\in\Mdeg}\Win r\Big)\big)
      \Compl\Inv\Leftu\\
    &=\Big(\sum_{(l,r)\in\Mdeg^2}\Tens{\Win l}{\Win r}\Big)
      \Compl\Inv\Leftu\\
    \text{and\quad}
    \Dbimonm\Compl\Wdiag
    &=\sum_{i\in\Mdeg}\Big(
      \sum_{(l,r)\in\Mdegsl i}\Tens{\Win l}{\Win r}\Big)
      \Compl\Inv\Leftu\Compl\Proj i\Compl\Wdiag\\
    &=\sum_{i\in\Mdeg}\Big(
      \sum_{(l,r)\in\Mdegsl i}\Tens{\Win l}{\Win r}\Big)
      \Compl\Inv\Leftu\\
    &=\Big(
      \sum_{(l,r)\in\Mdegdom}\Tens{\Win l}{\Win r}\Big)
      \Compl\Inv\Leftu
  \end{align*}
  where we raecall that \(\Mdegdom\) is the set of all
  \((l,r)\in\Mdeg^2\) such that \(l+r\) is defined.
  So this diagram commutes if the \Dpmonoid{} %
  is total. %
\end{proof}

Notice that, in \(\Op\cL\), \(\Dbimonco\) (see
\cref{prop:dbimon-comonoid}) is a commutative monoid and \(\Dbimonmo\)
(see \cref{th:dbimon-monoid-bimonoid}) is a commutative comonoid.
As explained in \cref{sec:monoid-to-monad}, 
since \(\Sfung\) is strong symmetric monoidal from \(\Op\cL\) to
\(\FUN\cL\cL\), it results from
\cref{thm:symmonoid-strong-mondoidal-functor-image}
that \(\Sfung\Dbimon\) inherits from \(\Dbimonco\) a symmetric monad
structure (notation: \(\Sfung\Dbimonco\)) and from \(\Dbimonmo\) a
symmetric comonad structure (notation: \(\Sfung\Dbimonmo\)).
\begin{theorem} \label{thm:sfung-dbimon}
  We have \(\Sfung{\Dbimonco}=\Sfunm\) and
  \(\Sfung\Dbimonmo=\Sfunc\).
  Moreover \(\Sfung\Sym=\Sflip\).
  If \(\Mdeg\) is total then \(\Sfung\Dbimonbi=\Sfunb\).
\end{theorem}

Remember indeed that it is only when \(\Mdeg\) is total (that is, is
an ordinary commutative monoid) that we are sure that \(\Dbimonbi\) is
a bi-commutative bimonoid and that \(\Sfunb\) is a symmetric bimonad,
see \cref{th:sfun-bimonad,th:dbimon-monoid-bimonoid}%
\footnote{We have not proven however that these latter properties imply
  the totality of \(\Mdeg\).}

\begin{proof}
  Simple verification.
\end{proof}

\begin{Example}
  Let us take \(\cL=\COH\).
  Then we have
  \begin{itemize}
  \item Counit %
    \(\Proj0=\{(0,\Sonelem)\}\in\COH(\Dbimon,\Sone)\);
  \item Comultiplication %
    \(\Dbimonm=\Set{(l+r,(l,r))\St l,r\in\Mdeg\text{ and }l+r\text{ defined}}
    \in\COH(\Dbimon,\Tens\Dbimon\Dbimon)\);
  \item Unit %
    \(\Wdiag=\Set{(\Sonelem,i)\St i\in\Mdeg}\in\COH(\Sone,\Dbimon)\);
  \item Multiplication %
    \(\Dbimonl=\Set{((i,i),i)\St i\in\Mdeg}
    \in\COH(\Tens\Dbimon\Dbimon,\Dbimon)\)
  \end{itemize}
  and the structure obtained in that way is a bimonoid iff \(\Mdeg\)
  is total.
\end{Example}

\begin{Example}
  Let us take \(\cL=\PCOH\) and let \(n\in\Natc\).
  An element \(u\) of \(\Pcoh{\Dbimon}\) can be written uniquely %
  \(u=\sum_{i\in\Mdeg} u_i\Base i\) where
    \((u_i\in\Intercc01)_{i=\in\Mdeg}\), so that %
    \(\Pcoh\Dbimon\Isom\Intercc01^\Mdeg\) as already mentioned.
  \begin{itemize}
  \item The unit \(\Proj0\in\PCOH(\Dbimon,\Sone)\) is characterized by %
    \(\Matappa{\Proj0}{u}=u_0\);
  \item The comultiplication
    \(\Dbimonm\in\PCOH(\Dbimon,\Tens\Dbimon\Dbimon)\) %
    is characterized by the equation %
    \(\Matappa\Dbimonm
    u=\sum_{i\in\Mdeg}u_i(\sum_{(l,r)\in\Mdegsl i}\Base{(l,r)})\);
  \item The unit \(\Wdiag\in\PCOH(\Sone,\Dbimon)\) is characterized by %
    \(\Matappa{\Wdiag}p=p(\sum_{i\in\Mdeg}\Base i)\);
  \item The multiplication
    \(\Dbimonl\in\PCOH(\Tens\Dbimon\Dbimon,\Dbimon)\) %
    is characterized by %
    \(\Matappa\Dbimonl{\Tensp uv}=\sum_{i\in\Mdeg}u_iv_i\Base i\).
  \end{itemize}
\end{Example}

\section{The \texorpdfstring{\DDIFF{}}{differentiation} structure}
\label{sec:diff-structure}

Using the notion of \SUMBLTY{} structure we introduced in our two
previous sections, we now introduce an operation of \DDIFF{} induced
by our \Pmonoid{} \(\Mdeg\).
As we shall see in \cref{sec:examples-rep-diff-resource-cat},
depending on the choice of \(\Mdeg\) this differentiation will
typically consist of traditional derivatives, of difference operators,
or a mix of both.
We adopt the viewpoint of differential \LL{} which consists in
associating such an operation with the exponential modality of
\LL{}.
Indeed, in differential \LL, differentiation is represented as a
morphism of type \(\Limpl A{\Oc A}\) which allows one to transform by
precomposition a morphism \(f:\Limpl{\Oc A}{B}\) (to be understood as
a non-linear morphism \(A\to B\)) into a linear morphism
\(\Limpl AB\): the derivative of \(f\) at \(0\).

In the present setting,
\DDIFF{} boils down to a distributive law between $\Sfun$ and the
exponential modality, with additional properties presented as a set of
categorical axioms.

\begin{convention}
  In the sequel, in order to avoid boring repetitions, %
  we assume that \(\cL\) is a \Pmsmc{} which is a resource category
  equipped with a \SUMBLTY{} structure for which we use the notations
  introduced in \cref{sec:summability-structure}.
  Remember that this means in particular that~\ref{ax:stens}
  and~\ref{ax:swith} hold, see
  \cref{sec:summability-tensor,sec:summability-cartesian}.
\end{convention}

To have an intuition about what happens now, assume that
\(\Mdeg=\Eset{0,1}\) equipped with \(\max\) as binary operation%
\footnote{Taking \(\Mdeg=\Eset{0,1}\) with \(1+1\) undefined instead
  would lead to the situation described in \cite{Ehrhard23a},
  which was the original motivation of all these developments.
  We prefer to develop another example here in order to illustrate the
  versatility of the theory.}%
; the main idea of \DDIFF{} is to associate with any nonlinear morphism
\(f\in\Kloc\cL(X,Y)\) a ``derivative'' %
\(\Dfun f\in\Kloc\cL(\Sfun X,\Sfun Y)\).
For this choice of \(\Mdeg\), this function \(\Dfun f\) maps a
summable pair \(x=(x_0,x_1)\) of elements of \(X\) to the summable
pair \(\Dfun f(x)=(f(x_0),f(x_0+x_1)-f(x_0))\), and this operation
\(\Dfun\) is functorial%
\footnote{More precisely, this is what happens for this choice of
  \(\Mdeg\) in the category \(\PCOH\) for instance, as we shall see in
  \cref{sec:representable-diff}.}.
In other words \(\Dfun\) is an extension (\cref{def:extension})
of the functor \(\Sfun\) to %
\(\Kloc\cL\) in the sense that if \(f\in\cL(X,Y)\), one has %
\(\Dfun(\Derfun(f))=\Derfun(\Sfun f)\).
Intuitively this condition means that the ``derivative'' of a linear
map is the map itself since when \(f\) is linear one has
\(f(x_0+x_1)-f(x_0)=f(x_0)+f(x_1)-f(x_0)=f(x_1)\).
Such extensions are in one-to-one correspondence
with \emph{distributive laws} between the functor \(\Sfun\) and the
comonad \(\Oc{\_}\) as explained in \cref{thm:extension-dl-iso}.

We will make further assumptions about this distributive law which
will also allow us to endow \(\Dfun\) with the monad structure of
\(\Sfun\), by simply applying the functor \(\Derfun\) to \(\Sin0\) and
\(\Sfunmonm\); we keep the same notations for these extended
morphisms.

With the same choice for \(\Mdeg\), the monad
\((\Dfun,\Sin0,\Sfunmonm)\) intuitively maps an object \(X\) of
\(\cL\) to the object \(\Sfun X\) of summable pairs \(x=(x_0,x_1)\) of
elements of \(X\), and a \(f\in\Kl\cL(X,Y)\) to the map
\(\Dfun f\in\Kl\cL(\Sfun X,\Sfun Y)\) such that
\(\Dfun f(x)=(f(x_0),f(x_0+x_1)-f(x_0))\).
Notice that \(\Dfun^2f\in\Kl\cL(\Sfun^2X,\Sfun^2Y))\) is given by
\begin{align}
  \Dfun^2f(x)
  =&(\Dfun f(x_0),\Dfun f(x_0+x_1)-\Dfun f(x_0))\nonumber\\
  =&((f(x_{00}),f(x_{00}+x_{01})-f(x_{00})),\nonumber\\
  &(f(x_{00}+x_{10}),f(x_{00}+x_{10}+x_{01}+x_{11})-f(x_{00}+x_{10}))\nonumber\\
   &\hspace{10em}-(f(x_{00}),f(x_{00}+x_{01})-f(x_{00})))\nonumber\\
  =&((f(x_{00}),f(x_{00}+x_{01})-f(x_{00})),(f(x_{00}+x_{10})-f(x_{00}),
  \nonumber\\
  &\hspace{5em}f(x_{00}+x_{10}+x_{01}+x_{11})-f(x_{00}+x_{10})-f(x_{00}+x_{01})
  +f(x_{00})))\label{eq:example-mdiff2-dfun2-expr}
\end{align}
The unit of the monad \(\Sin0\in\Kl\cL(X,\Sfun X)\) maps \(x\in X\) to
\((x,0)\in\Sfun X\), and the multiplication
\(\Sfunmonm\in\Kl\cL(\Sfun^2X,\Sfun X)\) maps \(x\in\Sfun^2X\) to
\((x_{00},x_{01}+x_{10}+x_{11})\).
One very important fact is that \(\Sfunmonm\) is natural that is, the
following diagram commutes in \(\Kl\cL\)
\begin{center}
  \begin{tikzcd}
    \Sfun^2X
    \ar[r,"\Dfun^2 f"]\ar[d,swap,"\Sfunmonm"]
    &
    \Sfun^2Y
    \ar[d,"\Sfunmonm"]
    \\
    \Sfun X
    \ar[r,"\Dfun f"']
    &
    \Sfun Y
  \end{tikzcd}
\end{center}
For our choice of \(\Mdeg\) this simply results from the fact that
\begin{align}
  &f(x_{00}+x_{01})-f(x_{00})+f(x_{00}+x_{10})-f(x_{00})\nonumber\\
  &\hspace{4em}
    +
    f(x_{00}+x_{10}+x_{01}+x_{11})-f(x_{00}+x_{10})-f(x_{00}+x_{01})
    +f(x_{00})\nonumber\\
  &=f(x_{00}+x_{10}+x_{01}+x_{11})-f(x_{00})
    \label{eq:example-mdiff2-sfunmonm-nat}\,.
\end{align}

This motivates the next definition.

\begin{definition} %
  \labeltext{(\(\Sdiff\)-chain)}{ax:sdchain}\ref{ax:sdchain} A
  pre-difference structure on \(\cL\) is a distributive law \(\Sdiff\)
  between the functor \(\Sfun\) and the comonad \(\Oc{\_}\), see
  \cref{sec:distr-law-comonad}.
\end{definition}
We use \(\Dfun\) for the extension of \(\Sfun\) to \(\Kl\cL\)
according to \cref{thm:extension-dl-iso}.
This basic condition \ref{ax:sdchain} is not sufficient for specifying
a difference operation.
Here are the additional conditions, where we assume that
\ref{ax:sdchain} holds.
Remember that the category $\DL \cL \Oc$ of distributive laws 
has a monoidal structure, described in~\cref{lemma:dist-law-monoidal-struct}.

\begin{definition} %
  \labeltext{(\(\Sdiff\)-local)}{ax:sdloc}\ref{ax:sdloc}
  \(\Sproj0\in\DL\cL\Oc((\Sfun,\Sdiff),\Sone)\).  
\end{definition}
This condition means that our difference operator, when applied
to a non-linear morphism \(f:X\to Y\), will yield a non-linear
morphism \(\Dfun f:\Sfun X\to\Sfun Y\) whose first component
\(\Sfun X\to Y\) maps all summable pairs \((x_0,x_1)\) to \(f(x_0)\).
The tangent bundle functor behaves exactly in the same way and this
condition expresses that the differential is computed locally (in this
differential geometric setting, \(x_0\) is an element of the manifold
\(X\) and \(x_1\) is a tangent vector at \(x_0\)).
This justifies our terminology.

\begin{definition} %
  \labeltext{(\(\Sdiff\)-Schwarz)}{ax:sdschwarz}\ref{ax:sdschwarz} %
  \(\Sflip\in
  \DL\cL\Oc(\Tens{(\Sfun,\Sdiff)}{(\Sfun,\Sdiff)},
  \Tens{(\Sfun,\Sdiff)}{(\Sfun,\Sdiff)})\).
\end{definition}
This condition means that the second difference function is a \emph{symmetric}
function.
More precisely, in the case \(\Mdeg=(\Eset{0,1},\max)\), the
expression of \(\Dfun^2f(x)\) given in
\cref{eq:example-mdiff2-dfun2-expr} is invariant by exchanging
\(x_{01}\) and \(x_{10}\), as well as the two central components of
the pair of pairs \(\Dfun^2f(x)\).

\begin{definition} %
  \labeltext{(\(\Sdiff\)-add)}{ax:sdadd}\ref{ax:sdadd} %
  \(\Sin0\in\DL\cL\Oc(\Sone,(\Sfun,\Sdiff))\) and
  \(\Sfunmonm\in
  \DL\cL\Oc(\Tens{(\Sfun,\Sdiff)}{(\Sfun,\Sdiff)},(\Sfun,\Sdiff))\).
\end{definition}

This means that the functor \(\Sfun\), when extended to \(\Kl\cL\),
has a monad structure which is obtained by applying the functor
\(\Derfun\) to the unit and multiplication of the monad \(\Sfun\).
In the case \(\Mdeg=(\Eset{0,1},\max)\), this means in particular that
\cref{eq:example-mdiff2-sfunmonm-nat} must hold.
More specifically, let $\Dmonu$ be the family of morphism
\(\Dmonu_X=\A\Derfun{\Sin0}\in\Kl\cL(X,\Dfun X)\), 
and $\Dmonm_X = \A\Derfun{\Sfunmonm_X} \in \Kl \cL(\Dfun \Dfun X, \Dfun X)$.
By abuse of notations, we write $\Dflip$ for the family of morphisms
\(\Derfun{\Sflip_X} \in \Kl \cL(\A{\Dfun^2} X, \A \Dfun X) \).

\begin{proposition} \label{prop:Dfun-monad}
  If \ref{ax:sdschwarz} and \ref{ax:sdadd} hold then
  \(\Dmonu : \Idfun \Natto \Dfun \), 
  \(\Dmonm : \Dfun^2 \Natto \Dfun \) and
  \(\Dflip : \Dfun^2 \Natto \Dfun^2 \) are natural transformations
  in $\Kl \cL$
  and \((\Dfun, \Dmonu, \Dmonm, \Dflip)\) is 
  a \Symmonad{} on \(\Kl\cL\).
\end{proposition}

\begin{proof}
  If \ref{ax:sdschwarz} and \ref{ax:sdadd} hold then
  \(((\Sfun,\Sdiff),\Sin0,\Sfunmonm,\Sflip)\) is a symmetric monoid in
  \(\DL\cL\Oc\).
  By \cref{thm:symmonoid-strong-mondoidal-functor-image}, 
  the image of this symmetric monoid 
  under the chain of strict monoidal functor given in
  \cref{eq:dl-chain-of-functors} induces a \Symmonad{} on \(\Kl\cL\).
  This \Symmonad{} is exactly \((\Dfun, \Dmonu, \Dmonm, \Dflip)\).
\end{proof}

The next condition was overlooked in~\cite{Ehrhard23a}, but the
corresponding condition was already recognized as important in the
theory of tangent categories~\cite{Rosicky84,CockettCruttwell14}.
\begin{definition} %
  \labeltext{(\(\Sdiff\)-lift)}{ax:sdlift}\ref{ax:sdlift} %
  \(\Sfunlift\in
  \DL\cL\Oc((\Sfun,\Sdiff),\Tens{(\Sfun,\Sdiff)}{(\Sfun,\Sdiff)})\).
\end{definition}

We can also ask a condition on the sum $\Ssum$, which is not 
always satisfied as we shall see.

\begin{definition} %
  \labeltext{(\(\Sdiff\)-sum)}{ax:sdsum}\ref{ax:sdsum} %
  \(\Ssum\in\DL\cL\Oc((\Sfun,\Sdiff),\Sone)\).
\end{definition}

Similarly to \ref{ax:sdadd} and \ref{ax:sdschwarz}, these two axioms
mean that the extension of $\Sfun$ to $\Kl \cL$ has a comonad
structure.
More precisely, let us write \(\Dlift\) and \(\Dsum\)
for the families of morphisms
\(\Derfun\Slift_X\in\Kl\cL(\Dfun X,\Dfun^2X)\) and
\(\Derfun\Ssum_X\in\Kl\cL(\Dfun X,X)\)
respectively.

\begin{proposition}
  \label{prop:Dfun-comonad}
  \label{prop:kleisli-bimonad}
  If \ref{ax:sdlift} holds then \(\Dlift : \Dfun \Natto 
  \Dfun^2\) is a natural transformation in $\Kl \cL$.
  If \ref{ax:sdsum} holds then \(\Dsum : \Dfun \Natto \Idfun \) 
  is a natural transformation in $\Kl \cL$.
  If \ref{ax:sdlift}, \ref{ax:sdsum} and \ref{ax:sdschwarz} hold then 
  \((\Dfun,\Dsum,\Dlift,\Dflip)\) is a \Symcomonad{} on $\Kl \cL$.
\end{proposition}

\begin{proof}
Same proof as \cref{prop:Dfun-monad}.
\end{proof}

So the axiom \ref{ax:sdsum} implies that for all \(f\in\Kl\cL(X,Y)\), one has
\begin{center}
  \begin{tikzcd}
    \Sfun X
    \ar[r,"\Dfun f"]
    \ar[d,swap,"\Ssum"]
    &
    \Sfun Y
    \ar[d,"\Ssum"]
    \\
    X
    \ar[r,"f"']
    &
    Y
  \end{tikzcd}
\end{center}
When \(\Mdeg\) is not a total monoid, this is too strong a
requirement, at least in the concrete models we know.
For instance if \(\Mdeg=(\Eset{0,1},+)\) with \(1+1\) undefined, then
one has typically \(\Dfun f(x)=(f(x_0),f'(x_0)\cdot x_1)\) and the
commutation of this diagram would mean
\(f(x_0+x_1)=f(x_0)+f'(x_0)\cdot x_1\) which holds only when \(f\) is
affine.

On the other hand when \(\Mdeg\) is total 
\ref{ax:sdsum} is a reasonable hypothesis.
For instance if \(\Mdeg=(\Eset{0,1},\max)\), it means
\(f(x_0+x_1)=f(x_0)+(f(x_0+x_1)-f(x_0))\) which is perfectly fine.
When \(\Mdeg=(\Nat,+)\), this diagram means that the morphisms for
\(\Kl\cL\) are analytic in the sense that they coincide with their
Taylor expansion, see~\cite{EhrhardWalch25}.

Recall from \cref{rem:Sfun-lax-tensor} and \cref{prop:S-strong-with}
that $\Sfun$ has a lax monoidal structure with regard to 
$\ITens$, and a strong monoidal structure with regard to 
$\IWith$.
Our last requirement is a compatibility condition between these two 
monoidal structures, the strong monoidal structure of $\Oc$, and the 
distributive law of our $\Mdeg$-differential structure. 
It becomes quite important when the SMC \(\cL\) is assumed to be
closed since, in that situation, the strong monoidal structure 
of $\Oc$ turns \(\Kl\cL\) into a cartesian closed category.

\begin{definition} %
  \labeltext{(\(\Sdiff\)-\(\IWith\))}{ax:sdwith}\ref{ax:sdwith} %
  \begin{equation*}
    \begin{tikzcd}
      \Tens{\Oc\Sfun X_1}{\Oc\Sfun X_2}
      \ar[r,"\Tens{\Sdiff_{X_1}}{\Sdiff_{X_2}}"]
      \ar[d,swap,"\Seelyt_{\Sfun X_1,\Sfun X_2}"]
      &[1em]
      \Tens{\Sfun\Oc X_1}{\Sfun\Oc X_2}
      \ar[r,"\Sfunlaxt_{\Oc X_1,\Oc X_2}"]
      &[1em]
      \Sfun\Tensp{\Oc X_1}{\Oc X_2}
      \ar[d,"\Sfun\Seelyt_{X_1,X_2}"]\\
      \Oc\Withp{\Sfun X_1}{\Sfun X_2}
      \ar[r,swap,"\Oc\Swithiso_{X_1, X_2}"]
      & \Oc{\Sfun\Withp{X_1}{X_2}}
      \ar[r,"\Sdiff_{\With{X_1}{X_2}}"']
      &
      \Sfun{\Oc\Withp{X_1}{X_2}}
    \end{tikzcd}
  \end{equation*}
\end{definition}

\begin{remark} 
  We might be tempted to consider as an axiom  
the following commutation  
 \begin{equation*}
    \begin{tikzcd}
      \Sone
      \ar[rr,"\Sin0"]\ar[d,swap,"\Seelyz"]
      &&
      \Sfun\Sone
      \ar[d,"\Sfun\Seelyz"]
      \\
      \Oc\Stop
      \ar[r,swap,"\Oc0"]
      &
      \Oc\Sfun\Stop
      \ar[r,"\Sdiff_\Stop"']
      &
      \Sfun\Oc\Stop
    \end{tikzcd}
  \end{equation*}
  which is the $0$-ary counterpart of the commutation \ref{ax:sdwith}.
  However, this diagram is an immediate consequence of \ref{ax:sdadd}, 
  because $0 \in \cL(\Stop, \Sfun \Stop)$ is equal to $\Sin 0$
  (due to the fact that $\Stop$ is terminal).
\end{remark}

\begin{definition}
  \label{def:diff-resource-cat}
  A \emph{\Lingdiffcat{}} is a resource category
  \(\cL\) equipped with a \SUMBLTY{} structure and a natural
  transformation \(\Sdiff\) which satisfies the conditions
  \ref{ax:sdchain}, \ref{ax:sdloc}, \ref{ax:sdschwarz}, \ref{ax:sdadd},
  \ref{ax:sdlift} and \ref{ax:sdwith}.
  A \Lingdiffcat{} is said to be \emph{\Resid{}} if \(\Mdeg\) is total
  and the property \ref{ax:sdsum} holds.
\end{definition}

\begin{proposition}
  \label{prop:Dfun-bimonad}
  If the \Lingdiffcat{} is \Resid{} then 
  \((\Dfun,\Dmonu,\Dmonm,\Dsum,\Dlift,\Dflip)\) is a \Symbimonad{}
  on $\Kl \cL$.
\end{proposition}

\begin{proof}
  Since $\Mdeg$ is total, it follows from \cref{th:sfun-bimonad} and
  our assumptions that
  \(((\Sfun,\Sdiff),\Sin0,\Sfunmonm, \Ssum, \Slift, \Sflip) \) is a
  symmetric bimonoid in $\DL \cL \Oc$.
  By \cref{thm:symmonoid-strong-mondoidal-functor-image}, the chain of
  strict monoidal functors given in \cref{eq:dl-chain-of-functors}
  sends this symmetric bimonoid to a \Symbimonad{} on $\Kl \cL$. This
  \Symbimonad{} is precisely
  \((\Dfun,\Dmonu,\Dmonm,\Dsum,\Dlift,\Dflip)\).
\end{proof}

\subsection{The \texorpdfstring{\DDIFF{}}{differentiation} 
structure, in the \Elementary{} case}
\label{sec:representable-diff}

Let \(\cL\) be a \Elementarily{} summable category (see
\cref{def:elementarily-summable-cat}) that we assume to be closed and let
\((\Sfun,(\Sproj i)_{i\in\Mdeg},\Ssum)\) be the associated \SUMBLTY
structure. 
We prove in this section that the existence of a \DDIFF{} as described 
in \cref{def:diff-resource-cat} boils down to the existence of 
a coalgebra structure on $\Dbimon$ subject to additional conditions
that express that $\Dbimon$ is a commutative bimonoid in the Eilenberg-Moore 
category $\Em \cL$.
The \DDIFF{} associated to this coalgebra will be the one 
obtained from this bimonoid structure, pushed through 
the chain of strong monoidal 
functors described in~\cref{eq:full-chain-of-functors}.
First thing first, let us describe these conditions more explicitly.
Let \(\Ddiff\in\cL(\Dbimon,\Oc\Dbimon)\) be a morphism.

\begin{definition} %
  \labeltext{(\(\Ddiff\)-chain)}{ax:dchain}\ref{ax:dchain} %
  \(\Dbimoncoalg=(\Dbimon,\Ddiff)\) is an object of \(\Em\cL\).
\end{definition}
We assume that this first condition holds and we introduce additional
conditions on \(\Ddiff\).
We use notions and notations introduced in \cref{sec:coalg-comonoids}.

\begin{definition}
  \labeltext{(\(\Ddiff\)-local)}{ax:dloc}\ref{ax:dloc} %
  \(\Win0\in\Em\cL(\Sone,\Dbimoncoalg)\).
\end{definition}

\begin{definition}
  \labeltext{(\(\Ddiff\)-add)}{ax:dadd}\ref{ax:dadd} %
  \(\Proj0\in\Em\cL(\Dbimoncoalg,\Sone)\) and
  \(\Dbimonm\in\Em\cL(\Dbimoncoalg,\Tens\Dbimoncoalg\Dbimoncoalg)\),
  that is \((\Dbimoncoalg,\Proj0,\Dbimonm)\) is a comonoid in $\Em \cL$, 
  which is cartesian by \cref{prop:em-excl-cartesian}.
\end{definition}

  It is well known that an object in a cartesian category has exactly one 
  commutative comonoid structure with 
  respect to the cartesian product.
  Thus, the coalgebra $\Dbimoncoalg$ has exactly one 
  commutative comonoid structure: the comonoid 
  $(\Dbimon, \Coalgw{\Dbimoncoalg}, \Coalgc{\Dbimoncoalg})$
  given in \cref{eq:coalgebra-to-comonoid}.
  Thus, \ref{ax:dadd} holds if and only if
  \(\Proj0=\Coalgw{\Dbimoncoalg}\) and \(\Dbimonm=\Coalgc{\Dbimoncoalg}\).

\begin{definition}
  \labeltext{(\(\Ddiff\)-lift)}{ax:dlift}\ref{ax:dlift} %
  \(\Dlift\in\Em\cL(\Tens{\Dbimoncoalg}{\Dbimoncoalg},\Dbimoncoalg)\).
\end{definition}

\begin{definition}
  \labeltext{(\(\Ddiff\)-sum)}{ax:dsum}\ref{ax:dsum} %
  \(\Wdiag\in\Em\cL(\Sone,\Dbimoncoalg)\).
\end{definition}

\begin{definition}
  We set \(\Sdiff=\Coalgdl\Dbimoncoalg\) which, by
  \cref{prop:coalg-to-distr-law}, is a distributive law between the
  comonad \(\Oc\) and the functor \(\Sfun\).
  In other words \((\Sfun,\Sdiff)=\CADL\Oc(\Dbimoncoalg)\) is an object of
  \(\DL\cL\Oc\).
\end{definition}

\begin{theorem}
  The following implications hold:
  \ref{ax:dloc}\(\Implies\)\ref{ax:sdloc},
  \ref{ax:dadd}\(\Implies\)\ref{ax:sdadd}, and
  \ref{ax:dlift}\(\Implies\)\ref{ax:sdlift}.
  Furthermore, \ref{ax:sdschwarz} and \ref{ax:sdwith} always holds.
\end{theorem}

\begin{proof}
  Let us prove that \ref{ax:dloc}\(\Implies\)\ref{ax:sdloc}.
  By definition, $\Sproj 0 \in \cL(\Limpl {\Dbimon} X, X)$ 
  is given by the following composition of 
  natural transformations
  \[ 
  (\Sproj 0)_X = \begin{tikzcd}
    {\Limplp \Dbimon X} &[1em] {\Limplp \Sone X} & X
    \arrow["{\Sfung(\Win 0)_X}", from=1-1, to=1-2]
    \arrow["{(\Sfungmonz)^{-1}}", from=1-2, to=1-3]
  \end{tikzcd} \]
  It immediately follows from 
  \cref{prop:coalg-to-distr-law} that 
  \(\Sproj0\in\DL\cL\Oc((\Sfun,\Sdiff),\Sone)\).
  Similarly, the implications
  \ref{ax:dadd}\(\Implies\)\ref{ax:sdadd} and
  \ref{ax:dlift}\(\Implies\)\ref{ax:sdlift} 
  as well as the property \ref{ax:sdschwarz} 
  are an immediate
  consequence of \cref{thm:sfung-dbimon,prop:coalg-to-distr-law}.  
  The last property \ref{ax:sdwith} is obtained by a-
  simple computation using the axioms of resource categories.
\end{proof}

\begin{definition}
  \label{def:rep-diff-resource-cat}
  A \emph{\Lingrdiffcat{}} is a
  \Elementarily{} \SUMBLE{} resource category \(\cL\) equipped with a
  morphism \(\Ddiff\in\cL(\Dbimon,\Oc\Dbimon)\) which satisfies
  \ref{ax:dchain}, \ref{ax:dloc}, \ref{ax:dadd} and \ref{ax:dlift}.
  Such a \Lingrdiffcat{} is said to be \emph{\Resid{}} if \(\Mdeg\) is
  total and \ref{ax:dsum} also holds.
\end{definition}

One important point of these definitions is the observation
in~\cite{Ehrhard23a} that such a \Lingrdiffcat{} structure on
\(\Dbimon\) is quite easy to obtain in Lafont resource categories.
\begin{theorem} %
  \label{th:lafont-elementary-differential-structure}
  If \(\cL\) is a Lafont resource category (\cref{def:Lafont}) which
  is \Elementarily{} \SUMBLE{}, then \(\cL\) has exactly one structure
  of \Lingrdiffcat{}, and if \(\Mdeg\) is total, then this
  \Lingrdiffcat{} is \Resid.
\end{theorem}
\begin{proof}
  We assume that \(\cL\) is a Lafont resource category and we use
  \((\Oc,\Der{},\Digg{})\) for the comonad induced on \(\cL\) by the
  corresponding adjunction.
  By \cref{prop:dbimon-comonoid}, \((\Dbimon,\Proj0,\Dbimonm)\) is a
  commutative comonoid in \(\cL\).
  Furthermore, \begin{enumerate}
  \item \label{prop:dbimon-struct-monoid-morphism-1} %
    the morphisms $\Proj 0 \in \cL(\Dbimon, \Sone)$ and
    $\Dbimonm \in \cL(\Dbimon, \Dbimon \ITens \Dbimon)$ are comonoid
    morphisms;
  \item \label{prop:dbimon-struct-monoid-morphism-2} %
    the morphisms $\Dbimonl \in \cL(\Dbimon \ITens \Dbimon, \Dbimon)$
    is a comonoid morphism;
  \item \label{prop:dbimon-struct-monoid-morphism-3} %
    if $\Mdeg$ is total, the morphism
    $\Dbimonu \in \cL(\Sone, \Dbimon)$ is a comonoid morphism;
    \item \label{prop:dbimon-struct-monoid-morphism-4} the morphism 
    $\Win 0 \in \cL(\Sone, \Dbimon)$ is a comonoid 
      morphism.
  \end{enumerate}
  \Eqref{prop:dbimon-struct-monoid-morphism-1} holds by standard properties 
  of commutative comonoids (the structural morphisms of a commutative 
  comonoid are always comonoid morphisms). %
  Both 
  \Eqref{prop:dbimon-struct-monoid-morphism-2} and
  \Eqref{prop:dbimon-struct-monoid-morphism-3}
  hold by \cref{th:dbimon-monoid-bimonoid}.
  Finally, \Eqref{prop:dbimon-struct-monoid-morphism-4} holds by a
  straightforward diagram chase.
  Now, recall from \cref{sec:Lafont} that there is an isomorphism
  $\CmonEmFun : \CCMON\cL \arrow \Em\cL$ that preserves products, and
  that acts trivially on morphisms.
  Thus, the comonoid $\Dbimon$ has a canonical $\Oc$-coalgebra structure, 
  for which the comonoid morphisms mentioned in 
  \Eqref{prop:dbimon-struct-monoid-morphism-1}%
  -\Eqref{prop:dbimon-struct-monoid-morphism-4}
  are also $\Oc$-coalgebra morphisms.
  This gives to $\cL$ the structure of a representable difference
  resource category.

  Concerning uniqueness, let us assume that $\Dbimon$ has a coalgebra
  structure $\Dbimoncoalg$ for which $\Proj 0 \in \cL(\Dbimon, \Sone)$
  and $\Dbimonm \in \cL(\Dbimon, \Dbimon \ITens \Dbimon)$ are
  coalgebra morphisms.
  Since $\Em \cL$ is cartesian (\cref{prop:em-excl-cartesian}), every
  coalgebra has a unique comonoid structure and thus we must have
  $(\Dbimon, \Proj 0, \Dbimonm) = \Coalcomo(\Dbimoncoalg)$, where
  $\Coalcomo : \Em \cL \arrow \CCMON \cL$ is the inverse of
  $\CmonEmFun$.
  It follows that
  $\Dbimoncoalg = \CmonEmFun(\Dbimon, \Proj 0, \Dbimonm)$.
\end{proof}

\subsection{Examples of \texorpdfstring{\Lingrdiffcats{}}{representable 
differentiations}}
\label{sec:examples-rep-diff-resource-cat}

We illustrate \cref{def:rep-diff-resource-cat} in coherence spaces and
probabilistic coherence spaces, for various choices of monoid \(\Mdeg\).

\subsubsection{Coherence spaces}
  \label{ex:coh-spaces-diff}
  We continue \cref{ex:coh-cart-smc-elem-sum}.
  The category \(\COH\) is a resource category.
  Its tensor product is defined by
  \(\Web{\Tens{E_1}{E_2}}=\Web{E_1}\times\Web{E_2}\) and coherence
  given by %
  \((a_1,a_2)\Coh{\Tens EF}(a'_1,a'_2)\) if
  \((a_i\Coh{E_i}a'_i)_{i=1,2}\) which is easily seen to be a functor:
  given %
  \((s_i\in\COH(E_i,F_i))_{i=1,2}\), the set %
  \(\Tens{s_1}{s_2}=\Set{((a_1,a_2),(b_1,b_2))
    \St ((a_i,b_i)\in s_i)_{i=1,2}}\) is an element of %
  \(\COH(\Tens{E_1}{E_2},\Tens{F_1}{F_2})\).
  The associated unit is \(\Sone=(\Set\Sonelem,\mathord=)\).

  The category \(\COH\) has a free exponential defined as follows:
  \(\Web{\Oc E}\) is the set of all \(m\in\Mfin{\Web E}\) such
  that \(\Supp m\in\Cl E\), and coherence given by %
  \(m_1\Coh{\Oc E}m_2\) if \(m_1+m_2\in\Web{\Oc E}\).
  Notice that \(\Web{\Oc\Dbimon}=\Mfin\Mdeg\) since \(i\Coh\Dbimon j\)
  for all \(i,j\in\Mdeg\), and \(m\Coh{\Oc\Dbimon}m'\) holds for all
  \(m,m'\in\Web{\Oc\Dbimon}\).
  The differential structure induced by the fact that this exponential
  is free is %
  \(\Ddiff\in\COH(\Dbimon,\Oc\Dbimon)\); it is given by
  \begin{equation*}
    \Ddiff=\{(i,\Mset{\List i1k})\St k\in\Nat,\,i,\List i1k\in\Mdeg
    \text{ and }i=i_1+\cdots+i_k\}\,.
  \end{equation*}
  We already know that \(\Ddiff\in\COH(\Dbimon,\Oc\Dbimon)\) by the
  universal property of \(\oc\), but this fact is also easy to check
  directly: if \((i,m),(i',m')\in\Ddiff\) and \(i\not=i'\) (that is
  \({i}\Scoh\Dbimon{i'}\)), then we have \(m\not=m'\), and hence
  \(m\Scoh{\Oc\Dbimon}m'\).

  The composition of morphisms described in
  \cref{eq:Sdiff-elem-left-mate} gives us, for each coherence space
  \(F\), %
  the morphism
  \(\widehat\Sdiff_F\in\COH(\Tens{\Oc F}{\Dbimon}, \Oc(\Tens F
  \Dbimon))\) given by
  \begin{multline*}
    \widehat\Sdiff_F
    =\Big\{\Big( (\Mset{a_1,\dots,a_k},i),\Mset{(a_1, i_1), \dots, (a_k, i_k)}\Big)\\
    \St k\in\Nat,\,i,\List i1k\in\Mdeg,\,i=i_1+\cdots+i_k\text{
      and } \Set{a_1, \dots, a_k}\in\Cl{F}\Big\}
  \end{multline*}
  and hence \(\Sdiff_E\in\COH(\Oc{\Sfun E},\Sfun{\Oc E})\) given in
  \cref{eq:Sdiff-elem-right-mate} (taking \(F=\Limplp\Dbimon E\)) is
  \begin{multline*}
    \Sdiff_E
    =\Big\{\Big(\Mset{(i_1,a_1),\dots,(i_k,a_k)},(i,\Mset{\List a1k})\Big)\\
    \St k\in\Nat,\,i,\List i1k\in\Mdeg,\,i=i_1+\cdots+i_k\text{
      and } \Set{(i_1,a_1),\dots,(i_k,a_k)}\in\Cl{\Limpl\Dbimon E}\Big\}\,.
  \end{multline*}
  This means that the ``differential''
  \(\Dfun t=(\Sfun t)\Compl\Sdiff_E\in\Kl\COH(\Sfun E,\Sfun F)\) %
  of \(t\in\Kl\COH(E,F)\) is given by
  \begin{align*}
    \Dfun t
    &=\{(\Mset{(i_1,a_1),\dots,(i_k,a_k)},(i,b))
    \in\Web{\Oc{\Sfun E}}\times\Web{\Sfun F}\\
    &\Textsep\St (\Mset{\List a1k},b)\in t,\, i,\List i1k\in\Mdeg,\
      i=i_1+\cdots+i_k\\
    &\Textsep\Textsep\text{and }
      \Set{(i_1,a_1),\dots,(i_k,a_k)}\in\Cl{\Limpl\Dbimon E}\}\\
    &=\{(\Mset{(i_1,a_1),\dots,(i_k,a_k)},(i,b))
    \in\Web{\Oc{\Sfun E}}\times\Web{\Sfun F}\\
    &\Textsep\St (\Mset{\List a1k},b)\in t,\, i,\List i1k\in\Mdeg,\
      i=i_1+\cdots+i_k\\
    &\Textsep\Textsep\text{and }
      i_j\not=i_l\Implies a_j\not=a_l\}\\
    &=\{(\Vect m,(i,b))\in\Mfin{\Web E}^{\Mdeg}\times(\Mdeg\times\Web F)
      \St(\sum_{j\in\Mdeg}m_j,b)\in t,\ 
      \sum_{j\in\Mdeg}(\Card{m_j})j=i\\
    &\Textsep\Textsep\text{and }
      j\not=l\Implies\Supp{m_j}\cap\Supp{m_l}=\emptyset\}
  \end{align*} where the sum \(\sum_{j\in\Mdeg}(\Card{m_j})j\)
  is computed in \(\Mdeg\) (notice that \(\Card{m_j}\in\Nat\) is the
  multiplicity of \(j\) in that sum, which is finite because, by the
  assumption that \((\sum_{j\in\Mdeg}m_j,b)\in t\), almost all \(m_j\)s
  are empty multisets).
  In the last line, we use the fact that, up to an obvious bijection,
  we can consider
  \(\Web{\Oc{\Sfun E}}\subseteq\Mfin{\Mdeg\times\Web E}\) as a subset
  of \(\Mfin{\Web E}^\Mdeg\).

  We can define a stable%
  \footnote{This means that \(\Fun t\) commutes with unions of
    directed families of cliques and with intersections of bounded
    non-empty finite families of cliques.} %
  function \(\Fun t:\Cl E\to\Cl F\) by %
  \(\Fun t(x)=\Set{b\in\Web F\St\exists m\in\Web{\Oc E} \ \Supp
    m\subseteq x\text{ and }(m,b)\in t}\).
  Then, under the identification
  \begin{align*}
    \Cl{\Sfun E}=\Set{\Vect x=(x_i)_{i\in\Mdeg}\in\Cl E
    \St\Union_{i\in\Mdeg} x_i\in\Cl E
    \text{ and }i\not=j\Implies x_i\cap x_j=\emptyset}
  \end{align*}
  the stable function 
  $\Fun{\Dfun t} :  \Cl{\Sfun E} \arrow  \Cl{\Sfun F}$ is given by
  \begin{align*}
    \Fun{\Dfun t}(\Vect x)&=
    \{(i,b)\in\Mdeg\times\Web F\St\exists\Vect m\in\Mfin{\Web E}^D
    \ (\sum_{j\in\Mdeg}m_j,b)\in t\text{, }\\
    &\quad\quad\sum_{j\in\Mdeg}(\Card{m_j})j=i
    \text{ and }\forall j\in\Mdeg\ \Supp{m_j}\subseteq x_j\}
  \end{align*}
  and it is hard to give in general a more explicit description of
  this clique.
  If we take \(\Mdeg=\Nat\) with \(i+_\Mdeg j=\max(i,j)\), then the
  condition \(\sum_{j\in\Mdeg}(\Card{m_j})j=i\) means that either
  \(i=0\) and \(m_j=\Msetempty\) for all \(j>0\), or \(i>1\) with
  \(m_i\not=\Msetempty\) and \(m_j=\Msetempty\) for all \(j>i\), and
  therefore
  \begin{align}
    \label{eq:Dfun-nat-max-desc}
    \Fun{\Dfun t}(\Vect x)_i=
    \begin{cases}
      \Fun t (x_0) & \text{if }i=0\\
      \Fun t(x_0+\cdots+x_i)\setminus\Fun t(x_0+\cdots+x_{i-1})
      & \text{if }i>0\,.
    \end{cases}
  \end{align}
  Indeed, if $(\sum_{j=0}^i m_j, b) \in t$ and $m_i \neq\Msetempty$
  then $(\sum_{j=0}^{i-1} m_j, b) \notin t$ because \(t\) is a clique
  in \(\Limpl{\Oc E}{F}\): this justifies the difference of sets in
  \cref{eq:Dfun-nat-max-desc}.
  We take now \(\Mdeg=\Nat\) with \(i+_\Mdeg j=i+j\).
  Let \(\Vect x\in\Cl{\Sfun E}\) and \(\Vect m\in\Mfin{\Web E}^\Mdeg\)
  be such that \((\sum_{j\in\Mdeg}m_j,b)\in t\) and
  \(\forall j\in\Mdeg\ \Supp{m_j}\subseteq x_j\), then
  \begin{itemize}
  \item the condition \(\sum_{j\in\Mdeg}(\Card{m_j})j=0\) is realized
    iff \(m_j=\Msetempty\) for all \(j>0\);
  \item the condition \(\sum_{j\in\Mdeg}(\Card{m_j})j=1\) is realized
    iff \(\Card{m_1}=1\) and \(m_j=\Msetempty\) for all \(j>1\) (no
    restrictions on \(m_0\));
  \item the condition \(\sum_{j\in\Mdeg}(\Card{m_j})j=2\) is realized
    iff \(\Card{m_1}=2\) and \(m_j=0\) for all \(j>1\), or \(m_1=0\),
    \(\Card{m_2}=1\) and \(m_j=\Msetempty\) for all \(j>2\) (no
    restrictions on \(m_0\)).
  \end{itemize}
  To describe the stable function \(\Fun{\Dfun t}\) we need an
  auxiliary notion of degree which can be extracted from \(t\) but not
  from \(\Fun t\).
  \begin{definition}
    Let \(x\in\Cl E\) and \(b\in\Web F\) be such that \(b\in\Fun t(x)\).
    We know that there is exactly one \(m\in\Mfin{x}\) such that
    \((m,b)\in t\), then for each \(a\in x\) we set
    \(\Cohdeg txab=m(a)\in\Nat\).
  \end{definition}
  Using this tool, we can write the three first components of the
  tuple \(\Fun{\Dfun t}(\Vect x)\) as follows, using \(\Fun t\) as
  much as possible
  \begin{align*}
    \Fun{\Dfun t}(\Vect x)_0
    &=\Fun t(x_0)\\
    \Fun{\Dfun t}(\Vect x)_1
    &=\{b\in\Fun t(x_0\cup\Eset a)\setminus\Fun t(x_0)
      \St a\in x_1\text{ and }\Cohdeg t{x_0\cup x_1}ab=1\}\\
    \Fun{\Dfun t}(\Vect x)_2
    &=\{b\in\Fun t(x_0\cup\Eset a)\setminus\Fun t(x_0)
      \St a\in x_1\text{ and }\Cohdeg t{x_0\cup x_1}ab=2\}\\
    &\quad\cup
      \{b\in\Fun t(x_0\cup\Eset{a,a'})
      \setminus(\Fun t(x_0\cup\Eset{a})\cup\Fun t(x_0\cup\Eset{a'}))
      \St a,a'\in x_1,\ a\not=a'\\
    &\quad\quad\text{and }\Cohdeg{t}{x_0\cup x_1}{a}{b}
      =\Cohdeg{t}{x_0\cup x_1}{a'}{b}=1\}\\
    &\quad\cup\{b\in\Fun t(x_0\cup\Eset a)\setminus\Fun t(x_0)
      \St a\in x_2\text{ and }\Cohdeg t{x_0\cup x_2}ab=1\}\,.
  \end{align*}%

  We can describe these expressions in another way, introducing a
  notation for iterated derivatives.
  For all $x_0, x_1, \ldots, x_n \in \Cl E$ such that
  $\Union_{i=1}^n x_i\in\Cl E$ we can define a clique
  $\Matappa{\Fun t^{(n)} (x_0)}{(x_1, \dots,x_n)} \in \Cl F$ by
  \begin{multline*}
    \Fun t^{(n)} (x_0) \cdot (x_1, \dots, x_n) 
    = \Set{b \in \Web F \St \exists m \in \Web{\Oc E}, 
      \exists a_1, \ldots, a_n \in \Web E \text{ such that } \\
      \forall i, a_i \in x_i \text{ and } \Supp m \subseteq x_0 
      \text{ and } (m+[a_1, \ldots, a_n], b) \in t}\,.
  \end{multline*}
  In particular, $\Fun t^{(0)} (x_0) = \Fun t (x_0)$.
  Then, we can easily check that $\Fun{\Dfun t}(\Vect x)$ rewrites
    as follows.
    \begin{equation*} \Fun{\Dfun t}(\Vect
      x)_i =\sum_{\Biind{\Vect n\in\Nat^{\Nat}}{\sum_{j\in\Nat}n_jj=i}}
    \Matappa{{\Fun t}^{(\sum_{j\in\Nat}n_j)}(0)}
    {(\Rep{x_0}{n_0},\Rep{x_1}{n_1},\dots)}
    \end{equation*}
    where \(\Rep xn\) denotes a repetition of \(n\) occurrences of
    \(x\) and the sums are \emph{disjoint} unions of cliques. 
      The functor $\Dfun$ can be though of as a \emph{qualitative} 
      Faa di Bruno formula, in the sense that the rational coefficients
   of the usual Faà~di~Bruno formula do not appear.
    For instance,
  \begin{align*}
      \Fun{\Dfun t}(\Vect x)_0 &= \Fun t(x_0) \\
      \Fun{\Dfun t}(\Vect x)_1 &= t'(x_0) \cdot x_1 \\
      \Fun{\Dfun t}(\Vect x)_2 &= 
        t^{(2)}(x_0) \cdot (x_1, x_1) + t'(x_0) \cdot x_2 \\
      \Fun{\Dfun t}(\Vect x)_3 &= 
        t^{(3)}(x_0) \cdot (x_1, x_1, x_1) 
        + t^{(2)}(x_0) \cdot (x_1, x_2) 
        + t'(x_0) \cdot x_3
  \end{align*}

  If we consider now the original exponential introduced by Girard for
  coherence spaces, where \(\Oc E\) is the set of finite
  \emph{cliques} (and not multicliques) of \(E\), that is
  \(\Oc E=\Ocg E\) using the notations of \cref{sec:diff-in-LL}, then
  we can define similarly \(\Sdiff_E\) when \(\Mdeg=(\Nat,\max)\), but
  not in the case \(\Mdeg=(\Nat,+)\) which is fundamentally
  quantitative (this corresponds to our remark in
  \cref{sec:diff-in-LL} that there is no well-behaved differentiation
  theory for stable functions, see~\cite{Ehrhard26}).
  In the first case, it suffices to take
    \begin{multline*}
    \Sdiff_E
    =\{(\{(i_1,a_1),\dots,(i_k,a_k)\},(i,\{\List a1k\}))\\
    \St k\in\Nat,\,i,\List i1k\in\Mdeg,\,i=\max(i_1,\cdots,i_k)\text{
      and } \Set{(i_1,a_1),\dots,(i_k,a_k)}\in\Cl{\Limpl\Dbimon E}
  \end{multline*}
  and then we have a description of \(\Fun{\Dfun t}(\Vect x)\) which
  is the same as \cref{eq:Dfun-nat-max-desc}.
  This functor %
  \(\Dfun\) is an \(\omega\)-ary version of the functor \(\Dfun\)
  introduced there in an \emph{ad hoc} way.
  Notice that one recovers the theory of %
  that section by taking \(\Mdeg=\{0,1\}\) with \(1+_\Mdeg 1 = 1\).

\subsubsection{Probabilistic coherence spaces}
  \label{ex:pcoh-cart-smc-elem-sum-cont}
  We continue \cref{ex:pcoh-cart-smc-elem-sum}.
  The category \(\PCOH\) is also a resource category.
  We have seen in \cref{ex:pcoh-cart-smc-elem-sum} that this category
  is a cartesian SMCC which is actually \Staraut{} for the dualizing
  object \(\Sbot=\Sone\).
  Then we define \(\Oc X\) by \(\Web{\Oc X}=\Mfin{\Web X}\) and %
  \(\Pcohp{\Oc X}=\Biorth{\Set{\Prom x\St x\in\Pcoh X}}\) where %
  \(\Prom x\in\Realpto{\Web{\Oc X}}\) is defined by
  \(\Prom x_m=x^m=\prod_{a\in\Web X}x_a^{m(a)}\) (see \cref{rk:pcs-bipolar}).
  Given \(t\in\PCOH(X,Y)\), it is easy to check that there is exactly one %
  \(\Oc t\in\Realpto{\Web{\Oc X}\times\Web{\Oc Y}}\) such that %
  \begin{align} %
    \label{eq:pcs-excl-prom-basic-req}
    \forall x\in\Pcoh X\quad\Matappa{\Oc t}{\Prom x}=\Promp{\Matappa tx}\,.
  \end{align}
  Explicitly, a simple computation using
  \cref{eq:pcs-excl-prom-basic-req} shows that
  \begin{align*}
    \forall (m,p)\in\Web{\Oc X}\times\Web{\Oc Y}\quad 
    (\Oc t)_{m,p}=\sum_{r\in\Mexpset mp}\Multinomb pr t^r
  \end{align*}
  where \(\Mexpset mp\) is the set of all
  \(r\in\Mfin{\Web X\times\Web Y}\) such that %
  \(\forall a\in\Web X\ \sum_{b\in\Web Y}r(a,b)=m(a)\) and %
  \(\forall b\in\Web Y\ \sum_{a\in\Web X}r(a,b)=p(b)\), and %
  \(\Multinomb pr
  =\prod_{b\in\Web Y}\frac{\Factor{p(b)}}
  {\prod_{a\in\Web X}\Factor{r(a,b)}}\in\Nat\)
  is a multinomial coefficient.
  It can be proven that if \(s,t\in\Kl\PCOH(\Oc X,Y)\) satisfy %
  \(\forall x\in\Pcoh X\ \Matappa s{\Prom x}=\Matappa t{\Prom x}\), %
  then \(s=t\) (as matrices).
  So the function \(\Fun t:\Pcoh X\to\Pcoh Y\) defined by %
  \(\Fun t(x)=\Matappa t{\Prom x}\) fully determines \(t\); %
  such a function \(\Pcoh X\to\Pcoh Y\) will be called an
  \emph{analytic} function since indeed we have
  \begin{align*}
    \Fun t(x)=\Big(\sum_{m\in\Web{\Oc X}}t_{m,b}x^m\Big)_{b\in\Web Y}
  \end{align*}
  meaning that \(\Fun t\) is defined as a (generalized) power series
  with nonnegative coefficients.

  This functor is a comonad with counit %
  \(\Der X\in\PCOH(\Oc X,X)\) characterized by \(\Fun{\Der X}(x)=x\) %
  and \(\Digg X\in\PCOH(\Oc X,\Occ X)\) by %
  \(\Fun{\Digg X}(x)=\Promm x\), that is, as matrices, %
  \((\Der X)_{m,a}=\Kronecker m{\Mset a}\) %
  (for \((m,a)\in\Web{\Oc X}\times\Web X\)) and %
  \((\Digg X)_{m,M}=\Kronecker m{\sum M}\).
  It is easily checked to be strong monoidal from %
  \((\PCOH,\Stop,\mathord{\IWith})\) to
  \((\PCOH,\Sone,\mathord{\ITens})\).
  For instance the Seely isomorphism %
  \(\Seelyt_{X,Y}\in\PCOH(\Tens{\Oc X}{\Oc Y},\Oc{\Withp XY}\) is
  fully characterized by %
  \(\Matappa{\Seelyt_{X,Y}}{\Tensp{\Prom x}{\Prom y}}=\Prom{(x,y)}\) %
  (identifying \(\Pcohp{\With XY}\) with \({\Pcoh X}\times{\Pcoh Y}\)).

  It was proved in~\cite{CrubilleEhrhardPaganiTasson16} that this
  exponential is the free one, that is, the SMC \(\PCOH\) is a Lafont
  category.
  So by \cref{th:lafont-elementary-differential-structure} the object
  \(\Dbimon=\Bwith_{i\in\Mdeg}\Sone\) has a structure of \(\oc\)-coalgebra
  \(\Ddiff\in\PCOH(\Dbimon,\Oc\Dbimon)\), which is
  characterized by
  \begin{align*}
    \Ddiff_{i,\Mset{\List i1k}}=\Kronecker i{i_1+\cdots+i_k}
  \end{align*}
  for \(i,\List i1k\in\Mdeg\).
  This structure turns \(\PCOH\) into a \Elementary{} coherent
  differential category.
  So we have an induced distributive law %
  \(\Sdiff_X\in\PCOH(\Oc{\Sfun X},\Sfun{\Oc X})\) and an easy
  computation shows that
  \begin{align*}
    (\Sdiff_X)_{\Vect m,(i,p)}
    =
    \frac{\Factor{(\sum_{j\in\Mdeg}m_j)}}{\prod_{j\in\Mdeg}\Factor{m_j}}
    \text{ if }\sum_{j\in\Mdeg}(\Card{m_j})j=i
    \text{\quad and\quad}
    (\Sdiff_X)_{\Vect m,(i,p)}=0\text{ otherwise}
  \end{align*}
  using the obvious bijection between
  \(\Web{\Oc{\Sfun X}}=\Mfin{\Mdeg\times\Web X}\) and
  \(\Mfin{\Web X}^{(\Mdeg)}\), the set off all
  \(\Vect m\in\Mfin{\Web X}^\Mdeg\) such that \(m_i=\Msetempty\) for
  almost all \(i\in\Mdeg\).
  It follows that the extension \(\Dfun\) of \(\Sfun\) to \(\Kl\PCOH\)
  acts as follows on morphisms.
  Let \(t\in\PCOH(\Oc X,Y)\) then
  \(\A\Dfun t\in\PCOH(\Oc{\Sfun X},\Sfun Y)\) is given by %
  \begin{align*}
    (\Dfun t)_{\Vect m,(i,b)}
    =
    \frac{\Factor{(\sum_{j\in\Mdeg}m_j)}}{\prod_{j\in\Mdeg}\Factor{m_j}}
    t_{\sum_{j\in\Mdeg}m_j,b}
    \text{ if }\sum_{j\in\Mdeg}(\Card{m_j})j=i
    \text{\quad and\quad}
    (\Dfun t)_{\Vect m,(i,b)}=0\text{ otherwise}
  \end{align*}
  This means that, given an element of \(\Pcohp{\Sfun X}\) that we
  consider as a family %
  \(\Vect x\in\Pcoh X^\Mdeg\) such that
  \(\sum_{j\in\Mdeg}x_j\in\Pcoh X\), we have
  \begin{align*}
    \Fun{\Dfun t}(\Vect x)
    =\sum_{(i,b)\in\Mdeg\times\Web Y}
    \Big(
    \sum_{\Biind{\Vect m\in\Mfin{\Mdeg\times\Web X}}
    {\sum_{j\in\Mdeg}(\Card{m_j})j=i}}
    \frac{\Factor{(\sum_{j\in\Mdeg}m_j)}}{\prod_{j\in\Mdeg}\Factor{m_j}}
    t_{\sum_{j\in\Mdeg}m_j,b}
    \prod_{j\in\Mdeg}x_j^{m_j}
    \Big)\Base{(i,b)}\,.
  \end{align*}
  Notice that if \(\sum_{j\in\Mdeg}(\Card{m_j})j=0\) then
  \(m_j=\Msetempty\) for all \(j\in\Mdeg\setminus\{0\}\) by our
  \Mpositivity{} assumption in \(\Mdeg\) and \(\Vect m\) is completely
  given by an \(m_0=m\in\Mfin{\Web X}\).
  It follows that \(\Fun{\Dfun t}(\Vect x)_{(0,b)}=\Fun t(x_0)_b\).

  We describe now the monad structure of \(\Dfun\) on \(\Kl\PCOH\).
  Remember that we identify \(\Pcohp{\Dfun X}\) with
  \(\{\Vect x\in\Pcoh X^\Mdeg\St\sum_{i\in\Mdeg} x_i\in\Pcoh X\}\),
  and hence we can identify \(\Pcohp{\Dfun^2X}\) with
  \(\{\Vect x\in\Pcoh X^{\Mdeg^2}\St\sum_{l,r\in\Mdeg^2}
  x_{l,r}\in\Pcoh X\}\).
  Then the unit \(\Dmonu_X\in\Kl\PCOH(X,\Dfun X)\) of the monad is
  characterized by \(\Fun{\Dmonu_X}(x)_j=\Kronecker j0x\) and its
  multiplication \(\Dmonm_X\in\Kl\PCOH(\Dfun^2X,\Dfun X)\) is characterized by
  \(\Fun{\Dmonm_X}(\Vect x)_i=\sum_{(l,r)\in\Mdegsl i}x_{l,r}\).
  When \(\Mdeg\) is total, the comonad structure of \(\Dfun\) on
  \(\Kl\PCOH\) can be described as follows (these definitions are
  always possible, but yield a comonad only if \(\Mdeg\) is total; if
  not, \(\Dsum\) is not natural in general).
  The counit \(\Dsum_X\in\Kl\PCOH(\Dfun X,X)\) is given by
  \(\Fun{\Dsum_X}(\Vect x)=\sum_{i\in\Mdeg}x_i\) and the
  comultiplication \(\Dlift_X\in\Kl\PCOH(\Dfun X,\Dfun^2X)\) is
  characterized by \(\Fun{\Dlift_X}(\Vect x)_{l,r}=\Kronecker lrx_l\).
  The symmetry \(\Dflip\in\Kl\PCOH(\Dfun^2X,\Dfun^2X)\) is given by
  \(\Fun{\Dflip_X}(\Vect x)_{l,r}=x_{r,l}\).
  Remember from \cref{prop:kleisli-bimonad} that when \(\Mdeg\) is
  total, \((\Dfun,\Dmonu,\Dmonm,\Dsum,\Dlift,\Dflip)\) is a
  \Symbimonad.

  Now we consider specific examples of \(\Mdeg\) for which we can
  write more explicit descriptions of these morphisms.
  
  \paragraph{The case \(\Mdeg=(\Nat,\max)\): finite differences.}
  Then the condition \(\sum_{j\in\Mdeg}(\Card{m_j})j=i\) on
  \(\Vect m\in\Mfin{\Web X}^{(\Mdeg)}\) means that
  \(m_i\not=\Msetempty\) and \(m_j=\Msetempty\) for all \(j>i\).
  Observe that by the multinomial formula we have
  \begin{align*}
    \Fun t(x_0+\cdots+x_i)
    =\sum_{b\in\Web Y}\Big(
    \sum_{\Vect m\in\Mfin{\Web X}^{\{0,\dots,i\}}}
    \frac{\Factor{(\sum_{j=0}^im_j)}}{\prod_{j=0}^i\Factor{m_j}}    
    t_{\big(\sum_{j=0}^im_j,b\big)}
    \prod_{j=0}^ix_j^{m_j}
    \Big)\Base b
  \end{align*}
  and it follows that
  \begin{align*}
    \Fun{\Dfun t}(\Vect x)_0
    &=\Fun t(x_0)\\
    \Fun{\Dfun t}(\Vect x)_{i}
    &=\Fun t(x_0+\cdots+x_{i})-\Fun t(x_0+\cdots+x_{i-1})
      \text{\quad if }i\in\Nat\setminus\Eset 0\,.
  \end{align*}
  The monad structure of \(\Dfun\) is given by
  \(\Fun{\Dmonu_X}(x)=(x,0,0,\dots)\) and
  \(\Fun{\Dmonm_X}(\Vect x)_i=
  x_{i,i}+\sum_{j=0}^{i-1}(x_{i,j}+x_{j,i})\), and the counit of the
  comonad is given by \(\Fun{\Dsum_X}(\Vect x)=\sum_{i\in\Nat}x_i\).
  The naturality of this counit means that
  \begin{align*}
    \Fun t\Big(\sum_{i=0}^\infty x_i\Big)
    =\sum_{i=0}^\infty\Big(\Fun t\big(\sum_{j=0}^ix_j\big)
    -\Fun t\big(\sum_{j=0}^{i-1}x_j\big)\Big)
  \end{align*}
  which simply means that \(\Fun t\) is Scott continuous (that is,
  commutes with lubs of increasing sequences in \(\Pcoh X\)).

  \paragraph{The case \(\Mdeg=(\Nat,+)\): Taylor and Faà~di~Bruno.}
  The map \(\Fun t:\Pcoh X\to\Pcoh Y\) admits all derivatives, in the
  sense that, for all \(x\in\Pcoh X\) and for all \(k\in\Nat\) there
  is a \(k\)-derivative of \(\Fun t\) at \(x\) which is a \(k\)-linear
  map \(\Fun t^{(k)}:\cT_x^{(k)}X\to\Pcoh Y\) where \(\cT_x^{(k)}X\)
  is the set of all tuples \(\List x1k\in\Pcoh X\) such that
  \(x+\sum_{i=1}^kx_i\in\Pcoh X\): considering
  \begin{align*}
    f(x_1,\dots,x_k)=\Fun t(x+x_1+\cdots+x_k)-\Fun t(x)
  \end{align*}
  as a power series in \(\List x1k\), then
  \(\Matappa{\Fun t^{(k)}(x)}{(\List x1k)}\) is the \(k\)-linear part
  of this power series.
  Explicitly, for \(b\in\Web Y\), let us compute, using the
  multinomial formula (setting \(x_0=x\)):
  \begin{align*}
    \Fun t(x+x_1+\cdots+x_k)_b
    &=\sum_{m\in\Mfin{\Web X}}t_{(m,b)}(x+x_1+\cdots+x_k)^m\\
    &=\sum_{m\in\Mfin{\Web X}}t_{(m,b)}
      \Big(\sum_{\Biind{\Vect m\in\Mfin{\Web X}^{\{0,\dots,k\}}}
      {m_0+m_1+\cdots+m_k=m}}\frac{\Factor m}{\prod_{i=0}^k\Factor{m_i}}
      \prod_{i=0}^kx_i^{m_i}\Big)\\
    &=\sum_{\Vect m\in\Mfin{\Web X}^{\{0,\dots,k\}}}
      \frac{(\Factor{\sum_{i=0}^km_i})}{\prod_{i=0}^k\Factor{m_i}}
      t_{(\sum_{i=0}^km_i,b)}\prod_{i=0}^kx_i^{m_i}
  \end{align*}
  and \(\Matappa{\Fun t^{(k)}(x)}{(\List x1k)}\) is the part of this
  summation which corresponds to vectors \(\Vect m\) of multisets such
  that \(\Card{m_i}=1\) for \(i=1,\dots,k\).
  Such a multiset being fully determined by its unique element, we
  have
  \begin{align*}
    \Matappa{\Fun{t}^{(k)}(x)}{(\List x1k)}_b
    =\sum_{\Biind{m\in\Mfin{\Web X}}{(\List a1k)\in\Web X^k}}
    \frac{\Factor{(m+\Mset{\List a1k})}}{\Factor m}
    t_{(m+\Mset{\List a1k},b)}x^m\prod_{i=1}^kx_{i,a_i}
  \end{align*}
  which is obviously symmetric in \(\List x1k\).
  Let us assume that \(x=0\) so that \(\cT_x^{(k)}\) is the set of
  \(\Vect x\in\Pcoh X^k\) such that \(\sum_{i=1}^kx_i\in\Pcoh X\).
  Being \(k\)-linear, the function
  \(\Fun{t}^{(k)}(0)\) can be extended to \(\Pcoh X^k\), with values in
  \(\Pcohc Y=\{y\in\Realpto{\Web Y}\St\exists\epsilon>0\ \epsilon
  y\in\Pcoh Y\}\).
  Indeed if \(\Vect x\in\Pcoh X^k\), then
  \((\frac 1kx_i)_{i=1}^k\in\cT^{(k)}_0\) and so we can define
  \(\Fun{t}^{(k)}(0)(\Vect x)\) as
  \(k^k\Matappa{\Fun{t}^{(k)}(0)}{(\frac 1kx_1,\dots,\frac 1kx_k)}\),
  using the definition above of \(\Fun{t}^{(k)}(0)\) on
  \(\cT^{(k)}_0\).

  Using these notions, a computation (whose details can be
  found in \cite{EhrhardWalch25}) shows that
  \(\Fun{\Dfun t}\) actually computes the Faà di Bruno formula:
  \begin{align*}
    \Fun{\Dfun t}(\Vect x)
    =\Big(\sum_{\Biind{\Vect n\in\Nat^{\Nat}}{\sum_{j\in\Nat}n_jj=i}}
    \frac1{\prod_{j=0}^\infty 
    \Factor{n_j}}\Matappa{{\Fun t}^{(\sum_{j\in\Nat}n_j)}(0)}
    {(\Rep{x_0}{n_0},\Rep{x_1}{n_1},\dots)}
    \Big)_{i\in\Nat}
  \end{align*}
  where \(\Rep xn\) is a the repetition of \(x\) 
  \(n\) times, so that the vector of parameters
  \( (\Rep{x_0}{n_0},\Rep{x_1}{n_1},\dots)\) has length
  \(\sum_{j\in\Nat}n_j\in\Nat\).
  Do well notice that in this formula we use the extended definition
  of derivatives at \(0\) since we have no guarantee that
  \(\sum_{i=0}^\infty n_ix_i\in\Pcoh X\).
  We know by the general theory of \DDIFF{} that
  \(\Fun{\Dfun t}(\Vect x)\in\Pcoh{\Sfun Y}\) whenever
  \(\Vect x\in\Pcoh{\Sfun X}\), in spite of the fact that the
  sub-expressions
  \(\Matappa{{\Fun t}^{(\sum_{j\in\Nat}n_j)}(0)}
  {(\Rep{x_0}{n_0},\Rep{x_1}{n_1},\dots)}\)
  do not necessarily denote elements of \(\Pcoh{Y}\).
  This is due to the presence of the combinatorial coefficients
  \(\Invp{\prod_{j=0}^\infty \Factor{n_j}}\) which ``tame'' these
  possibly too large values.

  To understand this Faà~di~Bruno formula it is essential to notice
  that, for each \(i\in\Nat\) and \(n_0\in\Nat\), there are only
  finitely many vectors \(\Vect n\in\Nat^{\Nat\setminus\Eset0}\) such
  \(\sum_{j\in\Nat}n_jj=i\).
  The first terms of this sequence are
  \begin{align*}
    \Fun{\Dfun t}(\Vect x)_0
    &=\Fun t(x_0)\\
    \Fun{\Dfun t}(\Vect x)_1
    &=\Matappa{\Fun t^{(1)}(x_0)}{x_1}\\
    \Fun{\Dfun t}(\Vect x)_2
    &=\frac12\Matappa{\Fun t^{(2)}(x_0)}{(x_1,x_1)}
      +\Matappa{\Fun t^{(1)}}{x_2}\\
    \Fun{\Dfun t}(\Vect x)_3
    &=\frac 16\Matappa{\Fun t^{(3)}(x_0)}{(x_1,x_1,x_1)}
      +\Matappa{\Fun t^{(2)}(x_0)}{(x_1,x_2)}
      +\Matappa{\Fun t^{(1)}}{x_3}\\
    \Fun{\Dfun t}(\Vect x)_4
    &=\frac 1{24}\Matappa{\Fun t^{(4)}(x_0)}{(x_1,x_1,x_1,x_1)}
      +\frac 12\Matappa{\Fun t^{(3)}(x_0)}{(x_1,x_1,x_2)}\\
    &\quad+\frac 12\Matappa{\Fun t^{(2)}(x_0)}{(x_2,x_2)}
      +\Matappa{\Fun t^{(2)}(x_0)}{(x_1,x_3)}
      +\Matappa{\Fun t^{(1)}(x_0)}{x_4}\,.
  \end{align*}
  For instance, concerning the last expression, observe that the
  sequences \(\Vect n=(n_0,n_1,\dots)\) such that
  \(\sum_{j\in\Nat}n_jj=4\) are exactly those of shape
  \((n_0,4,0,0,\dots)\) since \(4\times1=4\), \((n_0,2,1,0,0,\dots)\)
  since \(2\times 1+1\times2=4\), \((n_0,0,2,0,0,\dots)\) since
  \(2\times2=4\), \((n_0,1,0,1,0,0,\dots)\) since
  \(1\times 1+1\times 3=4\), \((n_0,0,0,0,1,0,0,\dots)\) since
  \(1\times 4=4\) and there are no others (\(n_0\) can be any element
  of \(\Nat\) since \(n_00\) is always \(0\)).
  In particular it is easy to see that
  \begin{align*}
    \Fun{\Dfun t}(0,x_1,0,0,\dots)
    =(\Fun t(0),\Matappa{\Fun t^{(1)}(0)}{x_1},
    \frac12\Matappa{\Fun t^{(2)}(0)}{(x_1,x_1)},
    \frac16\Matappa{\Fun t^{(3)}(0)}{(x_1,x_1,x_1)},\dots)
  \end{align*}
  and we get in \(\Fun{\Dfun t}(0,x_1,0,0,\dots)\) all the terms of
  the Taylor expansion of \(\Fun t\) at \(0\).

  The monad structure of \(\Dfun\) is given by
  \(\Fun{\Dmonu_X}(x)=(x,0,0,\dots)\) and
  \(\Fun{\Dmonm_X}(\Vect x)_i=
  \sum_{l+r=i}x_{l,r}\), and the counit of the
  comonad is given by \(\Fun{\Dsum_X}(\Vect x)=\sum_{i\in\Nat}x_i\).
  The naturality of this counit means that
  \begin{align*}
    \Fun t\Big(\sum_{i=0}^\infty x_i\Big)
    =\sum_{n=0}^\infty\frac1{\Factor n}
    \Matappa{\Fun t^{(n)}(x_0)}{(\Rep{x_1+x_2+\cdots}n)}
  \end{align*}
  which simply means that \(\Fun t\) obeys the Taylor formula.

  \paragraph{Truncated Taylor formula.}
  An interesting variation consists in taking
  \(\Mdeg=\{0,1,\dots,N\}\) with \(i+_\Mdeg j=i+j\) if \(i+j\leq N\)
  and \(i+_\Mdeg j\) undefined otherwise.
  Then \(\Pcohp{\Dfun X}\) can be identified with the set of all
  \(\Vect x\in\Pcoh X^{\{0,\dots,N\}}\) such that
  \(\sum_{i=0}^Nx_i\in\Pcoh X\), and
  \begin{align*}
    \Fun{\Dfun t}(\Vect x)
    =\Big(\sum_{\Biind{\Vect n\in\Nat^{\{0,\dots,N\}}}{\sum_{j=0}^N n_jj=i}}
    \frac1{\prod_{j=0}^N
    \Factor{n_j}}\Matappa{{\Fun t}^{(\sum_{j=0}^Nn_j)}(0)}
    {(\Rep{x_0}{n_0},\Rep{x_1}{n_1},\dots,\Rep{x_N}{n_N})}
    \Big)_{i=0}^N
  \end{align*}
  is the Faà di Bruno formula truncated at rank \(N\).
  In this case however the family of morphisms \(\Dsum_X\) is not
  natural because
  \[
    \Fun t\Big(\sum_{i=0}^N x_i\Big) =\sum_{n=0}^N\frac1{\Factor
      n}\Matappa{\Fun t^{(n)}(x_0)}{(\Rep{x_1+x_2+\cdots+x_N}n)}
  \]
  holds only if \(t\) is polynomial of degree \(\leq N\).
  This is of course consistent with the fact that in that case
  \(\Mdeg\) is not total.

  \paragraph{Taylor formula with a remainder.}
  We take now \(\Mdeg=\{0,1,\dots,N\}\) with
  \(i+_\Mdeg j=\inf(i+j,N)\).
  We have
  \begin{align*}
    \Fun{\Dfun t}(\Vect x)_i
    &=\sum_{\Biind{\Vect n\in\Nat^{\{0,\dots,N\}}}{\sum_{j=0}^N n_jj=i}}
    \frac1{\prod_{j=0}^N
    \Factor{n_j}}\Matappa{{\Fun t}^{(\sum_{j=0}^Nn_j)}(0)}
    {(\Rep{x_0}{n_0},\Rep{x_1}{n_1},\dots,\Rep{x_N}{n_N})}
    \text{\quad for }i<N\\
    \Fun{\Dfun t}(\Vect x)_N
    &=\sum_{\Biind{\Vect n\in\Nat^{\{0,\dots,N\}}}{\sum_{j=0}^N n_jj\geq N}}
    \frac1{\prod_{j=0}^N
    \Factor{n_j}}\Matappa{{\Fun t}^{(\sum_{j=0}^Nn_j)}(0)}
    {(\Rep{x_0}{n_0},\Rep{x_1}{n_1},\dots,\Rep{x_N}{n_N})}
    \\
    &=\Fun t(x_0+x_1+\cdots+x_N)-
      \sum_{\Biind{\Vect n\in\Nat^{\{0,\dots,N\}}}{\sum_{j=0}^N n_jj<N}}
    \frac1{\prod_{j=0}^N
    \Factor{n_j}}\Matappa{{\Fun t}^{(\sum_{j=0}^Nn_j)}(0)}
    {(\Rep{x_0}{n_0},\Rep{x_1}{n_1},\dots,\Rep{x_N}{n_N})}
  \end{align*}
  so that for instance
  \begin{align*}
    \Fun{\Dfun t}(0,x_1,0,\dots,0)
    &=(\Fun t(0),\Matappa{\Fun t^{(1)}(0)}{x_1},
    \frac12\Matappa{\Fun t^{(2)}(0)}{(x_1,x_1)},
      \frac16\Matappa{\Fun t^{(3)}(0)}{(x_1,x_1,x_1)},\dots,\\
    &\quad\quad
      \frac1{\Factor{(N-1)}}
      \Matappa{\Fun t^{(N-1)}(0)}{(\Rep{x_1}{N-1})},
      f(x_1)-\sum_{n=0}^{N-1}
      \frac1{\Factor n}\Matappa{\Fun t^{(n)}(0)}{(\Rep{x_1}n)})
  \end{align*}
  where we see the last element of this \((N+1)\)-tuple as the
  remainder of the Taylor expansion of order \(N-1\) of \(\Fun t\) at
  \(0\).
  The monad structure of \(\Dfun\) is given by
  \(\Fun{\Dmonu_X}(x)=(x,\Rep0{N-1})\) and
  \[
    \Fun{\Dmonm_X}(\Vect x)_i=
    \begin{cases}
      \sum_{l+r=i}x_{l,r}&\text{if }i<N\\
      \sum_{l+r\geq N}x_{l,r}&\text{if }i=N
    \end{cases}
  \]
  and the counit of the comonad is given by
  \(\Fun{\Dsum_X}(\Vect x)=\sum_{i=0}^Nx_i\); indeed \(\Dfun\) is a
  bimonad in that case since the monoid \(\Mdeg\) is total.

  \paragraph{A fancy example.}
  To illustrate the versatility of our setting, consider the fancy
  structure \(\Mdeg=(\{0,1,2\},+)\) with addition defined by
  \(1+1=2\), \(1+2=1\) and \(2+2=2\), which is easily checked to be a
  \Mpositive{} (total) commutative monoid.
  Notice that
  \begin{align*}
    n1=
    \begin{cases}
      1&\text{if }n\in 2\Nat+1\\
      2&\text{if }n\in 2\Nat+2
    \end{cases}
    \text{\quad and\quad}n2=2\text{ for all }n\in\Nat+1
  \end{align*}
  and it follows that
  \begin{align*}
    n_11+n_22=
    \begin{cases}
      0&\text{if }n_1=n_2=0\\
      2&\text{if }n_1=0\text{ and }n_2\in\Nat+1\\
      1&\text{if }n_1\in2\Nat+1\text{, for all }n_2\in\Nat\\
      2&\text{if }n_1\in2\Nat+2\text{, for all }n_2\in\Nat
    \end{cases}
  \end{align*}
  So we have
  \begin{align*}
    \Fun{\Dfun t}(x_0,x_1,x_2)_0
    &=\Fun t(x_0)\\
    \Fun{\Dfun t}(x_0,x_1,x_2)_{1,b}
    &=\sum_{\Biind{(m_0,m_1,m_2)\in\Mfin{\Web X}^3}{\Card{m_1}\in2\Nat+1}}
      \frac{\Factor{(m_0+m_1+m_2)}}{\Factor{m_0}\Factor{m_1}\Factor{m_2}}
      t_{(m_0+m_1+m_2,b)}x_0^{m_0}x_1^{m_1}x_2^{m_2}\\
    \Fun{\Dfun t}(x_0,x_1,x_2)_{2,b}
    &=\sum_{\Biind{(m_0,m_1,m_2)\in\Mfin{\Web X}^3}{\Card{m_1}\in2\Nat+2}}
      \frac{\Factor{(m_0+m_1+m_2)}}{\Factor{m_0}\Factor{m_1}\Factor{m_2}}
      t_{(m_0+m_1+m_2,b)}x_0^{m_0}x_1^{m_1}x_2^{m_2}\\
    &\quad+\sum_{\Biind{(m_0,m_2)\in\Mfin{\Web X}^2}{\Card{m_2}\not=0}}
      \frac{\Factor{(m_0+m_2)}}{\Factor{m_0}\Factor{m_2}}
      t_{(m_0+m_2,b)}x_0^{m_0}x_2^{m_2}
  \end{align*}
  and therefore
  \begin{align*}
    \Fun{\Dfun t}(x_0,x_1,x_2)
    =(\Fun t(x_0),
    &\frac12(\Fun t(x_0+x_1+x_2)-\Fun t(x_0-x_1+x_2)),\\
    &\frac12(\Fun t(x_0+x_1+x_2)+\Fun t(x_0-x_1+x_2))-\Fun t(x_0))
  \end{align*}

  The monad structure of \(\Dfun\) is given by
  \(\Fun{\Dmonu_X}(x)=(x,0,0)\) and
  \begin{align*}
    \Fun{\Dmonm_X}(\Vect
    x)=(x_{0,0},x_{1,0}+x_{0,1}+x_{2,1}+x_{1,2},x_{0,2}+x_{2,0}+x_{1,1}+x_{2,2})
  \end{align*}
  and the counit of the comonad is given by
  \(\Fun{\Dsum_X}(\Vect x)=x_0+x_1+x_2\).
  The naturality of this counit is obvious.

\section{Products, function spaces and fixpoints}

\label{sec:diff-lambda-cat}
This final section describes structures which are byproducts of the
general \DDIFF{} theory and are required to interpret a functional
programming language featuring differential constructs based on
\(\Mdeg\).
We decided not to present such a language in the present paper which
is already too long, it could be a generalization of the coherent
differential extension of \(\PCF\) developed in \cite{Ehrhard23b}.

\begin{remark}
  As explained in the introduction of this paper, we introduced in
  \cite{EhrhardWalch23,EhrhardWalch25} a \emph{cartesian} theory of
  coherent differentiation, which is a direct axiomatization of
  differentiation in the Kleisli category $\Kl \cL$ rather than the
  linear category $\cL$.
  The point of such an axiomatization is to generalize the ideas of
  coherent differentiation to cartesian categories that are not
  necessarily Kleisli categories of models of Linear Logic.

  We stuck to the linear theory in this paper because of size
  constraint, but a cartesian theory of $\Mdeg$-differentiation
  following the lines of \cite{EhrhardWalch23,EhrhardWalch25} is
  perfectly feasible.
  Furthermore, the content developed in this section would perfectly
  work for this cartesian theory, outside of any linear logical
  considerations.
\end{remark}

We explain first how to define \emph{partial derivatives} which are
required because, in our functional language, we must be able to
express derivatives of a term wrt.~a given \(\lambda\)-calculus
variable, as in the differential \(\lambda\)-calculus.
Next we make explicit an equation describing the interaction between
our \DDIFF{} operations and the \(\lambda\)-calculus operation of
application.
Last we assume that the Kleisli category \(\Kl\cL\) has fixpoint
operators ``à la Scott'' defined terms of least upper bounds, using an
order relation that we assume to be induced by the \Dpmonoid{}
structures hom-sets of our category \(\cL\) are equipped with.
In this setting we prove an equation describing the interaction between
\DDIFF{} and such fixpoint operators.

\begin{convention}
  In this section, we assume that \(\cL\) is a \Lingdiffcat{}
  (not necessarily representable).
\end{convention}

\subsection{Partial derivatives}
We have defined a global differentiation operator that maps a morphism 
\(f\in\Kl\cL(X_1\IWith\cdots\IWith X_n,Y)\)
to \(\Dfun f\in\Kl\cL(\Sfun(X_1\IWith\cdots\IWith X_n),\A\Sfun Y)\).
Up to the iso stipulated by \ref{ax:swith}), we thus have
\(\A\Dfun f\in\Kl\cL(\A\Sfun X_1\IWith\cdots\IWith \A\Sfun X_n,\A\Sfun
Y)\). Intuitively, this morphism $\Dfun f$ %
maps \((x_1,u_1),\dots,(x_n,u_n)\) to %
\((f(x_1,\dots,x_n),f'(x_1,\dots,x_n)\cdot(u_1,\dots,u_n))\) when
\(\Mdeg=\Eset{0,1}\) with \(1+1\) undefined (that we denote as
\(\Mder2\) in the sequel) and to %
\((f(x_1,\dots,x_n),f(x_1+u_1,\dots,x_n+u_n)-f(x_1,\dots,x_n))\) when
\(\Mdeg=\Eset{0,1}\) with \(1+1=1\) (that we denote as \(\Mdiff1\) in
the sequel).
For any \(i\in\Set{1,\dots,n}\) we also need to be able to define an
\(i\)th \emph{partial derivative} %
\(\A{\Dfunpart i}f \in\Kl\cL(X_1\IWith\cdots\IWith\A\Sfun
X_i\IWith\cdots\IWith X_n,\A\Sfun Y)\) %
such that, intuitively
\begin{align*}
  \A{\Dfunpart i}f : (x_1,\dots,(x_i,u),\dots,x_n)
  &\mapsto(f(x_1,\dots,x_n),\frac{\partial f(\List x1n)}{\partial x_i}\cdot u)
    \text{\quad when }\Mdeg=\Mder2\\
  \A{\Dfunpart i}f : (x_1,\dots,(x_i,u),\dots,x_n)
  &\mapsto(f(x_1,\dots,x_n),f(x_1,\dots,x_i+u_i,\dots,x_n)\\
  &\hspace{10em}-f(x_1,\dots,x_n))
    \text{\quad when }\Mdeg=\Mdiff1\,.
\end{align*}

When $\Mdeg=\Mder2$ we expect the derivative to be equal to the sum of the
partial derivatives.
\begin{align*}
  f'(x_1,\dots,x_n)\cdot(u_1,\dots,u_n)
  =\sum_{i=1}^n\frac{\partial f(\List x1n)}{\partial x_i}\cdot u_i.
\end{align*}
The corresponding equation for \(\Mdeg=\Mdiff1\) reads as follows, for
instance when \(n=3\),
\begin{align*}
  f(x_1+u_1,x_2+u_2,x_3+u_3)&-f(x_1,x_2,x_3)\\
  =&f(x_1+u_1,x_2+u_2,x_3+u_3)-f(x_1+u_1,x_2,x_3+u_3)\\
  &+f(x_1+u_1,x_2,x_3+u_3)-f(x_1,x_2,x_3+u_3)\\
  &+f(x_1,x_2,x_3+u_3)-f(x_1,x_2,x_3)
\end{align*}
and similarly for the other permutations of \((1,2,3)\).
The general theory of \DDIFF{} allows us to define such partial
derivatives and prove their expected properties without further
assumptions as we explain now.

We take \(n=2\) to simplify notations but the general case is not more
complicated conceptually. Recall from \cref{prop:S-strong-with} that
$(\Sfunm, \Sin 0, \Swithiso)$ is a strong monoidal monad with respect
to $\IWith$, and thus is in particular a lax monoidal monad.
It is well known that every lax monoidal monad has a canonical
commutative monad structure~\cite{Kock71}, given by
\(\Wstren1_{X_1,X_2} %
\in\cL(\With{\Sfun X_1}{X_2},\Sfun{\Withp{X_1}{X_2}})\) and
\(\Wstren2_{X_1,X_2} %
\in\cL(\With{X_1}{\Sfun X_2},\Sfun{\Withp{X_1}{X_2}})\) defined as the
following composition of morphisms
\begin{equation} \label{eq:S-with-commutative}
  \begin{tikzcd}[row sep = tiny]
    \A\Sfun{X_1}\IWith{X_2}
    \ar[r,"\With{\A\Sfun{X_1}}{\Sin0}"]
    &[2em]
    \A\Sfun{X_1}\IWith\A\Sfun{X_2}
    \ar[r,"\Swithiso_{X_1,X_2}"]
    &[2em]
    \Sfun{\Withp{X_1}{X_2}} \\ 
    {X_1}\IWith{\A \Sfun X_2}
    \ar[r,"\With{\Sin0}{\A\Sfun{X_2}}"]
    &[2em]
    \A\Sfun{X_1}\IWith\A\Sfun{X_2}
    \ar[r,"\Swithiso_{X_1,X_2}"]
    &[2em]
    \Sfun{\Withp{X_1}{X_2}}
  \end{tikzcd}.
\end{equation}
Intuitively \(\Wstren1((x_1,u_1),x_2)=((x_1,x_2),(u_1,0))\) and
\(\Wstren2(x_1,(x_2,u_2))=((x_1,x_2),(0,u_2))\).
It is also easily checked that \(\Wstren2\) can be obtained from %
\(\Wstren1\) using the symmetry isomorphism associated with
the product: %
\(\Wstren2
=\Sfun{\Tuple{\Proj2,\Proj1}}
\Comp\Wstren1
\Comp\Tuple{\Proj2,\Proj1}\).

\begin{theorem} \label{th:dfun-strength-with} The strong monoidal
  structure $(\Sin 0, \Swithiso)$ on $\Sfun$ given in
  \cref{prop:S-strong-with} extends in $\Kl \cL$ (in the sense of
  \cref{def:extension}) to a strong monoidal structure
  \[
    \Derfun(\Sin 0) \in \Kl \cL(\Top, \Dfun \Top) \quad \text{and}
    \quad \Derfun(\Swithiso_{X_1, X_2}) \in \Kl \cL(\Dfun X_1 \IWith
    \Dfun X_2, \Dfun(X_1 \IWith X_2))
  \]
  on the monad $\Dfun$ described in \cref{prop:Dfun-monad}.
  Similarly, the commutative monad structure \((\Wstren1, \Wstren2)\)
  on $\Sfun$ given in \cref{eq:S-with-commutative} above extends in
  $\Kl \cL$ to a commutative monad structure
  \begin{align*}
    \Dfunstr1_{X_1, X_2} 
    &= \Derfun \Wstren1_{X_1,X_2} \in \Kl\cL(\With{\A\Dfun X_1}{X_2},
      \Dfun{\Withp{X_1}{X_2}}) \\ 
    \Dfunstr2_{X_1, X_2} 
    & = \A\Derfun{\Wstren2_{X_1,X_2}} \in
      \Kl\cL(\With{X_1}{\A\Dfun X_2},\Dfun{\Withp{X_1}{X_2}}).
  \end{align*}
  on the monad $\Dfun$.
  Furthermore, the following diagram commutes in \(\Kl\cL\)
  \begin{equation} %
    \label{eq:com-strength-diff-kleisli}
    \begin{tikzcd}
      &[-2em]
      \With{\A\Dfun X_1}{\A\Dfun X_2}
      \ar[ld,swap,"\Dfunstr1_{X_1,\Dfun X_2}"]
      \ar[rd,"\Dfunstr2_{\Dfun X_1,X_2}"]
      &[-2em]\\
      \Dfun\Withp{X_1}{\Dfun X_2}
      \ar[d,swap,"\Dfun\Dfunstr2_{X_1,X_2}"]
      &
      &
      \Dfun\Withp{\Dfun X_1}{X_2}
      \ar[d,"\Dfun\Dfunstr1_{X_1,X_2}"]\\
      \Dfun^2\Withp{X_1}{X_2}
      \ar[rr,"\Sflip_{\With{X_1}{X_2}}"']
      &&
      \Dfun^2\Withp{X_1}{X_2}
    \end{tikzcd}
  \end{equation}
\end{theorem}

\begin{proof}
  Since $\Dfun$ is an extension of $\Sfun$,
  we get that in $\Kl \cL$ 
  \[ \Tuple{\Dfun(\Derfun \Proj 1), \Dfun(\Derfun \Proj 2)} 
  = \Derfun \Tuple{\Sfun(\Proj 1), \Sfun(\Proj 2)} \in 
  \Kl \cL(\Dfun (X_1 \IWith X_2), \Dfun X_1 \IWith \Dfun X_2) \]
  is an iso by \ref{ax:swith}. In particular, 
  $\Dfun$ has a strong monoidal structure, following the same 
  recipe as the one given for $\Sfun$ in \cref{prop:S-strong-with},
  and the strong monoidality 
  $\Dwithiso_{X_1, X_2} \in \cL(\Dfun X_1 \IWith \Dfun X_2, 
  \Dfun(X_1 \IWith X_2))$ is equal to $\Derfun \Swithiso$.
  Mirroring $\Sfun$, this strong monoidal structure on $\Dfun$
  induces by \cite{Kock71} a canonical commutative monad 
  structure $(\Dfunstr1, \Dfunstr2)$ defined as the following 
  compositions in $\Kl \cL$.
  \begin{equation*}
  \begin{tikzcd}[row sep = tiny]
    \A\Dfun{X_1}\IWith{X_2}
    \ar[r,"\With{\A\Dfun{X_1}}{\Derfun(\Sin0)}"]
    &[2em]
    \A\Dfun{X_1}\IWith\A\Dfun{X_2}
    \ar[r,"\Derfun(\Swithiso_{X_1,X_2})"]
    &[2em]
    \Dfun{\Withp{X_1}{X_2}} \\ 
    {X_1}\IWith{\A \Dfun X_2}
    \ar[r,"\With{\Derfun(\Sin0)}{\A\Dfun{X_2}}"]
    &[2em]
    \A\Dfun{X_1}\IWith\A\Dfun{X_2}
    \ar[r,"\Derfun(\Swithiso_{X_1,X_2})"]
    &[2em]
    \Dfun{\Withp{X_1}{X_2}}
  \end{tikzcd}.
\end{equation*}
which are respectively equal to $\Derfun(\Wstren{1}_{X_1, X_2})$ and
$\Derfun(\Wstren{2}_{X_1, X_2})$ by functoriality of $\Derfun$ and 
by definition of the product in $\Kl \cL$.
Finally, \cref{eq:com-strength-diff-kleisli} is the image by
  \(\Derfun\) of the diagram
  \begin{equation*}
    \begin{tikzcd}
      &[-2em]
      \With{\Sfun X_1}{\Sfun X_2}
      \ar[ld,swap,"\Wstren1_{X_1,\Sfun X_2}"]
      \ar[rd,"\Wstren2_{\Sfun X_1,X_2}"]
      &[-2em]\\
      \Sfun\Withp{X_1}{\Sfun X_2}
      \ar[d,swap,"\Sfun\Wstren2_{X_1,X_2}"]
      &
      &
      \Sfun\Withp{\Sfun X_1}{X_2}
      \ar[d,"\Sfun\Wstren1_{X_1,X_2}"]\\
      \Sfun^2\Withp{X_1}{X_2}
      \ar[rr,"\Sflip_{\With{X_1}{X_2}}"']
      &&
      \Sfun^2\Withp{X_1}{X_2}
    \end{tikzcd}    
  \end{equation*}
  whose commutation is easily proven using the joint monicity of %
  \((\Sproj i\Compl\Sproj j)_{(i,j)\in\Mdeg}\).
\end{proof}

We can use these strengths to define partial derivatives, as follows.
\begin{definition}
  \label{def:kleisli-partial-derivatives}
  The two partial derivatives of \(f\in\Kl\cL(\With{X_1}{X_2},Y)\) %
  are %
  \begin{align*}
    \Dfun_1 f
    &=\Dfun f
      \Comp\Dfunstr1_{X_1,X_2}\in\Kl\cL(\Dfun{X_1}\IWith X_2,\Dfun Y)\\
    \Dfun_2 f
    &=\Dfun f
      \Comp\Dfunstr2_{X_1,X_2}\in\Kl\cL({X_1}\IWith{\Dfun X_2},\Dfun Y)\,.
  \end{align*}
\end{definition}

When $\Mdeg=\Mder2$, the definition of $\Dfun_1 f$ corresponds intuitively to 
the function
\[ \Dfun_1 f ((x, u), y) 
= \Dfun f ((x, y), (u, 0)) = \big(f(x, y), f'(x, y) \cdot (u, 0) \big)
 = \big( f(x, y), \frac{\partial f(x, y)}{\partial x}\cdot u. \big) \]
 which matches with the intuition of partial derivatives 
 given at the beginning of this section.

\begin{proposition} \label{prop:sum-partial-derivative}
  If \(f\in\Kl\cL(\With{X_1}{X_2},Y)\) %
  we have
  \begin{align*}
    \Dfun f=\Dmonm\Comp\Dfun_1{\Dfun_2f}=\Dmonm\Comp\Dfun_2{\Dfun_1f}\,.
  \end{align*}
\end{proposition}
\begin{proof}
  Apply \cref{th:dfun-strength-with}.
\end{proof}

This formula expresses that the global derivative can be written as a
``sum'' (in the sense of the monad structure of \(\Sfun\)) of partial
derivatives.

\subsection{Linear and multilinear morphisms}

\label{sec:multilinearity}

There are three concurrent notions of ``linearity'' that coexist in the 
Kleisli category $\Kl \cL$. \begin{enumerate}
\item Additivity: $f \in \Kl \cL(X, Y)$ is additive if for all
  summable family $(g_i \in \Kl \cL(U, X))_{i\in I}$, we have
  $f \Comp (\sum_{i \in I} g_i) \sumsub \sum_{i \in I} (f \Comp g_i)$.
\item Linearity with respect to the $\Mdeg$-derivative:
  $f \in \Kl \cL(X, Y)$ is $\Mdeg$-linear if it is additive and if
  $\Derfun(\Sproj i) \Comp \Dfun f = f \Comp \Derfun(\Sproj i)$.
  For example, with $\Mdeg = \Set{0, 1}$ with \(1+1\) undefined, this
  expresses as expected that $f'(x) \cdot u = f(u)$.
\item Linearity with respect to the \(\LL\) structure:
  $f \in \Kl \cL(X, Y)$ is linear if $f = \Derfun h$ for some
  $h \in \cL(X, Y)$.
\end{enumerate}
Observe that linearity implies additivity, by \Pmonoid{}-enrichment
of $\cL$. Furthermore, we know by \ref{ax:dchain} that
$\Dfun (\Derfun h) = \Derfun (\Sfun h)$, and thus linearity implies
$\Mdeg$-linearity. So we have the following hierarchy.
\[ \text{linearity } \implies \text{ $\Mdeg$-linearity } \implies \text{ additivity} \]
The converse implications have no reason to be true in
general~\cite{BluteCockettSeely09}, even though they hold in most
models.
The goal of this section is to prove a ``coordinate-wise'' version of
this result.
We stick to the binary case of morphisms
$f \in \Kl \cL(X_1 \IWith X_2, Y)$ for the sake of conciseness, but a
generalization to $n$-ary morphisms is completely feasible.

\begin{definition} \label{def:additivity-coordinate} A morphism
  $f \in \Kl \cL(X_1 \IWith X_2, Y)$ is additive in its first
  argument if for all family of morphisms
  $(g_i \in \Kl \cL(U, X_1))_{i\in I}$,
  $f \Comp (\sum_{i \in I} g_i \IWith \Id_{X_2}) \sumsub \sum_{i \in
    I} f \Comp (g_i \IWith \Id_{X_2})$.
  We define additivity in the second argument similarly.
\end{definition}

\begin{definition} \label{def:D-linearity-coordinate} %
  A morphism $f \in \Kl \cL(X_1 \IWith X_2, Y)$ is $\Mdeg$-linear in
  its first argument if it is additive in that argument and if its
  partial derivative
  $\Dfun_1 f \in \Kl \cL(\Dfun X_1 \IWith X_2, \Dfun(X_1 \IWith X_2))$
  satisfies
  $\Derfun(\Sproj i) \Comp \Dfun_1 f = f \Comp (\Derfun(\Sproj i)
  \IWith \Idfun_{X_2})$
\end{definition}

For example, with $\Mdeg = \Set{0, 1}$ with $1+1$ undefined, this
expresses as expected that
$\frac{\partial f(x_1, x_2)}{\partial x_1} \cdot u = f(u, x_2)$.
When $\Mdeg = \Set{0,1}$ with \(1+1=1\) this expresses that $f$ is
additive in its first coordinate and that
$f(x_1 + u, x_2) - f(x_1, x_2) = f(u, x_2)$.
Observe that the latter equation is in fact an immediate consequence
of the left additivity of $f$, so in that example $\Mdeg$-linearity in
an argument coincides with additivity in that argument.

 \begin{definition} \label{def:linearity-coordinate}
  A morphism \(f\in\Kl\cL({X_1 \IWith X_2},Y)\) is
  linear in its first argument if there exists 
  \(h_1\in\cL({X_1\ITens \Oc X_2},Y)\) such that
  \(f\) coincides with the following composition of morphisms
  \begin{equation*}
    \begin{tikzcd}
      \Ap\Oc{X_1\IWith X_2}
      \ar[r,"\Inv{(\Seely^2)}"]
      &[1em]
      {\A\Oc{X_1}} \ITens{\A\Oc{X_2}}
      \ar[r,"\Der{X_1} \ITens \Id_{X_2}"]
      &[3em]
      X_1 \ITens \Oc X_2
      \ar[r,"h_1"]
      &[-0.6em]
      Y
    \end{tikzcd}
  \end{equation*}
  and in that case we write \(f=\Derfunm 1{h_1}\).
  We define linearity in the second argument similarly.
\end{definition}

Linearity in an argument immediately implies additivity in that
argument, by \Pmonoid{}-enrichment of $\cL$.
We now prove the promised result.

\begin{theorem} \label{thm:linearity-implies-d-linearity} %
  Every morphism that is linear in an argument is also $\Mdeg$-linear
  in that argument.
\end{theorem}

\begin{proof}
  The proof is done in~\cite[Sec.~6.4]{EhrhardWalch25} for
  $\Mdeg = (\Nat, +)$ and immediately carries to our setting.
  First, we can prove (see \cite{Ehrhard23b}) that the axiom
  \ref{ax:sdwith} which expresses compatibility between the monoidal
  structures of $\Sfun$ (with respect to $\ITens$ and $\IWith$), the
  Seely isomorphisms and the distributive law can be written
  equivalently as the following diagram on their associated strengths,
\begin{equation} \label{eq:sdwith-strength}
\begin{tikzcd}[column sep = large]
	{\Oc \Sfun X_1 \ITens \Oc X_2} & {\Sfun \Oc X_1 \ITens \Oc X_2} & {\Sfun(\Oc X_1 \ITens \Oc X_2)} \\
	{\Oc(\Sfun X_1 \IWith X_2)} & {\Oc \Sfun (X_1 \IWith X_2)} & {\Sfun \Oc (X_1 \IWith X_2)}
	\arrow["{\Sdiff_{X_1} \ITens \Id_{X_2}}", from=1-1, to=1-2]
	\arrow["{\Seelyt_{\Sfun X_1, X_2}}"', from=1-1, to=2-1]
	\arrow["{\Sfunstr^1_{\Oc X_1, \Oc X_2}}", from=1-2, to=1-3]
	\arrow["{\Sfun(\Seelyt_{X_1, X_2})}", from=1-3, to=2-3]
	\arrow["{\Oc \Wstren{1}_{X_1, X_2}}"', from=2-1, to=2-2]
	\arrow["{\Sdiff_{X_1 \IWith X_2}}"', from=2-2, to=2-3]
\end{tikzcd} 
\end{equation}
as well as a similar diagram for the other strengths $\Sfunstr^2$ and $\Wstren{2}$.
Now, let $f \in \cL(\Oc(X_1 \IWith X_2), Y)$ be linear in its first
argument, that is \(f=\Derfunm 1{h_1}\) for some
$h_1 \in \cL(X_1 \ITens \Oc X_2, Y)$.
We need to prove that
$\Derfun(\Sproj i) \Comp \Dfun_1 f = f \Comp (\Derfun(\Sproj i) \IWith
\Idfun_{X_2})$.
Observe that, by definition of $\Dfun_1 f$, of $\Derfun$ and by the triangle
identity of comonads, these compositions in the Kleisli category
$\Kl \cL$ boil down to the following compositions in $\cL$
\begin{align*}
  \Derfun(\Sproj i) \Comp \Dfun_1 f
  &= \Derfun(\Sproj i) \Comp \Dfun f \Comp \Derfun(\Wstren1_{X_1, X_2}) 
  = \Sproj i \Compl \Dfun f \Compl \Oc \Wstren1_{X_1, X_2}\\
  f \Comp (\Derfun \Sproj i \IWith \Id_{X_2})
  &= f \Compl \Oc (\Sproj i \IWith \Id_{X_2})
\end{align*}
and we can now check that these two morphisms are equal, by the following 
diagram chase where we unfold the equation $f = \Derfun_1 h_1$.
\[ 
\begin{tikzcd}
	{\Oc(\Sfun X_1 \IWith X_2)} & {\Oc(\Sfun X_1 \IWith X_2)} &[-0.5em] &[-0.5em] {\Oc \Sfun (X_1 \IWith X_2)} & {\Sfun \Oc(X_1 \IWith X_2)} \\
	& {\Oc \Sfun X_1 \ITens \Oc X_2} && {\Sfun \Oc X_1 \ITens \Oc X_2} & {\Sfun(\Oc X_1 \ITens \Oc X_2)} \\
	\\
	&& {\Sfun X_1 \ITens \Oc X_2} & {\Oc X_1 \ITens \Oc X_2} & {\Sfun(X_1 \ITens \Oc X_2)} \\
	&&&& {\Sfun Y} \\
	{\Oc(X_1 \IWith X_2)} & {\Oc X_1 \ITens \Oc X_2} && {X_1 \ITens \Oc X_2} & Y
	\arrow["\Id", from=1-1, to=1-2]
	\arrow["{\Oc(\Sproj i \IWith \Id_{X_2})}", from=1-1, to=6-1]
	\arrow[""{name=0, anchor=center, inner sep=0}, "{\Oc \Wstren{1}_{X_1, X_2}}", from=1-2, to=1-4]
	\arrow["{(\Seelyt_{\Sfun X_1, X_2})^{-1}}"', from=1-2, to=2-2]
	\arrow["{\Sdiff_{X_1 \IWith X_2}}", from=1-4, to=1-5]
	\arrow["{\Sfun(\Seelyt_{X_1, X_2})^{-1}}", from=1-5, to=2-5]
	\arrow[""{name=1, anchor=center, inner sep=0}, "{\Sdiff_{X_1} \ITens \Id_{\Oc X_2}}"', from=2-2, to=2-4]
	\arrow["{\Der{\Sfun X_1} \ITens \Oc X_2}"{description}, from=2-2, to=4-3]
	\arrow["{\Oc \Sproj i \ITens \Id_{\Oc X_2}}"', from=2-2, to=6-2]
	\arrow["{\Sfunstr^1_{\Oc X_1, X_2}}"', from=2-4, to=2-5]
	\arrow["{\quad \Sfun \Der{X_1} \ITens \Id_{\Oc X_2}}"{description}, from=2-4, to=4-3]
	\arrow["{\Sproj i \ITens \Id_{\Oc X_2}}", from=2-4, to=4-4]
	\arrow["{\Sproj i}", from=2-5, to=4-4]
	\arrow["{\Sfun(\Der{X_1} \ITens \Id_{X_2})}", from=2-5, to=4-5]
	\arrow["{\Sproj i \ITens \Id_{\Oc X_2}}"', from=4-3, to=6-4]
	\arrow["{\Der{X_1} \ITens X_2}"{description}, from=4-4, to=6-4]
	\arrow["{\Sfun h_1}", from=4-5, to=5-5]
	\arrow["{\Sproj i}", from=4-5, to=6-4]
	\arrow["{\Sproj i}", from=5-5, to=6-5]
	\arrow["{(\Seelyt_{X_1, X_2})^{-1}}"', from=6-1, to=6-2]
	\arrow["{\Der{X_1} \ITens \Id_{\Oc X_2}}"', from=6-2, to=6-4]
	\arrow["{h_1}"', from=6-4, to=6-5]
	\arrow["{\text{\ref{eq:sdwith-strength}}}"{description}, draw=none, from=0, to=2-4]
	\arrow["{\text{\ref{ax:sdchain}}}"{description}, draw=none, from=1, to=4-3]
\end{tikzcd}\]

\end{proof}

\subsection{The closed case: function spaces}
\label{sec:closed-differential}
So far we have not considered function space constructions, but the
reader acquainted with the denotational semantics of \LL{} probably
knows that both SMCs \(\COH\) and \(\PCOH\) are symmetric monoidal
closed categories, and that the associated Kleisli categories
\(\Kl\COH\) and \(\Kl\PCOH\) are cartesian closed.
More generally, when \(\cL\) is a resource category which is closed
(as an SMC), we know that \(\Kl\cL\) is a CCC%
\footnote{This is the categorical counterpart of the Girard's
  translation of intuitionistic logic in \LL.}.
We use the notation \(\Simpl XY\) for the object of morphisms from
\(X\) to \(Y\) in \(\Kl\cL\), which is \(\Limpl{\Oc X}{Y}\) and whose
associated evaluation morphism %
\(\Ev\in\Kl\cL(\With{(\Simpl XY)}{X},Y)\) is defined as the following
composition of morphisms in \(\cL\).
\begin{equation*}
  \begin{tikzcd}
    \Oc{\Withp{(\Simpl XY)}{X}}
    \ar[r,"\Invp{\Seelyt}"]
    &[1em]
    \Tens{\Oc\Limplp{\Oc X}{Y}}{\Oc X}
    \ar[r,"\Der{\Limpl{\Oc X}Y} \ITens \Id_{\Oc Y}"]
    &[3em]
    \Tens{\Limplp{\Oc X}{Y}}{\Oc X}
    \ar[r,"\Evlin"]
    &[-1em]
    Y
  \end{tikzcd}
\end{equation*}

Given \(f\in\Kl\cL(\With ZX,Y)=\cL(\Oc{\Withp ZX},Y)\), we have %
\(f\Compl\Seelyt\in\cL(\Tens{\Oc Z}{\Oc X},Y)\) and hence %
\(\Ap\Curlin{f\Compl\Seelyt}\in\Kl\cL(Z,\Simpl XY)\) and this morphism
\(\A\Cur f=\Ap\Curlin{f\Compl\Seelyt}\) is uniquely characterized by
the equation
\begin{equation} \label{eq:ev-cur}
  \Ev\Comp\Withp{\A\Cur f}{X}=f\,.
\end{equation}
By standard CCC constructions we get for all object $X$ a functor
$\Simpl{X}{\_} : \Kl \cL \arrow \Kl \cL$ that maps an object $Y$ to
$\Simpl X Y$ and a morphism $f \in \Kl \cL(Y, Z)$ to the morphism
$\Simplp X f = \Cur(f \Comp \Ev) \in \Kl \cL(\Simpl X Y, \Simpl X Z)$.
We can then check by unfolding the definitions on this equation that
the functor $\Simpl X {\_}$ is an extension to $\Kl \cL$
(\cref{def:extension}) of the functor $\Limpl{\Oc X}{\_}$ on $\cL$. In
other word, for all $h \in \cL(Y, Z)$,
\begin{equation} \label{eq:arrow-extension}
  \Simplp X {\Derfun(h)} = \Derfun(\Limpl {\Oc X} h)\,. 
\end{equation}

We now explain how this CCC structure interacts with our
$\Mdeg$-differentiation.
The commutative monad structure of $\Dfun$ given
\cref{th:dfun-strength-with} induces a canonical natural
transformation
\[
  \Dfunstrf_{X, Y}
  = \Cur(\Dfun \Ev \Comp \Dfunstr{1}_{\Simpl X Y, X})
  = \Cur(\Dfun_1 \Ev) \in \Kl \cL(\Dfun (\Simpl X Y), \Simpl{X}{\Dfun Y})\,.
\]
This natural transformation is the $\Dfun$ counterpart of the $\Sfun$
pointwise structure
$\Sfunstri_{X, Y} \in \cL(\Sfun(\Limpl X Y), \Limpl{X}{\Sfun Y})$ that
we introduced in \cref{eq:Sfunstri}.
In fact, the two are related as follows.

\begin{proposition}
$\Dfunstrf$ is a natural isomorphism and 
$\Dfunstrf_{X, Y} = \Derfun(\Sfunstri_{\Oc X, Y})$.
\end{proposition}

\begin{proof}
  Observe that by its very definition, the morphism
  $\Ev \in \Kl \cL(\Simpl X Y \IWith X, Y)$ is linear in its first
  argument, and thus is $\Mdeg$-linear in that argument by 
  \cref{thm:linearity-implies-d-linearity}.
  It follows that
  \[
    \Dfunstrf_{X, Y}
    = \Cur(\Dfun_1 \Ev)
    = \Cur \Stuple{\Ev \Comp
      (\Derfun \Sproj i \IWith \Id_{Y})}_{i \in \Mdeg}\,.
  \]
  Now we can check with a computation similar to that of
  \cref{prop:ss-pointwise-characterization,prop:ss-closure} that this
  morphism is an iso with inverse
  $\Stuple{\Simpl X {\Derfun(\Sproj i)}}_{i \in \Mdeg}$.
  Thus by \cref{eq:arrow-extension},
  \[
    (\Dfunstrf_{X, Y})^{-1}
    = \Stuple{\Simpl X {\Derfun(\Sproj i)}}_{i \in \Mdeg}
    = \Stuple{\Derfun(\Limpl{\Oc X}{\Sproj i})}_{i \in \Mdeg}
    = \Derfun \Stuple{\Limpl{\Oc X}{\Sproj i}}_{i \in \Mdeg}
    = \Derfun(\Sfunstri_{\Oc X, Y})^{-1}
  \]
  and thus $\Dfunstrf_{X, Y} = \Derfun(\Sfunstri_{\Oc X, Y})$.
\end{proof}

To summarize, we proved in \cref{th:dfun-strength-with} that $\Dfun$
is strong monoidal with respect to the cartesian product $\IWith$ in
$\Kl \cL$, with monoidality $\Derfun(\Swithiso_{Y_1, Y_2})$.
This strong monoidal structure induces a strength $\Dfunstr1$, and we
proved above that the pointwise structure $\Dfunstrf$ associated with
this strength is an iso.
This implies the following result.

\begin{theorem} \label{thm:kleisli-dfun-ccc}
  The Kleisli category of the monad $\Dfun$ is a cartesian closed category.
\end{theorem}

\begin{proof}
  The product of two objects $X_1$ and $X_2$ is simply given by the 
  product $X_1 \IWith X_2$ in $\Kl \cL$.
  The pairing of two morphisms $f \in \KLDFUN(X, Y_1)$ and 
  $g \in \KLDFUN(X, Y_2)$ is given by the following composition 
  in $\Kl \cL$
  \[ 
    \begin{tikzcd}
      X & {\Dfun Y_1 \IWith \Dfun Y_2} &[3em] {\Dfun(Y_1 \IWith Y_2).}
      \arrow["{\Pair{f}{g}}", from=1-1, to=1-2]
      \arrow["{\Derfun(\Swithiso_{Y_1, Y_2})}", from=1-2, to=1-3]
    \end{tikzcd} \]
  The internal hom of two objects $X$ and $Y$ is given by 
  the object $\Simpl{X}{\Dfun Y}$. The currying of a morphism
  $h \in \KLDFUN(U \IWith X, Y)$ is given by the following 
  composition in $\Kl \cL$ 
  \[ 
    \begin{tikzcd}
      U &[1em] {\Simpl{X}{\Dfun Y}} &[1em] {\Dfun(\Simpl X Y)}
      \arrow["{\Cur(h)}", from=1-1, to=1-2]
      \arrow["{(\Dfunstrf)^{-1}}", from=1-2, to=1-3]
    \end{tikzcd} 
  \]
  and the evaluation $\Ev \in \KLDFUN((\Simpl X {\Dfun Y}) \IWith X, Y)$
  is simply the evaluation of $\Kl \cL$.
  The fact that these constructions give a product and an internal hom
  can be checked by hand.
  Alternatively, this is the Kleisli category counterpart of 
  \emph{doctrinal adjunctions}~\cite{Kelly74-doctrinal}, a well
  known result about lifting of adjunctions to Eilenberg-Moore categories
  of monads.
  We refer the reader to \cite[Sec.~4.4]{Walch25-thesis} for a
  detailed explanation.
\end{proof}

This is a striking result since the Kleisli category of a monad is not
usually cartesian closed.
We see it as another piece of evidence pointing towards a
\Cohdifflambda{} whose semantics is given by the CCC associated with a
\Lingdiffcat{}.

For now, we focused on the currying $\Dfunstrf = \Cur(\Dfun_1 \Ev)$
of the first partial derivative of the evaluation. 
The currying of the second partial derivative is also quite
interesting and gives a morphism
\[
  \Dfunint
  =\Ap\Cur{\A\Dfun\Ev\Comp\Dfunstr2_{\Simpl XY,X}} 
  = \Cur(\Dfun_2 \Ev) \in \Kl\cL(\Simpl XY,\Simpl{\Dfun X}{\Dfun Y})
\]
which allows us to internalize the action of \(\Dfun\).

\newcommand\Context{U}
More precisely, let \(f\in\Kl\cL(\With U X,Y)\) so that 
$\Cur(f) \in \Kl \cL(U, \Simpl X Y)$.
There are two ways to define a morphism in
$\Kl \cL(U, \Simpl {\Dfun X} {\Dfun Y})$ that represents the
$\Mdeg$-differential of $\Cur(f)$ with respect to its argument in
$X$, namely,
\begin{align*}
  \Dfunint\Comp\A\Cur f
  \text{\quad and\quad}
  \Ap\Cur{\Dfun_2 f}=\Ap\Cur{\Dfun f\Comp\Dfunstr2_{Z,X}}\,.
\end{align*}
The first one ``curries then derives'' using the internalized action
of $\Dfun$, and the second ``derives then curries'' using the second
partial derivative.
We can easily prove that the two methods are equal, meaning that \(\cL\)
validates a form of ``differential \(\beta\)-reduction''.
\begin{proposition}
  For any \(f\in\Kl\cL(\With{U}X,Y)\) we have %
  \(\Dfunint\Comp\A\Cur f=\Ap\Cur{\Dfun f\Comp\Dfunstr2_{Z,X}}\).
\end{proposition}

Consider now \(f\in\Kl\cL(\With \Context X,\Simpl Y Z)\) and %
\(g\in\Kl\cL(\With \Context X,Y)\).  The morphism \(f\) should be seen as a
function \(Y\to Z\) depending on one parameter in \(\Context\) (a context) 
and a variable in \(X\).
The morphism \(g\) is an argument for $f$, similarly parameterized.
We can apply \(f\) to \(g\), defining %
\(\App fg=\Ev\Comp\Tuple{f,g}\in\Kl\cL(\With \Context X,Z)\) %
and then we can take the derivative of \(\App fg\) wrt.~the second
parameter which is
\begin{align*}
  \A{\Dfun_2}{(\App fg)}=\A\Dfun{(\App fg)}\Comp\Dfunstr2
  \in\Kl\cL(\With \Context{\A\Dfun X},\A\Dfun{Z})\,.
\end{align*}
We now want to express the derivative of this application wrt.~the
derivatives of $f$ and of $g$. %
Intuitively, the chain rules of the differential calculus tells us
(fixing the parameter in $\Context$) that
\[
  \frac{d f(x, g(x))}{d x}
  = \partial_1 f(x, g(x))
  + \partial_2 f(x, g(x)) \cdot \frac{d g(x)}{d x}\,.
\]

This idea corresponds in our formalism to \cref{th:dfun-part-app}
below. 
For all \(f\in\Kl\cL(\With \Context X,\Simpl Y Z)\) we define 
a morphism $\Dfun_{2, 3} f \in \Kl \cL(\With \Context \Dfun X, \Simpl{\Dfun
    Y}{\Dfun Z})$ as the following composition in $\Kl \cL$
  \begin{equation} \label{eq:partial-two-three}
    \begin{tikzcd} {\Context \IWith \Dfun X} &
      {\Dfun (\Simpl Y Z)} & {\Simpl Y {\Dfun Z}} & {\Simpl{\Dfun
          Y}{\Dfun \Dfun Z}} &[2em] {\Simpl{\Dfun X}{\Dfun Y}}
      \arrow["{\Dfun_2 f}", from=1-1, to=1-2] \arrow["{\Dfunstrf_{X,
          Y}}", from=1-2, to=1-3] \arrow["\Dfunint", from=1-3, to=1-4]
      \arrow["{\Simpl{\Dfun Y}{\Dmonm_Y}}", from=1-4, to=1-5]
          \end{tikzcd}
        \end{equation}
   which represents the derivative of
   $f \in \Kl \cL(\Context \IWith X, \Simpl Y Z)$ wrt.~both its
   argument in $X$ and in $Y$.

\begin{theorem} %
  \label{th:dfun-part-app}
  For all $f \in \Kl \cL(\Context \IWith X, \Simpl Y Z)$ and 
  \(g\in\Kl\cL(\With \Context X, Y)\), we have
  \begin{equation*}
    \A{\Dfun_2}{\App fg}
    =\App{\Dfun_{2, 3} f}
    {\A{\Dfun_2}g}\,.
  \end{equation*}
\end{theorem}

\begin{proof}
  The intuition on the chain rule above 
  is captured by \cref{prop:sum-partial-derivative} applied to the 
  evaluation morphism $\Ev$. 
  Concretely, remember that
  $\Dfunstrf_{Y, Z} \in \Kl \cL(\Dfun(\Simpl Y Z), \Simpl{Y}{\Dfun
    Z})$ and
  $\Dfunint_{Y, Z} \in \Kl \cL(\Simpl Y Z, \Simpl{\Dfun Y}{\Dfun Z})$
  are the currying of the partial derivatives of the evaluation:
  $\Dfunstrf_{Y, Z} = \Cur(\Dfun_1 \Ev)$ and
  $\Dfunint_{Y, Z} = \Cur(\Dfun_2 \Ev)$.
  From these observations, we can prove our contention by the following
  diagram chase
  \[ \begin{tikzcd}
	{\Context \IWith \Dfun X} & {} & {} & {\Dfun Z} \\
	\\
	{\Dfun(\Simpl Y Z) \IWith \Dfun Y} & {\Dfun(\Dfun (\Simpl Y Z) \IWith Y)} & {\Dfun \Dfun ((\Simpl Y Z) \IWith Y)} & {\Dfun \Dfun Z} \\
	{(\Simpl{Y}{\Dfun Z}) \IWith \Dfun Y} & {\Dfun ((\Simpl{Y}{\Dfun Z}) \IWith Y)} && {\Dfun \Dfun Z} \\
	{(\Simpl{\Dfun Y}{\Dfun \Dfun Z}) \IWith \Dfun Y} &&& {\Dfun \Dfun Z}
	\arrow["{\Dfun_2 (\App{f}{g})}", from=1-1, to=1-4]
	\arrow["{\Tuple{\Dfun_2 f, \Dfun_2 g}}"', from=1-1, to=3-1]
	\arrow["{\Dfun \Ev}"{description}, from=3-1, to=1-4]
	\arrow["{\Dfunstr{2}_{\Simpl Y Z, Y}}"', from=3-1, to=3-2]
	\arrow["{\Dfunstrf_{Y, Z} \IWith \Dfun Y}"', from=3-1, to=4-1]
	\arrow["{\text{by definition}}", draw=none, from=3-2, to=1-2]
	\arrow["{\Dfun \Dfunstr1_{\Simpl Y Z, Y}}"', from=3-2, to=3-3]
	\arrow["{\Dfun(\Dfunstrf_{Y}{Z} \IWith Y)}", from=3-2, to=4-2]
	\arrow["{\text{\cref{prop:sum-partial-derivative}}}"', draw=none, from=3-3, to=1-3]
	\arrow["{\Dfun \Dfun \Ev}", from=3-3, to=3-4]
	\arrow["{\Dmonm_Z}"', from=3-4, to=1-4]
	\arrow["{\Dfunstr{2}_{\Simpl{Y}{\Dfun Z}, Y}}"', from=4-1, to=4-2]
	\arrow["{\Dfunint_{Y, \Dfun Z} \IWith {\Dfun Y}}"', from=4-1, to=5-1]
	\arrow[""{name=0, anchor=center, inner sep=0}, "{\Dfun \Ev}"', from=4-2, to=4-4]
	\arrow["\Id"', from=4-4, to=3-4]
	\arrow[""{name=1, anchor=center, inner sep=0}, "\Ev"', from=5-1, to=5-4]
	\arrow["\Id"', from=5-4, to=4-4]
	\arrow["{\text{\cref{eq:ev-cur}}}"', draw=none, from=3-3, to=0]
	\arrow["{\text{\cref{eq:ev-cur}}}"{description}, draw=none, from=4-2, to=1]
      \end{tikzcd} \] since the bottom path is equal to
    \[
      \Dmonm_Z \Comp\App{\Dfunint\Comp\Dfunstrf_{Y, Z}
        \Comp\A{\Dfun_2}f}{\A{\Dfun_2}g}
      = \App{(\Simpl{\Dfun Y}{\Dmonm_Z})
        \Comp \Dfunint\Comp\Dfunstrf_{Y, Z}\Comp\A{\Dfun_2}f}{\A{\Dfun_2}g}
      = \App{\Dfun_{2, 3} f}{\A{\Dfun_2}g}
    \]
    by standard properties of CCCs.
\end{proof}

\subsection{Fixpoints}
\label{sec:fixpoints}

We start with stating a few standard results about fixpoints in a CCC
which is enriched in \(\omega\)-cpos.

\subsubsection{Reminder about fixpoint operators in a CCC}
\begin{definition}
  An \(\omega\)-cpo is a poset which has a least element \(0\) and
  where any monotone sequence has a lub.
  If \(D\) and \(D'\) are \(\omega\)-cpos, a function \(f:D\to D'\) is
  Scott-continuous
  (or simply continuous) if \(f\) is monotone and commutes with the
  lubs of monotone sequences.
\end{definition}

\begin{remark}
  In the literature, Scott continuity is usually defined as
  preservation of the lubs of arbitrary directed sets.
  However such lubs do not always exist in the situations we are
  interested in.
  This is specifically the case in categories arising in continuous
  probabilistic settings: in such situations lubs of monotone
  countable sequences can be assumed to exist thanks to the monotone
  convergence theorem though more general uncountable directed lubs do
  not always exist.
  We nevertheless use the term ``Scott continuity'' since the
  fundamental ideas of Dana~Scott are also central in such models.
\end{remark}

\begin{proposition}
  If \(D\) is an \(\omega\)-cpo and \(f:D\to D\) is continuous then
  \(f\) has a least fixpoint, which is \(\sup_{n\in\Nat}f^n(0)\).
\end{proposition}

\begin{definition} %
  \label{def:lambda-category}
  A \emph{\(\lambda\)-category} is a CCC which is enriched in
  \(\omega\)-cpos and in which the pairing operation %
  \(\cC(X,Y_1)\times\cC(X,Y_2)\to\cC(X,\With{Y_1}{Y_2})\) and the
  currying operation %
  \(\cC(\With ZX,Y)\to\cC(Z,\Impl XY)\) %
  are Scott continuous.
\end{definition}

\begin{theorem} %
  \label{th:lambda-cat-fixpoint-operator}
  Let \(\cC\) be a \(\lambda\)-category.
  For any objects \(X\) of $\cC$ there is a morphism %
  \(\cY\in\cC(\Impl XX,X)\) such that, for any morphism %
  \(f\in\cC(Z,\Simpl XX)\) %
  the morphism \(\cY\Comp f\in\cC(Z,X)\) %
  is the least morphism \(g\in\cC(Z,X)\) such that %
  \(\Ev\Comp\Tuple{f,g}=g\).
\end{theorem}
This is a standard result in semantics; we give the proof because we
think it helps understanding \cref{sec:diff-fix}.
\begin{proof}
  By induction we define a family of morphisms
  \((\cY_n\in\cC(\Simpl XX,X))_{n\in\Nat}\) by \(\cY_0=0\) and
  \(\cY_{n+1}=\Ev\Comp\Tuple{\Simpl XX,\cY_n}\).
  By monotonicity and Scott continuity of all the CCC operations in
  \(\cC\), this sequence is increasing in the poset
  \(\cC(\Simpl XX,X)\) and \(\cY=\sup_{n\in\Nat}\cY_n\) is the least
  element of \(\cC(\Simpl XX,X)\) such that
  \(\cY=\Ev\Comp\Tuple{\Simpl XX,\cY}\).
  Therefore
  \begin{align*}
    \cY\Comp f
    &=\sup_{n\in\Nat}(\cY_n\Comp f)\\
    &=\sup_{n\in\Nat}(\cY_{n+1}\Comp f)
      \text{\quad since }\cY_0=0\\
    &=\sup_{n\in\Nat}(\Ev\Comp\Tuple{f,\cY_n\Comp f})\\
    &=\sup_{n\in\Nat}g_n
  \end{align*}
  where \((g_n\in\cC(Z,X))_{n\in\Nat}\) is the (obviously monotone)
  sequence of morphisms inductively defined by \(g_0=0\) and %
  \(g_{n+1}=\Ev\Comp\Tuple{f,g_n}\).
  So \(g=\cY\Comp f\) is the least element of \(\cC(Z,X)\) such that
  \(g=\Ev\Compl\Tuple{f,g}\) by Scott continuity of the map
  \(\Ev\Comp\Tuple{f,\_}\).
\end{proof}

\subsubsection{The \texorpdfstring{\DDIFF{}}{differentiation} of fixpoints}
\label{sec:diff-fix}
\begin{definition}
  Let \(f,g\in\cL(X,Y)\), we write \(f\leq g\) if there is an
  \(h\in\cL(X,Y)\) such that \(f\) and \(h\) are summable and
  \(g=f+h\).
\end{definition}

\begin{lemma}
  The relation \(\leq\) is a preorder relation on \(\cL(X,Y)\).
  The composition operation %
  \(\cL(X,Y)\times\cL(Y,Z)\to\cL(X,Z)\) is monotone wrt.~this preorder
  relation and if \(\cL\) is a symmetric monoidal summable category
  (that is~\ref{ax:stens} holds), then the tensor product of \(\cL\)
  is monotone wrt.~this preorder relation. If \(\cL\) is a cartesian
  summable category (that is~\ref{ax:swith} holds), then the pairing
  operation \(\cL(X,Y_1)\times\cL(X,Y_2)\to\cL(X,\With{Y_1}{Y_2})\) is
  monotone.
\end{lemma}
\begin{proof}
  Monotonicity of composition results from \cref{lemma:summability-compl}.
  The two next statements are obvious consequences of~\ref{ax:stens}
  and~\ref{ax:swith} respectively.
\end{proof}

\begin{definition}
  We say that \(\cL\) \emph{is Scott} if the following conditions are
  satisfied:
  \begin{itemize}
  \item in all homsets, the relation \(\leq\) is an order relation;
  \item for all objects \(X\) and \(Y\), each monotone sequence
    \((f_n)_{n\in\Nat}\) of elements of \(\cL(X,Y)\) has a least upper
    bound \(\sup_{n\in\Nat}f_n\in\cL(X,Y)\);
  \item the composition operation %
    \(\cL(X,Y)\times\cL(Y,Z)\to\cL(X,Z)\) is Scott-continuous in the
    sense that it commutes with the lubs of monotone sequences (its
    domain being equipped with the product order relation);
  \item \(\ITens\) commutes with the lubs of monotone sequences;
  \item \(\Oc\_\) induces an increasing and Scott-continuous map on
    morphisms.
  \end{itemize}
\end{definition}

\begin{lemma}
  \label{lemma:curlin-continuous}
  If \(\cL\) is Scott then
  \(\Curlin:\cL(\Tens ZX,Y)\to\cL(Z,\Limpl XY)\) is continuous.
\end{lemma}
\begin{proof}
  Monotonicity results from \cref{rem:cur-additive}.
  Continuity results from the fact that the inverse of the map %
  \(\Curlin:\cL(\Tens ZX,Y)\to\cL(Z,\Limpl XY)\) is the function %
  \(f\mapsto\Evlin\Compl\Tensp fX\) which is continuous by our
  assumptions about \(\cL\).
\end{proof}

\begin{proposition}
  If \(\cL\) is Scott then the pairing operation %
  \(\cL(X,Y_1)\times\cL(X,Y_2)\to\cL(X,\With{Y_1}{Y_2})\) is
  continuous.
\end{proposition}
\begin{proof}
  By the universal property of the cartesian product.
\end{proof}

\begin{theorem}
  If \(\cL\) is Scott then the cartesian closed category \(\Kl\cL\) is
  a \(\lambda\)-category (in the sense of \cref{def:lambda-category}).
\end{theorem}
The proof is straightforward, using \cref{lemma:curlin-continuous}.

\begin{Example} \label{ex:fixpoint-models}
  Saying that \(s,t\in\Kl\COH(E,F)\) satisfy \(s\leq t\) simply means
  \(s\subseteq t\) and the fact that \(\COH\) is Scott comes from
  the fact that the set of cliques of a coherence space is closed
  under directed unions.
  
  The situation is completely similar in \(\PCOH\).
  Given \(x,y\in\Pcoh X\), we have \(x\leq y\) iff \(x_a\leq y_a\) for
  all \(a\in\Web X\) as easily checked.
  Therefore a monotone sequence \((x(n)\in\Pcoh X)_{n\in\Nat}\) has a
  lub in \(\Pcoh X\), namely
  \(x=(\sup_{n\in\Nat}x(n)_a)_{a\in\Web X}\in\Pcoh X\).
  It follows easily that \(\PCOH\) is Scott.
  As a consequence, for any probabilistic coherence space \(X\), we
  have %
  \(\cY\in\Kl\PCOH(\Simpl XX,X)\) which maps \(t\in\Kl\PCOH(X,X)\) to
  its least fixpoint \(\sup_{n\in\Nat}\Fun t^n(0)\).
  The fact that this fixpoint operator is itself an analytic morphism
  is a remarkable property of this semantics, and is deeply related to
  the fact that, in this semantics, the morphisms are matrices with
  nonnegative coefficients.

  Without this restriction, we could accept to have
  \(w\in\Kl\PCOH(\With\Sone\Sone,\Sone)\) with associated function
  \(\Pcoh\Sone\times\Pcoh\Sone\to\Pcoh\Sone\) given by
  \(\Fun w(u,v)=u+v-uv\), which is increasing and Scott continuous.
  Explicitly, for \(m\in\Web{\With\Sone\Sone}\), we have
  \(w_{m,\Sonelem}=1\) if
  \(m\in\Set{\Mset{(1,\Sonelem)},\Mset{(2,\Sonelem)}}\),
  \(w_{m,\Sonelem}=-1\) if \(m=\Mset{(1,\Sonelem),(2,\Sonelem)}\) and
  \(w_{m,\Sonelem}=0\) otherwise.
  If the fixpoint operator were accepted by this semantics, we would
  be able to define \(t\in\Kl\PCOH(\Sone,\Sone)\) such that
  \(\Fun t(u)=\Fun w(u,\Fun t(u))=u+\Fun t(u)-u\Fun t(u)\), that is %
  \(u(1-\Fun t(u))=0\).
  So we must have \(\Fun t(u)=1\) if \(u\in\Interoc 01\), and %
  \(\Fun t(0)=0\) because \(\Fun t(0)\) should be the least fixpoint
  of the function \(\Fun w(0,\_)\).
  So the function \(\Fun t:\Intercc01\to\Intercc01\) is not
  continuous, and \emph{a fortiori} cannot be described as a
  power series (even with possibly negative coefficients).

  Such a \(w\) is available in the Köthe sequence spaces model of
  \(\LL\)~\cite{Ehrhard02} whose associated Kleisli CCC has
  analytic functions as morphisms, but where general fixpoints are of
  course not available.
\end{Example}

Assume that \(\cL\) is Scott and closed (as an SMC).
Given \(f\in\Kl\cL(\With ZX,\Simpl YY)\), we can define %
\(\Ap\Sfix f\in\Kl\cL(\With ZX,Y)\) as %
\(\Ap\Sfix f=\cY\Comp f\) where %
\(\cY\in\Kl\cL(\Simpl YY,Y)\) comes from
\cref{th:lambda-cat-fixpoint-operator}.
Notice that 
\begin{equation*}
  \Ap\Sfix f=\App f{\Ap\Sfix f}\,.
\end{equation*}
Remember also that the family %
\((\Ap{\Sfix_n}f\in\Kl\cL(\With ZX,Y))_{n\in\Nat}\) %
defined inductively by \(\Ap{\Sfix_0}f=0\) and
\(\Ap{\Sfix_{n+1}}f=\App f{\Ap{\Sfix_n}f}\) %
is monotone and that
\begin{equation*}
  \Ap\Sfix f=\sup_{n\in\Nat}\Ap{\Sfix_n}f\,.
\end{equation*}

Remember that for each
\(f\in\Kl\cL(\With \Context X,\Simpl Y Y)\), the morphism
$\Dfun_{2, 3} f \in \Kl \cL(\With \Context \Dfun X, \Simpl{\Dfun
  Y}{\Dfun Y})$ defined in \cref{eq:partial-two-three} corresponds to
the partial derivative of $f$ with respect to its arguments in $X$ and
$Y$.
In particular, both $\Ap{\Dfun_2}{\Ap\Sfix f}$ (the differential of the 
fixpoint) and $\Ap{\Sfix}{\Dfun_{2, 3} f}$ (the fixpoint of the 
differential) belong to $\cL(\With \Context \Dfun X, \Dfun Y)$.

\begin{theorem}
  For each \(f\in\Kl\cL(\With ZX,\Simpl YY)\),
  \begin{equation*}
    \Ap{\Dfun_2}{\Ap\Sfix f}
    = \Ap{\Sfix}{\Dfun_{2, 3} f}.
  \end{equation*}
\end{theorem}
\begin{proof}
  It suffices to prove that for all $n \in \Nat$, %
  $\Ap{\Dfun_2}{\Ap{\Sfix_n}f}
    = \Ap{\Sfix_n}{\Dfun_{2, 3} f}$ and conclude by Scott continuity of 
  $\Dfun_2$. We proceed by induction.
  \
  The base case is obvious.
  Next we have
  \begin{align*}
    \Ap{\Dfun_2}{\Ap{\Sfix_{n+1}}f}
    &=\Ap{\Dfun_2}{\App{f}{\Ap{\Sfix_n}{f}}}\\
    &=\App{\Dfun_{2, 3} f}{\Ap{\Dfun_2}{\Ap{\Sfix_n}{f}}}
      \tag*{by \cref{th:dfun-part-app}}\\
    &=\App{\Dfun_{2, 3} f}
      {\Ap{\Sfix_n}{\Dfun_{2, 3} f}} \tag*{by inductive hypothesis} \\
    &=\Ap{\Sfix_{n+1}}{\Dfun_{2, 3} f} \tag*{by definition. \qedhere}
  \end{align*}

\end{proof}

\begin{remark}
  So the differential of a fixpoint can itself be written as a
  fixpoint of a differential, meaning that we can combine our coherent
  differential calculus with general fixpoints, in sharp contrast with
  the differential \(\lambda\)-calculus and \(\LL\).
\end{remark}

\section*{Conclusion}
Motivated by a striking similarity between coherent differentiation
\cite{Ehrhard23a} and coherent Taylor expansion \cite{EhrhardWalch25},
we have developed a general theory of \DDIFF{}, parameterized by the
choice of a partial commutative monoid \(\Mdeg\) satisfying a
\Mpositivity{} requirement.
The resulting theory subsumes of course both coherent differentiation
and coherent Taylor expansion, which correspond to
\(\Mdeg=\Eset{0,1}\) with \(1+1\) undefined and
\(\Mdeg=(\Nat,\mathord+)\) respectively.
But many other choices are possible for \(\Mdeg\) leading to other
``differential theories'' including various calculi of finite
differences and truncated Taylor expansion, or Taylor expansion with a
remainder as well as many unknown theories that we plan to explore in
further work.

One of the authors is currently developing a
syntactic counterpart of \DDIFF{} in the form of a functional
programming language, fortunately much simpler (in terms of the number
of syntactic constructs and reduction rules) than the one
of~\cite{Ehrhard23b}.
A report on this work will be available soon.
%


{
\small
\bibliographystyle{alpha}
\bibliography{newbiblio.bib}}

\end{document}